\documentclass[twocolumn,superscriptaddress,aps,showpacs,amsmath,amstex,amssymb,citeautoscript,floatfix,pra,preprintnumbers]{revtex4-2}
\pdfoutput=1
\usepackage[english]{babel}
 
\usepackage{letltxmacro}
\usepackage{latexsym}
\usepackage{booktabs}
\usepackage{siunitx}
\usepackage{multirow}
\usepackage{algpseudocode}
\usepackage{bbm}
 
\LetLtxMacro{\ORIGselectlanguage}{\selectlanguage}
\makeatletter
\DeclareRobustCommand{\selectlanguage}[1]{%
  \@ifundefined{alias@\string#1}
    {\ORIGselectlanguage{#1}}
    {\begingroup\edef\x{\endgroup
       \noexpand\ORIGselectlanguage{\@nameuse{alias@#1}}}\x}%
}
\newcommand{\definelanguagealias}[2]{%
  \@namedef{alias@#1}{#2}%
}
\makeatother
 
\definelanguagealias{en}{english}
\definelanguagealias{English}{english}
\usepackage{graphicx}
\usepackage{amsmath}
\usepackage{amsfonts}
\usepackage{amssymb}
\usepackage{amsthm}
\usepackage{cancel}
\usepackage{mathtools}
\usepackage{bm}
\usepackage{color}
\usepackage[percent]{overpic}
\usepackage{soul}
\usepackage{wasysym}
\usepackage{dsfont}
\usepackage{float}
\usepackage{physics}
\usepackage{hyperref}
\usepackage{comment}
\usepackage{enumitem}
\usepackage{textgreek}
\usepackage{tcolorbox}
\usepackage[dvipsnames]{xcolor}
 
\newcounter{alg}
 
\hypersetup{
  colorlinks   = true,
  urlcolor     = blue,
  linkcolor    = blue,
  citecolor    = red
}
\usepackage[mathscr]{euscript}
 
\usepackage{varwidth}
\usepackage[lined,boxed,ruled,norelsize,linesnumbered]{algorithm2e}
\usepackage{verbatim}
\usepackage[normalem]{ulem}
\usepackage{cleveref}
\usepackage{times}

\newtheorem{theorem}{Theorem}

\newtheorem{lemma}{Lemma}

\newtheorem{proposition}{Proposition}
\newtheorem{remark}{Remark}
\newtheorem{assumption}{Assumption}
\newtheorem{definition}{Definition}
\newtheorem{fact}{Fact}
\newtheorem*{lemma*}{Lemma}
\newtheorem{corollary}{Corollary}
\newtheorem*{theorem*}{Theorem}

\theoremstyle{plain}

\theoremstyle{plain}

\theoremstyle{plain}
\newtheorem*{lem*}{\protect\lemmaname}
\theoremstyle{plain}
\newtheorem*{thm*}{\protect\theoremname}
\theoremstyle{plain}

\theoremstyle{plain}

\renewcommand{\thealg}{\arabic{alg}}
\newtcolorbox[use counter=alg,
              crefname={algorithm}{algorithms},
              Crefname={Algorithm}{Algorithms}]
{alg}[2][]{%
  floatplacement=#1,
  float,
  colback=cyan!5!white,
  colframe=cyan!50!black,
  colbacktitle=cyan!85!black,
  fonttitle=\bfseries,
  title=Algorithm~\thealg: #2
}
 
\newcommand{\eqs}[1]{\begin{equation}\begin{split}#1\end{split}\end{equation}}

\newcommand{\appref}[1]{Appendix\,\ref{#1}}

\begin{document}
 
\title{
Provably Efficient Self-Calibrating Quantum Fault Tolerance
}
 
\author{Weiyuan Gong}
\affiliation{School of Engineering and Applied Sciences, Harvard University, Allston, MA 02134, USA}

\author{Hong-Ye Hu}
\email[]{hongyehu.physics@gmail.com}
\affiliation{Department of Physics, Harvard University, Cambridge, MA 02138, USA}

\begin{abstract}
Quantum error correction protects logical information only when every physical operation remains below the fault-tolerance threshold, a condition that must be maintained continuously rather than only at the initial calibration. In practice, however, analog control parameters inevitably drift because of environmental fluctuations. As future fault-tolerant quantum computations are expected to run for days or even months, interrupting computation for repeated recalibration becomes fundamentally impractical. A promising alternative is to integrate calibration directly into computation by repurposing syndrome measurements as a calibration signal \cite{sivak2025reinforcement}, but whether such self-calibration can be achieved with provable efficiency remains an open question. Here we establish a theoretical framework for such self-calibrating quantum fault tolerance. We prove that, for a broad class of control-induced errors, the detection rate defines a locally strongly convex surrogate objective for analog calibration with high probability. This geometric property enables a simple and efficient online optimization algorithm using only syndrome measurements collected during normal error correction. We prove convergence to an $\varepsilon$ detection rate within $O(1/\varepsilon^2)$ epochs for time-independent drifts and also establish guarantees for time-dependent drifts. We further show that the convergence rate is independent of the code distance for quantum low-density parity-check (LDPC) codes. Pulse-level simulations of neutral-atom arrays and large-scale circuit-level Clifford simulations confirm these theoretical predictions. Our results establish self-calibrating fault tolerance as a provably efficient paradigm in which the same syndrome measurements simultaneously protect logical information and stabilize the underlying hardware.
\end{abstract}
\maketitle
\raggedbottom

\let\realaddcontentsline\addcontentsline
\renewcommand{\addcontentsline}[3]{}
 
 
\section{Introduction}
\label{sec:introduction}

Quantum error correction (QEC) enables reliable quantum computation by encoding logical information into many physical qubits, allowing errors to be detected and corrected faster than they accumulate \cite{resilient_qc,fault-tolerance,fowler2009high,google2025quantum,bravyi2024high}. Recent experiments have demonstrated logical error suppression below threshold, establishing fault tolerance as a realistic route toward large-scale quantum computation \cite{harvard2024,harvard2026,google2021,google2025quantum,2024arXiv240402280P}. These demonstrations, however, implicitly assume that the underlying physical operations remain continuously within the calibrated low physical error regime. In practice, every quantum processor is controlled through analog parameters that inevitably drift because of environmental fluctuations \cite{PhysRevA.82.040305,PhysRevLett.116.020501,1qhb-r4fb,2025arXiv250304702C,995}. Therefore, fault tolerance requires both an initially successful calibration and the continuous stabilization of the analog control throughout the quantum computation.

The separation between calibration and computation becomes a fundamental limitation for long-duration fault-tolerant quantum computation. Useful logical algorithms and quantum memories may execute continuously for days or even months, far exceeding the stability timescale of existing quantum hardware. Interrupting the computation for dedicated calibration is, therefore, incompatible with fault-tolerant operation. A promising alternative is to integrate calibration directly into computation by repurposing the syndrome measurements already generated during quantum error correction as a calibration signal, an idea recently proposed and experimentally demonstrated in Ref.~\cite{sivak2025reinforcement}. 
However, a fundamental theoretical question remains unanswered: do syndrome measurements define an optimization landscape that admits provably efficient online calibration, or are they merely a heuristic proxy for logical performance?

Here we answer this question affirmatively. We show that, for a broad class of control-induced errors, the average detector-event rate, or detection rate, constitutes a good surrogate objective for analog calibration with a rigorous geometric structure. After randomized compiling or Pauli twirling \cite{emerson2007symmetrized,wallman2016noise,van2023probabilistic,Hu2025}, the detector-event landscape becomes locally strongly convex with high probability around the calibrated operating point under physically natural conditions. This result establishes a direct connection between syndrome statistics and analog control, providing the theoretical foundation for continuously steering hardware parameters without interrupting logical computation. Building on this surrogate structure, we formulate self-calibration as a zeroth-order online optimization problem, implement a simple and effective online convex optimization framework~\cite{hazan2016introduction}, and prove convergence guarantees for both one-time jumps and time-dependent drifting control parameters. For the surface code and a wide family of commonly considered quantum low-density parity-check (LDPC) codes~\cite{fowler2009high,bravyi2024high,panteleev2021degenerate,bravyi2024high,panteleev2022almostlinear,leverrier2022quantum}, we further propose an improved online optimization algorithm that utilizes locality information inspired by the masked gradient in Ref.\cite{sivak2025reinforcement} and rigorously prove that the convergence rate of our algorithm is independent of the code distance, establishing scalability toward large fault-tolerant processors.

We validate these theoretical predictions using pulse-level simulations of neutral-atom Rydberg gates, together with large-scale circuit-level Clifford simulations of the rotated surface code and quantum LDPC syndrome extraction \cite{gidney2021stim}. 
Our results establish self-calibrating fault tolerance as a \emph{provably efficient} paradigm in which the same syndrome measurements that protect logical information simultaneously stabilize the underlying hardware, enabling autonomous quantum computation over timescales far exceeding the intrinsic stability of today's quantum devices.

\section{Syndrome-guided calibration and online optimization guarantees}
\label{sec:theory}

A fault-tolerant processor is controlled by a high-dimensional vector of analog parameters, such as pulse amplitudes, detunings, and phases. We denote the control vector applied by the classical controller by $\vec\theta$, and the corresponding calibrated operating point by $\vec\theta^*$. Environmental fluctuations displace the hardware from this operating point, producing a control drift $\delta\vec\theta=\vec\theta-\vec\theta^*$. During a long computation, $\vec\theta^*$ may remain fixed, evolve gradually or change abruptly. This drift is generally not measured directly. Instead, the controller observes the syndrome stream generated by repeated QEC cycles.

The central question is whether this syndrome stream contains sufficient information to restore the controls with rigorous efficiency guarantees. Recent work \cite{sivak2025reinforcement} recognized that the logical error rate, although the ultimate measure of fault-tolerant performance, is poorly suited to real-time calibration: logical failures are rare below threshold, require many repetitions to estimate reliably, and generally cannot be identified during an unknown logical computation. To overcome this limitation, Ref.~\cite{sivak2025reinforcement} proposed using the detector-event rate as an online calibration signal and demonstrated its effectiveness experimentally. Here we ask the complementary theoretical question: does the detector signal define a locally restoring optimization landscape for the underlying analog controls that admits provably efficient online calibration?

\begin{figure}[t]
    \centering
    \includegraphics[width=1\linewidth]{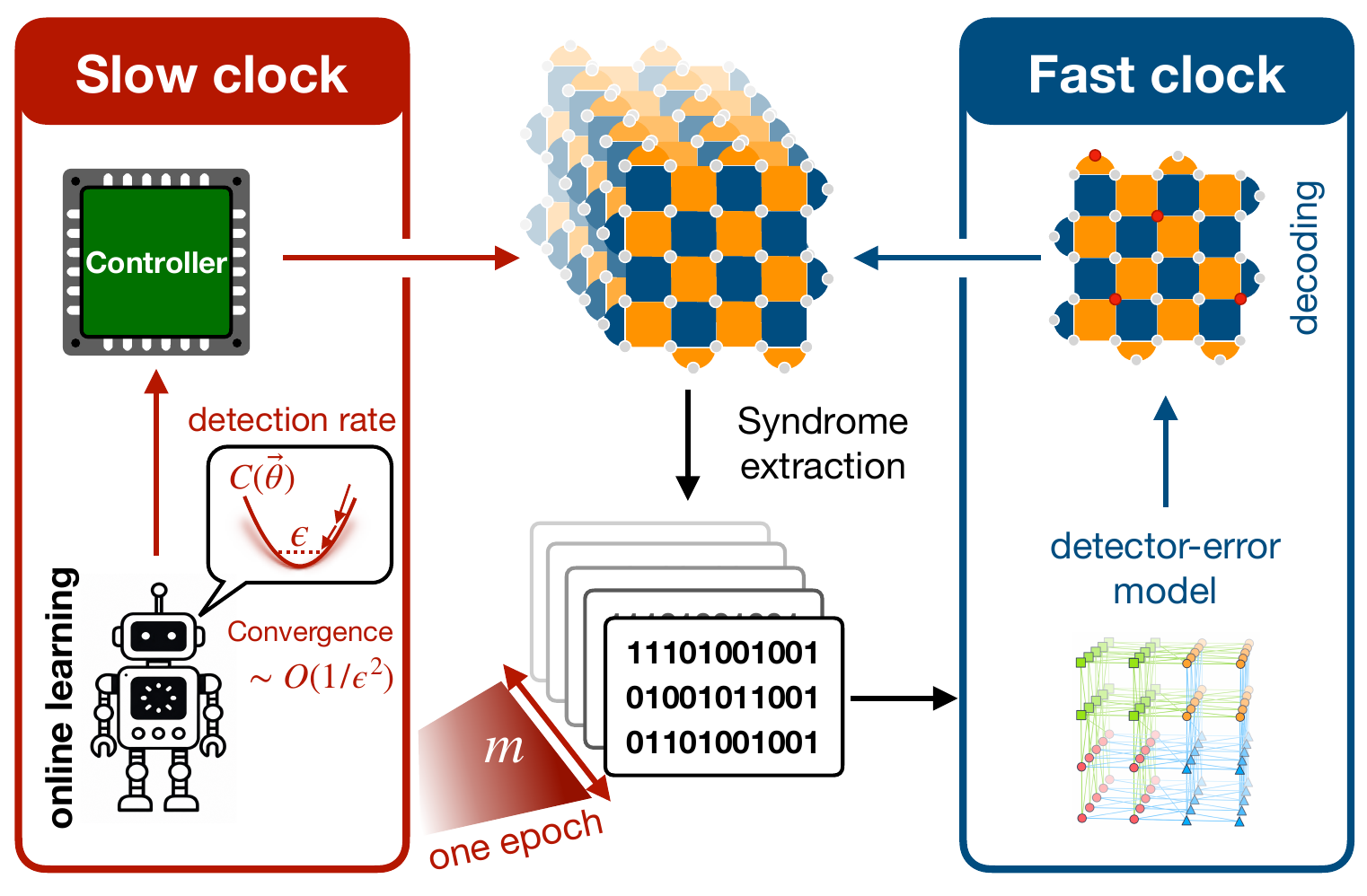}
    \caption{\textbf{Self-calibration during quantum error correction.} Each round of syndrome extraction feeds two processes at once: the decoder consumes the detection events to correct the logical state, while the same events drive online updates of the drifting control parameters. 
    Calibration thus runs as a closed loop that never halts the fault-tolerant computation.}
    \label{fig:theme_figure}
\end{figure}

We quantify this signal by the average detector-event rate, or detection rate for short. For $N_D$ detectors, let $DR_k(\vec\theta)$ denote the expected event rate of detector $D_k$. The syndrome-derived objective available to the controller~\cite{sivak2025reinforcement} is

\begin{align}\label{eq:loss_function}
C(\vec\theta)=\frac{1}{N_D}\sum_{k=1}^{N_D}DR_k(\vec\theta).
\end{align}
The role of $C$ is not assumed a priori to be equivalent to that of the logical error rate. Rather, the first part of this section establishes that, near a calibrated operating point and for broad classes of control-induced errors, $C$ possesses positive local curvature in every detector-visible control direction. This geometric property is what makes the syndrome stream suitable for calibration.

In practice, the controller has access only to finite-sample estimates of $C$. We divide the syndrome record into epochs, each containing $m$ consecutive QEC cycles, and hold the applied control vector fixed within each epoch as shown in \Cref{fig:theme_figure}. Averaging the observed detector events produces a noisy function-value estimate $\widehat C(\vec\theta)$. The window length $m$ sets a natural temporal trade-off: increasing $m$ suppresses statistical fluctuations, but reduces the rate at which the controller can respond to time-dependent drift. The controller therefore receives neither direct knowledge of $\delta\vec\theta$ nor gradients of $C$, but only noisy zeroth-order feedback from the ongoing QEC process.

We show that this restricted feedback is nevertheless sufficient for efficient self-calibration. We first establish the local convex structure of the detector-event objective. We then use this structure to derive convergence guarantees for a one-time displacement of the calibrated controls and dynamic-regret guarantees when the optimum moves in a time-dependent manner. Finally, we exploit the locality of detector regions in surface codes and general quantum low-density parity-check codes to obtain calibration convergence rates that do not deteriorate with increasing code distance with rigorous guarantees.

\subsection{Local convexity of the detection rate landscape}
\label{sec:surrogate_model_convexity}

We first establish the local geometry of the syndrome-derived objective in Eq.~\eqref{eq:loss_function}. The key observation is that, after randomized compiling or Pauli twirling, a small control displacement produces a positive quadratic response in every detector-visible error component. We begin with the contribution of a single gate to a single detector and then assemble these local responses into the global objective $C$.

For detector $D_k$, let $\mathcal R_k$ denote its detecting region, namely the set of gates whose errors can flip the detector. For a gate $i\in\mathcal R_k$ with local control displacement $\delta\vec\theta_i$, let $q_{ki}(\delta\vec\theta_i)$ be the probability that the induced error flips $D_k$. Under coherent unitary drift, randomized compiling or Pauli twirling maps the local error to a Pauli channel \cite{emerson2007symmetrized,wallman2016noise,van2023probabilistic}. To leading order,
\begin{align}
q_{ki}(\delta\vec\theta_i)
=
\sum_{P\in\mathcal F_{ki}}
\left(
\sum_j \delta\theta_{ij} h_{P,j}^{(i)}
\right)^2
+
O\!\left(\|\delta\vec\theta_i\|^3\right),
\label{eq:local_detector_response}
\end{align}
where $\mathcal F_{ki}$ is the set of Pauli errors on gate $i$ that flip $D_k$, and $h_{P,j}^{(i)}$ is the overlap between the Pauli operator $P$ and the $j$th control error generator. Hence,
\begin{align}
\nabla^2 q_{ki}(\vec 0)
=
2\sum_{P\in\mathcal F_{ki}}
\vec h_{P}^{(i)}
\vec h_{P}^{(i)\top}
\succeq 0.
\end{align}
Each detector-visible Pauli component therefore contributes a rank-one positive-semidefinite curvature matrix.

This structure extends beyond coherent errors. For a general completely positive trace-preserving (CPTP) channel, Pauli twirling converts the detector-relevant response into a sum of nonnegative contributions from coherent and dissipative error components. Leakage and atom loss have the same local structure when the readout is block diagonal between the computational and leaked subspaces, so that coherences between the two sectors do not directly enter the detector probability \cite{wu2022erasure,scholl2023erasure,chow2024circuit}. In each case, the detector-relevant error probabilities vanish to first order at the calibrated point. The leading Hessian of each detector rate is therefore obtained by summing positive-semidefinite local contributions, and averaging over detectors preserves this property.

We summarize the main result as follows, and the detailed proofs for coherent, CPTP, and leakage errors are given in Appendices~\ref{sec:unitary_convexity}, \ref{sec:cptp_convexity}, and \ref{sec:leakage_convexity}, respectively.

\begin{theorem}[Local convexity of the detector-event objective]
\label{thm:convexity_informal}
Consider control-induced coherent unitary errors, general CPTP errors, or leakage and atom loss with readout block diagonal between the computational and leaked subspaces. Then the Hessian of the average detector-event rate at the calibrated point satisfies
\begin{align}
\nabla^2 C(\vec 0)\succeq 0.
\end{align}
Any zero-curvature direction is either a gauge direction, along which $C$ is invariant, or an accidental degeneracy of the calibrated point. After quotienting out gauge directions, an arbitrarily small generic perturbation removes accidental degeneracies with probability one. Consequently, $C$ is generically strictly convex, i.e. $\nabla^2 C(\vec 0)\succ 0$, on the detector-visible parameter space within a neighborhood
\begin{align}
\|\delta\vec\theta\|\leq \theta_C^{(\mathrm{th})}.
\end{align}
\end{theorem}

\Cref{thm:convexity_informal} provides the structural basis for the optimization guarantees below. The detection rate landscape need not be globally convex.
Nevertheless, within the small-drift regime relevant to continuous calibration, every detector-visible control direction generically produces positive restoring curvature. Directions invisible to $C$ are either physically redundant and should be removed by working on the quotient parameter space, or arise from fine-tuned degeneracies eliminated by the smoothed-analysis argument in Appendix~\ref{sec:strong_convexity} \cite{carbery2001distributional,spielman2004smoothed,arthur2009k,chen2025quantum}.

\subsection{Provably efficient self-calibration after a one-time control drift}
\label{sec:time_independent_drift}

We first consider the simplest calibration scenario, in which the hardware experiences a one-time control drift and then remains stationary throughout the subsequent calibration procedure. Physically, this corresponds to a processor that has drifted away from its calibrated operating point because of environmental fluctuations, after which the controller continuously updates the analog controls using the syndrome stream generated during QEC.

By Theorem~\ref{thm:convexity_informal}, the detector-event objective is locally strictly convex in the detector-visible parameter space. Consequently, within the trusted neighborhood established by Theorem~\ref{thm:convexity_informal}, it admits the quadratic expansion
\begin{align}
C(\delta\vec\theta)
=
C_0+
\delta\vec\theta^\top
\Omega_C
\delta\vec\theta
+
O\!\left(\|\delta\vec\theta\|^3\right),
\label{eq:quadratic}
\end{align}
where $C_0$ is independent of the control drift, and $\Omega_C$ is the local Hessian of the detector-event landscape.

The controller, however, does not observe either the drift $\delta\vec\theta$ or the gradient $\nabla C$. Instead, it only receives noisy estimates of $C$ from finite batches of syndrome measurements. We therefore employ a projected simultaneous perturbation stochastic approximation (SPSA) algorithm \cite{spall1992spsa,duchi2015optimal}, which estimates the gradient using only two noisy function evaluations per epoch while projecting every iterate back into the trusted convex region. The complete procedure is summarized in Algorithm~\ref{alg:spsa_offline_main}.

Our main result shows that this limited feedback is nevertheless sufficient for efficient self-calibration.

\begin{theorem}[Provably efficient self-calibration after a one-time control drift]
\label{thm:offline_informal}
Assume that the calibrated operating point remains fixed during the calibration procedure and that the control drift stays inside the locally convex region established in Theorem~\ref{thm:convexity_informal}. Then there exists a choice of step size $\eta$ such that the projected SPSA iterates satisfy
\begin{align*}
\mathbb E[C(\overline{\delta\vec\theta_T})]-C(\vec0)
\lesssim
\frac{\sqrt d\,
\mathrm{poly}
\!\left(
\|\Omega_C\|,
\theta_C^{(\mathrm{th})},
\frac1{mN_D},
\frac1\lambda
\right)}
{\sqrt T},
\end{align*}
where $T$ is the number of syndrome-measurement epochs, $d$ is the total number of control parameters, $m$ is the number of QEC cycles per epoch, $N_D$ is the number of detectors, and $\lambda$ is the SPSA perturbation radius. Consequently, achieving an expected excess detector-event rate of at most $\varepsilon$ requires only
\[
T=O(\varepsilon^{-2})
\]
syndrome-measurement epochs.
\end{theorem}

Theorem~\ref{thm:offline_informal} establishes that self-calibration is provably efficient: the detector events already generated during normal QEC operation provide sufficient information to restore the analog controls without dedicated calibration experiments or direct access to individual physical error channels. The convergence rate depends only polynomially on the local curvature, the detector statistics and the control dimension, while retaining the optimal zeroth-order scaling $O(\varepsilon^{-2})$.

The local convexity established above does not imply global convexity. Far from the calibrated operating point, the detector-event landscape may contain multiple local minima, and finding the global optimum is generally NP-hard \cite{bittel2021training}. Fortunately, practical calibration is inherently local: fault-tolerant processors are initialized near a calibrated operating point, and continuous calibration only needs to compensate relatively small drifts. For completeness, Appendix~\ref{sec:nonconvex_offline} extends the analysis beyond the local convex regime and shows that a perturbed accelerated gradient-descent procedure can efficiently converge to an approximate local minimum from noisy detector-event queries.

\begin{alg}[htbp,label={alg:spsa_offline_main}]{Projected SPSA for one-time control drift}
\textbf{Input:} epochs $T$, QEC cycles per epoch $m$, perturbation radius $\lambda<\theta_C^{(th)}$, step size $\eta$; the current control $\vec\theta_1$ inside the trusted convex region.\\[2pt]
\textbf{For} each epoch $t=1,\dots,T$\textbf{:}
\begin{enumerate}[leftmargin=*,itemsep=1pt,topsep=2pt]
\item Sample a random vector $\vec u_t\in\{\pm 1\}^d$ uniformly.
\item Collect two independent batches of $m$ QEC cycles at the perturbed controls $\vec\theta_t\pm\lambda\vec u_t$ and record the empirical detector-event rates $\hat C(\delta\vec\theta_t\pm\lambda\vec u_t)$.
\item Estimate gradient $\hat g_t$ from $\hat C(\delta\vec\theta_t\pm\lambda\vec u_t)$.
\item Update the control to $\vec\theta_{t}-\eta\hat g_t$ and project the drift back into the convex region to get $\vec\theta_{t+1}$.
\end{enumerate}
\textbf{Output:} averaged control drift $\overline{\delta\vec\theta_T}=\frac{1}{T}\sum_{t=1}^{T}\delta\vec\theta_t$.
\end{alg}

\subsection{Provably efficient tracking of time-dependent control drift}
\label{sec:time_dependent_drift}

We now consider the more realistic situation in which the calibrated operating point itself evolves during the computation owing to continuous environmental fluctuations. In contrast to the one-time drift considered above, the controller must now simultaneously estimate and track a moving optimum while relying only on the syndrome stream generated during ongoing QEC.

Let $\vec\theta_r^*$ denote the optimal control vector at the QEC cycle $r$. During each calibration epoch, the controller applies a single control vector $\vec\theta_t$ over a window of $m$ consecutive QEC cycles, while the optimal operating point may vary within this window. The corresponding epoch loss is therefore the average detection rate over the window:
\begin{align}
C_t(\vec\theta_t)
=
\frac1m
\sum_{r\in W_t}
C_r(\vec\theta_t;\vec\theta_r^*).
\end{align}

Unlike the stationary case, convergence to a fixed optimum is no longer possible for time-dependent drifts. Instead, the relevant question is whether the controller can continuously track the moving calibration target. We quantify this capability using the dynamic regret
\begin{align}
\mathrm{DynReg}(T)
=
\sum_{t=1}^{T}
\left[
C_t(\vec\theta_t)
-
C_t(\vec v_t)
\right],
\end{align}
where $\{\vec v_t\}$ is an arbitrary comparator sequence inside the trusted convex region, which could also be the optimal control vectors. Its total motion is characterized by the path length
\begin{align}
P_T=\sum_{t=1}^{T-1}\|\vec v_{t+1}-\vec v_t\|.
\end{align}
Small $P_T$ corresponds to slowly drifting hardware, whereas large $P_T$ describes rapidly varying operating conditions. We can apply the same projected SPSA feedback loop introduced in the previous subsection, now using the epoch-dependent detector-event objective and outputting $\delta\vec\theta_t$ in each epoch. And we can prove the following guarantee for time-dependent drifts.

\begin{theorem}[Provably efficient tracking of time-dependent control drift]
\label{thm:online_informal}

Assume that the optimal operating point remains inside the locally convex region established in Theorem~\ref{thm:convexity_informal}. Then there exists a choice of step size $\eta$ such that the projected SPSA iterates satisfy

\begin{align*}
\frac1T
\mathbb E
\!\left[
\mathrm{DynReg}(T)
\right]
\lesssim
\frac{\sqrt d}{\sqrt T}
\sqrt{
\theta_C^{(\mathrm{th})2}
+
\theta_C^{(\mathrm{th})}
P_T
}.
\end{align*}

In particular, the average dynamic regret vanishes whenever the path length satisfies
$
P_T=o(T).
$

\end{theorem}

Theorem~\ref{thm:online_informal} establishes that self-calibration can continuously track slowly drifting hardware using only detector events collected during normal QEC operation. The first term in the regret bound is the unavoidable cost of learning from noisy zeroth-order feedback, while the second quantifies the additional difficulty introduced by the motion of the calibration target. Consequently, self-calibration remains effective whenever the operating point evolves more slowly than the information accumulated from the syndrome stream.

The tracking guarantee is fundamentally limited by the temporal resolution of the detector data. Increasing the epoch size reduces statistical fluctuations but slows the controller's response, whereas shorter epochs provide faster feedback at the expense of larger measurement noise. Thus, no online calibration protocol can accurately follow arbitrarily rapid analog drift using finite syndrome statistics.

Beyond the locally convex regime, continuously tracking the global optimum becomes a nonconvex online optimization problem. For piecewise stationary drifts with finitely many abrupt changes, restarting the offline calibration procedure after each change recovers approximate local minima on every stationary segment. The corresponding analysis is presented in Appendix~\ref{sec:nonconvex_online}.












\subsection{Scalable self-calibration for local quantum codes}
\label{sec:imp_local}

The results above establish that self-calibration is provably efficient, but they treat the detector-event objective as a generic $d$-dimensional optimization problem. Consequently, the convergence bounds contain an explicit dependence on the number of control parameters. This naturally raises the question of scalability: does calibration become progressively more difficult as the fault-tolerant processor grows? 

Recent circuit-level Clifford simulations of rotated surface codes in Ref.~\cite{sivak2025reinforcement} suggested that, when combined with masked gradients, the reinforcement learning exhibits convergence that is largely independent of the code distance. Here, we establish this observation on rigorous theoretical grounds. Specifically, we prove that locality-aware SPSA achieves convergence guarantees that are independent of the code distance for local quantum error-correcting codes. The key additional ingredient is the locality of syndrome information. Each detector is sensitive only to a constant-size neighborhood of physical operations, while each analog control parameter influences only a bounded number of nearby detectors. Therefore, although the total number of detectors and control parameters increases with the code distance, the amount of detector information relevant to each individual control parameter remains constant.

More precisely, we assume that each detector-event rate depends on at most $s$ scalar control coordinates, and each scalar control coordinate affects at most $c$ detector-event rates, where both $s$ and $c$ are independent of the code distance. This condition is naturally satisfied whenever the stabilizer check weight and qubit degree are bounded and each gate is described by only a constant number of analog control parameters. It therefore applies to the surface code \cite{kitaev2003fault,dennis2002topological,fowler2012surface}, bivariate-bicycle (BB) codes \cite{mackay2004sparse,kovalev2013quantum,bravyi2024high}, lifted-product codes \cite{panteleev2021degenerate,panteleev2022almostlinear,panteleev2022asymptotically}, and more generally to quantum LDPC codes \cite{tillich2014quantum,breuckmann2021quantum,panteleev2022asymptotically,leverrier2022quantum}.

This locality allows substantially more efficient gradient estimation. Rather than estimating every gradient component from the globally averaged detector-event rate, we construct locality-aware gradient estimates using only the detector events influenced by the corresponding control parameter. Graph coloring is then used to partition compatible control coordinates so that multiple perturbations can be performed simultaneously without introducing interference. The resulting locality-aware SPSA estimator remains unbiased within the local quadratic regime, while its variance depends only on the local connectivity parameters $s$ and $c$ instead of the total control dimension.

The resulting scalability guarantee is summarized below, and the details of the locality-aware SPSA can be found in Algorithm~\ref{alg:spsa_local_main}, which makes use of the locality parameters $s$ and $c$ and reduces the optimization to optimizing at most $c(s-1)+1$ subsets of parameters with a convergence rate independent of the code distance.

\begin{alg}[t!,label={alg:spsa_local_main}]{Locality-aware SPSA for local codes}
\textbf{Preprocessing:} Connect two control coordinates whenever some detector depends on both. 
Color the graph with $\chi\leq c(s-1)+1$ colors $G_1,\dots,G_\chi$ so that no detector sees two coordinates of the same color.\\[2pt]
\textbf{For} each epoch $t$ and each color class $G_a$\textbf{:}
\begin{enumerate}[leftmargin=*,itemsep=1pt,topsep=2pt]
\item Sample random vector $\vec u_t^{(a)}$ for coordinates in $G_a$.
\item Query the two perturbed controls $\vec\theta_t+\lambda\vec u_t^{(a)}$ and $\vec\theta_t-\lambda\vec u_t^{(a)}$ and record the empirical rates $\widehat{DR}_k$.
\item Estimate the gradient entry of each coordinate $j\in G_a$ from only the at most $c$ detectors that coordinate $j$ can affect, masking out all other detectors.
\item Use the masked gradient estimator $\hat g_t^{\mathrm{loc}}$ to update the controls.
\end{enumerate}
\end{alg}

\begin{theorem}[Code-distance-independent self-calibration]
\label{theorem:LSPSA}

Assume the detector-locality condition described above. Then the locality-aware SPSA algorithm achieves the same $O(\varepsilon^{-2})$ convergence guarantees as Theorems~\ref{thm:offline_informal} and \ref{thm:online_informal}, while replacing the explicit dependence on the control dimension by constants determined only by the detector-locality parameters $s$ and $c$.

Consequently, for families of local quantum error-correcting codes with bounded $s$ and $c$, the calibration complexity is independent of the code distance.

\end{theorem}

Theorem~\ref{theorem:LSPSA} identifies detector locality as the mechanism that enables scalable self-calibration. As the code size grows, the amount of syndrome information relevant to each control parameter remains constant, so the difficulty of calibration does not increase with the size of the code. In particular, for quantum LDPC codes, continuous calibration can be maintained without sacrificing the asymptotic scaling established in the previous subsections, which we also verify in the next section with numerical simulation on qLDPC codes. The same locality argument also extends to the nonconvex setting, where it accelerates convergence toward approximate local minima (Appendix~\ref{sec:nonconvex_imp}).

\section{Numerical validation}

The previous section establishes a theoretical framework for self-calibration together with rigorous convergence guarantees, online tracking and scalability. Here we validate these predictions in two complementary numerical settings. We first perform pulse-level simulations of neutral-atom quantum processors, where realistic analog control drifts naturally arise from imperfect laser control and leakage outside the computational subspace. We then perform large-scale circuit-level Clifford simulations of fault-tolerant quantum error-correction circuits to investigate the scaling of self-calibration with increasing code size. 

\subsection{Pulse-level validation of self-calibration}

To validate the theoretical framework under realistic analog control errors, we perform exact pulse-level simulations of a neutral-atom quantum processor implementing the $[\![4,1,2]\!]$ quantum detection code. This platform provides a stringent test of self-calibration because the native entangling gates are implemented through laser pulses, whose amplitudes and phases are intrinsically susceptible to experimental drift. The numerical model explicitly resolves the underlying pulse dynamics and naturally captures coherent over-rotations, control imperfections and leakage outside the computational subspace.

The simulated process is illustrated in \Cref{fig:rydberg}. Logical syndrome extraction is implemented using native Rydberg controlled-$Z$ gates optimized following Ref.~\cite{Jandura2022timeoptimaltwothree}. We model independent control drifts for every entangling gate. Specifically, laser-amplitude drifts are introduced as offsets in the pulse amplitude, while phase drifts are modeled by perturbing the leading Fourier components of the optimized phase profile. Typical drifted control pulses are visualized in \Cref{fig:rydberg} (b). Since each atom is simulated as a three-level system, leakage into the Rydberg state is naturally incorporated without introducing additional phenomenological error channels.

We first consider the one-time drift scenario analyzed in Theorem~\ref{thm:offline_informal}. At the beginning of the simulation, the control parameters are perturbed from their calibrated values and subsequently remain fixed. Figure~\ref{fig:rydberg}(c) shows that the detection rate decreases rapidly as the SPSA controller continuously updates the analog controls using only syndrome measurements. Simultaneously, the control parameters converge toward their calibrated values, confirming that detection rate feedback alone is sufficient to restore the operating point. The observed convergence is consistent with the predicted $O(\varepsilon^{-2})$ complexity and is even faster in practice.
Although we choose a constant learning rate $\eta=0.005$ for this simulation, one could also use more advanced learning rate schedulers, such as Adam \cite{kingma2015adam}, to speed up the convergence in practice while keeping the $\sqrt{T}$ dynamic regret scaling as in \Cref{thm:online_informal} in theory.

\begin{figure}[t]
    \centering
    \includegraphics[width=1\linewidth]{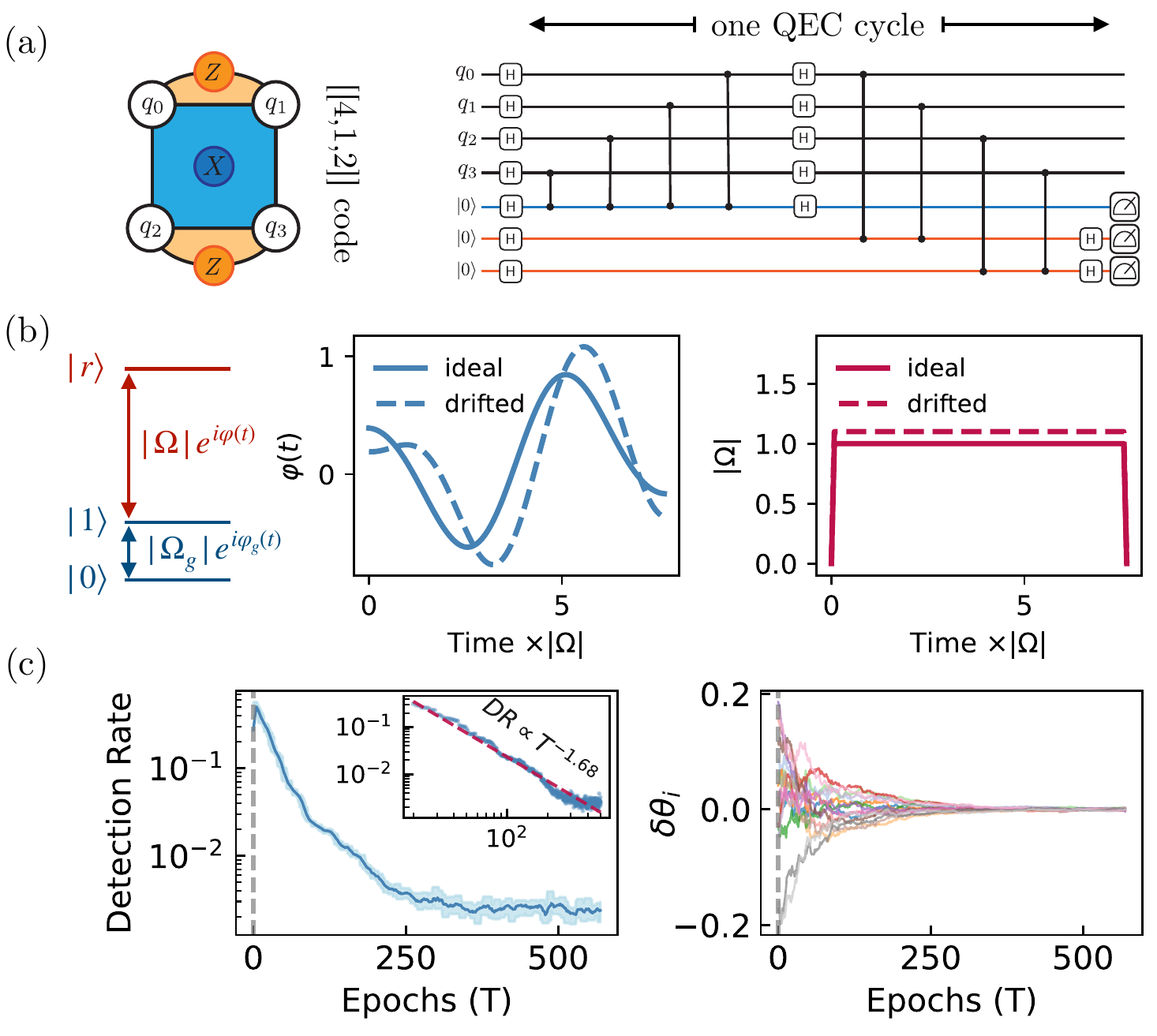}
    \caption{
\textbf{Pulse-level validation of provably efficient self-calibration.}
(a) The $[\![4,1,2]\!]$ quantum detection code and one cycle of syndrome extraction.
(b) Pulse-level neutral-atom control model. Native controlled-$Z$ gates are implemented by optimized laser amplitude $|\Omega|$ and phase $\varphi(t)$; dashed curves illustrate representative control drifts.
(c) Following an initial displacement of the analog controls, syndrome-guided SPSA continuously restores the calibrated operating point. The detector-event rate decreases throughout the calibration procedure (left), while all control parameters converge toward their calibrated values (right), validating Theorem~\ref{thm:offline_informal}. In this simulation, we choose a constant learning rate $\eta=0.005$.
}
    \label{fig:rydberg}
\end{figure}

We next consider time-dependent drifting controls, corresponding to the online calibration problem studied in Theorem~\ref{thm:online_informal}. We investigate both slowly varying linear drifts and periodic sinusoidal drifts, representing two representative classes of long-timescale hardware fluctuations. As shown in \Cref{fig:time-dependent}, the learned control parameters closely track the drifting optimum throughout the computation. Without online calibration, the detector-event rate rapidly increases as calibration errors accumulate. By contrast, syndrome-guided SPSA maintains the detection rate at a consistently low level over thousands of QEC cycles, demonstrating that the controller can continuously compensate for realistic time-dependent hardware drift.

These pulse-level simulations provide direct numerical support for the theoretical framework developed above. At the microscopic hardware level, detector events generated during normal QEC operation contain sufficient information not only to recover from a sudden miscalibration but also to continuously track slowly varying analog control drifts without interrupting the computation.

 \begin{figure}[t]
    \centering
    \includegraphics[width=1\linewidth]{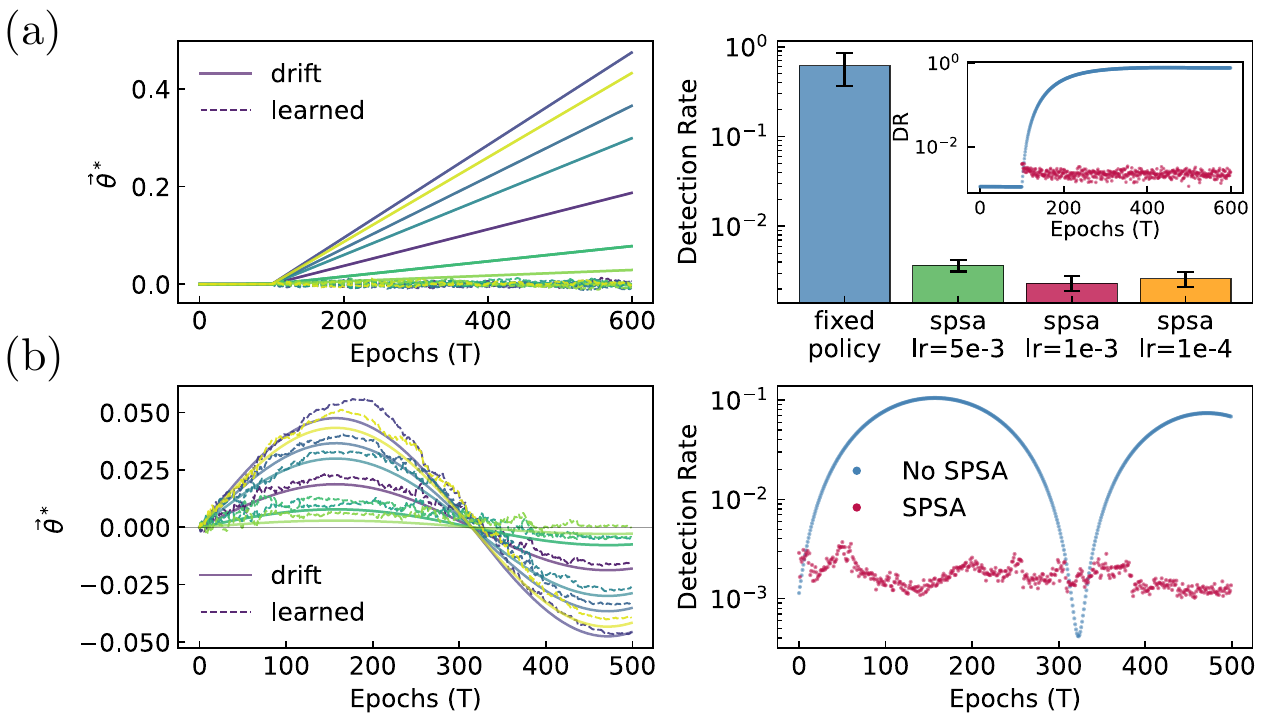}
    \caption{
\textbf{Tracking continuously drifting analog controls.}
(a) Slowly varying linear drifts. Solid curves denote the drifting optimal controls, while dashed curves show the controls learned from syndrome measurements. The detector-event rate remains low throughout the computation under online calibration.
(b) Sinusoidal control drifts. Syndrome-guided SPSA continuously tracks the time-dependent optimum and suppresses the accumulation of calibration errors, validating the online tracking guarantee of Theorem~\ref{thm:online_informal}. In both simulations, we choose a constant learning rate $\eta=0.005$.
}
    \label{fig:time-dependent}
\end{figure}

\subsection{Scalable self-calibration for quantum LDPC codes}

The previous pulse-level simulations verify that syndrome-guided calibration efficiently restores realistic analog control drifts. We now turn to the scalability of the framework on large fault-tolerant quantum processors. Theorem~\ref{theorem:LSPSA} predicts that, by exploiting the locality of syndrome information, the convergence rate of self-calibration becomes independent of the code distance for local quantum codes. Here we verify this prediction using large-scale circuit-level Clifford simulations.

We consider both rotated surface codes and the BB codes \cite{bravyi2024high} under a standard circuit-level noise model. Every single-qubit and two-qubit gate is followed by depolarizing noise with baseline error probability $\varepsilon_i^{(0)}=10^{-3}$. To model analog control imperfections, each gate is additionally assigned an independent control-induced error channel whose strength follows the quadratic detector-event landscape established in Eq.~\eqref{eq:quadratic},
\begin{align}
\varepsilon_i
=
\varepsilon_i^{(0)}
+
\delta\vec\theta_i^{\top}
\Omega_i
\delta\vec\theta_i,
\end{align}
where $\Omega_i\succeq0$ is a random positive-semidefinite sensitivity matrix. Each gate is assumed to possess five independent analog control parameters, representing realistic multi-parameter pulse calibration.

 \begin{figure}[t]
    \centering
    \includegraphics[width=1\linewidth]{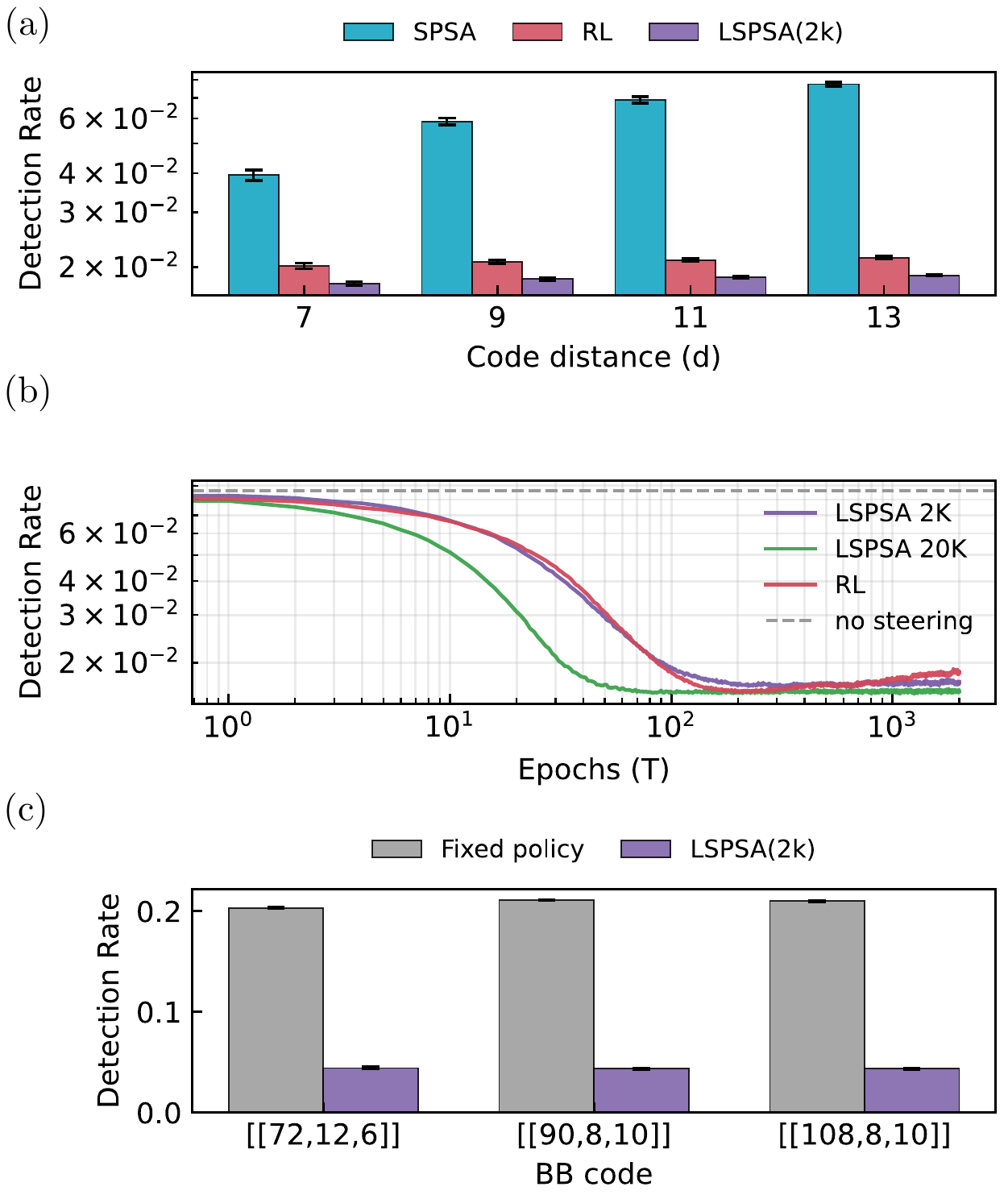}
    \caption{
\textbf{Scalable self-calibration for quantum LDPC codes.}
(a) Detector-event rate after calibration for rotated surface codes of increasing code distance. Conventional SPSA exhibits the predicted degradation with code distance, whereas locality-aware SPSA (LSPSA) remains essentially independent of system size and achieves performance comparable to or better than the reinforcement-learning (RL) controller.
(b) Representative convergence curves. Increasing the number of syndrome measurements accelerates SPSA by reducing statistical noise. LSPSA rapidly converges to a stable low detector-event rate, while the RL controller exhibits a small increase after convergence because of continued exploration.
(c) Application to BB codes. The calibrated detector-event rate remains nearly unchanged as the code size increases, confirming the code-distance-independent scaling predicted by Theorem~\ref{theorem:LSPSA}.
}
    \label{fig:qldpc}
\end{figure}

We first study rotated surface codes with increasing code distance. \Cref{fig:qldpc} (a) compares three calibration strategies: conventional SPSA, locality-aware SPSA (LSPSA), and the reinforcement-learning (RL) controller introduced in Ref.~\cite{sivak2025reinforcement}. We note that the implementation of the reinforcement-learning controller in Ref.~\cite{sivak2025reinforcement} is not publicly available. The RL results reported here are therefore obtained using our own implementation based on the algorithm described in Ref.~\cite{sivak2025reinforcement} (details in Appendix~\ref{app:num_rl}). As predicted by Theorem~\ref{theorem:LSPSA}, conventional SPSA converges more slowly as the code distance increases because the number of control parameters grows approximately as $O(d^2)$, leading to the expected $O(d)$ scaling. In contrast, the performance of LSPSA remains essentially unchanged across all simulated code distances, demonstrating that detector locality removes the explicit dependence on the system size. Using fewer syndrome measurements per estimation of detection rate, LSPSA achieves performance comparable to or better than the reinforcement learning controller.

Figure~\ref{fig:qldpc} (b) further compares the convergence dynamics of different calibration strategies. Increasing the number of syndrome measurements improves the convergence speed of SPSA by reducing the statistical uncertainty of the detector-event objective. More importantly, LSPSA converges to a stable low detection rate and remains there throughout the calibration procedure. By comparison, the reinforcement-learning controller exhibits a slight increase in detection rate after reaching the optimum. This behavior is consistent with the continued exploration required by policy optimization, whereas gradient-based self-calibration naturally stabilizes once the detector-event landscape has been minimized.

Finally, we verify the scalability of the approach beyond surface codes by applying LSPSA to BB codes \cite{mackay2004sparse,kovalev2013quantum,bravyi2024high}. As shown in \Cref{fig:qldpc} (c), the calibrated detection rate remains essentially unchanged as the code size increases, in excellent agreement with Theorem~\ref{theorem:LSPSA}. Therefore, these circuit-level simulations confirm that detector locality enables scalable self-calibration whose complexity remains independent of the code distance, making continuous self-calibration compatible with large-scale fault-tolerant quantum computation.

\section{Discussion}
\label{sec:discussion}

Fault-tolerant quantum computation has traditionally been viewed as consisting of two largely separate tasks: maintaining logical information through quantum error correction and maintaining hardware performance through repeated calibration. 
Our results establish a theoretical framework that unifies these two processes. 
Instead of pausing the computation for dedicated calibration runs, our framework lets a single syndrome stream serve both ends: the detection events the decoder uses to protect the logical state also supply the signal that keeps the analog controls in tune.
In this picture, calibration is no longer an auxiliary procedure performed between computations, but an integrated part of fault-tolerant operation itself.

While detector events have recently been used experimentally as calibration signals \cite{sivak2025reinforcement}, it has remained unclear whether they possess sufficient geometric structure to support provably efficient online optimization. We show that, after randomized compiling, the detection rate landscape is generically locally convex for a broad class of control-induced errors. This geometric property directly leads to rigorous guarantees for efficient self-calibration, online tracking of slowly drifting hardware and scalability to quantum LDPC codes.

Our results also identify several directions for extending the present framework. The current analysis focuses on local optimization around a calibrated operating point, which is the regime most relevant to the continuous correction of small hardware drifts but does not address global calibration far from a well-functioning operating point. Extending the theory to more strongly nonconvex landscapes and correlated drift errors will therefore be important for a more complete theory of autonomous calibration.

A first practical direction is to develop more expressive controllers beyond the simple zeroth-order algorithms considered here. Our results provide a rigorous objective on which deep-neural-network-based or model-based controllers can be trained and evaluated. Such controllers could learn platform-specific structures and exploit temporal and spatial correlations in the syndrome stream to reduce the number of QEC measurements required for each update. This could substantially accelerate control optimization while preserving the direct connection between the learning signal and fault-tolerant performance established in this work.

A second direction is the joint design of calibration and decoding. In present architectures, the decoder operates on a fast timescale to infer logical corrections, whereas the calibration controller updates analog parameters on a slower timescale. These two feedback loops are usually designed separately, even though they rely on the same syndrome stream and respond to the same underlying noise processes. Integrating the fast decoding loop with the slow calibration loop, as illustrated in \Cref{fig:theme_figure}, could enable a unified controller that simultaneously addresses baseline errors and control drift errors. More broadly, advances in large language models (LLMs) and AI agents suggest the possibility of agent-based control systems that coordinate calibration, decoding, and hardware diagnostics across the full fault-tolerant stack \cite{2026arXiv260328052L, 2026arXiv260622376X,2026arXiv260725145I,bhardwaj2026high}. In this setting, the detection rate provides an objective and experimentally accessible benchmark for evaluating whether an agent can maintain the processor within its fault-tolerant operating regime. Self-calibration therefore becomes a cornerstone of quantum computers that not only correct baseline errors but also autonomously diagnose and control the physical processes that generate them.

 
 
\section*{Acknowledgements}
 
We thank Wenlan Chen, Hsin-Yuan Huang, Tongyang Li, John Preskill, Muqing Xu, Yi-Zhuang You, Tao Zhang, Wenjun Zhang, Chen Zhao, and Harry Zhou for the insightful discussions and feedback during the preparation of this work. We thank Volodymyr Sivak for his comments on the draft. WG and HYH conceived the project and developed the theory of this work while HYH was a research fellow at Harvard University.
HYH investigated the numerical performance of the algorithm after joining IOP-CAS. HYH thanks Lei Wang for the computational support.
A part of this work was carried out while WG was visiting the California Institute of Technology.
We acknowledge ChatGPT and Claude for assistance in discussing and refining proof ideas, as well as improving the presentation. 
The authors are solely responsible for the proofs and the results.

\let\oldaddcontentsline\addcontentsline
\renewcommand{\addcontentsline}[3]{}
\bibliography{main.bib}

\begin{thebibliography}{75}%
\makeatletter
\providecommand \@ifxundefined [1]{%
 \@ifx{#1\undefined}
}%
\providecommand \@ifnum [1]{%
 \ifnum #1\expandafter \@firstoftwo
 \else \expandafter \@secondoftwo
 \fi
}%
\providecommand \@ifx [1]{%
 \ifx #1\expandafter \@firstoftwo
 \else \expandafter \@secondoftwo
 \fi
}%
\providecommand \natexlab [1]{#1}%
\providecommand \enquote  [1]{``#1''}%
\providecommand \bibnamefont  [1]{#1}%
\providecommand \bibfnamefont [1]{#1}%
\providecommand \citenamefont [1]{#1}%
\providecommand \href@noop [0]{\@secondoftwo}%
\providecommand \href [0]{\begingroup \@sanitize@url \@href}%
\providecommand \@href[1]{\@@startlink{#1}\@@href}%
\providecommand \@@href[1]{\endgroup#1\@@endlink}%
\providecommand \@sanitize@url [0]{\catcode `\\12\catcode `\$12\catcode `\&12\catcode `\#12\catcode `\^12\catcode `\_12\catcode `\%12\relax}%
\providecommand \@@startlink[1]{}%
\providecommand \@@endlink[0]{}%
\providecommand \url  [0]{\begingroup\@sanitize@url \@url }%
\providecommand \@url [1]{\endgroup\@href {#1}{\urlprefix }}%
\providecommand \urlprefix  [0]{URL }%
\providecommand \Eprint [0]{\href }%
\providecommand \doibase [0]{https://doi.org/}%
\providecommand \selectlanguage [0]{\@gobble}%
\providecommand \bibinfo  [0]{\@secondoftwo}%
\providecommand \bibfield  [0]{\@secondoftwo}%
\providecommand \translation [1]{[#1]}%
\providecommand \BibitemOpen [0]{}%
\providecommand \bibitemStop [0]{}%
\providecommand \bibitemNoStop [0]{.\EOS\space}%
\providecommand \EOS [0]{\spacefactor3000\relax}%
\providecommand \BibitemShut  [1]{\csname bibitem#1\endcsname}%
\let\auto@bib@innerbib\@empty
\bibitem [{\citenamefont {Sivak}\ \emph {et~al.}(2026)\citenamefont {Sivak}, \citenamefont {Morvan}, \citenamefont {Broughton}, \citenamefont {Cortiñas}, \citenamefont {Bausch}, \citenamefont {Senior}, \citenamefont {Neeley}, \citenamefont {Eickbusch}, \citenamefont {Shutty}, \citenamefont {Beni}, \citenamefont {Spencer}, \citenamefont {Heras}, \citenamefont {Edlich}, \citenamefont {Abanin}, \citenamefont {Abbas}, \citenamefont {Acharya}, \citenamefont {Aigeldinger}, \citenamefont {Alcaraz}, \citenamefont {Alcaraz}, \citenamefont {Andersen}, \citenamefont {Ansmann}, \citenamefont {Arute}, \citenamefont {Arya}, \citenamefont {Askew}, \citenamefont {Astrakhantsev}, \citenamefont {Atalaya}, \citenamefont {Ballard}, \citenamefont {Bardin}, \citenamefont {Bates}, \citenamefont {Bengtsson}, \citenamefont {Karimi}, \citenamefont {Bilmes}, \citenamefont {Bilodeau}, \citenamefont {Borjans}, \citenamefont {Bourassa}, \citenamefont {Bovaird}, \citenamefont {Bowers}, \citenamefont {Brill}, \citenamefont {Brooks},
  \citenamefont {Browne}, \citenamefont {Buchea}, \citenamefont {Buckley}, \citenamefont {Burger}, \citenamefont {Burkett}, \citenamefont {Bushnell}, \citenamefont {Busnaina}, \citenamefont {Cabrera}, \citenamefont {Campero}, \citenamefont {Chang}, \citenamefont {Chen}, \citenamefont {Chiaro}, \citenamefont {Chih}, \citenamefont {Cleland}, \citenamefont {Cochrane}, \citenamefont {Cockrell}, \citenamefont {Cogan}, \citenamefont {Collins}, \citenamefont {Conner}, \citenamefont {Cook}, \citenamefont {Courtney}, \citenamefont {Crook}, \citenamefont {Curtin}, \citenamefont {Damyanov}, \citenamefont {Das}, \citenamefont {Debroy}, \citenamefont {Demura}, \citenamefont {Donohoe}, \citenamefont {Drozdov}, \citenamefont {Dunsworth}, \citenamefont {Ehimhen}, \citenamefont {Elbag}, \citenamefont {Ella}, \citenamefont {Elzouka}, \citenamefont {Enriquez}, \citenamefont {Erickson}, \citenamefont {Ferreira}, \citenamefont {Flores}, \citenamefont {Burgos}, \citenamefont {Forati}, \citenamefont {Ford}, \citenamefont {Fowler},
  \citenamefont {Foxen}, \citenamefont {Fukami}, \citenamefont {Fung}, \citenamefont {Fuste}, \citenamefont {Ganjam}, \citenamefont {Garcia}, \citenamefont {Garrick}, \citenamefont {Gasca}, \citenamefont {Gehring}, \citenamefont {Geiger}, \citenamefont {Élie Genois}, \citenamefont {Giang}, \citenamefont {Gilboa}, \citenamefont {Goeders}, \citenamefont {Gonzales}, \citenamefont {Gosula}, \citenamefont {de~Graaf}, \citenamefont {Dau}, \citenamefont {Graumann}, \citenamefont {Grebel}, \citenamefont {Greene}, \citenamefont {Gross}, \citenamefont {Guerrero}, \citenamefont {Guevel}, \citenamefont {Ha}, \citenamefont {Habegger}, \citenamefont {Hadick}, \citenamefont {Hadjikhani}, \citenamefont {Hamilton}, \citenamefont {Harrigan}, \citenamefont {Harrington}, \citenamefont {Hartshorn}, \citenamefont {Heslin}, \citenamefont {Heu}, \citenamefont {Higgott}, \citenamefont {Hiltermann}, \citenamefont {Huang}, \citenamefont {Hucka}, \citenamefont {Hudspeth}, \citenamefont {Huff}, \citenamefont {Huggins}, \citenamefont
  {Jeffrey}, \citenamefont {Jevons}, \citenamefont {Jiang}, \citenamefont {Jin}, \citenamefont {Joshi}, \citenamefont {Juhas}, \citenamefont {Kabel}, \citenamefont {Kafri}, \citenamefont {Kang}, \citenamefont {Kang}, \citenamefont {Karamlou}, \citenamefont {Kaufman}, \citenamefont {Kechedzhi}, \citenamefont {Khattar}, \citenamefont {Khezri}, \citenamefont {Kim}, \citenamefont {Knaut}, \citenamefont {Kobrin}, \citenamefont {Kostritsa}, \citenamefont {Kreikebaum}, \citenamefont {Kudo}, \citenamefont {Kueffler}, \citenamefont {Kumar}, \citenamefont {Kurilovich}, \citenamefont {Kutsko}, \citenamefont {Lacroix}, \citenamefont {Landhuis}, \citenamefont {Lange-Dei}, \citenamefont {Langley}, \citenamefont {Laptev}, \citenamefont {Lau}, \citenamefont {Ledford}, \citenamefont {Lee}, \citenamefont {Lee}, \citenamefont {Lester}, \citenamefont {Leung}, \citenamefont {Li}, \citenamefont {Li}, \citenamefont {Li}, \citenamefont {Lill}, \citenamefont {Livingston}, \citenamefont {Lloyd}, \citenamefont {Locharla}, \citenamefont
  {Lorenzo}, \citenamefont {Lundahl}, \citenamefont {Lunt}, \citenamefont {Madhuk}, \citenamefont {Maiti}, \citenamefont {Maloney}, \citenamefont {Mandrà}, \citenamefont {Martin}, \citenamefont {Martin}, \citenamefont {Mascot}, \citenamefont {Das}, \citenamefont {Maslov}, \citenamefont {Mathews}, \citenamefont {Maxfield}, \citenamefont {McClean}, \citenamefont {McEwen}, \citenamefont {Meeks}, \citenamefont {Miao}, \citenamefont {Minev}, \citenamefont {Molavi}, \citenamefont {Molina}, \citenamefont {Montazeri}, \citenamefont {Neill}, \citenamefont {Newman}, \citenamefont {Nguyen}, \citenamefont {Nguyen}, \citenamefont {Ni}, \citenamefont {Niu}, \citenamefont {Oas}, \citenamefont {Orosco}, \citenamefont {Ottosson}, \citenamefont {Pagano}, \citenamefont {Paolo}, \citenamefont {Peek}, \citenamefont {Peterson}, \citenamefont {Pizzuto}, \citenamefont {Portoles}, \citenamefont {Potter}, \citenamefont {Pritchard}, \citenamefont {Qian}, \citenamefont {Quintana}, \citenamefont {Ranadive}, \citenamefont {Reagor},
  \citenamefont {Resnick}, \citenamefont {Rhodes}, \citenamefont {Riley}, \citenamefont {Roberts}, \citenamefont {Rodriguez}, \citenamefont {Ropes}, \citenamefont {Rose}, \citenamefont {Rosenberg}, \citenamefont {Rosenfeld}, \citenamefont {Rosenstock}, \citenamefont {Rossi}, \citenamefont {Roushan}, \citenamefont {Rower}, \citenamefont {Salazar}, \citenamefont {Sankaragomathi}, \citenamefont {Sarihan}, \citenamefont {Satzinger}, \citenamefont {Schaefer}, \citenamefont {Schroeder}, \citenamefont {Schurkus}, \citenamefont {Shahingohar}, \citenamefont {Shearn}, \citenamefont {Shorter}, \citenamefont {Shvarts}, \citenamefont {Small}, \citenamefont {Smith}, \citenamefont {Sobel}, \citenamefont {Spells}, \citenamefont {Springer}, \citenamefont {Sterling}, \citenamefont {Suchard}, \citenamefont {Szasz}, \citenamefont {Sztein}, \citenamefont {Taylor}, \citenamefont {Thiruraman}, \citenamefont {Thor}, \citenamefont {Timucin}, \citenamefont {Tomita}, \citenamefont {Torres}, \citenamefont {Torunbalci}, \citenamefont
  {Tran}, \citenamefont {Vaishnav}, \citenamefont {Vargas}, \citenamefont {Vdovichev}, \citenamefont {Vidal}, \citenamefont {Heidweiller}, \citenamefont {Voorhees}, \citenamefont {Waltman}, \citenamefont {Waltz}, \citenamefont {Wang}, \citenamefont {Ware}, \citenamefont {Watson}, \citenamefont {Wei}, \citenamefont {Weidel}, \citenamefont {White}, \citenamefont {Wong}, \citenamefont {Woo}, \citenamefont {Wood}, \citenamefont {Woodson}, \citenamefont {Xing}, \citenamefont {Yao}, \citenamefont {Yeh}, \citenamefont {Ying}, \citenamefont {Yoo}, \citenamefont {Yosri}, \citenamefont {Young}, \citenamefont {Young}, \citenamefont {Zalcman}, \citenamefont {Zhang}, \citenamefont {Zhang}, \citenamefont {Zhu}, \citenamefont {Zobrist}, \citenamefont {Zou}, \citenamefont {Babbush}, \citenamefont {Bacon}, \citenamefont {Boixo}, \citenamefont {Chen}, \citenamefont {Chen}, \citenamefont {Devoret}, \citenamefont {Hansen}, \citenamefont {Hilton}, \citenamefont {Jones}, \citenamefont {Kelly}, \citenamefont {Korotkov},
  \citenamefont {Lucero}, \citenamefont {Megrant}, \citenamefont {Neven}, \citenamefont {Oliver}, \citenamefont {Ramachandran}, \citenamefont {Smelyanskiy},\ and\ \citenamefont {Klimov}}]{sivak2025reinforcement}%
  \BibitemOpen
  \bibfield  {author} {\bibinfo {author} {\bibfnamefont {V.}~\bibnamefont {Sivak}}, \bibinfo {author} {\bibfnamefont {A.}~\bibnamefont {Morvan}}, \bibinfo {author} {\bibfnamefont {M.}~\bibnamefont {Broughton}}, \bibinfo {author} {\bibfnamefont {R.~G.}\ \bibnamefont {Cortiñas}}, \bibinfo {author} {\bibfnamefont {J.}~\bibnamefont {Bausch}}, \bibinfo {author} {\bibfnamefont {A.~W.}\ \bibnamefont {Senior}}, \bibinfo {author} {\bibfnamefont {M.}~\bibnamefont {Neeley}}, \bibinfo {author} {\bibfnamefont {A.}~\bibnamefont {Eickbusch}}, \bibinfo {author} {\bibfnamefont {N.}~\bibnamefont {Shutty}}, \bibinfo {author} {\bibfnamefont {L.~A.}\ \bibnamefont {Beni}}, \bibinfo {author} {\bibfnamefont {J.~S.}\ \bibnamefont {Spencer}}, \bibinfo {author} {\bibfnamefont {F.~J.~H.}\ \bibnamefont {Heras}}, \bibinfo {author} {\bibfnamefont {T.}~\bibnamefont {Edlich}}, \bibinfo {author} {\bibfnamefont {D.}~\bibnamefont {Abanin}}, \bibinfo {author} {\bibfnamefont {A.}~\bibnamefont {Abbas}}, \bibinfo {author} {\bibfnamefont
  {R.}~\bibnamefont {Acharya}}, \bibinfo {author} {\bibfnamefont {G.}~\bibnamefont {Aigeldinger}}, \bibinfo {author} {\bibfnamefont {R.}~\bibnamefont {Alcaraz}}, \bibinfo {author} {\bibfnamefont {S.}~\bibnamefont {Alcaraz}}, \bibinfo {author} {\bibfnamefont {T.~I.}\ \bibnamefont {Andersen}}, \bibinfo {author} {\bibfnamefont {M.}~\bibnamefont {Ansmann}}, \bibinfo {author} {\bibfnamefont {F.}~\bibnamefont {Arute}}, \bibinfo {author} {\bibfnamefont {K.}~\bibnamefont {Arya}}, \bibinfo {author} {\bibfnamefont {W.}~\bibnamefont {Askew}}, \bibinfo {author} {\bibfnamefont {N.}~\bibnamefont {Astrakhantsev}}, \bibinfo {author} {\bibfnamefont {J.}~\bibnamefont {Atalaya}}, \bibinfo {author} {\bibfnamefont {B.}~\bibnamefont {Ballard}}, \bibinfo {author} {\bibfnamefont {J.~C.}\ \bibnamefont {Bardin}}, \bibinfo {author} {\bibfnamefont {H.}~\bibnamefont {Bates}}, \bibinfo {author} {\bibfnamefont {A.}~\bibnamefont {Bengtsson}}, \bibinfo {author} {\bibfnamefont {M.~B.}\ \bibnamefont {Karimi}}, \bibinfo {author} {\bibfnamefont
  {A.}~\bibnamefont {Bilmes}}, \bibinfo {author} {\bibfnamefont {S.}~\bibnamefont {Bilodeau}}, \bibinfo {author} {\bibfnamefont {F.}~\bibnamefont {Borjans}}, \bibinfo {author} {\bibfnamefont {A.}~\bibnamefont {Bourassa}}, \bibinfo {author} {\bibfnamefont {J.}~\bibnamefont {Bovaird}}, \bibinfo {author} {\bibfnamefont {D.}~\bibnamefont {Bowers}}, \bibinfo {author} {\bibfnamefont {L.}~\bibnamefont {Brill}}, \bibinfo {author} {\bibfnamefont {P.}~\bibnamefont {Brooks}}, \bibinfo {author} {\bibfnamefont {D.~A.}\ \bibnamefont {Browne}}, \bibinfo {author} {\bibfnamefont {B.}~\bibnamefont {Buchea}}, \bibinfo {author} {\bibfnamefont {B.~B.}\ \bibnamefont {Buckley}}, \bibinfo {author} {\bibfnamefont {T.}~\bibnamefont {Burger}}, \bibinfo {author} {\bibfnamefont {B.}~\bibnamefont {Burkett}}, \bibinfo {author} {\bibfnamefont {N.}~\bibnamefont {Bushnell}}, \bibinfo {author} {\bibfnamefont {J.}~\bibnamefont {Busnaina}}, \bibinfo {author} {\bibfnamefont {A.}~\bibnamefont {Cabrera}}, \bibinfo {author} {\bibfnamefont
  {J.}~\bibnamefont {Campero}}, \bibinfo {author} {\bibfnamefont {H.-S.}\ \bibnamefont {Chang}}, \bibinfo {author} {\bibfnamefont {S.}~\bibnamefont {Chen}}, \bibinfo {author} {\bibfnamefont {B.}~\bibnamefont {Chiaro}}, \bibinfo {author} {\bibfnamefont {L.-Y.}\ \bibnamefont {Chih}}, \bibinfo {author} {\bibfnamefont {A.~Y.}\ \bibnamefont {Cleland}}, \bibinfo {author} {\bibfnamefont {B.}~\bibnamefont {Cochrane}}, \bibinfo {author} {\bibfnamefont {M.}~\bibnamefont {Cockrell}}, \bibinfo {author} {\bibfnamefont {J.}~\bibnamefont {Cogan}}, \bibinfo {author} {\bibfnamefont {R.}~\bibnamefont {Collins}}, \bibinfo {author} {\bibfnamefont {P.}~\bibnamefont {Conner}}, \bibinfo {author} {\bibfnamefont {H.}~\bibnamefont {Cook}}, \bibinfo {author} {\bibfnamefont {W.}~\bibnamefont {Courtney}}, \bibinfo {author} {\bibfnamefont {A.~L.}\ \bibnamefont {Crook}}, \bibinfo {author} {\bibfnamefont {B.}~\bibnamefont {Curtin}}, \bibinfo {author} {\bibfnamefont {M.}~\bibnamefont {Damyanov}}, \bibinfo {author} {\bibfnamefont
  {S.}~\bibnamefont {Das}}, \bibinfo {author} {\bibfnamefont {D.~M.}\ \bibnamefont {Debroy}}, \bibinfo {author} {\bibfnamefont {S.}~\bibnamefont {Demura}}, \bibinfo {author} {\bibfnamefont {P.}~\bibnamefont {Donohoe}}, \bibinfo {author} {\bibfnamefont {I.}~\bibnamefont {Drozdov}}, \bibinfo {author} {\bibfnamefont {A.}~\bibnamefont {Dunsworth}}, \bibinfo {author} {\bibfnamefont {V.}~\bibnamefont {Ehimhen}}, \bibinfo {author} {\bibfnamefont {A.~M.}\ \bibnamefont {Elbag}}, \bibinfo {author} {\bibfnamefont {L.}~\bibnamefont {Ella}}, \bibinfo {author} {\bibfnamefont {M.}~\bibnamefont {Elzouka}}, \bibinfo {author} {\bibfnamefont {D.}~\bibnamefont {Enriquez}}, \bibinfo {author} {\bibfnamefont {C.}~\bibnamefont {Erickson}}, \bibinfo {author} {\bibfnamefont {V.~S.}\ \bibnamefont {Ferreira}}, \bibinfo {author} {\bibfnamefont {M.}~\bibnamefont {Flores}}, \bibinfo {author} {\bibfnamefont {L.~F.}\ \bibnamefont {Burgos}}, \bibinfo {author} {\bibfnamefont {E.}~\bibnamefont {Forati}}, \bibinfo {author} {\bibfnamefont
  {J.}~\bibnamefont {Ford}}, \bibinfo {author} {\bibfnamefont {A.~G.}\ \bibnamefont {Fowler}}, \bibinfo {author} {\bibfnamefont {B.}~\bibnamefont {Foxen}}, \bibinfo {author} {\bibfnamefont {M.}~\bibnamefont {Fukami}}, \bibinfo {author} {\bibfnamefont {A.~W.~L.}\ \bibnamefont {Fung}}, \bibinfo {author} {\bibfnamefont {L.}~\bibnamefont {Fuste}}, \bibinfo {author} {\bibfnamefont {S.}~\bibnamefont {Ganjam}}, \bibinfo {author} {\bibfnamefont {G.}~\bibnamefont {Garcia}}, \bibinfo {author} {\bibfnamefont {C.}~\bibnamefont {Garrick}}, \bibinfo {author} {\bibfnamefont {R.}~\bibnamefont {Gasca}}, \bibinfo {author} {\bibfnamefont {H.}~\bibnamefont {Gehring}}, \bibinfo {author} {\bibfnamefont {R.}~\bibnamefont {Geiger}}, \bibinfo {author} {\bibnamefont {Élie Genois}}, \bibinfo {author} {\bibfnamefont {W.}~\bibnamefont {Giang}}, \bibinfo {author} {\bibfnamefont {D.}~\bibnamefont {Gilboa}}, \bibinfo {author} {\bibfnamefont {J.~E.}\ \bibnamefont {Goeders}}, \bibinfo {author} {\bibfnamefont {E.~C.}\ \bibnamefont
  {Gonzales}}, \bibinfo {author} {\bibfnamefont {R.}~\bibnamefont {Gosula}}, \bibinfo {author} {\bibfnamefont {S.~J.}\ \bibnamefont {de~Graaf}}, \bibinfo {author} {\bibfnamefont {A.~G.}\ \bibnamefont {Dau}}, \bibinfo {author} {\bibfnamefont {D.}~\bibnamefont {Graumann}}, \bibinfo {author} {\bibfnamefont {J.}~\bibnamefont {Grebel}}, \bibinfo {author} {\bibfnamefont {A.}~\bibnamefont {Greene}}, \bibinfo {author} {\bibfnamefont {J.~A.}\ \bibnamefont {Gross}}, \bibinfo {author} {\bibfnamefont {J.}~\bibnamefont {Guerrero}}, \bibinfo {author} {\bibfnamefont {L.~L.}\ \bibnamefont {Guevel}}, \bibinfo {author} {\bibfnamefont {T.}~\bibnamefont {Ha}}, \bibinfo {author} {\bibfnamefont {S.}~\bibnamefont {Habegger}}, \bibinfo {author} {\bibfnamefont {T.}~\bibnamefont {Hadick}}, \bibinfo {author} {\bibfnamefont {A.}~\bibnamefont {Hadjikhani}}, \bibinfo {author} {\bibfnamefont {M.~C.}\ \bibnamefont {Hamilton}}, \bibinfo {author} {\bibfnamefont {M.~P.}\ \bibnamefont {Harrigan}}, \bibinfo {author} {\bibfnamefont {S.~D.}\
  \bibnamefont {Harrington}}, \bibinfo {author} {\bibfnamefont {J.}~\bibnamefont {Hartshorn}}, \bibinfo {author} {\bibfnamefont {S.}~\bibnamefont {Heslin}}, \bibinfo {author} {\bibfnamefont {P.}~\bibnamefont {Heu}}, \bibinfo {author} {\bibfnamefont {O.}~\bibnamefont {Higgott}}, \bibinfo {author} {\bibfnamefont {R.}~\bibnamefont {Hiltermann}}, \bibinfo {author} {\bibfnamefont {H.-Y.}\ \bibnamefont {Huang}}, \bibinfo {author} {\bibfnamefont {M.}~\bibnamefont {Hucka}}, \bibinfo {author} {\bibfnamefont {C.}~\bibnamefont {Hudspeth}}, \bibinfo {author} {\bibfnamefont {A.}~\bibnamefont {Huff}}, \bibinfo {author} {\bibfnamefont {W.~J.}\ \bibnamefont {Huggins}}, \bibinfo {author} {\bibfnamefont {E.}~\bibnamefont {Jeffrey}}, \bibinfo {author} {\bibfnamefont {S.}~\bibnamefont {Jevons}}, \bibinfo {author} {\bibfnamefont {Z.}~\bibnamefont {Jiang}}, \bibinfo {author} {\bibfnamefont {X.}~\bibnamefont {Jin}}, \bibinfo {author} {\bibfnamefont {C.}~\bibnamefont {Joshi}}, \bibinfo {author} {\bibfnamefont {P.}~\bibnamefont
  {Juhas}}, \bibinfo {author} {\bibfnamefont {A.}~\bibnamefont {Kabel}}, \bibinfo {author} {\bibfnamefont {D.}~\bibnamefont {Kafri}}, \bibinfo {author} {\bibfnamefont {H.}~\bibnamefont {Kang}}, \bibinfo {author} {\bibfnamefont {K.}~\bibnamefont {Kang}}, \bibinfo {author} {\bibfnamefont {A.~H.}\ \bibnamefont {Karamlou}}, \bibinfo {author} {\bibfnamefont {R.}~\bibnamefont {Kaufman}}, \bibinfo {author} {\bibfnamefont {K.}~\bibnamefont {Kechedzhi}}, \bibinfo {author} {\bibfnamefont {T.}~\bibnamefont {Khattar}}, \bibinfo {author} {\bibfnamefont {M.}~\bibnamefont {Khezri}}, \bibinfo {author} {\bibfnamefont {S.}~\bibnamefont {Kim}}, \bibinfo {author} {\bibfnamefont {C.~M.}\ \bibnamefont {Knaut}}, \bibinfo {author} {\bibfnamefont {B.}~\bibnamefont {Kobrin}}, \bibinfo {author} {\bibfnamefont {F.}~\bibnamefont {Kostritsa}}, \bibinfo {author} {\bibfnamefont {J.~M.}\ \bibnamefont {Kreikebaum}}, \bibinfo {author} {\bibfnamefont {R.}~\bibnamefont {Kudo}}, \bibinfo {author} {\bibfnamefont {B.}~\bibnamefont {Kueffler}},
  \bibinfo {author} {\bibfnamefont {A.}~\bibnamefont {Kumar}}, \bibinfo {author} {\bibfnamefont {V.~D.}\ \bibnamefont {Kurilovich}}, \bibinfo {author} {\bibfnamefont {V.}~\bibnamefont {Kutsko}}, \bibinfo {author} {\bibfnamefont {N.}~\bibnamefont {Lacroix}}, \bibinfo {author} {\bibfnamefont {D.}~\bibnamefont {Landhuis}}, \bibinfo {author} {\bibfnamefont {T.}~\bibnamefont {Lange-Dei}}, \bibinfo {author} {\bibfnamefont {B.~W.}\ \bibnamefont {Langley}}, \bibinfo {author} {\bibfnamefont {P.}~\bibnamefont {Laptev}}, \bibinfo {author} {\bibfnamefont {K.-M.}\ \bibnamefont {Lau}}, \bibinfo {author} {\bibfnamefont {J.}~\bibnamefont {Ledford}}, \bibinfo {author} {\bibfnamefont {J.}~\bibnamefont {Lee}}, \bibinfo {author} {\bibfnamefont {K.}~\bibnamefont {Lee}}, \bibinfo {author} {\bibfnamefont {B.~J.}\ \bibnamefont {Lester}}, \bibinfo {author} {\bibfnamefont {W.}~\bibnamefont {Leung}}, \bibinfo {author} {\bibfnamefont {L.}~\bibnamefont {Li}}, \bibinfo {author} {\bibfnamefont {W.~Y.}\ \bibnamefont {Li}}, \bibinfo {author}
  {\bibfnamefont {M.}~\bibnamefont {Li}}, \bibinfo {author} {\bibfnamefont {A.~T.}\ \bibnamefont {Lill}}, \bibinfo {author} {\bibfnamefont {W.~P.}\ \bibnamefont {Livingston}}, \bibinfo {author} {\bibfnamefont {M.~T.}\ \bibnamefont {Lloyd}}, \bibinfo {author} {\bibfnamefont {A.}~\bibnamefont {Locharla}}, \bibinfo {author} {\bibfnamefont {L.~D.}\ \bibnamefont {Lorenzo}}, \bibinfo {author} {\bibfnamefont {D.}~\bibnamefont {Lundahl}}, \bibinfo {author} {\bibfnamefont {A.}~\bibnamefont {Lunt}}, \bibinfo {author} {\bibfnamefont {S.}~\bibnamefont {Madhuk}}, \bibinfo {author} {\bibfnamefont {A.}~\bibnamefont {Maiti}}, \bibinfo {author} {\bibfnamefont {A.}~\bibnamefont {Maloney}}, \bibinfo {author} {\bibfnamefont {S.}~\bibnamefont {Mandrà}}, \bibinfo {author} {\bibfnamefont {L.~S.}\ \bibnamefont {Martin}}, \bibinfo {author} {\bibfnamefont {O.}~\bibnamefont {Martin}}, \bibinfo {author} {\bibfnamefont {E.}~\bibnamefont {Mascot}}, \bibinfo {author} {\bibfnamefont {P.~M.}\ \bibnamefont {Das}}, \bibinfo {author}
  {\bibfnamefont {D.}~\bibnamefont {Maslov}}, \bibinfo {author} {\bibfnamefont {M.}~\bibnamefont {Mathews}}, \bibinfo {author} {\bibfnamefont {C.}~\bibnamefont {Maxfield}}, \bibinfo {author} {\bibfnamefont {J.~R.}\ \bibnamefont {McClean}}, \bibinfo {author} {\bibfnamefont {M.}~\bibnamefont {McEwen}}, \bibinfo {author} {\bibfnamefont {S.}~\bibnamefont {Meeks}}, \bibinfo {author} {\bibfnamefont {K.~C.}\ \bibnamefont {Miao}}, \bibinfo {author} {\bibfnamefont {Z.~K.}\ \bibnamefont {Minev}}, \bibinfo {author} {\bibfnamefont {R.}~\bibnamefont {Molavi}}, \bibinfo {author} {\bibfnamefont {S.}~\bibnamefont {Molina}}, \bibinfo {author} {\bibfnamefont {S.}~\bibnamefont {Montazeri}}, \bibinfo {author} {\bibfnamefont {C.}~\bibnamefont {Neill}}, \bibinfo {author} {\bibfnamefont {M.}~\bibnamefont {Newman}}, \bibinfo {author} {\bibfnamefont {A.}~\bibnamefont {Nguyen}}, \bibinfo {author} {\bibfnamefont {M.}~\bibnamefont {Nguyen}}, \bibinfo {author} {\bibfnamefont {C.-H.}\ \bibnamefont {Ni}}, \bibinfo {author} {\bibfnamefont
  {M.~Y.}\ \bibnamefont {Niu}}, \bibinfo {author} {\bibfnamefont {L.}~\bibnamefont {Oas}}, \bibinfo {author} {\bibfnamefont {R.}~\bibnamefont {Orosco}}, \bibinfo {author} {\bibfnamefont {K.}~\bibnamefont {Ottosson}}, \bibinfo {author} {\bibfnamefont {A.}~\bibnamefont {Pagano}}, \bibinfo {author} {\bibfnamefont {A.~D.}\ \bibnamefont {Paolo}}, \bibinfo {author} {\bibfnamefont {S.}~\bibnamefont {Peek}}, \bibinfo {author} {\bibfnamefont {D.}~\bibnamefont {Peterson}}, \bibinfo {author} {\bibfnamefont {A.}~\bibnamefont {Pizzuto}}, \bibinfo {author} {\bibfnamefont {E.}~\bibnamefont {Portoles}}, \bibinfo {author} {\bibfnamefont {R.}~\bibnamefont {Potter}}, \bibinfo {author} {\bibfnamefont {O.}~\bibnamefont {Pritchard}}, \bibinfo {author} {\bibfnamefont {M.}~\bibnamefont {Qian}}, \bibinfo {author} {\bibfnamefont {C.}~\bibnamefont {Quintana}}, \bibinfo {author} {\bibfnamefont {A.}~\bibnamefont {Ranadive}}, \bibinfo {author} {\bibfnamefont {M.~J.}\ \bibnamefont {Reagor}}, \bibinfo {author} {\bibfnamefont
  {R.}~\bibnamefont {Resnick}}, \bibinfo {author} {\bibfnamefont {D.~M.}\ \bibnamefont {Rhodes}}, \bibinfo {author} {\bibfnamefont {D.}~\bibnamefont {Riley}}, \bibinfo {author} {\bibfnamefont {G.}~\bibnamefont {Roberts}}, \bibinfo {author} {\bibfnamefont {R.}~\bibnamefont {Rodriguez}}, \bibinfo {author} {\bibfnamefont {E.}~\bibnamefont {Ropes}}, \bibinfo {author} {\bibfnamefont {L.~B.~D.}\ \bibnamefont {Rose}}, \bibinfo {author} {\bibfnamefont {E.}~\bibnamefont {Rosenberg}}, \bibinfo {author} {\bibfnamefont {E.}~\bibnamefont {Rosenfeld}}, \bibinfo {author} {\bibfnamefont {D.}~\bibnamefont {Rosenstock}}, \bibinfo {author} {\bibfnamefont {E.}~\bibnamefont {Rossi}}, \bibinfo {author} {\bibfnamefont {P.}~\bibnamefont {Roushan}}, \bibinfo {author} {\bibfnamefont {D.~A.}\ \bibnamefont {Rower}}, \bibinfo {author} {\bibfnamefont {R.}~\bibnamefont {Salazar}}, \bibinfo {author} {\bibfnamefont {K.}~\bibnamefont {Sankaragomathi}}, \bibinfo {author} {\bibfnamefont {M.~C.}\ \bibnamefont {Sarihan}}, \bibinfo {author}
  {\bibfnamefont {K.~J.}\ \bibnamefont {Satzinger}}, \bibinfo {author} {\bibfnamefont {M.}~\bibnamefont {Schaefer}}, \bibinfo {author} {\bibfnamefont {S.}~\bibnamefont {Schroeder}}, \bibinfo {author} {\bibfnamefont {H.~F.}\ \bibnamefont {Schurkus}}, \bibinfo {author} {\bibfnamefont {A.}~\bibnamefont {Shahingohar}}, \bibinfo {author} {\bibfnamefont {M.~J.}\ \bibnamefont {Shearn}}, \bibinfo {author} {\bibfnamefont {A.}~\bibnamefont {Shorter}}, \bibinfo {author} {\bibfnamefont {V.}~\bibnamefont {Shvarts}}, \bibinfo {author} {\bibfnamefont {S.}~\bibnamefont {Small}}, \bibinfo {author} {\bibfnamefont {W.~C.}\ \bibnamefont {Smith}}, \bibinfo {author} {\bibfnamefont {D.~A.}\ \bibnamefont {Sobel}}, \bibinfo {author} {\bibfnamefont {B.}~\bibnamefont {Spells}}, \bibinfo {author} {\bibfnamefont {S.}~\bibnamefont {Springer}}, \bibinfo {author} {\bibfnamefont {G.}~\bibnamefont {Sterling}}, \bibinfo {author} {\bibfnamefont {J.}~\bibnamefont {Suchard}}, \bibinfo {author} {\bibfnamefont {A.}~\bibnamefont {Szasz}}, \bibinfo
  {author} {\bibfnamefont {A.}~\bibnamefont {Sztein}}, \bibinfo {author} {\bibfnamefont {M.}~\bibnamefont {Taylor}}, \bibinfo {author} {\bibfnamefont {J.~P.}\ \bibnamefont {Thiruraman}}, \bibinfo {author} {\bibfnamefont {D.}~\bibnamefont {Thor}}, \bibinfo {author} {\bibfnamefont {D.}~\bibnamefont {Timucin}}, \bibinfo {author} {\bibfnamefont {E.}~\bibnamefont {Tomita}}, \bibinfo {author} {\bibfnamefont {A.}~\bibnamefont {Torres}}, \bibinfo {author} {\bibfnamefont {M.~M.}\ \bibnamefont {Torunbalci}}, \bibinfo {author} {\bibfnamefont {H.}~\bibnamefont {Tran}}, \bibinfo {author} {\bibfnamefont {A.}~\bibnamefont {Vaishnav}}, \bibinfo {author} {\bibfnamefont {J.}~\bibnamefont {Vargas}}, \bibinfo {author} {\bibfnamefont {S.}~\bibnamefont {Vdovichev}}, \bibinfo {author} {\bibfnamefont {G.}~\bibnamefont {Vidal}}, \bibinfo {author} {\bibfnamefont {C.~V.}\ \bibnamefont {Heidweiller}}, \bibinfo {author} {\bibfnamefont {M.}~\bibnamefont {Voorhees}}, \bibinfo {author} {\bibfnamefont {S.}~\bibnamefont {Waltman}}, \bibinfo
  {author} {\bibfnamefont {J.}~\bibnamefont {Waltz}}, \bibinfo {author} {\bibfnamefont {S.~X.}\ \bibnamefont {Wang}}, \bibinfo {author} {\bibfnamefont {B.}~\bibnamefont {Ware}}, \bibinfo {author} {\bibfnamefont {J.~D.}\ \bibnamefont {Watson}}, \bibinfo {author} {\bibfnamefont {Y.}~\bibnamefont {Wei}}, \bibinfo {author} {\bibfnamefont {T.}~\bibnamefont {Weidel}}, \bibinfo {author} {\bibfnamefont {T.}~\bibnamefont {White}}, \bibinfo {author} {\bibfnamefont {K.}~\bibnamefont {Wong}}, \bibinfo {author} {\bibfnamefont {B.~W.~K.}\ \bibnamefont {Woo}}, \bibinfo {author} {\bibfnamefont {C.~J.}\ \bibnamefont {Wood}}, \bibinfo {author} {\bibfnamefont {M.}~\bibnamefont {Woodson}}, \bibinfo {author} {\bibfnamefont {C.}~\bibnamefont {Xing}}, \bibinfo {author} {\bibfnamefont {Z.~J.}\ \bibnamefont {Yao}}, \bibinfo {author} {\bibfnamefont {P.}~\bibnamefont {Yeh}}, \bibinfo {author} {\bibfnamefont {B.}~\bibnamefont {Ying}}, \bibinfo {author} {\bibfnamefont {J.}~\bibnamefont {Yoo}}, \bibinfo {author} {\bibfnamefont
  {N.}~\bibnamefont {Yosri}}, \bibinfo {author} {\bibfnamefont {E.}~\bibnamefont {Young}}, \bibinfo {author} {\bibfnamefont {G.}~\bibnamefont {Young}}, \bibinfo {author} {\bibfnamefont {A.}~\bibnamefont {Zalcman}}, \bibinfo {author} {\bibfnamefont {R.}~\bibnamefont {Zhang}}, \bibinfo {author} {\bibfnamefont {Y.}~\bibnamefont {Zhang}}, \bibinfo {author} {\bibfnamefont {N.}~\bibnamefont {Zhu}}, \bibinfo {author} {\bibfnamefont {N.}~\bibnamefont {Zobrist}}, \bibinfo {author} {\bibfnamefont {Z.}~\bibnamefont {Zou}}, \bibinfo {author} {\bibfnamefont {R.}~\bibnamefont {Babbush}}, \bibinfo {author} {\bibfnamefont {D.}~\bibnamefont {Bacon}}, \bibinfo {author} {\bibfnamefont {S.}~\bibnamefont {Boixo}}, \bibinfo {author} {\bibfnamefont {Y.}~\bibnamefont {Chen}}, \bibinfo {author} {\bibfnamefont {Z.}~\bibnamefont {Chen}}, \bibinfo {author} {\bibfnamefont {M.}~\bibnamefont {Devoret}}, \bibinfo {author} {\bibfnamefont {M.}~\bibnamefont {Hansen}}, \bibinfo {author} {\bibfnamefont {J.}~\bibnamefont {Hilton}}, \bibinfo
  {author} {\bibfnamefont {C.}~\bibnamefont {Jones}}, \bibinfo {author} {\bibfnamefont {J.}~\bibnamefont {Kelly}}, \bibinfo {author} {\bibfnamefont {A.~N.}\ \bibnamefont {Korotkov}}, \bibinfo {author} {\bibfnamefont {E.}~\bibnamefont {Lucero}}, \bibinfo {author} {\bibfnamefont {A.}~\bibnamefont {Megrant}}, \bibinfo {author} {\bibfnamefont {H.}~\bibnamefont {Neven}}, \bibinfo {author} {\bibfnamefont {W.~D.}\ \bibnamefont {Oliver}}, \bibinfo {author} {\bibfnamefont {G.}~\bibnamefont {Ramachandran}}, \bibinfo {author} {\bibfnamefont {V.}~\bibnamefont {Smelyanskiy}},\ and\ \bibinfo {author} {\bibfnamefont {P.~V.}\ \bibnamefont {Klimov}},\ }\bibfield  {title} {\bibinfo {title} {Reinforcement learning control of quantum error correction},\ }\href {https://doi.org/10.1038/s41586-026-10759-2} {\bibfield  {journal} {\bibinfo  {journal} {Nature}\ }\textbf {\bibinfo {volume} {655}},\ \bibinfo {pages} {879} (\bibinfo {year} {2026})}\BibitemShut {NoStop}%
\bibitem [{\citenamefont {Knill}\ \emph {et~al.}(1998)\citenamefont {Knill}, \citenamefont {Laflamme},\ and\ \citenamefont {Zurek}}]{resilient_qc}%
  \BibitemOpen
  \bibfield  {author} {\bibinfo {author} {\bibfnamefont {E.}~\bibnamefont {Knill}}, \bibinfo {author} {\bibfnamefont {R.}~\bibnamefont {Laflamme}},\ and\ \bibinfo {author} {\bibfnamefont {W.~H.}\ \bibnamefont {Zurek}},\ }\bibfield  {title} {\bibinfo {title} {Resilient quantum computation},\ }\href {https://doi.org/10.1126/science.279.5349.342} {\bibfield  {journal} {\bibinfo  {journal} {Science}\ }\textbf {\bibinfo {volume} {279}},\ \bibinfo {pages} {342} (\bibinfo {year} {1998})}\BibitemShut {NoStop}%
\bibitem [{\citenamefont {Aharonov}\ and\ \citenamefont {Ben-Or}(1997)}]{fault-tolerance}%
  \BibitemOpen
  \bibfield  {author} {\bibinfo {author} {\bibfnamefont {D.}~\bibnamefont {Aharonov}}\ and\ \bibinfo {author} {\bibfnamefont {M.}~\bibnamefont {Ben-Or}},\ }\bibfield  {title} {\bibinfo {title} {Fault-tolerant quantum computation with constant error},\ }in\ \href {https://doi.org/10.1145/258533.258579} {\emph {\bibinfo {booktitle} {Proceedings of the Twenty-Ninth Annual ACM Symposium on Theory of Computing}}},\ \bibinfo {series and number} {STOC '97}\ (\bibinfo  {publisher} {Association for Computing Machinery},\ \bibinfo {address} {New York, NY, USA},\ \bibinfo {year} {1997})\ p.\ \bibinfo {pages} {176–188}\BibitemShut {NoStop}%
\bibitem [{\citenamefont {Fowler}\ \emph {et~al.}(2009)\citenamefont {Fowler}, \citenamefont {Stephens},\ and\ \citenamefont {Groszkowski}}]{fowler2009high}%
  \BibitemOpen
  \bibfield  {author} {\bibinfo {author} {\bibfnamefont {A.~G.}\ \bibnamefont {Fowler}}, \bibinfo {author} {\bibfnamefont {A.~M.}\ \bibnamefont {Stephens}},\ and\ \bibinfo {author} {\bibfnamefont {P.}~\bibnamefont {Groszkowski}},\ }\bibfield  {title} {\bibinfo {title} {High-threshold universal quantum computation on the surface code},\ }\href {https://journals.aps.org/pra/abstract/10.1103/PhysRevA.80.052312} {\bibfield  {journal} {\bibinfo  {journal} {Physical Review A}\ }\textbf {\bibinfo {volume} {80}},\ \bibinfo {pages} {052312} (\bibinfo {year} {2009})}\BibitemShut {NoStop}%
\bibitem [{\citenamefont {{Google Quantum AI and Collaborators}}(2025)}]{google2025quantum}%
  \BibitemOpen
  \bibfield  {author} {\bibinfo {author} {\bibnamefont {{Google Quantum AI and Collaborators}}},\ }\bibfield  {title} {\bibinfo {title} {Quantum error correction below the surface code threshold},\ }\href {https://doi.org/10.1038/s41586-024-08449-y} {\bibfield  {journal} {\bibinfo  {journal} {Nature}\ }\textbf {\bibinfo {volume} {638}},\ \bibinfo {pages} {920} (\bibinfo {year} {2025})}\BibitemShut {NoStop}%
\bibitem [{\citenamefont {Bravyi}\ \emph {et~al.}(2024)\citenamefont {Bravyi}, \citenamefont {Cross}, \citenamefont {Gambetta}, \citenamefont {Maslov}, \citenamefont {Rall},\ and\ \citenamefont {Yoder}}]{bravyi2024high}%
  \BibitemOpen
  \bibfield  {author} {\bibinfo {author} {\bibfnamefont {S.}~\bibnamefont {Bravyi}}, \bibinfo {author} {\bibfnamefont {A.~W.}\ \bibnamefont {Cross}}, \bibinfo {author} {\bibfnamefont {J.~M.}\ \bibnamefont {Gambetta}}, \bibinfo {author} {\bibfnamefont {D.}~\bibnamefont {Maslov}}, \bibinfo {author} {\bibfnamefont {P.}~\bibnamefont {Rall}},\ and\ \bibinfo {author} {\bibfnamefont {T.~J.}\ \bibnamefont {Yoder}},\ }\bibfield  {title} {\bibinfo {title} {High-threshold and low-overhead fault-tolerant quantum memory},\ }\href {https://www.nature.com/articles/s41586-024-07107-7} {\bibfield  {journal} {\bibinfo  {journal} {Nature}\ }\textbf {\bibinfo {volume} {627}},\ \bibinfo {pages} {778} (\bibinfo {year} {2024})}\BibitemShut {NoStop}%
\bibitem [{\citenamefont {Bluvstein}\ \emph {et~al.}(2024)\citenamefont {Bluvstein}, \citenamefont {Evered}, \citenamefont {Geim}, \citenamefont {Li}, \citenamefont {Zhou}, \citenamefont {Manovitz}, \citenamefont {Ebadi}, \citenamefont {Cain}, \citenamefont {Kalinowski}, \citenamefont {Hangleiter}, \citenamefont {Bonilla~Ataides}, \citenamefont {Maskara}, \citenamefont {Cong}, \citenamefont {Gao}, \citenamefont {Sales~Rodriguez}, \citenamefont {Karolyshyn}, \citenamefont {Semeghini}, \citenamefont {Gullans}, \citenamefont {Greiner}, \citenamefont {Vuleti{\'c}},\ and\ \citenamefont {Lukin}}]{harvard2024}%
  \BibitemOpen
  \bibfield  {author} {\bibinfo {author} {\bibfnamefont {D.}~\bibnamefont {Bluvstein}}, \bibinfo {author} {\bibfnamefont {S.~J.}\ \bibnamefont {Evered}}, \bibinfo {author} {\bibfnamefont {A.~A.}\ \bibnamefont {Geim}}, \bibinfo {author} {\bibfnamefont {S.~H.}\ \bibnamefont {Li}}, \bibinfo {author} {\bibfnamefont {H.}~\bibnamefont {Zhou}}, \bibinfo {author} {\bibfnamefont {T.}~\bibnamefont {Manovitz}}, \bibinfo {author} {\bibfnamefont {S.}~\bibnamefont {Ebadi}}, \bibinfo {author} {\bibfnamefont {M.}~\bibnamefont {Cain}}, \bibinfo {author} {\bibfnamefont {M.}~\bibnamefont {Kalinowski}}, \bibinfo {author} {\bibfnamefont {D.}~\bibnamefont {Hangleiter}}, \bibinfo {author} {\bibfnamefont {J.~P.}\ \bibnamefont {Bonilla~Ataides}}, \bibinfo {author} {\bibfnamefont {N.}~\bibnamefont {Maskara}}, \bibinfo {author} {\bibfnamefont {I.}~\bibnamefont {Cong}}, \bibinfo {author} {\bibfnamefont {X.}~\bibnamefont {Gao}}, \bibinfo {author} {\bibfnamefont {P.}~\bibnamefont {Sales~Rodriguez}}, \bibinfo {author} {\bibfnamefont
  {T.}~\bibnamefont {Karolyshyn}}, \bibinfo {author} {\bibfnamefont {G.}~\bibnamefont {Semeghini}}, \bibinfo {author} {\bibfnamefont {M.~J.}\ \bibnamefont {Gullans}}, \bibinfo {author} {\bibfnamefont {M.}~\bibnamefont {Greiner}}, \bibinfo {author} {\bibfnamefont {V.}~\bibnamefont {Vuleti{\'c}}},\ and\ \bibinfo {author} {\bibfnamefont {M.~D.}\ \bibnamefont {Lukin}},\ }\bibfield  {title} {\bibinfo {title} {Logical quantum processor based on reconfigurable atom arrays},\ }\href {https://doi.org/10.1038/s41586-023-06927-3} {\bibfield  {journal} {\bibinfo  {journal} {Nature}\ }\textbf {\bibinfo {volume} {626}},\ \bibinfo {pages} {58} (\bibinfo {year} {2024})}\BibitemShut {NoStop}%
\bibitem [{\citenamefont {Bluvstein}\ \emph {et~al.}(2026)\citenamefont {Bluvstein}, \citenamefont {Geim}, \citenamefont {Li}, \citenamefont {Evered}, \citenamefont {Bonilla~Ataides}, \citenamefont {Baranes}, \citenamefont {Gu}, \citenamefont {Manovitz}, \citenamefont {Xu}, \citenamefont {Kalinowski}, \citenamefont {Majidy}, \citenamefont {Kokail}, \citenamefont {Maskara}, \citenamefont {Trapp}, \citenamefont {Stewart}, \citenamefont {Hollerith}, \citenamefont {Zhou}, \citenamefont {Gullans}, \citenamefont {Yelin}, \citenamefont {Greiner}, \citenamefont {Vuleti{\'c}}, \citenamefont {Cain},\ and\ \citenamefont {Lukin}}]{harvard2026}%
  \BibitemOpen
  \bibfield  {author} {\bibinfo {author} {\bibfnamefont {D.}~\bibnamefont {Bluvstein}}, \bibinfo {author} {\bibfnamefont {A.~A.}\ \bibnamefont {Geim}}, \bibinfo {author} {\bibfnamefont {S.~H.}\ \bibnamefont {Li}}, \bibinfo {author} {\bibfnamefont {S.~J.}\ \bibnamefont {Evered}}, \bibinfo {author} {\bibfnamefont {J.~P.}\ \bibnamefont {Bonilla~Ataides}}, \bibinfo {author} {\bibfnamefont {G.}~\bibnamefont {Baranes}}, \bibinfo {author} {\bibfnamefont {A.}~\bibnamefont {Gu}}, \bibinfo {author} {\bibfnamefont {T.}~\bibnamefont {Manovitz}}, \bibinfo {author} {\bibfnamefont {M.}~\bibnamefont {Xu}}, \bibinfo {author} {\bibfnamefont {M.}~\bibnamefont {Kalinowski}}, \bibinfo {author} {\bibfnamefont {S.}~\bibnamefont {Majidy}}, \bibinfo {author} {\bibfnamefont {C.}~\bibnamefont {Kokail}}, \bibinfo {author} {\bibfnamefont {N.}~\bibnamefont {Maskara}}, \bibinfo {author} {\bibfnamefont {E.~C.}\ \bibnamefont {Trapp}}, \bibinfo {author} {\bibfnamefont {L.~M.}\ \bibnamefont {Stewart}}, \bibinfo {author} {\bibfnamefont
  {S.}~\bibnamefont {Hollerith}}, \bibinfo {author} {\bibfnamefont {H.}~\bibnamefont {Zhou}}, \bibinfo {author} {\bibfnamefont {M.~J.}\ \bibnamefont {Gullans}}, \bibinfo {author} {\bibfnamefont {S.~F.}\ \bibnamefont {Yelin}}, \bibinfo {author} {\bibfnamefont {M.}~\bibnamefont {Greiner}}, \bibinfo {author} {\bibfnamefont {V.}~\bibnamefont {Vuleti{\'c}}}, \bibinfo {author} {\bibfnamefont {M.}~\bibnamefont {Cain}},\ and\ \bibinfo {author} {\bibfnamefont {M.~D.}\ \bibnamefont {Lukin}},\ }\bibfield  {title} {\bibinfo {title} {A fault-tolerant neutral-atom architecture for universal quantum computation},\ }\href {https://doi.org/10.1038/s41586-025-09848-5} {\bibfield  {journal} {\bibinfo  {journal} {Nature}\ }\textbf {\bibinfo {volume} {649}},\ \bibinfo {pages} {39} (\bibinfo {year} {2026})}\BibitemShut {NoStop}%
\bibitem [{\citenamefont {Chen}\ \emph {et~al.}(2021)\citenamefont {Chen}, \citenamefont {Satzinger}, \citenamefont {Atalaya}, \citenamefont {Korotkov}, \citenamefont {Dunsworth}, \citenamefont {Sank}, \citenamefont {Quintana}, \citenamefont {McEwen}, \citenamefont {Barends}, \citenamefont {Klimov}, \citenamefont {Hong}, \citenamefont {Jones}, \citenamefont {Petukhov}, \citenamefont {Kafri}, \citenamefont {Demura}, \citenamefont {Burkett}, \citenamefont {Gidney}, \citenamefont {Fowler}, \citenamefont {Paler}, \citenamefont {Putterman}, \citenamefont {Aleiner}, \citenamefont {Arute}, \citenamefont {Arya}, \citenamefont {Babbush}, \citenamefont {Bardin}, \citenamefont {Bengtsson}, \citenamefont {Bourassa}, \citenamefont {Broughton}, \citenamefont {Buckley}, \citenamefont {Buell}, \citenamefont {Bushnell}, \citenamefont {Chiaro}, \citenamefont {Collins}, \citenamefont {Courtney}, \citenamefont {Derk}, \citenamefont {Eppens}, \citenamefont {Erickson}, \citenamefont {Farhi}, \citenamefont {Foxen}, \citenamefont
  {Giustina}, \citenamefont {Greene}, \citenamefont {Gross}, \citenamefont {Harrigan}, \citenamefont {Harrington}, \citenamefont {Hilton}, \citenamefont {Ho}, \citenamefont {Huang}, \citenamefont {Huggins}, \citenamefont {Ioffe}, \citenamefont {Isakov}, \citenamefont {Jeffrey}, \citenamefont {Jiang}, \citenamefont {Kechedzhi}, \citenamefont {Kim}, \citenamefont {Kitaev}, \citenamefont {Kostritsa}, \citenamefont {Landhuis}, \citenamefont {Laptev}, \citenamefont {Lucero}, \citenamefont {Martin}, \citenamefont {McClean}, \citenamefont {McCourt}, \citenamefont {Mi}, \citenamefont {Miao}, \citenamefont {Mohseni}, \citenamefont {Montazeri}, \citenamefont {Mruczkiewicz}, \citenamefont {Mutus}, \citenamefont {Naaman}, \citenamefont {Neeley}, \citenamefont {Neill}, \citenamefont {Newman}, \citenamefont {Niu}, \citenamefont {O'Brien}, \citenamefont {Opremcak}, \citenamefont {Ostby}, \citenamefont {Pat{\'o}}, \citenamefont {Redd}, \citenamefont {Roushan}, \citenamefont {Rubin}, \citenamefont {Shvarts}, \citenamefont
  {Strain}, \citenamefont {Szalay}, \citenamefont {Trevithick}, \citenamefont {Villalonga}, \citenamefont {White}, \citenamefont {Yao}, \citenamefont {Yeh}, \citenamefont {Yoo}, \citenamefont {Zalcman}, \citenamefont {Neven}, \citenamefont {Boixo}, \citenamefont {Smelyanskiy}, \citenamefont {Chen}, \citenamefont {Megrant}, \citenamefont {Kelly},\ and\ \citenamefont {{Google Quantum AI}}}]{google2021}%
  \BibitemOpen
  \bibfield  {author} {\bibinfo {author} {\bibfnamefont {Z.}~\bibnamefont {Chen}}, \bibinfo {author} {\bibfnamefont {K.~J.}\ \bibnamefont {Satzinger}}, \bibinfo {author} {\bibfnamefont {J.}~\bibnamefont {Atalaya}}, \bibinfo {author} {\bibfnamefont {A.~N.}\ \bibnamefont {Korotkov}}, \bibinfo {author} {\bibfnamefont {A.}~\bibnamefont {Dunsworth}}, \bibinfo {author} {\bibfnamefont {D.}~\bibnamefont {Sank}}, \bibinfo {author} {\bibfnamefont {C.}~\bibnamefont {Quintana}}, \bibinfo {author} {\bibfnamefont {M.}~\bibnamefont {McEwen}}, \bibinfo {author} {\bibfnamefont {R.}~\bibnamefont {Barends}}, \bibinfo {author} {\bibfnamefont {P.~V.}\ \bibnamefont {Klimov}}, \bibinfo {author} {\bibfnamefont {S.}~\bibnamefont {Hong}}, \bibinfo {author} {\bibfnamefont {C.}~\bibnamefont {Jones}}, \bibinfo {author} {\bibfnamefont {A.}~\bibnamefont {Petukhov}}, \bibinfo {author} {\bibfnamefont {D.}~\bibnamefont {Kafri}}, \bibinfo {author} {\bibfnamefont {S.}~\bibnamefont {Demura}}, \bibinfo {author} {\bibfnamefont {B.}~\bibnamefont
  {Burkett}}, \bibinfo {author} {\bibfnamefont {C.}~\bibnamefont {Gidney}}, \bibinfo {author} {\bibfnamefont {A.~G.}\ \bibnamefont {Fowler}}, \bibinfo {author} {\bibfnamefont {A.}~\bibnamefont {Paler}}, \bibinfo {author} {\bibfnamefont {H.}~\bibnamefont {Putterman}}, \bibinfo {author} {\bibfnamefont {I.}~\bibnamefont {Aleiner}}, \bibinfo {author} {\bibfnamefont {F.}~\bibnamefont {Arute}}, \bibinfo {author} {\bibfnamefont {K.}~\bibnamefont {Arya}}, \bibinfo {author} {\bibfnamefont {R.}~\bibnamefont {Babbush}}, \bibinfo {author} {\bibfnamefont {J.~C.}\ \bibnamefont {Bardin}}, \bibinfo {author} {\bibfnamefont {A.}~\bibnamefont {Bengtsson}}, \bibinfo {author} {\bibfnamefont {A.}~\bibnamefont {Bourassa}}, \bibinfo {author} {\bibfnamefont {M.}~\bibnamefont {Broughton}}, \bibinfo {author} {\bibfnamefont {B.~B.}\ \bibnamefont {Buckley}}, \bibinfo {author} {\bibfnamefont {D.~A.}\ \bibnamefont {Buell}}, \bibinfo {author} {\bibfnamefont {N.}~\bibnamefont {Bushnell}}, \bibinfo {author} {\bibfnamefont {B.}~\bibnamefont
  {Chiaro}}, \bibinfo {author} {\bibfnamefont {R.}~\bibnamefont {Collins}}, \bibinfo {author} {\bibfnamefont {W.}~\bibnamefont {Courtney}}, \bibinfo {author} {\bibfnamefont {A.~R.}\ \bibnamefont {Derk}}, \bibinfo {author} {\bibfnamefont {D.}~\bibnamefont {Eppens}}, \bibinfo {author} {\bibfnamefont {C.}~\bibnamefont {Erickson}}, \bibinfo {author} {\bibfnamefont {E.}~\bibnamefont {Farhi}}, \bibinfo {author} {\bibfnamefont {B.}~\bibnamefont {Foxen}}, \bibinfo {author} {\bibfnamefont {M.}~\bibnamefont {Giustina}}, \bibinfo {author} {\bibfnamefont {A.}~\bibnamefont {Greene}}, \bibinfo {author} {\bibfnamefont {J.~A.}\ \bibnamefont {Gross}}, \bibinfo {author} {\bibfnamefont {M.~P.}\ \bibnamefont {Harrigan}}, \bibinfo {author} {\bibfnamefont {S.~D.}\ \bibnamefont {Harrington}}, \bibinfo {author} {\bibfnamefont {J.}~\bibnamefont {Hilton}}, \bibinfo {author} {\bibfnamefont {A.}~\bibnamefont {Ho}}, \bibinfo {author} {\bibfnamefont {T.}~\bibnamefont {Huang}}, \bibinfo {author} {\bibfnamefont {W.~J.}\ \bibnamefont
  {Huggins}}, \bibinfo {author} {\bibfnamefont {L.~B.}\ \bibnamefont {Ioffe}}, \bibinfo {author} {\bibfnamefont {S.~V.}\ \bibnamefont {Isakov}}, \bibinfo {author} {\bibfnamefont {E.}~\bibnamefont {Jeffrey}}, \bibinfo {author} {\bibfnamefont {Z.}~\bibnamefont {Jiang}}, \bibinfo {author} {\bibfnamefont {K.}~\bibnamefont {Kechedzhi}}, \bibinfo {author} {\bibfnamefont {S.}~\bibnamefont {Kim}}, \bibinfo {author} {\bibfnamefont {A.}~\bibnamefont {Kitaev}}, \bibinfo {author} {\bibfnamefont {F.}~\bibnamefont {Kostritsa}}, \bibinfo {author} {\bibfnamefont {D.}~\bibnamefont {Landhuis}}, \bibinfo {author} {\bibfnamefont {P.}~\bibnamefont {Laptev}}, \bibinfo {author} {\bibfnamefont {E.}~\bibnamefont {Lucero}}, \bibinfo {author} {\bibfnamefont {O.}~\bibnamefont {Martin}}, \bibinfo {author} {\bibfnamefont {J.~R.}\ \bibnamefont {McClean}}, \bibinfo {author} {\bibfnamefont {T.}~\bibnamefont {McCourt}}, \bibinfo {author} {\bibfnamefont {X.}~\bibnamefont {Mi}}, \bibinfo {author} {\bibfnamefont {K.~C.}\ \bibnamefont {Miao}},
  \bibinfo {author} {\bibfnamefont {M.}~\bibnamefont {Mohseni}}, \bibinfo {author} {\bibfnamefont {S.}~\bibnamefont {Montazeri}}, \bibinfo {author} {\bibfnamefont {W.}~\bibnamefont {Mruczkiewicz}}, \bibinfo {author} {\bibfnamefont {J.}~\bibnamefont {Mutus}}, \bibinfo {author} {\bibfnamefont {O.}~\bibnamefont {Naaman}}, \bibinfo {author} {\bibfnamefont {M.}~\bibnamefont {Neeley}}, \bibinfo {author} {\bibfnamefont {C.}~\bibnamefont {Neill}}, \bibinfo {author} {\bibfnamefont {M.}~\bibnamefont {Newman}}, \bibinfo {author} {\bibfnamefont {M.~Y.}\ \bibnamefont {Niu}}, \bibinfo {author} {\bibfnamefont {T.~E.}\ \bibnamefont {O'Brien}}, \bibinfo {author} {\bibfnamefont {A.}~\bibnamefont {Opremcak}}, \bibinfo {author} {\bibfnamefont {E.}~\bibnamefont {Ostby}}, \bibinfo {author} {\bibfnamefont {B.}~\bibnamefont {Pat{\'o}}}, \bibinfo {author} {\bibfnamefont {N.}~\bibnamefont {Redd}}, \bibinfo {author} {\bibfnamefont {P.}~\bibnamefont {Roushan}}, \bibinfo {author} {\bibfnamefont {N.~C.}\ \bibnamefont {Rubin}}, \bibinfo
  {author} {\bibfnamefont {V.}~\bibnamefont {Shvarts}}, \bibinfo {author} {\bibfnamefont {D.}~\bibnamefont {Strain}}, \bibinfo {author} {\bibfnamefont {M.}~\bibnamefont {Szalay}}, \bibinfo {author} {\bibfnamefont {M.~D.}\ \bibnamefont {Trevithick}}, \bibinfo {author} {\bibfnamefont {B.}~\bibnamefont {Villalonga}}, \bibinfo {author} {\bibfnamefont {T.}~\bibnamefont {White}}, \bibinfo {author} {\bibfnamefont {Z.~J.}\ \bibnamefont {Yao}}, \bibinfo {author} {\bibfnamefont {P.}~\bibnamefont {Yeh}}, \bibinfo {author} {\bibfnamefont {J.}~\bibnamefont {Yoo}}, \bibinfo {author} {\bibfnamefont {A.}~\bibnamefont {Zalcman}}, \bibinfo {author} {\bibfnamefont {H.}~\bibnamefont {Neven}}, \bibinfo {author} {\bibfnamefont {S.}~\bibnamefont {Boixo}}, \bibinfo {author} {\bibfnamefont {V.}~\bibnamefont {Smelyanskiy}}, \bibinfo {author} {\bibfnamefont {Y.}~\bibnamefont {Chen}}, \bibinfo {author} {\bibfnamefont {A.}~\bibnamefont {Megrant}}, \bibinfo {author} {\bibfnamefont {J.}~\bibnamefont {Kelly}},\ and\ \bibinfo {author}
  {\bibnamefont {{Google Quantum AI}}},\ }\bibfield  {title} {\bibinfo {title} {Exponential suppression of bit or phase errors with cyclic error correction},\ }\href {https://doi.org/10.1038/s41586-021-03588-y} {\bibfield  {journal} {\bibinfo  {journal} {Nature}\ }\textbf {\bibinfo {volume} {595}},\ \bibinfo {pages} {383} (\bibinfo {year} {2021})}\BibitemShut {NoStop}%
\bibitem [{\citenamefont {Paetznick}\ \emph {et~al.}(2024)\citenamefont {Paetznick}, \citenamefont {Da~Silva}, \citenamefont {Ryan-Anderson}, \citenamefont {Bello-Rivas}, \citenamefont {Campora~III}, \citenamefont {Chernoguzov}, \citenamefont {Dreiling}, \citenamefont {Foltz}, \citenamefont {Frachon}, \citenamefont {Gaebler} \emph {et~al.}}]{2024arXiv240402280P}%
  \BibitemOpen
  \bibfield  {author} {\bibinfo {author} {\bibfnamefont {A.}~\bibnamefont {Paetznick}}, \bibinfo {author} {\bibfnamefont {M.}~\bibnamefont {Da~Silva}}, \bibinfo {author} {\bibfnamefont {C.}~\bibnamefont {Ryan-Anderson}}, \bibinfo {author} {\bibfnamefont {J.}~\bibnamefont {Bello-Rivas}}, \bibinfo {author} {\bibfnamefont {J.}~\bibnamefont {Campora~III}}, \bibinfo {author} {\bibfnamefont {A.}~\bibnamefont {Chernoguzov}}, \bibinfo {author} {\bibfnamefont {J.}~\bibnamefont {Dreiling}}, \bibinfo {author} {\bibfnamefont {C.}~\bibnamefont {Foltz}}, \bibinfo {author} {\bibfnamefont {F.}~\bibnamefont {Frachon}}, \bibinfo {author} {\bibfnamefont {J.}~\bibnamefont {Gaebler}}, \emph {et~al.},\ }\bibfield  {title} {\bibinfo {title} {Demonstration of logical qubits and repeated error correction with better-than-physical error rates},\ }\href {https://arxiv.org/abs/2404.02280} {\bibfield  {journal} {\bibinfo  {journal} {arXiv:2404.02280}\ } (\bibinfo {year} {2024})}\BibitemShut {NoStop}%
\bibitem [{\citenamefont {Chow}\ \emph {et~al.}(2010)\citenamefont {Chow}, \citenamefont {DiCarlo}, \citenamefont {Gambetta}, \citenamefont {Motzoi}, \citenamefont {Frunzio}, \citenamefont {Girvin},\ and\ \citenamefont {Schoelkopf}}]{PhysRevA.82.040305}%
  \BibitemOpen
  \bibfield  {author} {\bibinfo {author} {\bibfnamefont {J.~M.}\ \bibnamefont {Chow}}, \bibinfo {author} {\bibfnamefont {L.}~\bibnamefont {DiCarlo}}, \bibinfo {author} {\bibfnamefont {J.~M.}\ \bibnamefont {Gambetta}}, \bibinfo {author} {\bibfnamefont {F.}~\bibnamefont {Motzoi}}, \bibinfo {author} {\bibfnamefont {L.}~\bibnamefont {Frunzio}}, \bibinfo {author} {\bibfnamefont {S.~M.}\ \bibnamefont {Girvin}},\ and\ \bibinfo {author} {\bibfnamefont {R.~J.}\ \bibnamefont {Schoelkopf}},\ }\bibfield  {title} {\bibinfo {title} {Optimized driving of superconducting artificial atoms for improved single-qubit gates},\ }\href {https://doi.org/10.1103/PhysRevA.82.040305} {\bibfield  {journal} {\bibinfo  {journal} {Physical Review A}\ }\textbf {\bibinfo {volume} {82}},\ \bibinfo {pages} {040305(R)} (\bibinfo {year} {2010})}\BibitemShut {NoStop}%
\bibitem [{\citenamefont {Chen}\ \emph {et~al.}(2016)\citenamefont {Chen}, \citenamefont {Kelly}, \citenamefont {Quintana}, \citenamefont {Barends}, \citenamefont {Campbell}, \citenamefont {Chen}, \citenamefont {Chiaro}, \citenamefont {Dunsworth}, \citenamefont {Fowler}, \citenamefont {Lucero}, \citenamefont {Jeffrey}, \citenamefont {Megrant}, \citenamefont {Mutus}, \citenamefont {Neeley}, \citenamefont {Neill}, \citenamefont {O'Malley}, \citenamefont {Roushan}, \citenamefont {Sank}, \citenamefont {Vainsencher}, \citenamefont {Wenner}, \citenamefont {White}, \citenamefont {Korotkov},\ and\ \citenamefont {Martinis}}]{PhysRevLett.116.020501}%
  \BibitemOpen
  \bibfield  {author} {\bibinfo {author} {\bibfnamefont {Z.}~\bibnamefont {Chen}}, \bibinfo {author} {\bibfnamefont {J.}~\bibnamefont {Kelly}}, \bibinfo {author} {\bibfnamefont {C.}~\bibnamefont {Quintana}}, \bibinfo {author} {\bibfnamefont {R.}~\bibnamefont {Barends}}, \bibinfo {author} {\bibfnamefont {B.}~\bibnamefont {Campbell}}, \bibinfo {author} {\bibfnamefont {Y.}~\bibnamefont {Chen}}, \bibinfo {author} {\bibfnamefont {B.}~\bibnamefont {Chiaro}}, \bibinfo {author} {\bibfnamefont {A.}~\bibnamefont {Dunsworth}}, \bibinfo {author} {\bibfnamefont {A.~G.}\ \bibnamefont {Fowler}}, \bibinfo {author} {\bibfnamefont {E.}~\bibnamefont {Lucero}}, \bibinfo {author} {\bibfnamefont {E.}~\bibnamefont {Jeffrey}}, \bibinfo {author} {\bibfnamefont {A.}~\bibnamefont {Megrant}}, \bibinfo {author} {\bibfnamefont {J.}~\bibnamefont {Mutus}}, \bibinfo {author} {\bibfnamefont {M.}~\bibnamefont {Neeley}}, \bibinfo {author} {\bibfnamefont {C.}~\bibnamefont {Neill}}, \bibinfo {author} {\bibfnamefont {P.~J.~J.}\ \bibnamefont
  {O'Malley}}, \bibinfo {author} {\bibfnamefont {P.}~\bibnamefont {Roushan}}, \bibinfo {author} {\bibfnamefont {D.}~\bibnamefont {Sank}}, \bibinfo {author} {\bibfnamefont {A.}~\bibnamefont {Vainsencher}}, \bibinfo {author} {\bibfnamefont {J.}~\bibnamefont {Wenner}}, \bibinfo {author} {\bibfnamefont {T.~C.}\ \bibnamefont {White}}, \bibinfo {author} {\bibfnamefont {A.~N.}\ \bibnamefont {Korotkov}},\ and\ \bibinfo {author} {\bibfnamefont {J.~M.}\ \bibnamefont {Martinis}},\ }\bibfield  {title} {\bibinfo {title} {Measuring and suppressing quantum state leakage in a superconducting qubit},\ }\href {https://doi.org/10.1103/PhysRevLett.116.020501} {\bibfield  {journal} {\bibinfo  {journal} {Physical Review Letters}\ }\textbf {\bibinfo {volume} {116}},\ \bibinfo {pages} {020501} (\bibinfo {year} {2016})}\BibitemShut {NoStop}%
\bibitem [{\citenamefont {Hellings}\ \emph {et~al.}(2025)\citenamefont {Hellings}, \citenamefont {Lacroix}, \citenamefont {Remm}, \citenamefont {Boell}, \citenamefont {Herrmann}, \citenamefont {Laz\ifmmode~\u{a}\else \u{a}\fi{}r}, \citenamefont {Krinner}, \citenamefont {Swiadek}, \citenamefont {Andersen}, \citenamefont {Eichler},\ and\ \citenamefont {Wallraff}}]{1qhb-r4fb}%
  \BibitemOpen
  \bibfield  {author} {\bibinfo {author} {\bibfnamefont {C.}~\bibnamefont {Hellings}}, \bibinfo {author} {\bibfnamefont {N.}~\bibnamefont {Lacroix}}, \bibinfo {author} {\bibfnamefont {A.}~\bibnamefont {Remm}}, \bibinfo {author} {\bibfnamefont {R.}~\bibnamefont {Boell}}, \bibinfo {author} {\bibfnamefont {J.}~\bibnamefont {Herrmann}}, \bibinfo {author} {\bibfnamefont {S.}~\bibnamefont {Laz\ifmmode~\u{a}\else \u{a}\fi{}r}}, \bibinfo {author} {\bibfnamefont {S.}~\bibnamefont {Krinner}}, \bibinfo {author} {\bibfnamefont {F.}~\bibnamefont {Swiadek}}, \bibinfo {author} {\bibfnamefont {C.~K.}\ \bibnamefont {Andersen}}, \bibinfo {author} {\bibfnamefont {C.}~\bibnamefont {Eichler}},\ and\ \bibinfo {author} {\bibfnamefont {A.}~\bibnamefont {Wallraff}},\ }\bibfield  {title} {\bibinfo {title} {Calibrating magnetic flux control in superconducting circuits by compensating distortions on timescales from nanoseconds up to tens of microseconds},\ }\href {https://doi.org/10.1103/1qhb-r4fb} {\bibfield  {journal} {\bibinfo
  {journal} {Physical Review Research}\ }\textbf {\bibinfo {volume} {7}},\ \bibinfo {pages} {043142} (\bibinfo {year} {2025})}\BibitemShut {NoStop}%
\bibitem [{\citenamefont {Chen}\ \emph {et~al.}(2025{\natexlab{a}})\citenamefont {Chen}, \citenamefont {Lee}, \citenamefont {Liu}, \citenamefont {Marinelli}, \citenamefont {Naik}, \citenamefont {Kang}, \citenamefont {Goss}, \citenamefont {Kim}, \citenamefont {Santiago},\ and\ \citenamefont {Siddiqi}}]{2025arXiv250304702C}%
  \BibitemOpen
  \bibfield  {author} {\bibinfo {author} {\bibfnamefont {L.}~\bibnamefont {Chen}}, \bibinfo {author} {\bibfnamefont {K.-H.}\ \bibnamefont {Lee}}, \bibinfo {author} {\bibfnamefont {C.-H.}\ \bibnamefont {Liu}}, \bibinfo {author} {\bibfnamefont {B.}~\bibnamefont {Marinelli}}, \bibinfo {author} {\bibfnamefont {R.~K.}\ \bibnamefont {Naik}}, \bibinfo {author} {\bibfnamefont {Z.}~\bibnamefont {Kang}}, \bibinfo {author} {\bibfnamefont {N.}~\bibnamefont {Goss}}, \bibinfo {author} {\bibfnamefont {H.}~\bibnamefont {Kim}}, \bibinfo {author} {\bibfnamefont {D.~I.}\ \bibnamefont {Santiago}},\ and\ \bibinfo {author} {\bibfnamefont {I.}~\bibnamefont {Siddiqi}},\ }\bibfield  {title} {\bibinfo {title} {Scalable and site-specific frequency tuning of two-level system defects in superconducting qubit arrays},\ }\href {https://arxiv.org/abs/2503.04702} {\bibfield  {journal} {\bibinfo  {journal} {arXiv:2503.04702}\ } (\bibinfo {year} {2025}{\natexlab{a}})}\BibitemShut {NoStop}%
\bibitem [{\citenamefont {Evered}\ \emph {et~al.}(2023)\citenamefont {Evered}, \citenamefont {Bluvstein}, \citenamefont {Kalinowski}, \citenamefont {Ebadi}, \citenamefont {Manovitz}, \citenamefont {Zhou}, \citenamefont {Li}, \citenamefont {Geim}, \citenamefont {Wang}, \citenamefont {Maskara}, \citenamefont {Levine}, \citenamefont {Semeghini}, \citenamefont {Greiner}, \citenamefont {Vuleti{\'c}},\ and\ \citenamefont {Lukin}}]{995}%
  \BibitemOpen
  \bibfield  {author} {\bibinfo {author} {\bibfnamefont {S.~J.}\ \bibnamefont {Evered}}, \bibinfo {author} {\bibfnamefont {D.}~\bibnamefont {Bluvstein}}, \bibinfo {author} {\bibfnamefont {M.}~\bibnamefont {Kalinowski}}, \bibinfo {author} {\bibfnamefont {S.}~\bibnamefont {Ebadi}}, \bibinfo {author} {\bibfnamefont {T.}~\bibnamefont {Manovitz}}, \bibinfo {author} {\bibfnamefont {H.}~\bibnamefont {Zhou}}, \bibinfo {author} {\bibfnamefont {S.~H.}\ \bibnamefont {Li}}, \bibinfo {author} {\bibfnamefont {A.~A.}\ \bibnamefont {Geim}}, \bibinfo {author} {\bibfnamefont {T.~T.}\ \bibnamefont {Wang}}, \bibinfo {author} {\bibfnamefont {N.}~\bibnamefont {Maskara}}, \bibinfo {author} {\bibfnamefont {H.}~\bibnamefont {Levine}}, \bibinfo {author} {\bibfnamefont {G.}~\bibnamefont {Semeghini}}, \bibinfo {author} {\bibfnamefont {M.}~\bibnamefont {Greiner}}, \bibinfo {author} {\bibfnamefont {V.}~\bibnamefont {Vuleti{\'c}}},\ and\ \bibinfo {author} {\bibfnamefont {M.~D.}\ \bibnamefont {Lukin}},\ }\bibfield  {title} {\bibinfo {title}
  {High-fidelity parallel entangling gates on a neutral-atom quantum computer},\ }\href {https://doi.org/10.1038/s41586-023-06481-y} {\bibfield  {journal} {\bibinfo  {journal} {Nature}\ }\textbf {\bibinfo {volume} {622}},\ \bibinfo {pages} {268} (\bibinfo {year} {2023})}\BibitemShut {NoStop}%
\bibitem [{\citenamefont {Emerson}\ \emph {et~al.}(2007)\citenamefont {Emerson}, \citenamefont {Silva}, \citenamefont {Moussa}, \citenamefont {Ryan}, \citenamefont {Laforest}, \citenamefont {Baugh}, \citenamefont {Cory},\ and\ \citenamefont {Laflamme}}]{emerson2007symmetrized}%
  \BibitemOpen
  \bibfield  {author} {\bibinfo {author} {\bibfnamefont {J.}~\bibnamefont {Emerson}}, \bibinfo {author} {\bibfnamefont {M.}~\bibnamefont {Silva}}, \bibinfo {author} {\bibfnamefont {O.}~\bibnamefont {Moussa}}, \bibinfo {author} {\bibfnamefont {C.}~\bibnamefont {Ryan}}, \bibinfo {author} {\bibfnamefont {M.}~\bibnamefont {Laforest}}, \bibinfo {author} {\bibfnamefont {J.}~\bibnamefont {Baugh}}, \bibinfo {author} {\bibfnamefont {D.~G.}\ \bibnamefont {Cory}},\ and\ \bibinfo {author} {\bibfnamefont {R.}~\bibnamefont {Laflamme}},\ }\bibfield  {title} {\bibinfo {title} {Symmetrized characterization of noisy quantum processes},\ }\href {https://doi.org/10.1126/science.1145699} {\bibfield  {journal} {\bibinfo  {journal} {Science}\ }\textbf {\bibinfo {volume} {317}},\ \bibinfo {pages} {1893} (\bibinfo {year} {2007})}\BibitemShut {NoStop}%
\bibitem [{\citenamefont {Wallman}\ and\ \citenamefont {Emerson}(2016)}]{wallman2016noise}%
  \BibitemOpen
  \bibfield  {author} {\bibinfo {author} {\bibfnamefont {J.~J.}\ \bibnamefont {Wallman}}\ and\ \bibinfo {author} {\bibfnamefont {J.}~\bibnamefont {Emerson}},\ }\bibfield  {title} {\bibinfo {title} {Noise tailoring for scalable quantum computation via randomized compiling},\ }\href {https://journals.aps.org/pra/abstract/10.1103/PhysRevA.94.052325} {\bibfield  {journal} {\bibinfo  {journal} {Physical Review A}\ }\textbf {\bibinfo {volume} {94}},\ \bibinfo {pages} {052325} (\bibinfo {year} {2016})}\BibitemShut {NoStop}%
\bibitem [{\citenamefont {Van Den~Berg}\ \emph {et~al.}(2023)\citenamefont {Van Den~Berg}, \citenamefont {Minev}, \citenamefont {Kandala},\ and\ \citenamefont {Temme}}]{van2023probabilistic}%
  \BibitemOpen
  \bibfield  {author} {\bibinfo {author} {\bibfnamefont {E.}~\bibnamefont {Van Den~Berg}}, \bibinfo {author} {\bibfnamefont {Z.~K.}\ \bibnamefont {Minev}}, \bibinfo {author} {\bibfnamefont {A.}~\bibnamefont {Kandala}},\ and\ \bibinfo {author} {\bibfnamefont {K.}~\bibnamefont {Temme}},\ }\bibfield  {title} {\bibinfo {title} {Probabilistic error cancellation with sparse {Pauli--Lindblad} models on noisy quantum processors},\ }\href {https://www.nature.com/articles/s41567-023-02042-2} {\bibfield  {journal} {\bibinfo  {journal} {Nature Physics}\ }\textbf {\bibinfo {volume} {19}},\ \bibinfo {pages} {1116} (\bibinfo {year} {2023})}\BibitemShut {NoStop}%
\bibitem [{\citenamefont {Hu}\ \emph {et~al.}(2025)\citenamefont {Hu}, \citenamefont {Gu}, \citenamefont {Majumder}, \citenamefont {Ren}, \citenamefont {Zhang}, \citenamefont {Wang}, \citenamefont {You}, \citenamefont {Minev}, \citenamefont {Yelin},\ and\ \citenamefont {Seif}}]{Hu2025}%
  \BibitemOpen
  \bibfield  {author} {\bibinfo {author} {\bibfnamefont {H.-Y.}\ \bibnamefont {Hu}}, \bibinfo {author} {\bibfnamefont {A.}~\bibnamefont {Gu}}, \bibinfo {author} {\bibfnamefont {S.}~\bibnamefont {Majumder}}, \bibinfo {author} {\bibfnamefont {H.}~\bibnamefont {Ren}}, \bibinfo {author} {\bibfnamefont {Y.}~\bibnamefont {Zhang}}, \bibinfo {author} {\bibfnamefont {D.~S.}\ \bibnamefont {Wang}}, \bibinfo {author} {\bibfnamefont {Y.-Z.}\ \bibnamefont {You}}, \bibinfo {author} {\bibfnamefont {Z.}~\bibnamefont {Minev}}, \bibinfo {author} {\bibfnamefont {S.~F.}\ \bibnamefont {Yelin}},\ and\ \bibinfo {author} {\bibfnamefont {A.}~\bibnamefont {Seif}},\ }\bibfield  {title} {\bibinfo {title} {Demonstration of robust and efficient quantum property learning with shallow shadows},\ }\href {https://doi.org/10.1038/s41467-025-57349-w} {\bibfield  {journal} {\bibinfo  {journal} {Nature Communications}\ }\textbf {\bibinfo {volume} {16}},\ \bibinfo {pages} {2943} (\bibinfo {year} {2025})}\BibitemShut {NoStop}%
\bibitem [{\citenamefont {Hazan}(2016)}]{hazan2016introduction}%
  \BibitemOpen
  \bibfield  {author} {\bibinfo {author} {\bibfnamefont {E.}~\bibnamefont {Hazan}},\ }\bibfield  {title} {\bibinfo {title} {Introduction to online convex optimization},\ }\href {https://www.emerald.com/ftopt/article/2/3-4/157/1324322} {\bibfield  {journal} {\bibinfo  {journal} {Foundations and Trends in Optimization}\ }\textbf {\bibinfo {volume} {2}},\ \bibinfo {pages} {157} (\bibinfo {year} {2016})}\BibitemShut {NoStop}%
\bibitem [{\citenamefont {Panteleev}\ and\ \citenamefont {Kalachev}(2021)}]{panteleev2021degenerate}%
  \BibitemOpen
  \bibfield  {author} {\bibinfo {author} {\bibfnamefont {P.}~\bibnamefont {Panteleev}}\ and\ \bibinfo {author} {\bibfnamefont {G.}~\bibnamefont {Kalachev}},\ }\bibfield  {title} {\bibinfo {title} {Degenerate quantum {LDPC} codes with good finite length performance},\ }\href {https://quantum-journal.org/papers/q-2021-11-22-585/} {\bibfield  {journal} {\bibinfo  {journal} {Quantum}\ }\textbf {\bibinfo {volume} {5}},\ \bibinfo {pages} {585} (\bibinfo {year} {2021})}\BibitemShut {NoStop}%
\bibitem [{\citenamefont {Panteleev}\ and\ \citenamefont {Kalachev}(2022{\natexlab{a}})}]{panteleev2022almostlinear}%
  \BibitemOpen
  \bibfield  {author} {\bibinfo {author} {\bibfnamefont {P.}~\bibnamefont {Panteleev}}\ and\ \bibinfo {author} {\bibfnamefont {G.}~\bibnamefont {Kalachev}},\ }\bibfield  {title} {\bibinfo {title} {Quantum {LDPC} codes with almost linear minimum distance},\ }\href {https://ieeexplore.ieee.org/iel7/18/4667673/09567703.pdf} {\bibfield  {journal} {\bibinfo  {journal} {IEEE Transactions on Information Theory}\ }\textbf {\bibinfo {volume} {68}},\ \bibinfo {pages} {213} (\bibinfo {year} {2022}{\natexlab{a}})}\BibitemShut {NoStop}%
\bibitem [{\citenamefont {Leverrier}\ and\ \citenamefont {Z{\'e}mor}(2022)}]{leverrier2022quantum}%
  \BibitemOpen
  \bibfield  {author} {\bibinfo {author} {\bibfnamefont {A.}~\bibnamefont {Leverrier}}\ and\ \bibinfo {author} {\bibfnamefont {G.}~\bibnamefont {Z{\'e}mor}},\ }\bibfield  {title} {\bibinfo {title} {Quantum {Tanner} codes},\ }in\ \href {https://doi.org/10.1109/FOCS54457.2022.00117} {\emph {\bibinfo {booktitle} {Proceedings of 2022 IEEE 63rd Annual Symposium on Foundations of Computer Science (FOCS)}}}\ (\bibinfo {organization} {IEEE},\ \bibinfo {year} {2022})\ pp.\ \bibinfo {pages} {872--883}\BibitemShut {NoStop}%
\bibitem [{\citenamefont {Gidney}(2021)}]{gidney2021stim}%
  \BibitemOpen
  \bibfield  {author} {\bibinfo {author} {\bibfnamefont {C.}~\bibnamefont {Gidney}},\ }\bibfield  {title} {\bibinfo {title} {{Stim}: a fast stabilizer circuit simulator},\ }\href {https://doi.org/10.22331/q-2021-07-06-497} {\bibfield  {journal} {\bibinfo  {journal} {Quantum}\ }\textbf {\bibinfo {volume} {5}},\ \bibinfo {pages} {497} (\bibinfo {year} {2021})}\BibitemShut {NoStop}%
\bibitem [{\citenamefont {Wu}\ \emph {et~al.}(2022)\citenamefont {Wu}, \citenamefont {Kolkowitz}, \citenamefont {Puri},\ and\ \citenamefont {Thompson}}]{wu2022erasure}%
  \BibitemOpen
  \bibfield  {author} {\bibinfo {author} {\bibfnamefont {Y.}~\bibnamefont {Wu}}, \bibinfo {author} {\bibfnamefont {S.}~\bibnamefont {Kolkowitz}}, \bibinfo {author} {\bibfnamefont {S.}~\bibnamefont {Puri}},\ and\ \bibinfo {author} {\bibfnamefont {J.~D.}\ \bibnamefont {Thompson}},\ }\bibfield  {title} {\bibinfo {title} {Erasure conversion for fault-tolerant quantum computing in alkaline earth {Rydberg} atom arrays},\ }\href {https://doi.org/10.1038/s41467-022-32094-6} {\bibfield  {journal} {\bibinfo  {journal} {Nature Communications}\ }\textbf {\bibinfo {volume} {13}},\ \bibinfo {pages} {4657} (\bibinfo {year} {2022})}\BibitemShut {NoStop}%
\bibitem [{\citenamefont {Scholl}\ \emph {et~al.}(2023)\citenamefont {Scholl}, \citenamefont {Shaw}, \citenamefont {Tsai}, \citenamefont {Finkelstein}, \citenamefont {Choi},\ and\ \citenamefont {Endres}}]{scholl2023erasure}%
  \BibitemOpen
  \bibfield  {author} {\bibinfo {author} {\bibfnamefont {P.}~\bibnamefont {Scholl}}, \bibinfo {author} {\bibfnamefont {A.~L.}\ \bibnamefont {Shaw}}, \bibinfo {author} {\bibfnamefont {R.~B.-S.}\ \bibnamefont {Tsai}}, \bibinfo {author} {\bibfnamefont {R.}~\bibnamefont {Finkelstein}}, \bibinfo {author} {\bibfnamefont {J.}~\bibnamefont {Choi}},\ and\ \bibinfo {author} {\bibfnamefont {M.}~\bibnamefont {Endres}},\ }\bibfield  {title} {\bibinfo {title} {Erasure conversion in a high-fidelity {Rydberg} quantum simulator},\ }\href {https://doi.org/10.1038/s41586-023-06516-4} {\bibfield  {journal} {\bibinfo  {journal} {Nature}\ }\textbf {\bibinfo {volume} {622}},\ \bibinfo {pages} {273} (\bibinfo {year} {2023})}\BibitemShut {NoStop}%
\bibitem [{\citenamefont {Chow}\ \emph {et~al.}(2024)\citenamefont {Chow}, \citenamefont {Buchemmavari}, \citenamefont {Omanakuttan}, \citenamefont {Little}, \citenamefont {Pandey}, \citenamefont {Deutsch},\ and\ \citenamefont {Jau}}]{chow2024circuit}%
  \BibitemOpen
  \bibfield  {author} {\bibinfo {author} {\bibfnamefont {M.~N.}\ \bibnamefont {Chow}}, \bibinfo {author} {\bibfnamefont {V.}~\bibnamefont {Buchemmavari}}, \bibinfo {author} {\bibfnamefont {S.}~\bibnamefont {Omanakuttan}}, \bibinfo {author} {\bibfnamefont {B.~J.}\ \bibnamefont {Little}}, \bibinfo {author} {\bibfnamefont {S.}~\bibnamefont {Pandey}}, \bibinfo {author} {\bibfnamefont {I.~H.}\ \bibnamefont {Deutsch}},\ and\ \bibinfo {author} {\bibfnamefont {Y.-Y.}\ \bibnamefont {Jau}},\ }\bibfield  {title} {\bibinfo {title} {Circuit-based leakage-to-erasure conversion in a neutral-atom quantum processor},\ }\href {https://journals.aps.org/prxquantum/abstract/10.1103/PRXQuantum.5.040343} {\bibfield  {journal} {\bibinfo  {journal} {PRX Quantum}\ }\textbf {\bibinfo {volume} {5}},\ \bibinfo {pages} {040343} (\bibinfo {year} {2024})}\BibitemShut {NoStop}%
\bibitem [{\citenamefont {Carbery}\ and\ \citenamefont {Wright}(2001)}]{carbery2001distributional}%
  \BibitemOpen
  \bibfield  {author} {\bibinfo {author} {\bibfnamefont {A.}~\bibnamefont {Carbery}}\ and\ \bibinfo {author} {\bibfnamefont {J.}~\bibnamefont {Wright}},\ }\bibfield  {title} {\bibinfo {title} {Distributional and {$L_q$} norm inequalities for polynomials over convex bodies in {$\mathbb{R}_n$}},\ }\href {https://projecteuclid.org/journalArticle/Download?urlId=10.4310%2FMRL.2001.v8.n3.a1} {\bibfield  {journal} {\bibinfo  {journal} {Mathematical Research Letters}\ }\textbf {\bibinfo {volume} {8}},\ \bibinfo {pages} {233} (\bibinfo {year} {2001})}\BibitemShut {NoStop}%
\bibitem [{\citenamefont {Spielman}\ and\ \citenamefont {Teng}(2004)}]{spielman2004smoothed}%
  \BibitemOpen
  \bibfield  {author} {\bibinfo {author} {\bibfnamefont {D.~A.}\ \bibnamefont {Spielman}}\ and\ \bibinfo {author} {\bibfnamefont {S.-H.}\ \bibnamefont {Teng}},\ }\bibfield  {title} {\bibinfo {title} {Smoothed analysis of algorithms: Why the simplex algorithm usually takes polynomial time},\ }\href {https://doi.org/10.1145/990308.990310} {\bibfield  {journal} {\bibinfo  {journal} {Journal of the ACM (JACM)}\ }\textbf {\bibinfo {volume} {51}},\ \bibinfo {pages} {385} (\bibinfo {year} {2004})}\BibitemShut {NoStop}%
\bibitem [{\citenamefont {Arthur}\ \emph {et~al.}(2009)\citenamefont {Arthur}, \citenamefont {Manthey},\ and\ \citenamefont {R{\"o}glin}}]{arthur2009k}%
  \BibitemOpen
  \bibfield  {author} {\bibinfo {author} {\bibfnamefont {D.}~\bibnamefont {Arthur}}, \bibinfo {author} {\bibfnamefont {B.}~\bibnamefont {Manthey}},\ and\ \bibinfo {author} {\bibfnamefont {H.}~\bibnamefont {R{\"o}glin}},\ }\bibfield  {title} {\bibinfo {title} {K-means has polynomial smoothed complexity},\ }in\ \href {https://doi.org/10.1109/FOCS.2009.14} {\emph {\bibinfo {booktitle} {Proceedings of the 50th Annual IEEE Symposium on Foundations of Computer Science (FOCS 2009)}}}\ (\bibinfo {organization} {IEEE},\ \bibinfo {year} {2009})\ pp.\ \bibinfo {pages} {405--414}\BibitemShut {NoStop}%
\bibitem [{\citenamefont {Chen}\ \emph {et~al.}(2025{\natexlab{b}})\citenamefont {Chen}, \citenamefont {Cotler},\ and\ \citenamefont {Huang}}]{chen2025quantum}%
  \BibitemOpen
  \bibfield  {author} {\bibinfo {author} {\bibfnamefont {S.}~\bibnamefont {Chen}}, \bibinfo {author} {\bibfnamefont {J.}~\bibnamefont {Cotler}},\ and\ \bibinfo {author} {\bibfnamefont {H.-Y.}\ \bibnamefont {Huang}},\ }\bibfield  {title} {\bibinfo {title} {Quantum probe tomography},\ }\href {https://arxiv.org/abs/2510.08499} {\bibfield  {journal} {\bibinfo  {journal} {arXiv:2510.08499}\ } (\bibinfo {year} {2025}{\natexlab{b}})}\BibitemShut {NoStop}%
\bibitem [{\citenamefont {Spall}(1992)}]{spall1992spsa}%
  \BibitemOpen
  \bibfield  {author} {\bibinfo {author} {\bibfnamefont {J.~C.}\ \bibnamefont {Spall}},\ }\bibfield  {title} {\bibinfo {title} {Multivariate stochastic approximation using a simultaneous perturbation gradient approximation},\ }\href {https://doi.org/10.1109/9.119632} {\bibfield  {journal} {\bibinfo  {journal} {IEEE Transactions on Automatic Control}\ }\textbf {\bibinfo {volume} {37}},\ \bibinfo {pages} {332} (\bibinfo {year} {1992})}\BibitemShut {NoStop}%
\bibitem [{\citenamefont {Duchi}\ \emph {et~al.}(2015)\citenamefont {Duchi}, \citenamefont {Jordan}, \citenamefont {Wainwright},\ and\ \citenamefont {Wibisono}}]{duchi2015optimal}%
  \BibitemOpen
  \bibfield  {author} {\bibinfo {author} {\bibfnamefont {J.~C.}\ \bibnamefont {Duchi}}, \bibinfo {author} {\bibfnamefont {M.~I.}\ \bibnamefont {Jordan}}, \bibinfo {author} {\bibfnamefont {M.~J.}\ \bibnamefont {Wainwright}},\ and\ \bibinfo {author} {\bibfnamefont {A.}~\bibnamefont {Wibisono}},\ }\bibfield  {title} {\bibinfo {title} {Optimal rates for zero-order convex optimization: The power of two function evaluations},\ }\href {https://ieeexplore.ieee.org/abstract/document/7055287/} {\bibfield  {journal} {\bibinfo  {journal} {IEEE Transactions on Information Theory}\ }\textbf {\bibinfo {volume} {61}},\ \bibinfo {pages} {2788} (\bibinfo {year} {2015})}\BibitemShut {NoStop}%
\bibitem [{\citenamefont {Bittel}\ and\ \citenamefont {Kliesch}(2021)}]{bittel2021training}%
  \BibitemOpen
  \bibfield  {author} {\bibinfo {author} {\bibfnamefont {L.}~\bibnamefont {Bittel}}\ and\ \bibinfo {author} {\bibfnamefont {M.}~\bibnamefont {Kliesch}},\ }\bibfield  {title} {\bibinfo {title} {Training variational quantum algorithms is {NP}-hard},\ }\href {https://journals.aps.org/prl/abstract/10.1103/PhysRevLett.127.120502} {\bibfield  {journal} {\bibinfo  {journal} {Physical Review Letters}\ }\textbf {\bibinfo {volume} {127}},\ \bibinfo {pages} {120502} (\bibinfo {year} {2021})}\BibitemShut {NoStop}%
\bibitem [{\citenamefont {Kitaev}(2003)}]{kitaev2003fault}%
  \BibitemOpen
  \bibfield  {author} {\bibinfo {author} {\bibfnamefont {A.~Y.}\ \bibnamefont {Kitaev}},\ }\bibfield  {title} {\bibinfo {title} {Fault-tolerant quantum computation by anyons},\ }\href {https://www.sciencedirect.com/science/article/pii/S0003491602000180} {\bibfield  {journal} {\bibinfo  {journal} {Annals of Physics}\ }\textbf {\bibinfo {volume} {303}},\ \bibinfo {pages} {2} (\bibinfo {year} {2003})}\BibitemShut {NoStop}%
\bibitem [{\citenamefont {Dennis}\ \emph {et~al.}(2002)\citenamefont {Dennis}, \citenamefont {Kitaev}, \citenamefont {Landahl},\ and\ \citenamefont {Preskill}}]{dennis2002topological}%
  \BibitemOpen
  \bibfield  {author} {\bibinfo {author} {\bibfnamefont {E.}~\bibnamefont {Dennis}}, \bibinfo {author} {\bibfnamefont {A.}~\bibnamefont {Kitaev}}, \bibinfo {author} {\bibfnamefont {A.}~\bibnamefont {Landahl}},\ and\ \bibinfo {author} {\bibfnamefont {J.}~\bibnamefont {Preskill}},\ }\bibfield  {title} {\bibinfo {title} {Topological quantum memory},\ }\href {https://pubs.aip.org/aip/jmp/article-abstract/43/9/4452/230976/Topological-quantum-memory?} {\bibfield  {journal} {\bibinfo  {journal} {Journal of Mathematical Physics}\ }\textbf {\bibinfo {volume} {43}},\ \bibinfo {pages} {4452} (\bibinfo {year} {2002})}\BibitemShut {NoStop}%
\bibitem [{\citenamefont {Fowler}\ \emph {et~al.}(2012)\citenamefont {Fowler}, \citenamefont {Mariantoni}, \citenamefont {Martinis},\ and\ \citenamefont {Cleland}}]{fowler2012surface}%
  \BibitemOpen
  \bibfield  {author} {\bibinfo {author} {\bibfnamefont {A.~G.}\ \bibnamefont {Fowler}}, \bibinfo {author} {\bibfnamefont {M.}~\bibnamefont {Mariantoni}}, \bibinfo {author} {\bibfnamefont {J.~M.}\ \bibnamefont {Martinis}},\ and\ \bibinfo {author} {\bibfnamefont {A.~N.}\ \bibnamefont {Cleland}},\ }\bibfield  {title} {\bibinfo {title} {Surface codes: Towards practical large-scale quantum computation},\ }\href {https://link.aps.org/doi/10.1103/PhysRevA.86.032324} {\bibfield  {journal} {\bibinfo  {journal} {Physical Review A}\ }\textbf {\bibinfo {volume} {86}},\ \bibinfo {pages} {032324} (\bibinfo {year} {2012})}\BibitemShut {NoStop}%
\bibitem [{\citenamefont {MacKay}\ \emph {et~al.}(2004)\citenamefont {MacKay}, \citenamefont {Mitchison},\ and\ \citenamefont {McFadden}}]{mackay2004sparse}%
  \BibitemOpen
  \bibfield  {author} {\bibinfo {author} {\bibfnamefont {D.~J.~C.}\ \bibnamefont {MacKay}}, \bibinfo {author} {\bibfnamefont {G.}~\bibnamefont {Mitchison}},\ and\ \bibinfo {author} {\bibfnamefont {P.~L.}\ \bibnamefont {McFadden}},\ }\bibfield  {title} {\bibinfo {title} {Sparse-graph codes for quantum error correction},\ }\href {https://ieeexplore.ieee.org/document/1337106/} {\bibfield  {journal} {\bibinfo  {journal} {IEEE Transactions on Information Theory}\ }\textbf {\bibinfo {volume} {50}},\ \bibinfo {pages} {2315} (\bibinfo {year} {2004})}\BibitemShut {NoStop}%
\bibitem [{\citenamefont {Kovalev}\ and\ \citenamefont {Pryadko}(2013)}]{kovalev2013quantum}%
  \BibitemOpen
  \bibfield  {author} {\bibinfo {author} {\bibfnamefont {A.~A.}\ \bibnamefont {Kovalev}}\ and\ \bibinfo {author} {\bibfnamefont {L.~P.}\ \bibnamefont {Pryadko}},\ }\bibfield  {title} {\bibinfo {title} {Quantum {Kronecker} sum-product low-density parity-check codes with finite rate},\ }\href {https://link.aps.org/doi/10.1103/PhysRevA.88.012311} {\bibfield  {journal} {\bibinfo  {journal} {Physical Review A}\ }\textbf {\bibinfo {volume} {88}},\ \bibinfo {pages} {012311} (\bibinfo {year} {2013})}\BibitemShut {NoStop}%
\bibitem [{\citenamefont {Panteleev}\ and\ \citenamefont {Kalachev}(2022{\natexlab{b}})}]{panteleev2022asymptotically}%
  \BibitemOpen
  \bibfield  {author} {\bibinfo {author} {\bibfnamefont {P.}~\bibnamefont {Panteleev}}\ and\ \bibinfo {author} {\bibfnamefont {G.}~\bibnamefont {Kalachev}},\ }\bibfield  {title} {\bibinfo {title} {Asymptotically good quantum and locally testable classical {LDPC} codes},\ }in\ \href {https://dl.acm.org/doi/10.1145/3519935.3520017} {\emph {\bibinfo {booktitle} {Proceedings of the 54th Annual ACM SIGACT Symposium on Theory of Computing (STOC 2022)}}}\ (\bibinfo {year} {2022})\ pp.\ \bibinfo {pages} {375--388}\BibitemShut {NoStop}%
\bibitem [{\citenamefont {Tillich}\ and\ \citenamefont {Z{\'e}mor}(2014)}]{tillich2014quantum}%
  \BibitemOpen
  \bibfield  {author} {\bibinfo {author} {\bibfnamefont {J.-P.}\ \bibnamefont {Tillich}}\ and\ \bibinfo {author} {\bibfnamefont {G.}~\bibnamefont {Z{\'e}mor}},\ }\bibfield  {title} {\bibinfo {title} {Quantum {LDPC} codes with positive rate and minimum distance proportional to the square root of the blocklength},\ }\href {https://ieeexplore.ieee.org/document/6671468} {\bibfield  {journal} {\bibinfo  {journal} {IEEE Transactions on Information Theory}\ }\textbf {\bibinfo {volume} {60}},\ \bibinfo {pages} {1193} (\bibinfo {year} {2014})}\BibitemShut {NoStop}%
\bibitem [{\citenamefont {Breuckmann}\ and\ \citenamefont {Eberhardt}(2021)}]{breuckmann2021quantum}%
  \BibitemOpen
  \bibfield  {author} {\bibinfo {author} {\bibfnamefont {N.~P.}\ \bibnamefont {Breuckmann}}\ and\ \bibinfo {author} {\bibfnamefont {J.~N.}\ \bibnamefont {Eberhardt}},\ }\bibfield  {title} {\bibinfo {title} {Quantum low-density parity-check codes},\ }\href {https://link.aps.org/doi/10.1103/PRXQuantum.2.040101} {\bibfield  {journal} {\bibinfo  {journal} {PRX Quantum}\ }\textbf {\bibinfo {volume} {2}},\ \bibinfo {pages} {040101} (\bibinfo {year} {2021})}\BibitemShut {NoStop}%
\bibitem [{\citenamefont {Jandura}\ and\ \citenamefont {Pupillo}(2022)}]{Jandura2022timeoptimaltwothree}%
  \BibitemOpen
  \bibfield  {author} {\bibinfo {author} {\bibfnamefont {S.}~\bibnamefont {Jandura}}\ and\ \bibinfo {author} {\bibfnamefont {G.}~\bibnamefont {Pupillo}},\ }\bibfield  {title} {\bibinfo {title} {Time-{O}ptimal {T}wo- and {T}hree-{Q}ubit {G}ates for {R}ydberg {A}toms},\ }\href {https://doi.org/10.22331/q-2022-05-13-712} {\bibfield  {journal} {\bibinfo  {journal} {{Quantum}}\ }\textbf {\bibinfo {volume} {6}},\ \bibinfo {pages} {712} (\bibinfo {year} {2022})}\BibitemShut {NoStop}%
\bibitem [{\citenamefont {Kingma}\ and\ \citenamefont {Ba}(2015)}]{kingma2015adam}%
  \BibitemOpen
  \bibfield  {author} {\bibinfo {author} {\bibfnamefont {D.~P.}\ \bibnamefont {Kingma}}\ and\ \bibinfo {author} {\bibfnamefont {J.}~\bibnamefont {Ba}},\ }\bibfield  {title} {\bibinfo {title} {Adam: A method for stochastic optimization},\ }in\ \href {https://arxiv.org/abs/1412.6980} {\emph {\bibinfo {booktitle} {International Conference on Learning Representations (ICLR)}}}\ (\bibinfo {year} {2015})\BibitemShut {NoStop}%
\bibitem [{\citenamefont {{Lee}}\ \emph {et~al.}(2026)\citenamefont {{Lee}}, \citenamefont {{Nair}}, \citenamefont {{Zhang}}, \citenamefont {{Lee}}, \citenamefont {{Khattab}},\ and\ \citenamefont {{Finn}}}]{2026arXiv260328052L}%
  \BibitemOpen
  \bibfield  {author} {\bibinfo {author} {\bibfnamefont {Y.}~\bibnamefont {{Lee}}}, \bibinfo {author} {\bibfnamefont {R.}~\bibnamefont {{Nair}}}, \bibinfo {author} {\bibfnamefont {Q.}~\bibnamefont {{Zhang}}}, \bibinfo {author} {\bibfnamefont {K.}~\bibnamefont {{Lee}}}, \bibinfo {author} {\bibfnamefont {O.}~\bibnamefont {{Khattab}}},\ and\ \bibinfo {author} {\bibfnamefont {C.}~\bibnamefont {{Finn}}},\ }\bibfield  {title} {\bibinfo {title} {{Meta-Harness: End-to-End Optimization of Model Harnesses}},\ }\href {https://ui.adsabs.harvard.edu/abs/2026arXiv260328052L} {\bibfield  {journal} {\bibinfo  {journal} {arXiv:2603.28052}\ } (\bibinfo {year} {2026})}\BibitemShut {NoStop}%
\bibitem [{\citenamefont {{Xu}}\ \emph {et~al.}(2026)\citenamefont {{Xu}}, \citenamefont {{Han}}, \citenamefont {{Ou}}, \citenamefont {{Ye}}, \citenamefont {{Shen}}, \citenamefont {{Gao}}, \citenamefont {{Wang}}, \citenamefont {{Che}}, \citenamefont {{Song}}, \citenamefont {{Liu}}, \citenamefont {{Wang}}, \citenamefont {{Zhang}}, \citenamefont {{Zhang}},\ and\ \citenamefont {{Yu}}}]{2026arXiv260622376X}%
  \BibitemOpen
  \bibfield  {author} {\bibinfo {author} {\bibfnamefont {H.}~\bibnamefont {{Xu}}}, \bibinfo {author} {\bibfnamefont {J.}~\bibnamefont {{Han}}}, \bibinfo {author} {\bibfnamefont {S.}~\bibnamefont {{Ou}}}, \bibinfo {author} {\bibfnamefont {C.}~\bibnamefont {{Ye}}}, \bibinfo {author} {\bibfnamefont {Z.}~\bibnamefont {{Shen}}}, \bibinfo {author} {\bibfnamefont {J.}~\bibnamefont {{Gao}}}, \bibinfo {author} {\bibfnamefont {Y.}~\bibnamefont {{Wang}}}, \bibinfo {author} {\bibfnamefont {T.}~\bibnamefont {{Che}}}, \bibinfo {author} {\bibfnamefont {Y.}~\bibnamefont {{Song}}}, \bibinfo {author} {\bibfnamefont {W.}~\bibnamefont {{Liu}}}, \bibinfo {author} {\bibfnamefont {L.}~\bibnamefont {{Wang}}}, \bibinfo {author} {\bibfnamefont {L.-F.}\ \bibnamefont {{Zhang}}}, \bibinfo {author} {\bibfnamefont {P.}~\bibnamefont {{Zhang}}},\ and\ \bibinfo {author} {\bibfnamefont {H.-F.}\ \bibnamefont {{Yu}}},\ }\bibfield  {title} {\bibinfo {title} {{Vibe Calibration: Autonomous Bring-up of a 112-Qubit Superconducting Quantum Processor
  by a Skill-Orchestrating Language Agent}},\ }\href {https://ui.adsabs.harvard.edu/abs/2026arXiv260622376X} {\bibfield  {journal} {\bibinfo  {journal} {arXiv:2606.22376}\ } (\bibinfo {year} {2026})}\BibitemShut {NoStop}%
\bibitem [{\citenamefont {{Isogawa}}\ \emph {et~al.}(2026)\citenamefont {{Isogawa}}, \citenamefont {{Okabe}}, \citenamefont {{Phadetsuwannukun}}, \citenamefont {{Li}},\ and\ \citenamefont {{Cappellaro}}}]{2026arXiv260725145I}%
  \BibitemOpen
  \bibfield  {author} {\bibinfo {author} {\bibfnamefont {T.}~\bibnamefont {{Isogawa}}}, \bibinfo {author} {\bibfnamefont {R.}~\bibnamefont {{Okabe}}}, \bibinfo {author} {\bibfnamefont {N.}~\bibnamefont {{Phadetsuwannukun}}}, \bibinfo {author} {\bibfnamefont {M.}~\bibnamefont {{Li}}},\ and\ \bibinfo {author} {\bibfnamefont {P.}~\bibnamefont {{Cappellaro}}},\ }\bibfield  {title} {\bibinfo {title} {{Agentic AI for Scientific Reasoning in Autonomous Quantum Sensing Experiments}},\ }\href {https://ui.adsabs.harvard.edu/abs/2026arXiv260725145I} {\bibfield  {journal} {\bibinfo  {journal} {arXiv:2607.25145}\ } (\bibinfo {year} {2026})}\BibitemShut {NoStop}%
\bibitem [{\citenamefont {Bhardwaj}\ \emph {et~al.}(2026)\citenamefont {Bhardwaj}, \citenamefont {Ma}, \citenamefont {Meister}, \citenamefont {King}, \citenamefont {Bluvstein}, \citenamefont {Preskill}, \citenamefont {Cain}, \citenamefont {Xu},\ and\ \citenamefont {Huang}}]{bhardwaj2026high}%
  \BibitemOpen
  \bibfield  {author} {\bibinfo {author} {\bibfnamefont {A.}~\bibnamefont {Bhardwaj}}, \bibinfo {author} {\bibfnamefont {M.}~\bibnamefont {Ma}}, \bibinfo {author} {\bibfnamefont {N.}~\bibnamefont {Meister}}, \bibinfo {author} {\bibfnamefont {R.}~\bibnamefont {King}}, \bibinfo {author} {\bibfnamefont {D.}~\bibnamefont {Bluvstein}}, \bibinfo {author} {\bibfnamefont {J.}~\bibnamefont {Preskill}}, \bibinfo {author} {\bibfnamefont {M.}~\bibnamefont {Cain}}, \bibinfo {author} {\bibfnamefont {Q.}~\bibnamefont {Xu}},\ and\ \bibinfo {author} {\bibfnamefont {H.-Y.}\ \bibnamefont {Huang}},\ }\bibfield  {title} {\bibinfo {title} {High-rate qldpc processors},\ }\href {https://arxiv.org/abs/2607.28795} {\bibfield  {journal} {\bibinfo  {journal} {arXiv:2607.28795}\ } (\bibinfo {year} {2026})}\BibitemShut {NoStop}%
\bibitem [{\citenamefont {Cong}\ \emph {et~al.}(2022)\citenamefont {Cong}, \citenamefont {Levine}, \citenamefont {Keesling}, \citenamefont {Bluvstein}, \citenamefont {Wang},\ and\ \citenamefont {Lukin}}]{cong2022hardware}%
  \BibitemOpen
  \bibfield  {author} {\bibinfo {author} {\bibfnamefont {I.}~\bibnamefont {Cong}}, \bibinfo {author} {\bibfnamefont {H.}~\bibnamefont {Levine}}, \bibinfo {author} {\bibfnamefont {A.}~\bibnamefont {Keesling}}, \bibinfo {author} {\bibfnamefont {D.}~\bibnamefont {Bluvstein}}, \bibinfo {author} {\bibfnamefont {S.-T.}\ \bibnamefont {Wang}},\ and\ \bibinfo {author} {\bibfnamefont {M.~D.}\ \bibnamefont {Lukin}},\ }\bibfield  {title} {\bibinfo {title} {Hardware-efficient, fault-tolerant quantum computation with rydberg atoms},\ }\href {https://journals.aps.org/prx/abstract/10.1103/PhysRevX.12.021049} {\bibfield  {journal} {\bibinfo  {journal} {Physical Review X}\ }\textbf {\bibinfo {volume} {12}},\ \bibinfo {pages} {021049} (\bibinfo {year} {2022})}\BibitemShut {NoStop}%
\bibitem [{\citenamefont {Zinkevich}(2003)}]{zinkevich2003online}%
  \BibitemOpen
  \bibfield  {author} {\bibinfo {author} {\bibfnamefont {M.}~\bibnamefont {Zinkevich}},\ }\bibfield  {title} {\bibinfo {title} {Online convex programming and generalized infinitesimal gradient ascent},\ }in\ \href {https://cdn.aaai.org/ICML/2003/ICML03-120.pdf} {\emph {\bibinfo {booktitle} {Proceedings of the 20th International Conference on Machine Learning (ICML 2003)}}}\ (\bibinfo {year} {2003})\ pp.\ \bibinfo {pages} {928--936}\BibitemShut {NoStop}%
\bibitem [{\citenamefont {Farhi}\ \emph {et~al.}(2014)\citenamefont {Farhi}, \citenamefont {Goldstone},\ and\ \citenamefont {Gutmann}}]{farhi2014quantum}%
  \BibitemOpen
  \bibfield  {author} {\bibinfo {author} {\bibfnamefont {E.}~\bibnamefont {Farhi}}, \bibinfo {author} {\bibfnamefont {J.}~\bibnamefont {Goldstone}},\ and\ \bibinfo {author} {\bibfnamefont {S.}~\bibnamefont {Gutmann}},\ }\bibfield  {title} {\bibinfo {title} {A quantum approximate optimization algorithm},\ }\href {https://arxiv.org/abs/1411.4028} {\bibfield  {journal} {\bibinfo  {journal} {arXiv:1411.4028}\ } (\bibinfo {year} {2014})}\BibitemShut {NoStop}%
\bibitem [{\citenamefont {McClean}\ \emph {et~al.}(2016)\citenamefont {McClean}, \citenamefont {Romero}, \citenamefont {Babbush},\ and\ \citenamefont {Aspuru-Guzik}}]{mcclean2016theory}%
  \BibitemOpen
  \bibfield  {author} {\bibinfo {author} {\bibfnamefont {J.~R.}\ \bibnamefont {McClean}}, \bibinfo {author} {\bibfnamefont {J.}~\bibnamefont {Romero}}, \bibinfo {author} {\bibfnamefont {R.}~\bibnamefont {Babbush}},\ and\ \bibinfo {author} {\bibfnamefont {A.}~\bibnamefont {Aspuru-Guzik}},\ }\bibfield  {title} {\bibinfo {title} {The theory of variational hybrid quantum-classical algorithms},\ }\href {https://iopscience.iop.org/article/10.1088/1367-2630/18/2/023023/meta} {\bibfield  {journal} {\bibinfo  {journal} {New Journal of Physics}\ }\textbf {\bibinfo {volume} {18}},\ \bibinfo {pages} {023023} (\bibinfo {year} {2016})}\BibitemShut {NoStop}%
\bibitem [{\citenamefont {Cerezo}\ \emph {et~al.}(2021{\natexlab{a}})\citenamefont {Cerezo}, \citenamefont {Arrasmith}, \citenamefont {Babbush}, \citenamefont {Benjamin}, \citenamefont {Endo}, \citenamefont {Fujii}, \citenamefont {McClean}, \citenamefont {Mitarai}, \citenamefont {Yuan}, \citenamefont {Cincio} \emph {et~al.}}]{cerezo2021variational}%
  \BibitemOpen
  \bibfield  {author} {\bibinfo {author} {\bibfnamefont {M.}~\bibnamefont {Cerezo}}, \bibinfo {author} {\bibfnamefont {A.}~\bibnamefont {Arrasmith}}, \bibinfo {author} {\bibfnamefont {R.}~\bibnamefont {Babbush}}, \bibinfo {author} {\bibfnamefont {S.~C.}\ \bibnamefont {Benjamin}}, \bibinfo {author} {\bibfnamefont {S.}~\bibnamefont {Endo}}, \bibinfo {author} {\bibfnamefont {K.}~\bibnamefont {Fujii}}, \bibinfo {author} {\bibfnamefont {J.~R.}\ \bibnamefont {McClean}}, \bibinfo {author} {\bibfnamefont {K.}~\bibnamefont {Mitarai}}, \bibinfo {author} {\bibfnamefont {X.}~\bibnamefont {Yuan}}, \bibinfo {author} {\bibfnamefont {L.}~\bibnamefont {Cincio}}, \emph {et~al.},\ }\bibfield  {title} {\bibinfo {title} {Variational quantum algorithms},\ }\href {https://www.nature.com/articles/s42254-021-00348-9} {\bibfield  {journal} {\bibinfo  {journal} {Nature Reviews Physics}\ }\textbf {\bibinfo {volume} {3}},\ \bibinfo {pages} {625} (\bibinfo {year} {2021}{\natexlab{a}})}\BibitemShut {NoStop}%
\bibitem [{\citenamefont {McClean}\ \emph {et~al.}(2018)\citenamefont {McClean}, \citenamefont {Boixo}, \citenamefont {Smelyanskiy}, \citenamefont {Babbush},\ and\ \citenamefont {Neven}}]{mcclean2018barren}%
  \BibitemOpen
  \bibfield  {author} {\bibinfo {author} {\bibfnamefont {J.~R.}\ \bibnamefont {McClean}}, \bibinfo {author} {\bibfnamefont {S.}~\bibnamefont {Boixo}}, \bibinfo {author} {\bibfnamefont {V.~N.}\ \bibnamefont {Smelyanskiy}}, \bibinfo {author} {\bibfnamefont {R.}~\bibnamefont {Babbush}},\ and\ \bibinfo {author} {\bibfnamefont {H.}~\bibnamefont {Neven}},\ }\bibfield  {title} {\bibinfo {title} {Barren plateaus in quantum neural network training landscapes},\ }\href {https://www.nature.com/articles/s41467-018-07090-4} {\bibfield  {journal} {\bibinfo  {journal} {Nature Communications}\ }\textbf {\bibinfo {volume} {9}},\ \bibinfo {pages} {4812} (\bibinfo {year} {2018})}\BibitemShut {NoStop}%
\bibitem [{\citenamefont {Cerezo}\ \emph {et~al.}(2021{\natexlab{b}})\citenamefont {Cerezo}, \citenamefont {Sone}, \citenamefont {Volkoff}, \citenamefont {Cincio},\ and\ \citenamefont {Coles}}]{cerezo2021cost}%
  \BibitemOpen
  \bibfield  {author} {\bibinfo {author} {\bibfnamefont {M.}~\bibnamefont {Cerezo}}, \bibinfo {author} {\bibfnamefont {A.}~\bibnamefont {Sone}}, \bibinfo {author} {\bibfnamefont {T.}~\bibnamefont {Volkoff}}, \bibinfo {author} {\bibfnamefont {L.}~\bibnamefont {Cincio}},\ and\ \bibinfo {author} {\bibfnamefont {P.~J.}\ \bibnamefont {Coles}},\ }\bibfield  {title} {\bibinfo {title} {Cost function dependent barren plateaus in shallow parametrized quantum circuits},\ }\href {https://www.nature.com/articles/s41467-021-21728-w} {\bibfield  {journal} {\bibinfo  {journal} {Nature Communications}\ }\textbf {\bibinfo {volume} {12}},\ \bibinfo {pages} {1791} (\bibinfo {year} {2021}{\natexlab{b}})}\BibitemShut {NoStop}%
\bibitem [{\citenamefont {Ortiz~Marrero}\ \emph {et~al.}(2021)\citenamefont {Ortiz~Marrero}, \citenamefont {Kieferov{\'a}},\ and\ \citenamefont {Wiebe}}]{ortiz2021entanglement}%
  \BibitemOpen
  \bibfield  {author} {\bibinfo {author} {\bibfnamefont {C.}~\bibnamefont {Ortiz~Marrero}}, \bibinfo {author} {\bibfnamefont {M.}~\bibnamefont {Kieferov{\'a}}},\ and\ \bibinfo {author} {\bibfnamefont {N.}~\bibnamefont {Wiebe}},\ }\bibfield  {title} {\bibinfo {title} {Entanglement-induced barren plateaus},\ }\href {https://link.aps.org/doi/10.1103/PRXQuantum.2.040316} {\bibfield  {journal} {\bibinfo  {journal} {PRX quantum}\ }\textbf {\bibinfo {volume} {2}},\ \bibinfo {pages} {040316} (\bibinfo {year} {2021})}\BibitemShut {NoStop}%
\bibitem [{\citenamefont {You}\ and\ \citenamefont {Wu}(2021)}]{you2021exponentially}%
  \BibitemOpen
  \bibfield  {author} {\bibinfo {author} {\bibfnamefont {X.}~\bibnamefont {You}}\ and\ \bibinfo {author} {\bibfnamefont {X.}~\bibnamefont {Wu}},\ }\bibfield  {title} {\bibinfo {title} {Exponentially many local minima in quantum neural networks},\ }in\ \href {https://proceedings.mlr.press/v139/you21c} {\emph {\bibinfo {booktitle} {Proceedings of the 38th International Conference on Machine Learning (ICML 2021)}}},\ Vol.\ \bibinfo {volume} {139}\ (\bibinfo {organization} {PMLR},\ \bibinfo {year} {2021})\ pp.\ \bibinfo {pages} {12144--12155}\BibitemShut {NoStop}%
\bibitem [{\citenamefont {Anschuetz}(2021)}]{anschuetz2021critical}%
  \BibitemOpen
  \bibfield  {author} {\bibinfo {author} {\bibfnamefont {E.~R.}\ \bibnamefont {Anschuetz}},\ }\bibfield  {title} {\bibinfo {title} {Critical points in quantum generative models},\ }\href {https://arxiv.org/abs/2109.06957} {\bibfield  {journal} {\bibinfo  {journal} {arXiv:2109.06957}\ } (\bibinfo {year} {2021})}\BibitemShut {NoStop}%
\bibitem [{\citenamefont {Anschuetz}\ and\ \citenamefont {Kiani}(2022)}]{anschuetz2022quantum}%
  \BibitemOpen
  \bibfield  {author} {\bibinfo {author} {\bibfnamefont {E.~R.}\ \bibnamefont {Anschuetz}}\ and\ \bibinfo {author} {\bibfnamefont {B.~T.}\ \bibnamefont {Kiani}},\ }\bibfield  {title} {\bibinfo {title} {Quantum variational algorithms are swamped with traps},\ }\href {https://www.nature.com/articles/s41467-022-35364-5} {\bibfield  {journal} {\bibinfo  {journal} {Nature Communications}\ }\textbf {\bibinfo {volume} {13}},\ \bibinfo {pages} {7760} (\bibinfo {year} {2022})}\BibitemShut {NoStop}%
\bibitem [{\citenamefont {You}\ \emph {et~al.}(2022)\citenamefont {You}, \citenamefont {Chakrabarti},\ and\ \citenamefont {Wu}}]{you2022convergence}%
  \BibitemOpen
  \bibfield  {author} {\bibinfo {author} {\bibfnamefont {X.}~\bibnamefont {You}}, \bibinfo {author} {\bibfnamefont {S.}~\bibnamefont {Chakrabarti}},\ and\ \bibinfo {author} {\bibfnamefont {X.}~\bibnamefont {Wu}},\ }\bibfield  {title} {\bibinfo {title} {A convergence theory for over-parameterized variational quantum eigensolvers},\ }\href {https://arxiv.org/abs/2205.12481} {\bibfield  {journal} {\bibinfo  {journal} {arXiv:2205.12481}\ } (\bibinfo {year} {2022})}\BibitemShut {NoStop}%
\bibitem [{\citenamefont {Russell}\ \emph {et~al.}(2016)\citenamefont {Russell}, \citenamefont {Rabitz},\ and\ \citenamefont {Wu}}]{russell2016quantum}%
  \BibitemOpen
  \bibfield  {author} {\bibinfo {author} {\bibfnamefont {B.}~\bibnamefont {Russell}}, \bibinfo {author} {\bibfnamefont {H.}~\bibnamefont {Rabitz}},\ and\ \bibinfo {author} {\bibfnamefont {R.}~\bibnamefont {Wu}},\ }\bibfield  {title} {\bibinfo {title} {Quantum control landscapes are almost always trap free},\ }\href {https://arxiv.org/abs/1608.06198} {\bibfield  {journal} {\bibinfo  {journal} {arXiv:1608.06198}\ } (\bibinfo {year} {2016})}\BibitemShut {NoStop}%
\bibitem [{\citenamefont {Wu}\ \emph {et~al.}(2011)\citenamefont {Wu}, \citenamefont {Hsieh},\ and\ \citenamefont {Rabitz}}]{wu2011role}%
  \BibitemOpen
  \bibfield  {author} {\bibinfo {author} {\bibfnamefont {R.-B.}\ \bibnamefont {Wu}}, \bibinfo {author} {\bibfnamefont {M.~A.}\ \bibnamefont {Hsieh}},\ and\ \bibinfo {author} {\bibfnamefont {H.}~\bibnamefont {Rabitz}},\ }\bibfield  {title} {\bibinfo {title} {Role of controllability in optimizing quantum dynamics},\ }\href {https://journals.aps.org/pra/abstract/10.1103/PhysRevA.83.062306} {\bibfield  {journal} {\bibinfo  {journal} {Physical Review A}\ }\textbf {\bibinfo {volume} {83}},\ \bibinfo {pages} {062306} (\bibinfo {year} {2011})}\BibitemShut {NoStop}%
\bibitem [{\citenamefont {Cong}\ \emph {et~al.}(2019)\citenamefont {Cong}, \citenamefont {Choi},\ and\ \citenamefont {Lukin}}]{cong2019quantum}%
  \BibitemOpen
  \bibfield  {author} {\bibinfo {author} {\bibfnamefont {I.}~\bibnamefont {Cong}}, \bibinfo {author} {\bibfnamefont {S.}~\bibnamefont {Choi}},\ and\ \bibinfo {author} {\bibfnamefont {M.~D.}\ \bibnamefont {Lukin}},\ }\bibfield  {title} {\bibinfo {title} {Quantum convolutional neural networks},\ }\href {https://www.nature.com/articles/s41567-019-0648-8} {\bibfield  {journal} {\bibinfo  {journal} {Nature Physics}\ }\textbf {\bibinfo {volume} {15}},\ \bibinfo {pages} {1273} (\bibinfo {year} {2019})}\BibitemShut {NoStop}%
\bibitem [{\citenamefont {Kim}\ \emph {et~al.}(2021)\citenamefont {Kim}, \citenamefont {Kim},\ and\ \citenamefont {Rosa}}]{kim2021universal}%
  \BibitemOpen
  \bibfield  {author} {\bibinfo {author} {\bibfnamefont {J.}~\bibnamefont {Kim}}, \bibinfo {author} {\bibfnamefont {J.}~\bibnamefont {Kim}},\ and\ \bibinfo {author} {\bibfnamefont {D.}~\bibnamefont {Rosa}},\ }\bibfield  {title} {\bibinfo {title} {Universal effectiveness of high-depth circuits in variational eigenproblems},\ }\href {https://journals.aps.org/prresearch/abstract/10.1103/PhysRevResearch.3.023203} {\bibfield  {journal} {\bibinfo  {journal} {Physical Review Research}\ }\textbf {\bibinfo {volume} {3}},\ \bibinfo {pages} {023203} (\bibinfo {year} {2021})}\BibitemShut {NoStop}%
\bibitem [{\citenamefont {Kim}\ and\ \citenamefont {Oz}(2022)}]{kim2022quantum}%
  \BibitemOpen
  \bibfield  {author} {\bibinfo {author} {\bibfnamefont {J.}~\bibnamefont {Kim}}\ and\ \bibinfo {author} {\bibfnamefont {Y.}~\bibnamefont {Oz}},\ }\bibfield  {title} {\bibinfo {title} {Quantum energy landscape and circuit optimization},\ }\href {https://journals.aps.org/pra/abstract/10.1103/PhysRevA.106.052424} {\bibfield  {journal} {\bibinfo  {journal} {Physical Review A}\ }\textbf {\bibinfo {volume} {106}},\ \bibinfo {pages} {052424} (\bibinfo {year} {2022})}\BibitemShut {NoStop}%
\bibitem [{\citenamefont {Wiersema}\ \emph {et~al.}(2020)\citenamefont {Wiersema}, \citenamefont {Zhou}, \citenamefont {de~Sereville}, \citenamefont {Carrasquilla}, \citenamefont {Kim},\ and\ \citenamefont {Yuen}}]{wiersema2020exploring}%
  \BibitemOpen
  \bibfield  {author} {\bibinfo {author} {\bibfnamefont {R.}~\bibnamefont {Wiersema}}, \bibinfo {author} {\bibfnamefont {C.}~\bibnamefont {Zhou}}, \bibinfo {author} {\bibfnamefont {Y.}~\bibnamefont {de~Sereville}}, \bibinfo {author} {\bibfnamefont {J.~F.}\ \bibnamefont {Carrasquilla}}, \bibinfo {author} {\bibfnamefont {Y.~B.}\ \bibnamefont {Kim}},\ and\ \bibinfo {author} {\bibfnamefont {H.}~\bibnamefont {Yuen}},\ }\bibfield  {title} {\bibinfo {title} {Exploring entanglement and optimization within the hamiltonian variational ansatz},\ }\href {https://link.aps.org/doi/10.1103/PRXQuantum.1.020319} {\bibfield  {journal} {\bibinfo  {journal} {PRX Quantum}\ }\textbf {\bibinfo {volume} {1}},\ \bibinfo {pages} {020319} (\bibinfo {year} {2020})}\BibitemShut {NoStop}%
\bibitem [{\citenamefont {Jin}\ \emph {et~al.}(2018{\natexlab{a}})\citenamefont {Jin}, \citenamefont {Netrapalli},\ and\ \citenamefont {Jordan}}]{jin2018accelerated}%
  \BibitemOpen
  \bibfield  {author} {\bibinfo {author} {\bibfnamefont {C.}~\bibnamefont {Jin}}, \bibinfo {author} {\bibfnamefont {P.}~\bibnamefont {Netrapalli}},\ and\ \bibinfo {author} {\bibfnamefont {M.~I.}\ \bibnamefont {Jordan}},\ }\bibfield  {title} {\bibinfo {title} {Accelerated gradient descent escapes saddle points faster than gradient descent},\ }in\ \href {https://proceedings.mlr.press/v75/jin18a.html} {\emph {\bibinfo {booktitle} {Proceedings of the 31st Conference On Learning Theory}}},\ Vol.~\bibinfo {volume} {75}\ (\bibinfo  {publisher} {PMLR},\ \bibinfo {year} {2018})\ pp.\ \bibinfo {pages} {1042--1085}\BibitemShut {NoStop}%
\bibitem [{\citenamefont {Jin}\ \emph {et~al.}(2017)\citenamefont {Jin}, \citenamefont {Ge}, \citenamefont {Netrapalli}, \citenamefont {Kakade},\ and\ \citenamefont {Jordan}}]{jin2017escape}%
  \BibitemOpen
  \bibfield  {author} {\bibinfo {author} {\bibfnamefont {C.}~\bibnamefont {Jin}}, \bibinfo {author} {\bibfnamefont {R.}~\bibnamefont {Ge}}, \bibinfo {author} {\bibfnamefont {P.}~\bibnamefont {Netrapalli}}, \bibinfo {author} {\bibfnamefont {S.~M.}\ \bibnamefont {Kakade}},\ and\ \bibinfo {author} {\bibfnamefont {M.~I.}\ \bibnamefont {Jordan}},\ }\bibfield  {title} {\bibinfo {title} {How to escape saddle points efficiently},\ }in\ \href {https://proceedings.mlr.press/v70/jin17a.html} {\emph {\bibinfo {booktitle} {Proceedings of the 34th International Conference on Machine Learning (ICML 2017)}}}\ (\bibinfo {organization} {PMLR},\ \bibinfo {year} {2017})\ pp.\ \bibinfo {pages} {1724--1732}\BibitemShut {NoStop}%
\bibitem [{\citenamefont {Jin}\ \emph {et~al.}(2018{\natexlab{b}})\citenamefont {Jin}, \citenamefont {Liu}, \citenamefont {Ge},\ and\ \citenamefont {Jordan}}]{jin2018local}%
  \BibitemOpen
  \bibfield  {author} {\bibinfo {author} {\bibfnamefont {C.}~\bibnamefont {Jin}}, \bibinfo {author} {\bibfnamefont {L.~T.}\ \bibnamefont {Liu}}, \bibinfo {author} {\bibfnamefont {R.}~\bibnamefont {Ge}},\ and\ \bibinfo {author} {\bibfnamefont {M.~I.}\ \bibnamefont {Jordan}},\ }\bibfield  {title} {\bibinfo {title} {On the local minima of the empirical risk},\ }in\ \href {https://proceedings.neurips.cc/paper_files/paper/2018/hash/da4902cb0bc38210839714ebdcf0efc3-Abstract.html} {\emph {\bibinfo {booktitle} {Proceedings of the 32nd International Conference on Neural Information Processing Systems (NIPS 2018)}}},\ Vol.~\bibinfo {volume} {31}\ (\bibinfo {year} {2018})\ pp.\ \bibinfo {pages} {4901--4910}\BibitemShut {NoStop}%
\bibitem [{\citenamefont {Shalev-Shwartz}\ and\ \citenamefont {Singer}(2006)}]{shalev2006convex}%
  \BibitemOpen
  \bibfield  {author} {\bibinfo {author} {\bibfnamefont {S.}~\bibnamefont {Shalev-Shwartz}}\ and\ \bibinfo {author} {\bibfnamefont {Y.}~\bibnamefont {Singer}},\ }\bibfield  {title} {\bibinfo {title} {Convex repeated games and fenchel duality},\ }in\ \href {https://proceedings.neurips.cc/paper/2006/hash/1cfead9959b76ce44a847c850b61c587-Abstract.html} {\emph {\bibinfo {booktitle} {Proceedings of the 20th International Conference on Neural Information Processing Systems (NIPS 2006)}}},\ Vol.~\bibinfo {volume} {19}\ (\bibinfo {year} {2006})\BibitemShut {NoStop}%
\bibitem [{\citenamefont {H{\'e}liou}\ \emph {et~al.}(2020)\citenamefont {H{\'e}liou}, \citenamefont {Martin}, \citenamefont {Mertikopoulos},\ and\ \citenamefont {Rahier}}]{heliou2020online}%
  \BibitemOpen
  \bibfield  {author} {\bibinfo {author} {\bibfnamefont {A.}~\bibnamefont {H{\'e}liou}}, \bibinfo {author} {\bibfnamefont {M.}~\bibnamefont {Martin}}, \bibinfo {author} {\bibfnamefont {P.}~\bibnamefont {Mertikopoulos}},\ and\ \bibinfo {author} {\bibfnamefont {T.}~\bibnamefont {Rahier}},\ }\bibfield  {title} {\bibinfo {title} {Online non-convex optimization with imperfect feedback},\ }in\ \href {https://proceedings.neurips.cc/paper/2020/hash/c7c46d4baf816bfb07c7f3bf96d88544-Abstract.html} {\emph {\bibinfo {booktitle} {Proceedings of the 34th International Conference on Neural Information Processing Systems (NIPS 2020)}}},\ Vol.~\bibinfo {volume} {33}\ (\bibinfo {year} {2020})\ pp.\ \bibinfo {pages} {17224--17235}\BibitemShut {NoStop}%
\bibitem [{\citenamefont {H{\'e}liou}\ \emph {et~al.}(2021)\citenamefont {H{\'e}liou}, \citenamefont {Martin}, \citenamefont {Mertikopoulos},\ and\ \citenamefont {Rahier}}]{heliou2021zeroth}%
  \BibitemOpen
  \bibfield  {author} {\bibinfo {author} {\bibfnamefont {A.}~\bibnamefont {H{\'e}liou}}, \bibinfo {author} {\bibfnamefont {M.}~\bibnamefont {Martin}}, \bibinfo {author} {\bibfnamefont {P.}~\bibnamefont {Mertikopoulos}},\ and\ \bibinfo {author} {\bibfnamefont {T.}~\bibnamefont {Rahier}},\ }\bibfield  {title} {\bibinfo {title} {Zeroth-order non-convex learning via hierarchical dual averaging},\ }in\ \href {https://proceedings.mlr.press/v139/heliou21a.html} {\emph {\bibinfo {booktitle} {Proceedings of the 38th International Conference on Machine Learning (ICML 2021)}}}\ (\bibinfo {organization} {PMLR},\ \bibinfo {year} {2021})\ pp.\ \bibinfo {pages} {4192--4202}\BibitemShut {NoStop}%
\bibitem [{\citenamefont {Gao}\ \emph {et~al.}(2018)\citenamefont {Gao}, \citenamefont {Li},\ and\ \citenamefont {Zhang}}]{gao2018online}%
  \BibitemOpen
  \bibfield  {author} {\bibinfo {author} {\bibfnamefont {X.}~\bibnamefont {Gao}}, \bibinfo {author} {\bibfnamefont {X.}~\bibnamefont {Li}},\ and\ \bibinfo {author} {\bibfnamefont {S.}~\bibnamefont {Zhang}},\ }\bibfield  {title} {\bibinfo {title} {Online learning with non-convex losses and non-stationary regret},\ }in\ \href {https://proceedings.mlr.press/v84/gao18a.html} {\emph {\bibinfo {booktitle} {Proceedings of the Twenty-First International Conference on Artificial Intelligence and Statistics}}}\ (\bibinfo {organization} {PMLR},\ \bibinfo {year} {2018})\ pp.\ \bibinfo {pages} {235--243}\BibitemShut {NoStop}%
\bibitem [{\citenamefont {Guan}\ \emph {et~al.}(2024)\citenamefont {Guan}, \citenamefont {Zhou},\ and\ \citenamefont {Liang}}]{guan2023hardness}%
  \BibitemOpen
  \bibfield  {author} {\bibinfo {author} {\bibfnamefont {Z.}~\bibnamefont {Guan}}, \bibinfo {author} {\bibfnamefont {Y.}~\bibnamefont {Zhou}},\ and\ \bibinfo {author} {\bibfnamefont {Y.}~\bibnamefont {Liang}},\ }\bibfield  {title} {\bibinfo {title} {On the hardness of online nonconvex optimization with single oracle feedback},\ }in\ \href {https://iclr.cc/virtual/2024/poster/18061} {\emph {\bibinfo {booktitle} {Proceedings of the Twelfth International Conference on Learning Representations (ICLR 2024)}}}\ (\bibinfo {year} {2024})\BibitemShut {NoStop}%
\bibitem [{\citenamefont {Levine}\ \emph {et~al.}(2019)\citenamefont {Levine}, \citenamefont {Keesling}, \citenamefont {Semeghini}, \citenamefont {Omran}, \citenamefont {Wang}, \citenamefont {Ebadi}, \citenamefont {Bernien}, \citenamefont {Greiner}, \citenamefont {Vuleti\ifmmode~\acute{c}\else \'{c}\fi{}}, \citenamefont {Pichler},\ and\ \citenamefont {Lukin}}]{Levine_Pichler}%
  \BibitemOpen
  \bibfield  {author} {\bibinfo {author} {\bibfnamefont {H.}~\bibnamefont {Levine}}, \bibinfo {author} {\bibfnamefont {A.}~\bibnamefont {Keesling}}, \bibinfo {author} {\bibfnamefont {G.}~\bibnamefont {Semeghini}}, \bibinfo {author} {\bibfnamefont {A.}~\bibnamefont {Omran}}, \bibinfo {author} {\bibfnamefont {T.~T.}\ \bibnamefont {Wang}}, \bibinfo {author} {\bibfnamefont {S.}~\bibnamefont {Ebadi}}, \bibinfo {author} {\bibfnamefont {H.}~\bibnamefont {Bernien}}, \bibinfo {author} {\bibfnamefont {M.}~\bibnamefont {Greiner}}, \bibinfo {author} {\bibfnamefont {V.}~\bibnamefont {Vuleti\ifmmode~\acute{c}\else \'{c}\fi{}}}, \bibinfo {author} {\bibfnamefont {H.}~\bibnamefont {Pichler}},\ and\ \bibinfo {author} {\bibfnamefont {M.~D.}\ \bibnamefont {Lukin}},\ }\bibfield  {title} {\bibinfo {title} {Parallel implementation of high-fidelity multiqubit gates with neutral atoms},\ }\href {https://doi.org/10.1103/PhysRevLett.123.170503} {\bibfield  {journal} {\bibinfo  {journal} {Physical Review Letters}\ }\textbf {\bibinfo
  {volume} {123}},\ \bibinfo {pages} {170503} (\bibinfo {year} {2019})}\BibitemShut {NoStop}%
\end{thebibliography}%
\let\addcontentsline\oldaddcontentsline
 

\clearpage
\newpage
\onecolumngrid

\let\addcontentsline\realaddcontentsline

\tableofcontents

\begin{appendix}
 
\section{Surrogate error model and properties}
\label{app:appendix_1}

Here, we introduce the surrogate error model equipped with an objective function generalized from Ref.~\cite{sivak2025reinforcement}, and show under what conditions it is convex with respect to the control parameters. 
 We also prove the smoothness and Lipschitzness properties of the surrogate objective function.
 
\subsection{The surrogate error model}\label{sec:surrogate}

Our starting point is the surrogate error model. 
Following the reinforcement-learning calibration approach of Ref.~\cite{sivak2025reinforcement}, we describe a QEC circuit through the parameters of its gates, which are held in a classical controller. 
Physical errors incurred during the computation are modeled as stochastic displacements of these parameters away from their calibrated settings. 
The controller responds by probing the parameters and adjusting them to drive the errors back down. 
Because each adjustment shifts the observed syndrome statistics, calibration reduces to learning, from these statistics alone, how to move the parameters toward their optimal values. 
We formalize this below and determine the conditions under which the induced objective is convex.

Fault-tolerant circuit design arranges for every physical error to leave a trace in the stabilizer (or check) measurement record. 
This trace is captured by detectors~\cite{gidney2021stim}, each defined as a group of measurements whose combined parity is fixed in the absence of errors. 
A parity flip, or more generally, a detection event, localizes an error to a space-time region of the circuit, its detecting region. 
Since a detecting region can be triggered by several distinct error mechanisms, a decoder reconstructs the most probable error configuration from the observed detection pattern, and the logical error rate (LER) measures the residual error once this correction is applied.

Although the LER is the quantity one ultimately wants to minimize, it is ill-suited as an optimization target here, for the reasons already noted in Ref.~\cite{sivak2025reinforcement}: it is exponentially expensive to resolve below threshold~\cite{google2025quantum,harvard2026}. 
It must be optimized over a very large parameter set, and it cannot be evaluated in real time when the logical state is unknown. 
Following Ref.~\cite{sivak2025reinforcement}, we therefore optimize a surrogate: the objective $C$ is the average detection-event rate across the circuit, estimated empirically over a batch of QEC cycles, and serves as a locally computable proxy for the LER.
Our contribution is not this objective itself but the analysis, given below, of when it is convex and how efficiently it can be optimized.

Technically, let $r=1,2,\dots$ index QEC cycles. 
The physical time is $\tau_r = r\Delta t$, where $\Delta t$ is the \emph{fixed} duration of one syndrome-extraction cycle determined by gate speed and compilation on the quantum devices. 
There are $N_D$ detectors indexed by $D_k$ with $k=1,..,N_D$ across the QEC circuit.
In each cycle $r$, each detector $D_k$ produces a binary detection event $\mathcal D_{k,r}\in\{0,1\}$.
The per-cycle empirical detection fraction is
\begin{align}\label{eq:model_Y}
Y_r := \frac{1}{N_D}\sum_{k=1}^{N_D} \mathcal D_{k,r}\in[0,1].
\end{align}
In the time-dependent setting, we only see a noisy time series $\{\mathcal D_{k,r}\}$ and $\{Y_r\}$; to get a usable estimation, we must average over a window.
Let an epoch $t=1,2,\dots,T$ contain $m$ consecutive QEC cycles. We define a window corresponding to the epoch $t$ to be
\begin{align}\label{eq:model_window}
\mathcal W_t := \{r:r=(t-1)m+1,\dots,tm\}.
\end{align}
The surrogate objective $C=\mathbb E[Y_r]$ is an expectation over detector outcomes computed from $m$ cycles in each window as
\begin{align}\label{eq:model_Cemp}
\hat{C}:=\frac{1}{m}\sum_{r\in\mathcal W_t} Y_r\in[0,1]
\end{align}
We define $DR_k=\mathbb E[\mathcal D_{k,r}]$ to be the expected detection rate for detector $D_k$ in the cycle $r$, and thus $C=\frac{1}{N_D}\sum_{k=1}^{N_D}DR_k$.
Within epoch $t$, we start by implementing a control parameter perturbation, estimate $C$ from collecting detectors across the $m$ QEC cycle in the window while holding the control unchanged, and deciding the control parameter perturbation in the next epoch based on observed $C$.
The surrogate objective function $C$ can be efficiently estimated without knowing the logical state, compared to logical error rate estimation.
Moreover, to ensure efficient multi-dimensional optimization of $C$, we denote the parameter as a vector $\vec{\theta}$, and assume that it is sparse, i.e., polynomial in the system size, in the practical cases.
This assumption is natural when we consider a quantum error circuit with gate count scaling polynomially in the system size and each gate (usually a single- and two-qubit gate) parametrized by a few parameters.

\subsection{Control induced unitary error}\label{sec:unitary_convexity}
As a warm-up, we consider the model where the erroneous control results in a unitary error. 
We will prove that, for a QEC circuit equipped with randomized compiling or Pauli twirling~\cite{wallman2016noise,emerson2007symmetrized,van2023probabilistic}, the surrogate objective function $C$ is a strictly convex function of the physical control parameters within an explicit neighborhood of the optimal calibration point.

\paragraph{The error model setup.}Consider a QEC circuit with gates indexed by $i$. Each gate $i$ acts on $n_i$ qubits ($d_i = 2^{n_i}$) with ideal unitary $U_i^* = U_i(\vec{\theta}_i^{*})$ and $\vec{\theta}$ collects all $\vec{\theta}_i$'s.
Assume the error is caused by a control drift $\delta\vec{\theta}_i = \vec{\theta}_i - \vec{\theta}_i^{*}$, the error unitary is:
\begin{align*}
V_i(\delta\vec{\theta}_i) = U_i(\vec{\theta}_i^{*} + \delta\vec{\theta}_i) U_i(\vec{\theta}_i^{*})^\dagger = e^{-i\mathcal{H}_i(\delta\vec{\theta}_i)}, \qquad V_i(\vec{0}) = I.
\end{align*}
Here, the map $\delta\vec{\theta}_i \mapsto \mathcal{H}_i(\delta\vec{\theta}_i)$ is smooth from Schr\"{o}dinger's equation but generically nonlinear. The Taylor expansion for $\mathcal{H}_i(\delta\vec{\theta}_i)$ is given by
\begin{align*}
\mathcal{H}_i(\delta\vec{\theta}_i) = \sum_j \delta \theta_{i,j} G_{i,j} + \frac{1}{2}\sum_{j,k} \delta \theta_{i,j}\delta \theta_{i,k} G_{i,jk}^{(2)} + O(\norm{\delta\vec{\theta}_i}^3),
\end{align*}
where $\delta \theta_{i,j}$ is the $j$-th entry of $\delta \vec{\theta}_{i}$, and $G_{i,j}$ and $G_{i,jk}^{(2)}$ are the coefficients. 

We note that randomized compiling produces a Pauli channel as follows:
\begin{proposition}\label{prop:Pauli_twirling_unitary}
After Pauli twirling, the error channel for gate $i$ becomes a Pauli channel:
\begin{align*}
\mathcal{T}_i(\rho) = \sum_{P \in \mathcal{P}_{n_i}} p_P^{(i)}(\delta\vec{\theta}_i) P\rho P, \qquad p_P^{(i)} \geq 0,\quad \sum_P p_P^{(i)} = 1,
\end{align*}
where $\mathcal{P}_{n_i}$ denote the set of all Pauli observables on $n_i$ qubits, with exact Pauli error probabilities:
\begin{align*}
p_P^{(i)} = \frac{1}{d_i^2}\sum_{Q \in \mathcal{P}_{n_i}} (-1)^{s(P,Q)} c_Q^{(i)}, \qquad c_Q^{(i)} = \frac{1}{d_i}\Tr(Q V_i Q V_i^\dagger),
\end{align*}
where $s(P,Q)=1$ if $P$ and $Q$ anti-commute, and $s(P,Q)=0$ if $P$ and $Q$ commute.
\end{proposition}

\begin{proof}
We note that the twirled channel $\mathcal{T}_i(\rho) = \frac{1}{d_i^2}\sum_P P V_i P \rho P V_i^\dagger P$ acts diagonally on Paulis as $\mathcal{T}_i(Q) = c_Q Q$ where $c_Q = \frac{1}{d_i}\tr(QV_iQV_i^\dagger)$. The diagonality follows from symplectic character orthogonality: $\frac{1}{d_i^2}\sum_P (-1)^{s(P, Q' \oplus Q)} = \delta_{Q',Q}$. The complete positivity of $\mathcal{T}_i$, which is a convex combination of completely positive maps, ensures $p_P \geq 0$.
\end{proof}

For detector $D_k$, let $\mathcal{R}_k$ denote its detecting region that consists the set of gates whose errors can flip $D_k$. For each gate $i \in \mathcal{R}_k$, let $\mathcal{F}_k^{(i)} \subset \mathcal{P}_{n_i} \setminus \{I\}$ be the set of nontrivial Pauli errors on gate $i$ that flip detector $D_k$. Define the detector-relevant error probability:
\begin{align}
q_{ki}(\delta\vec{\theta}_i) = \sum_{P \in \mathcal{F}_k^{(i)}} p_P^{(i)}(\delta\vec{\theta}_i).
\label{eq:qki_def_unitary}
\end{align}
This is the total probability that gate $i$ produces an error flipping $D_k$. For independent errors across gates, the exact detection event rate is:
\begin{align}\label{eq:DRk_unitary}
DR_k(\delta\vec{\theta}) = \frac{1 - \prod_{i \in \mathcal{R}_k}\big(1 - 2q_{ki}(\delta\vec{\theta}_i)\big)}{2}.
\end{align}
In the case when the errors are correlated across different gates, we still have $DR_k(\delta\vec \theta)\approx \sum_{i\in\mathcal{R}_k}q_{k_i}(\delta\vec{\theta}_i)$, which is the same as Eq.~\eqref{eq:DRk_unitary} up to the second-order in the small error regime we consider.
To show that $C(\delta\vec{\theta})$ is convex within an explicit neighborhood of the optimal calibration point, we only need to show that $DR_k(\delta\vec \theta)$ is strictly convex in $\delta\vec{\theta}$ for $\|\delta\vec{\theta}\|$ below a certain threshold.

\paragraph{Convexity of the detection rate $DR_k$ at $\delta\vec{\theta}=0$.}We first analyze the behavior of each $p_P^{(i)}$ near the optimal point $\vec{\theta}^*$. 
We note that the optimal point corresponds to $\vec{\theta}=\vec{\theta}^*$ and $\delta\vec{\theta}=0$.
We have the following lemma:
\begin{lemma}[Properties of $p_P^{(i)}$ at the optimum]
\label{lem:pP_hessian}
Given the traceless error generators assumption, i.e. $\Tr(G_{i,j}) = 0$ for all $i,j$, for each nontrivial Pauli $P \neq I$, we have
\begin{enumerate}
\item $p_P^{(i)}(\vec{0}) = 0$.
\item $\nabla p_P^{(i)}(\vec{0}) = \vec{0}$.
\item The Hessian at the optimum is the rank-1 positive semidefinite matrix:
\begin{align}
\frac{\partial^2 p_P^{(i)}}{\partial(\delta \theta_{i,j})\partial(\delta \theta_{i,k})}\bigg|_{\vec{0}} = \frac{2}{d_i^2}\Tr(P G_{i,j})\Tr(P G_{i,k}) = 2 h_{P,j}^{(i)} h_{P,k}^{(i)},
\label{eq:hessian_pP}
\end{align}
where we define the Pauli-parameter coupling vector:
\begin{align*}
h_{P,j}^{(i)} := \frac{\Tr(P G_{i,j})}{d_i}.
\end{align*}
\item The nonlinear response operators $G_{i,jk}^{(2)}, G_{i,jkl}^{(3)}, \ldots$ do not contribute to the Hessian at the optimum.
\end{enumerate}
\end{lemma}

\begin{proof}
According to \Cref{prop:Pauli_twirling_unitary}, the eigenvalues of the Pauli channel are $c_Q^{(i)} = \frac{1}{d_i}\Tr(Q V_i Q V_i^\dagger)$, and the Pauli error probabilities are obtained by symplectic Fourier transform: $p_P^{(i)} = \frac{1}{d_i^2}\sum_Q (-1)^{s(P,Q)} c_Q^{(i)}$. 
Using this knowledge, we prove the claimed properties sequentially
\begin{enumerate}
\item At $\delta\vec{\theta}_i = \vec{0}$, $V_i = I$, so $c_Q^{(i)} = 1$ for all $Q$. Then $p_P^{(i)}(\vec{0}) = \frac{1}{d_i^2}\sum_Q (-1)^{s(P,Q)} = \delta_{P,I}$, which vanishes for $P \neq I$ by character orthogonality.
\item Using $V_i = I - i\mathcal{H}_i + O(\mathcal{H}_i^2)$, we have
\begin{align*}
\frac{\partial c_Q^{(i)}}{\partial(\delta \theta_{i,j})}\bigg|_{\vec{0}} = \frac{1}{d_i}\Tr\big(Q(-iG_{i,j})Q + Q^2(iG_{i,j})\big) = \frac{-i}{d_i}\big(\Tr(G_{i,j}) - \Tr(G_{i,j})\big) = 0,
\end{align*}
using $Q^2 = I$ and cyclicity. 
Since $\nabla c_Q^{(i)}|_{\vec{0}} = 0$ for all $Q$, the linear combination $\nabla p_P^{(i)}|_{\vec{0}} = 0$.
\item To the second-order expansion, we have $V_i = I - i\mathcal{H}_i - \frac{1}{2}\mathcal{H}_i^2 + O(\mathcal{H}_i^3)$ and $V_i^\dagger = I + i\mathcal{H}_i - \frac{1}{2}\mathcal{H}_i^2 + O(\mathcal{H}_i^3)$. 
Thus:
\begin{align*}
QV_iQV_i^\dagger = Q(I - i\mathcal{H}_i - \tfrac{1}{2}\mathcal{H}_i^2)Q(I + i\mathcal{H}_i - \tfrac{1}{2}\mathcal{H}_i^2) + O(\mathcal{H}_i^3).
\end{align*}
Expanding to second order and using $Q^2 = I$:
\begin{align*}
QV_iQV_i^\dagger = I + i[\mathcal{H}_i, Q]Q + \mathcal{H}_i^2 - Q\mathcal{H}_i^2Q - Q\mathcal{H}_iQ\mathcal{H}_i + O(\mathcal{H}_i^3).
\end{align*}
We define $\widetilde{\mathcal{H}} := Q\mathcal{H}_i Q$ as the conjugated Hamiltonian. Using $QRQ = (-1)^{s(Q,R)} R$ for Pauli $R$, the Pauli components of $\widetilde{\mathcal{H}}$ are sign-flipped versions of those of $\mathcal{H}_i$. 
Taking the trace, we have
\begin{align*}
\begin{split}
c_Q^{(i)} &= \frac{1}{d_i}\Tr(QV_iQV_i^\dagger) = \frac{1}{d_i}\Tr(e^{-i\widetilde{\mathcal{H}}} e^{i\mathcal{H}_i})\\
&= \frac{1}{d_i}\Tr\Big[(I - i\widetilde{\mathcal{H}} - \tfrac{1}{2}\widetilde{\mathcal{H}}^2)(I + i\mathcal{H}_i - \tfrac{1}{2}\mathcal{H}_i^2)\Big] + O(\mathcal{H}_i^3) \\
&= 1 + \frac{i}{d_i}\Tr(\mathcal{H}_i - \widetilde{\mathcal{H}}) - \frac{1}{2d_i}\Tr(\mathcal{H}_i^2 + \widetilde{\mathcal{H}}^2) + \frac{1}{d_i}\Tr(\widetilde{\mathcal{H}}\mathcal{H}_i) + O(\mathcal{H}_i^3).
\end{split}
\end{align*}
Since $\Tr(\widetilde{\mathcal{H}}) = \Tr(\mathcal{H}_i) = 0$ due to the tracelessness and $\Tr(\widetilde{\mathcal{H}}^2) = \Tr(\mathcal{H}_i^2)$ due to the cyclicity, we have
\begin{align*}
c_Q^{(i)} = 1 - \frac{1}{d_i}\Tr(\mathcal{H}_i^2) + \frac{1}{d_i}\Tr(Q\mathcal{H}_i Q\mathcal{H}_i) + O(\mathcal{H}_i^3).
\end{align*}
We expand $\mathcal{H}_i$ in the Pauli basis: $\mathcal{H}_i = \sum_R \alpha_R R$ where $\alpha_R = \Tr(R\mathcal{H}_i)/d_i$. Using $QRQ = (-1)^{s(Q,R)}R$:
\begin{align*}
\Tr(Q\mathcal{H}_i Q\mathcal{H}_i) = d_i\sum_R (-1)^{s(Q,R)}\alpha_R^2.
\end{align*}
Therefore:
\begin{align*}
c_Q^{(i)} = 1 - 2\sum_{R: s(Q,R)=1} \alpha_R^2 + O(\mathcal{H}_i^3).
\end{align*}

Applying the symplectic Fourier transform to get $p_P^{(i)}$:
\begin{align*}
\begin{split}
p_P^{(i)} &= \frac{1}{d_i^2}\sum_Q (-1)^{s(P,Q)} c_Q^{(i)}\\
&= \underbrace{\frac{1}{d_i^2}\sum_Q (-1)^{s(P,Q)}}_{=\delta_{P,I}=0} - \frac{2}{d_i^2}\sum_R \alpha_R^2 \underbrace{\sum_{Q:s(Q,R)=1}(-1)^{s(P,Q)}}_{=-\frac{d_i^2}{2}\delta_{P,R}} + O(\mathcal{H}_i^3).
\end{split}
\end{align*}
The claim that the restricted characters sum $\sum_{Q:s(Q,R)=1}(-1)^{s(P,Q)}$ evaluates to $-\frac{d_i^2}{2}\delta_{P,R}$ can be seen as follows. 
Split $\sum_Q = \sum_{Q:s(Q,R)=0} + \sum_{Q:s(Q,R)=1}$ and we have
\begin{align*}
\sum_{Q:s(Q,R)=0} (-1)^{s(P,Q)} + \sum_{Q:s(Q,R)=1} (-1)^{s(P,Q)} = d^2\delta_{P,I} = 0,
\end{align*}
and
\begin{align*}
\sum_{Q:s(Q,R)=0} (-1)^{s(P,Q)} - \sum_{Q:s(Q,R)=1} (-1)^{s(P,Q)} = d^2\delta_{P,R}.
\end{align*}
The second identity follows from multiplying by $(-1)^{s(Q,R)}$ as $\sum_Q (-1)^{s(P\oplus R, Q)} = d^2\delta_{P,R}$. 
Subtracting the above two identities, we have $2\sum_{Q:s(Q,R)=1}(-1)^{s(P,Q)} = -d^2\delta_{P,R}$.
Therefore,
\begin{align*}
p_P^{(i)} = \alpha_P^2 + O(\|\mathcal{H}_i\|^3) = \left(\frac{\Tr(P\mathcal{H}_i)}{d_i}\right)^2 + O(\|\mathcal{H}_i\|^3).
\end{align*}
Now substitute the Taylor expansion $\mathcal{H}_i = \sum_j \delta \theta_{i,j} G_{i,j} + O(\norm{\delta\vec{\theta}_i}^2)$:
\begin{align*}
\alpha_P = \frac{\Tr(P\mathcal{H}_i)}{d_i} = \sum_j \delta \theta_{i,j} h_{P,j}^{(i)} + O(\norm{\delta\vec{\theta}_i}^2),
\end{align*}
where $h_{P,j}^{(i)} = \Tr(P G_{i,j})/d_i$. 
Therefore $\alpha_P^2 = (\sum_j \delta \theta_{i,j} h_{P,j}^{(i)})^2 + O(\norm{\delta\vec{\theta}_i}^3)$, and the Hessian at the origin is:
\begin{align*}
\frac{\partial^2(\alpha_P^2)}{\partial(\delta \theta_{i,j})\partial(\delta \theta_{i,k})}\bigg|_{\vec{0}} = 2 h_{P,j}^{(i)} h_{P,k}^{(i)}.
\end{align*}
\item The second-order Taylor coefficients $G_{i,jk}^{(2)}$ contribute to $\alpha_P$ at order $O(\norm{\delta\vec{\theta}_i}^2)$, which contributes to $\alpha_P^2$ at order $O(\norm{\delta\vec{\theta}_i}^3)$. 
These do not affect the Hessian at the origin.
\end{enumerate}
\end{proof}

With \Cref{lem:pP_hessian}, summing Eq.~\eqref{eq:hessian_pP} over $P \in \mathcal{F}_k^{(i)}$ and using linearity of differentiation, we have the following corollary on the Hessian of the detector-relevant error probability $q_{ki}$.
\begin{corollary}[Hessian of $q_{ki}$ at the optimum]
\label{coro:qki_hessian}
The detector-relevant error probability $q_{ki}(\delta\vec{\theta}_i) = \sum_{P \in \mathcal{F}_k^{(i)}} p_P^{(i)}(\delta\vec{\theta}_i)$ satisfies $q_{ki}(\vec{0}) = 0$, $\nabla q_{ki}(\vec{0}) = \vec{0}$, and has Hessian at the optimum:
\begin{align*}
\nabla^2 q_{ki}\big|_{\vec{0}} = 2\sum_{P \in \mathcal{F}_k^{(i)}} \vec{h}_P^{(i)} \big(\vec{h}_P^{(i)}\big)^\top = 2 H_k^{(i)} \big(H_k^{(i)}\big)^\top,
\end{align*}
where $\vec{h}_P^{(i)} \in \mathbb{R}^{\abs{\delta\vec{\theta}}_i}$ has components $h_{P,j}^{(i)} = \Tr(P G_{i,j})/d_i$, and $H_k^{(i)}$ is the $\abs{\delta\vec{\theta}}_i \times |\mathcal{F}_k^{(i)}|$ matrix whose columns are $\{\vec{h}_P^{(i)}\}_{P \in \mathcal{F}_k^{(i)}}$. 
This Hessian is positive semidefinite as it is a sum of rank-1 PSD matrices.
\end{corollary}
\noindent We note that the matrix $\nabla^2 q_{ki}|_{\vec{0}}$ is positive definite if and only if $H_k^{(i)}$ has full row rank $\abs{\delta\vec{\theta}}_i$. 
Equivalently, the vectors $\{\vec{h}_P^{(i)}\}_{P \in \mathcal{F}_k^{(i)}}$ must span $\mathbb{R}^{\abs{\delta\vec{\theta}}_i}$. 
Physically, this means that every direction in parameter space produces at least one Pauli error that flips detector $D_k$.
For a single detector, this is a strong requirement. 
However, the surrogate objective $C$ always averages over \emph{all} detectors, so the relevant convexity condition for $C$ will be weaker.

We then consider $DR_k(\delta\vec{\theta})$. We rewrite the detection rate as $DR_k = F(q_{k1}, q_{k2}, \ldots)$, where $F(\vec{q}) = (1 - \prod_i(1-2q_i))/2$ is a function of the individual detector-relevant error probabilities. 
We compute the Hessian via the chain rule, and we have
\begin{align}\label{eq:derivativeF}
\frac{\partial F}{\partial q_i} = \prod_{j \neq i}(1 - 2q_j), \quad
\frac{\partial^2 F}{\partial q_i^2} = 0, \quad
\frac{\partial^2 F}{\partial q_i\partial q_{i'}} = -2\prod_{j \neq i,i'}(1 - 2q_j)\ \text{for}\ \quad i \neq i'.
\end{align}
Therefore, the Hessian of $DR_k$ with respect to $\delta\vec{\theta}$ has the block structure in blocks indexed by gate pairs $(i, i')$:
\begin{align}
\nabla^2_{\delta\vec{\theta}} DR_k = \sum_i \frac{\partial F}{\partial q_{ki}} \nabla^2_{\delta\vec{\theta}_i} q_{ki} + \sum_{i \neq i'} \frac{\partial^2 F}{\partial q_{ki}\partial q_{ki'}} (\nabla_{\delta\vec{\theta}_i} q_{ki})(\nabla_{\delta\vec{\theta}_{i'}} q_{ki'})^\top.
\label{eq:DR_hessian_chain}
\end{align}
Note that the $\partial^2 F/\partial q_{ki}^2 = 0$ term has been dropped. The first sum is block-diagonal with each term living in the $\delta\vec{\theta}_i$ subspace. The second sum contributes off-diagonal blocks between different gates' parameter spaces.
Now, we are ready to show that $\delta\vec{\theta}=0$ is the optimum of $DR_k(\delta\vec{\theta})$, and present the properties for the Hessian of $DR_k(\delta\vec{\theta})$ at $\delta\vec{\theta}=0$.
\begin{lemma}[Hessian of $DR_k$ at the optimum]
\label{lem:DR_hessian_origin}
At $\delta\vec{\theta} = \vec{0}$:
\begin{align*}
\nabla^2 DR_k\big|_{\vec{0}} = \sum_{i \in \mathcal{R}_k} \nabla^2 q_{ki}\big|_{\vec{0}} = 2\sum_{i \in \mathcal{R}_k}\sum_{P \in \mathcal{F}_k^{(i)}} \vec{h}_P^{(i)} \big(\vec{h}_P^{(i)}\big)^\top \succeq 0.
\end{align*}
The off-diagonal (cross-gate) terms vanish identically because $\nabla q_{ki}|_{\vec{0}} = \vec{0}$.
\end{lemma}
\begin{proof}
At $\delta\vec{\theta} = \vec{0}$: all $q_{ki} = 0$, so $\partial F/\partial q_{ki} = 1$. 
Also $\nabla q_{ki}|_{\vec{0}} = \vec{0}$ from the second property of \Cref{lem:pP_hessian}, so the entire second sum in Eq.~\eqref{eq:DR_hessian_chain} vanishes. 
The first sum reduces to $\sum_i \nabla^2 q_{ki}|_{\vec{0}}$, which is the stated expression by using \Cref{coro:qki_hessian}.
\end{proof}

\paragraph{Convexity of the detection rate $DR_k$ away from $\delta\vec{\theta}=0$ and the convexity radius.}We have shown in \Cref{lem:DR_hessian_origin} that $DR_k(\delta\vec{\theta})$ is convex at the optimum $\delta\vec{\theta}=0$. 
Now, we bound the Hessian of $DR_k$ for nonzero $\delta\vec{\theta}$ to establish a convexity region. 
For convenience, we rewrite the gradient vectors as $\vec{u}_i := \nabla_{\delta\vec{\theta}_i} q_{ki} \in \mathbb{R}^{\abs{\delta\vec{\theta}}_i}$, which live in orthogonal subspaces for different $i$.
The cross-gate contribution in Eq.~\eqref{eq:DR_hessian_chain} can be rewritten using the identity:
\begin{align*}
\sum_{i \neq i'} \vec{u}_i \vec{u}_{i'}^\top = \Big(\sum_i \vec{u}_i\Big)\Big(\sum_{i'} \vec{u}_{i'}\Big)^\top - \sum_i \vec{u}_i \vec{u}_i^\top.
\end{align*}
For small errors, we have both $\vec{u}_i$ and $1-\prod_{j \neq i,i'}(1-2q_{kj})$ as small as $O(\norm{\delta\vec{\theta}})$.
We expand $\nabla^2 DR_k$ to the second-order, and obtain
\begin{align}
\begin{split}
\nabla^2 DR_k &=\sum_i (1-2\bar{q}_{k\setminus i}) \nabla^2 q_{ki}-2\sum_{i\neq i'}\left(\prod_{j \neq i,i'}(1-2q_{kj})\right)\vec{u}_i\vec{u}_{i'}^\top\\
&= \underbrace{\sum_i (1-2\bar{q}_{k\setminus i}) \nabla^2 q_{ki}}_{\text{(I): block-diagonal, PSD}} + \underbrace{2\sum_i \vec{u}_i\vec{u}_i^\top}_{\text{(II): block-diagonal, PSD}} - \underbrace{2\Big(\sum_i \vec{u}_i\Big)\Big(\sum_i \vec{u}_i\Big)^\top}_{\text{(III): rank-1, PSD}}+O(\norm{\delta\vec{\theta}}^3),
\end{split}
\label{eq:hessian_three_terms}
\end{align}
where $\bar{q}_{k\setminus i} = \sum_{j \neq i} q_{kj}$ to leading order in the small-error regime. 
We note that terms (I) and (II) are positive semidefinite given that $(1-2\bar{q}_{k\setminus i})>0$, while term (III) is negative semidefinite. 

We then characterize the regime where term (I) and term (II) dominate term (III) in Eq.~\eqref{eq:hessian_three_terms}.
We first prove the following scaling properties of these three terms.
\begin{proposition}[Scaling of the three terms]
\label{prop:three_terms}
In the regime $\|\delta\vec{\theta}_i\| \leq r$ for all $i \in \mathcal{R}_k$:
\begin{itemize}
\item[\textbf{(I)}] The block-diagonal Hessian term scales as $O(1)$. Its minimum eigenvalue is at least $\lambda_k^{\min} - O(r)$, where
\begin{align*}
\lambda_k^{\min} := \lambda_{\min}\Big(\sum_{i \in \mathcal{R}_k} \nabla^2 q_{ki}\big|_{\vec{0}}\Big),\quad
\lambda_k^{\max} := \lambda_{\max}\Big(\sum_{i \in \mathcal{R}_k} \nabla^2 q_{ki}\big|_{\vec{0}}\Big)
\end{align*}
are the minimum and maximum eigenvalues of the Hessian of $DR_k$ at the optimum.
\item[\textbf{(II)}] The correction term $2\sum_i \vec{u}_i\vec{u}_i^\top$ is $O(r^2)$ since each $\vec{u}_i = \nabla q_{ki} = O(r)$.
\item[\textbf{(III)}] The rank-1 negative term $2(\sum_i \vec{u}_i)(\sum_i \vec{u}_i)^\top$ is also $O(r^2)$, and its single nonzero eigenvalue is bounded by $2|\mathcal{R}_k| \sum_i \|\vec{u}_i\|^2 = O(r^2)$.
\end{itemize}
\end{proposition}
\begin{proof}
For term (I), the deviation of the block-diagonal term from its value at the origin has two contributions. 
The first is from $\partial F/\partial q_{ki}$ deviating from $1$.
Since $\partial F/\partial q_{ki} = 1 - 2\bar{q}_{k\setminus i}$ where $\bar{q}_{k\setminus i} = O(r^2)$ (each $q_{kj} = O(r^2)$), this shifts the coefficient of $\nabla^2 q_{ki}$ by $O(r^2)$. 
The second is from $\nabla^2 q_{ki}$ itself, deviating from its value at the origin. 
This Lipschitz deviation is $O(r)$ controlled by third derivatives. 
Combining these two contributions, we have $\|(\text{I}) - \nabla^2 DR_k|_{\vec{0}}\| = O(r)$.

For term (II) and term (III), since $\nabla q_{ki}|_{\vec{0}} = 0$, we have $\|\vec{u}_i\| = \|\nabla q_{ki}\| = O(r)$. 
Thus both terms have eigenvalues $O(r^2)$, which is subleading compared to the $O(1)$ minimum eigenvalue of term (I).
\end{proof}

Given the scaling properties in \Cref{prop:three_terms}, we can show the following convexity radius for $DR_k(\delta\vec{\theta})$ near $\delta\vec{\theta}=0$.
To better quantitatively use \Cref{prop:three_terms} for term (I), we will need to introduce the Hessian Lipschitz factor $L_k^{\text{Hessian}}$, which is defined as $\norm{\nabla^2 DR_k(\delta\vec{\theta}_1)-\nabla^2 DR_k(\delta\vec{\theta}_2)}\leq L_k^{\text{Hessian}}\norm{\delta\vec{\theta}_1-\delta\vec{\theta}_2}$.
\begin{lemma}[Convexity radius of $DR_k$]
\label{lem:DRk_convexity}
If $\lambda_k^{\min} > 0$, then $DR_k$ is strictly convex on the ball $\|\delta\vec{\theta}\| < \theta_k^{(th)}$, where
\begin{align*}
\theta_k^{(th)} = \frac{\lambda_k^{\min}}{L_k},
\end{align*}
where $L_k=L_k^{\mathrm{Hessian}}+\abs{\mathcal{R}_k}\lambda_k^{\max}+2\abs{\mathcal{R}_k}\sum_{i \in \mathcal{R}_k} \sum_j \|G_{i,j}\|^2$.
\end{lemma}

\begin{proof}
The minimum eigenvalue of $\nabla^2 DR_k(\delta\vec{\theta})$ satisfies
\begin{align*}
\begin{split}
\lambda_{\min}(\nabla^2 DR_k) &\geq \lambda_{\min}(\text{term (I)}) - \|\text{term (III)}\|\\
&\geq (\lambda_k^{\min} - L_k^{(\mathrm{Hess})}r) - \abs{\mathcal{R}_k}\lambda_k^{\max}r^2- 2\abs{\mathcal{R}_k}\sum_{i \in \mathcal{R}_k} \sum_j \|G_{i,j}\|^2\cdot r^2,\\
&\geq \lambda_k^{\min}-\left(L_k^{(\mathrm{Hess})}+\abs{\mathcal{R}_k}\lambda_k^{\max}+2\abs{\mathcal{R}_k}\sum_{i \in \mathcal{R}_k} \sum_j \|G_{i,j}\|^2\right)\cdot r\\
&=\lambda_k^{\min}-L_k\cdot r.
\end{split}
\end{align*}
where we used that (II) $\succeq 0$ and combine the third point of \Cref{prop:three_terms}. 
This is positive when $r < \lambda_k^{\min}/L_k$.
\end{proof}

\paragraph{Convexity of the surrogate objective function $C$.}Given that $DR_k$ is convex at $\delta\vec{\theta}$ and the convexity radius, we now consider the convexity of the surrogate objective function $C$.
Recall that the surrogate objective is the average detection rate
\begin{align*}
C(\delta\vec{\theta}) = \frac{1}{N_D}\sum_{k=1}^{N_D} DR_k(\delta\vec{\theta}).
\end{align*}
Its Hessian is the average of the individual Hessians:
\begin{align}\label{eq:nabla2C_unitary}
\nabla^2 C = \frac{1}{N_D}\sum_k \nabla^2 DR_k.
\end{align}
For each gate $i$, we define the detectable projection matrix $(M^{(i)}_{jk})_{j,k}$:
\begin{align*}
M^{(i)}_{jk} := \frac{1}{N_D}\sum_{\ell:i \in \mathcal{R}_\ell} \sum_{P \in \mathcal{F}_\ell^{(i)}} h_{P,j}^{(i)} h_{P,k}^{(i)} = \frac{1}{N_D d_i^2}\sum_{\ell:i \in \mathcal{R}_\ell} \sum_{P \in \mathcal{F}_\ell^{(i)}} \Tr(P G_{i,j})\Tr(P G_{i,k}).
\end{align*}
This is a sum of rank-1 PSD matrices, hence $M^{(i)} \succeq 0$. The Hessian of $C$ at the optimum is block-diagonal with blocks $2M^{(i)}$:
\begin{align*}
\nabla^2 C\big|_{\vec{0}} = 2\ \mathrm{diag}\big(M^{(1)}, M^{(2)}, \ldots\big).
\end{align*}

To formally study the convexity of $C$, we require the following assumption.
\begin{assumption}[Detectability condition for unitary error]
\label{ass:detectability}
We say the QEC circuit satisfies the \textbf{detectability condition} if $M^{(i)} \succ 0$ for all gates $i$. Equivalently:
\begin{quote}
For every gate $i$ and every nonzero parameter perturbation $\vec{v} \in \mathbb{R}^{\abs{\delta\vec{\theta}}_i}$, the error Hamiltonian $\sum_j v_j G_{i,j}$ has nonzero overlap with at least one Pauli operator $P$ that is detectable by some detector in the circuit.
\end{quote}
\end{assumption}
\noindent We remark that \Cref{ass:detectability} is a special case of the later assumption \Cref{ass:detectability_cptp} when the error is a CPTP map. 
In Appendix~\ref{sec:strong_convexity}, we show that \Cref{ass:detectability_cptp} and its special case \Cref{ass:detectability}, are almost satisfied for any parametrization.
In particular, there are two cases that lead to $\vec v^\top\nabla^2 C(\vec 0)\vec v=0$ for some unit $\vec v$. 
The first one is the structural case, where moving along $\vec v$ will never change the value of $C$, and thus $\vec v$ is a gauge direction of $C$.
We can consider optimizing in the quotient space of $\vec v$, and the quotient $\nabla^2 C(\vec 0)$ is positive definite.
The second case is the non-structured case, where $\vec v^\top\nabla^2 C(\vec 0)\vec v=0$ only at the particular point $\vec 0$.
We then show that even a small amount of (possibly stochastic) noise will ensure $\nabla^2 C(\vec 0)$ to be positive definite almost for sure.

We can then obtain our first main result on the convexity of the surrogate objective function $C$ for control induced unitary error using \Cref{lem:DRk_convexity} and Eq.~\eqref{eq:nabla2C_unitary}.

\begin{theorem}[Convexity of the surrogate objective for unitary error]\label{thm:convexity_unitary}
Assume the detectability condition in \Cref{ass:detectability} holds, and we denote:
\begin{align*}
\lambda_{C}^{\min} := \min_i \lambda_{\min}(M^{(i)}) > 0.
\end{align*}
Then the surrogate objective $C(\delta\vec{\theta})$ is strictly convex on the ball $\|\delta\vec{\theta}\| < \theta_C^{(th)}$, where
\begin{align*}
\theta_C^{(th)} = \frac{2\lambda_{C}^{\min}}{L_C},
\end{align*}
with $L_C = \frac{1}{N_D}\sum_k L_k$ as defined in \Cref{lem:DRk_convexity}. This guarantees that within the ball $\|\delta\vec{\theta}\| < \theta_C^{(th)} $:
\begin{enumerate}
\item $C$ has a \textbf{unique global minimum} at $\delta\vec{\theta} = \vec{0}$.
\item $\nabla^2 C$ is positive semidefinite at any point, making $C$ a convex function within this regime.
\end{enumerate}
\end{theorem}

\subsection{Control induced general CPTP map error}\label{sec:cptp_convexity}
We further extend \Cref{thm:convexity_unitary} to the case when the error is a general complete positive trace-preserving (CPTP) channel. 

\paragraph{The error model setup.}Again, we index each gate in the QEC circuit by $i$, and further assume that it acts on $n_i$ qubits and has parameters $\vec{\theta}_i$. 
For each gate $i$, we consider an $n_i$-qubit system with Hilbert space dimension $d_i = 2^{n_i}$. 
The error of a gate is a CPTP map $\mathcal{E}_{\vec{\theta}_i}$ depending smoothly on drift parameters $\vec{\theta}_i$, with $\mathcal{E}_{i,\vec{0}} = \mathrm{id}$ the identity channel:
\begin{align*}
\mathcal{E}_{i,\delta\vec{\theta}_i}(\rho) = \sum_\alpha K_{i,\alpha}(\delta\vec{\theta}_i) \rho K_{i,\alpha}(\delta\vec{\theta}_i)^\dagger, \qquad \sum_\alpha K_{i,\alpha}^\dagger K_{i,\alpha} = I_{d_i}.
\end{align*}
All Kraus operators are $d_i \times d_i$ coefficient matrices. Since $\mathcal{E}_{i,\vec{0}} = \mathrm{id}$, the Kraus operators expand as:
\begin{align*}
K_{i,0}(\delta\vec{\theta}_i) = I_{d_i} - i\sum_j \delta \theta_{i,j} G_{i,j} + \sum_{j,k} \delta \theta_{i,j} \delta \theta_{i,k} G_{i,jk} + O(\norm{\delta\vec{\theta}_i}^3), \quad K_{i,\alpha\geq 1}(\delta\vec{\theta}_i) = \sum_j \delta \theta_{i,j} L_{i,\alpha,j} + O(\norm{\delta\vec{\theta}_i}^2),
\end{align*}
where $G_{i,j}$, $G_{i,jk}$, and $L_{i,\alpha,j}$ are $d \times d$ coefficient operators. 
The CPTP condition $\sum_\alpha K_{i,\alpha}^\dagger K_{i,\alpha} = I_{d_i}$ gives $G_{i,j}^\dagger = G_{i,j}$ (Hermitian) at the first order. 
At the second order, it gives
\begin{align*}
G_{i,jk} + G_{i,jk}^\dagger = -G_{i,j} G_{i,k} + \sum_{\alpha \geq 1} L_{i,\alpha,j}^\dagger L_{i,\alpha,k}.
\end{align*}
We also assume $\Tr(G_{i,j}) = 0$ for all $j$, which holds for standard error generators.

The $n_i$-qubit Pauli group $\mathcal{P}_{n_i}$ consists of $d_i^2$ operators satisfying $P^2 = I_{d_i}$, $\Tr(P) = d_i\delta_{P,I}$, and $PQP = (-1)^{s(P,Q)} Q$ for Paulis $P, Q$, where $s(P,Q) \in \{0,1\}$ is the symplectic inner product. 
The Pauli-twirled channel reads
\begin{align*}
\mathcal{T}_{\mathcal{E}_i,\delta\vec{\theta}_i}(\rho) = \sum_{P \in \mathcal{P}_{n_i}} p_P^{(i)}(\delta\vec{\theta}_i) P\rho P, \qquad p_P^{(i)} \geq 0,\quad \sum_P p_P^{(i)} = 1,
\end{align*}
where
\begin{align}
p_P^{(i)} = \frac{1}{d_i^2}\sum_{Q \in \mathcal{P}_{n_i}} (-1)^{s(P,Q)} c_Q^{(i)}, \qquad
c_Q^{(i)}(\delta\vec{\theta}_i) = \sum_\alpha \frac{1}{d_i}\Tr\big(Q K_{i,\alpha} Q K_{i,\alpha}^\dagger\big), \qquad Q \in \mathcal{P}_{n_i}.
\label{eq:cQ}
\end{align}

Again, for a detector $D_k$, let $\mathcal{R}_k$ denote its detecting region that consists of the set of gates whose errors can flip $D_k$. 
For each gate $i \in \mathcal{R}_k$, let $\mathcal{F}_k^{(i)} \subset \mathcal{P}_{n_i} \setminus \{I\}$ be the set of nontrivial Pauli errors on gate $i$ that flip detector $D_k$. 
When the gate errors are independent, the detector-relevant error probability is again 
\begin{align*}
q_{ki}(\delta\vec{\theta}_i) = \sum_{P \in \mathcal{F}_k^{(i)}} p_P^{(i)}(\delta\vec{\theta}_i),
\end{align*}
and the exact detection event rate is 
\begin{align*}
DR_k(\delta\vec{\theta}) = \frac{1 - \prod_{i \in \mathcal{R}_k}\big(1 - 2q_{ki}(\delta\vec{\theta}_i)\big)}{2}.
\end{align*}
For the correlated error case, the above equation can still approximately estimate the detection event rate up to the second order according to \appref{sec:unitary_convexity}.
And the surrogate objective function is $C(\delta\vec{\theta})=N_D^{-1}\sum_{k=1}^{N_D}DR_k(\delta\vec{\theta})$. Next, we follow the same route as the unitary error case in \appref{sec:unitary_convexity} to present the convexity property of $C$ when the noise is a general CPTP map.

\paragraph{Convexity of the surrogate objective function $C$.}To start with, we prove the following analog of \Cref{lem:pP_hessian} in the CPTP error case.
\begin{lemma}[Properties of $p_P^{(i)}$ at the optimum]\label{lem:pP_hessian_cptp}
For any nontrivial Pauli $P\neq I$ in the Pauli group $\mathcal{P}_{n_i}$, we have
\begin{align*}
p_{P}^{(i)}(\delta\vec{\theta}_i) = \bigg|\sum_j \delta \theta_{i,j} h_{P,j}^{(i)}\bigg|^2 + \sum_{\alpha \geq 1}\bigg|\sum_j \delta \theta_{i,j} \ell_{\alpha,P,j}^{(i)}\bigg|^2 + O(\norm{\delta\vec{\theta}_i}^3),
\end{align*}
where
\begin{align*}
h_{P,j}^{(i)} = \frac{\Tr(P G_{i,j})}{d_i} \in \mathbb{R}, \qquad \ell_{\alpha,P,j}^{(i)} = \frac{\Tr(P L_{i,\alpha,j})}{d_i} \in \mathbb{C}.
\end{align*}
In particular: $p_{P}^{(i)}(\vec{0}) = 0$, $\nabla p^{(i)}_{P}(\vec{0}) = \vec{0}$, and
\begin{align}
\nabla^2 p_{P}^{(i)}\big|_{\vec{0}} = 2\vec{h}_{P}^{(i)}\vec{h}_{P}^{(i)\top} + 2\sum_\alpha \Re\big(\vec{\ell}_{\alpha,P}^{(i)}\vec{\ell}_{\alpha,P}^{(i)\dagger}\big) \succeq 0,
\label{eq:hessian_pP_kraus}
\end{align}
where $\vec{h}_{P}^{(i)}$ and $\vec{\ell}_{\alpha,P}^{(i)}$ are vectors with the $j$-th entry being $h_{P,j}^{(i)}$ and $\ell^{(i)}_{\alpha,P,j}$, respectively.
\end{lemma}

\begin{proof}
Before we proceed to the main proof, we will need to prove two facts.
\begin{enumerate}
\item For any $d_i \times d_i$ operators $A$, $B$, and nontrivial Pauli $P \neq I$:
\begin{align}
\frac{1}{d_i^2}\sum_{Q \in \mathcal{P}_{n_i}} (-1)^{s(P,Q)}\ \frac{1}{d_i}\Tr(Q A Q B^\dagger) = \frac{\Tr(P A)}{d_i}\cdot \frac{\overline{\Tr(P B)}}{d_i}.
\label{eq:identity}
\end{align}
To see this, we expand $A = \sum_R a_R R$ in the Pauli basis with $a_R = \Tr(RA)/d_i$, and $B = \sum_S b_S S$ in the Pauli basis with $b_S = \Tr(SB)/d_i$. 
Using $QRQ = (-1)^{s(Q,R)}R$:
\begin{align*}
\frac{1}{d_i}\Tr(QAQB^\dagger) = \sum_R (-1)^{s(Q,R)} a_R \bar{b}_R.
\end{align*}
Apply $\frac{1}{d_i^2}\sum_Q (-1)^{s(P,Q)+s(Q,R)} = \delta_{P,R}$ (character orthogonality). This selects $R = P$, giving $a_{P}\bar{b}_{P}$ and proves Eq.~\eqref{eq:identity}.
\item For any $P \neq I$ independent of $Q$:
\begin{align}
\sum_Q (-1)^{s(P,Q)}=0.
\label{eq:Qindep}
\end{align}
\end{enumerate}

We expand $c_Q^{(i)}$ from Eq.~\eqref{eq:cQ} to second order in $\delta\vec{\theta}$, splitting into $\alpha = 0$ and $\alpha \geq 1$. 
\begin{itemize}
\item \textbf{From $K_{i,\alpha \geq 1}$:} Using $K_{i,\alpha} = \sum_j \delta \theta_{i,j} L_{i,\alpha,j} + O(\norm{\delta\vec{\theta}}^2)$:
\begin{align*}
\sum_{\alpha \geq 1} \frac{1}{d_i}\Tr(QK_{i,\alpha} QK_{i,\alpha}^\dagger) = \sum_{\alpha \geq 1}\sum_{j,k} \delta \theta_{i,j} \delta \theta_{i,k} \frac{1}{d_i}\Tr(QL_{i,\alpha,j}QL_{i,\alpha,k}^\dagger) + O(\norm{\delta\vec{\theta}}^3).
\end{align*}
Apply Eq.~\eqref{eq:identity} ($A = L_{i,\alpha,j}$, $B = L_{i,\alpha,k}$) then sum:
\begin{align*}
\xrightarrow{\text{Eq.~\eqref{eq:cQ}}} \sum_\alpha \sum_{j,k} \delta \theta_{i,j} \delta \theta_{i,k} \ell_{\alpha,P,j}^{(i)}\overline{\ell_{\alpha,P,k}^{(i)}} = \sum_\alpha \Big|\sum_j \delta \theta_{i,j} \ell_{\alpha,P,j}^{(i)}\Big|^2.
\end{align*}
\item \textbf{From $K_{i,0}$:} Using $K_{i,0} = I_{d_i} - i\sum_j \delta \theta_{i,j} G_{i,j} + \sum_{jk}\delta \theta_{i,j}\delta \theta_{i,k} G_{i,jk} + O(\norm{\delta\vec{\theta}}^3)$:
\begin{align*}
\frac{1}{d_i}\Tr(QK_{i,0}QK_{i,0}^\dagger) = 1 + \sum_{j,k}\delta \theta_{i,j}\delta \theta_{i,k}\bigg[\underbrace{\frac{\Tr(QG_{i,j}QG_{i,k})}{d_i}}_{\text{$Q$-dependent}} + \underbrace{\frac{\Tr(G_{i,jk} + G_{i,jk}^\dagger)}{d_i}}_{\text{$Q$-independent}}\bigg] + O(\norm{\delta\vec{\theta}}^3).
\end{align*}
The first-order terms vanish by $\Tr(G_{i,j}) = 0$.
The $Q$-dependent part, we apply Eq.~\eqref{eq:identity} with $A = G_{i,j}$ and $B = G_{i,k}$:
\begin{align*}
\xrightarrow{\text{Eq.~\eqref{eq:cQ}}} \sum_{j,k}\delta \theta_{i,j}\delta \theta_{i,k} h_{P,j}^{(i)} h_{P,k}^{(i)} = \Big|\sum_j \delta \theta_{i,j} h_{P,j}^{(i)}\Big|^2. 
\end{align*}
The $Q$-independent part: $\Tr(G_{i,jk} + G_{i,jk}^\dagger)/d_i$ does not depend on $Q$, so by Eq.~\eqref{eq:Qindep} it vanishes for $P \neq I$.
\item \textbf{Zeroth-order contribution:} The constant $c_Q^{(i)} = 1 + O(\norm{\delta\vec{\theta}}^2)$ gives $p_{P} = \delta_{P,I} + O(\norm{\delta\vec{\theta}}^2) = 0$ for $P \neq I$.
\end{itemize}
Combining all contributions gives \Cref{lem:pP_hessian_cptp}. 
The Hessian~\eqref{eq:hessian_pP_kraus} follows from $\nabla^2|f|^2\big|_{\vec{0}} = 2\Re(\vec{c}\vec{c}^\dagger)$ for $f = \sum_j \delta \theta_{i,j} c_j$ with $f(\vec{0}) = 0$, applied to each squared term. 
Each $\Re(\vec{c}\vec{c}^\dagger) \succeq 0$ since it is the real part of a rank-1 PSD Hermitian matrix.
\end{proof}

Using \Cref{lem:pP_hessian_cptp} and an analog to \Cref{coro:qki_hessian}, we have the detector-relevant error probability $q_{ki}(\delta\vec{\theta}_i) = \sum_{P \in \mathcal{F}_k^{(i)}} p_P^{(i)}(\delta\vec{\theta}_i)$ satisfies $q_{ki}(\vec{0}) = 0$, $\nabla q_{ki}(\vec{0}) = \vec{0}$, and has Hessian at the optimum:
\begin{align*}
\nabla^2 q_{ki}\big|_{\vec{0}} = 2\sum_{P \in \mathcal{F}_k^{(i)}} \left(\vec{h}_P^{(i)} \big(\vec{h}_P^{(i)}\big)^\top+\sum_{\alpha\geq 1} \Re\big(\vec{\ell}_{\alpha,P}^{(i)}\vec{\ell}_{\alpha,P}^{(i)\dagger}\big)\right) = 2 \left(H_k^{(i)} \big(H_k^{(i)}\big)^\top+\sum_{\alpha\geq 1}\Re\big(\mathcal{L}_{\alpha,k}^{(i)}\mathcal{L}_{\alpha,k}^{(i)\dagger}\big)\right).
\end{align*}
Here, $H_k^{(i)}$ is the $\abs{\delta\vec{\theta}}_i \times |\mathcal{F}_k^{(i)}|$ matrix whose columns are $\{\vec{h}_P^{(i)}\}_{P \in \mathcal{F}_k^{(i)}}$ where $\vec{h}_P^{(i)}$ is the vector with the $j$-th entry $h^{(i)}_{P,j}$. 
Similarly, $\mathcal{L}_{\alpha,k}^{(i)}$ is the $\abs{\delta\vec{\theta}}_i \times |\mathcal{F}_k^{(i)}|$ matrix whose columns are $\{\vec{\ell}_{\alpha,P}^{(i)}\}_{P \in \mathcal{F}_k^{(i)}}$ where $\vec{\ell}_{\alpha,P}^{(i)}$ is the vector with the $j$-th entry $\ell_{\alpha,P,j}^{(i)}$.
This Hessian is positive semidefinite as it is a sum of rank-1 PSD matrices.

We then consider $DR_k(\delta\vec{\theta})$. We rewrite the detection rate as $DR_k = F(q_{k1}, q_{k2}, \ldots)$, where $F(\vec{q}) = (1 - \prod_i(1-2q_i))/2$ is a function of the individual detector-relevant error probabilities. 
Using the same chain rule, we can obtain an analog of \Cref{lem:DR_hessian_origin}.
In particular, we have the following Hessian of $DR_k$ with respect to $\delta\vec{\theta}$ has the block structure in blocks indexed by gate pairs $(i, i')$:
\begin{align}
\nabla^2 DR_k = \sum_i \frac{\partial F}{\partial q_{ki}} \nabla^2 q_{ki} + \sum_{i \neq i'} \frac{\partial^2 F}{\partial q_{ki}\partial q_{ki'}} (\nabla q_{ki})(\nabla' q_{ki'})^\top.
\label{eq:DR_hessian_chain_cptp}
\end{align}
And at $\delta\vec{\theta} = \vec{0}$, we have
\begin{align*}
\nabla^2 DR_k\big|_{\vec{0}} = \sum_{i \in \mathcal{R}_k} \nabla^2 q_{ki}\big|_{\vec{0}} = 2\sum_{i \in \mathcal{R}_k}\sum_{P \in \mathcal{F}_k^{(i)}} \left(\vec{h}_P^{(i)} \big(\vec{h}_P^{(i)}\big)^\top+\sum_{\alpha\geq 1} \Re\big(\vec{\ell}_{\alpha,P}^{(i)}\vec{\ell}_{\alpha,P}^{(i)\dagger}\big)\right) \succeq 0.
\end{align*}
The off-diagonal (cross-gate) terms vanish identically because $\nabla q_{ki}|_{\vec{0}} = \vec{0}$.
We then proceed to the convexity of the detection rate $DR_k$ away from $\delta\vec{\theta}=0$.
Again, , we rewrite the gradient vectors as $\vec{u}_i := \nabla_{\delta\vec{\theta}_i} q_{ki} \in \mathbb{R}^{\abs{\delta\vec{\theta}}_i}$, which live in orthogonal subspaces for different $i$.
The cross-gate contribution in Eq.~\eqref{eq:DR_hessian_chain_cptp} can be rewritten using the identity $\sum_{i \neq i'} \vec{u}_i \vec{u}_{i'}^\top = \Big(\sum_i \vec{u}_i\Big)\Big(\sum_{i'} \vec{u}_{i'}\Big)^\top - \sum_i \vec{u}_i \vec{u}_i^\top$.
For small errors, we have both $\vec{u}_i$ and $1-\prod_{j \neq i,i'}(1-2q_{kj})$ as small as $O(\norm{\delta\vec{\theta}})$.
We expand $\nabla^2 DR_k$ to the second-order, and obtain the same form as Eq.~\eqref{eq:hessian_three_terms}
\begin{align*}
\nabla^2 DR_k &= \underbrace{\sum_i (1-2\bar{q}_{k\setminus i}) \nabla^2 q_{ki}}_{\text{(I): block-diagonal, PSD}} + \underbrace{2\sum_i \vec{u}_i\vec{u}_i^\top}_{\text{(II): block-diagonal, PSD}} - \underbrace{2\Big(\sum_i \vec{u}_i\Big)\Big(\sum_i \vec{u}_i\Big)^\top}_{\text{(III): rank-1, PSD}}+O(\norm{\delta\vec{\theta}}^3),
\end{align*}
where $\bar{q}_{k\setminus i} = \sum_{j \neq i} q_{kj}$ to leading order in the small-error regime. 
We note that terms (I) and (II) are positive semidefinite given that $(1-2\bar{q}_{k\setminus i})>0$, while term (III) is negative semidefinite. 

We then compute the convexity radius of $DR_k$  and characterize the regime where term (I) (and term (II)) dominate term (III) in Eq.~\eqref{eq:hessian_three_terms}.
We introduce the Hessian Lipschitz factor $L_k^{\text{Hessian}}$, which is defined as $\norm{\nabla^2 DR_k(\delta\vec{\theta}_1)-\nabla^2 DR_k(\delta\vec{\theta}_2)}\leq L_k^{\text{Hessian}}\norm{\delta\vec{\theta}_1-\delta\vec{\theta}_2}$.
We also denote 
\begin{align*}
\lambda_k^{\min} := \lambda_{\min}\Big(\sum_{i \in \mathcal{R}_k} \nabla^2 q_{ki}\big|_{\vec{0}}\Big)\quad,
\lambda_k^{\max} := \lambda_{\max}\Big(\sum_{i \in \mathcal{R}_k} \nabla^2 q_{ki}\big|_{\vec{0}}\Big)
\end{align*}
as the minimum and the maximum eigenvalue of the Hessian of $DR_k$ at the optimum.
Therefore, in the regime $\|\delta\vec{\theta}_i\| \leq r$ for all $i \in \mathcal{R}_k$, we have a similar radius as \Cref{lem:DRk_convexity}.
If $\lambda_k^{\min} > 0$, then $DR_k$ is strictly convex on the ball $\|\delta\vec{\theta}\| < \theta_k^{(th)}$, where
\begin{align}
\theta_k^{(th)} = \frac{\lambda_k^{\min}}{L_k},
\label{eq:pk_star_cptp}
\end{align}
where $L_k=L_k^{\mathrm{Hessian}}+\abs{\mathcal{R}_k}\lambda_k^{\max}+2\abs{\mathcal{R}_k}\sum_{i \in \mathcal{R}_k}\sum_j(\|G_{i,j}\|^2+\sum_{\alpha}\norm{L_{i,\alpha,j}}^2)$.

Given that $DR_k$ is convex at $\delta\vec{\theta}$ and the convexity radius, we now consider the convexity of the surrogate objective function $C$.
Recall that the surrogate objective is the average detection rate $C(\delta\vec{\theta}) = \frac{1}{N_D}\sum_{k=1}^{N_D} DR_k(\delta\vec{\theta})$.
Its Hessian is the average of the individual Hessians $\nabla^2 C = \frac{1}{N_D}\sum_k \nabla^2 DR_k$.
For each gate $i$, we define the detectable projection matrix $(M^{(i),\text{CPTP}}_{jk})_{j,k}$:
\begin{align}\label{eq:M_cptp}
M^{(i),\text{CPTP}}_{jk} := \frac{1}{N_D}\sum_{\ell: i \in R_\ell} \sum_{P \in \mathcal{F}_\ell^{(i)}} \left(h_{P,j}^{(i)} h_{P,k}^{(i)}+\sum_{\alpha\geq 1}\ell_{\alpha,P,j}^{(i)}\ell_{\alpha,P,k}^{(i)}\right).
\end{align}
This is a sum of rank-1 PSD matrices, hence $M^{(i),\text{CPTP}} \succeq 0$. The Hessian of $C$ at the optimum is block-diagonal with blocks $2M^{(i)}$:
\begin{align}
\label{eq:nabla2C_cptp}
\nabla^2 C\big|_{\vec{0}} = 2\ \mathrm{diag}\big(M^{(1),\text{CPTP}}, M^{(2),\text{CPTP}}, \ldots\big).
\end{align}
To formally study the convexity of $C$, we require the following analog assumption of \Cref{ass:detectability}.
\begin{assumption}[Detectability condition for CPTP error]
\label{ass:detectability_cptp}
We say the QEC circuit satisfies the \textbf{detectability condition} if $M^{(i),\text{CPTP}} \succ 0$ for all gates $i$. Equivalently:
\begin{quote}
For every gate $i$ and every nonzero parameter perturbation $\vec{v} \in \mathbb{R}^{\abs{\delta\vec{\theta}}_i}$, the error channel $\mathcal{E}_{i,\delta\vec{\theta}_i}(\cdot)$ has nonzero overlap with at least one Pauli operator $P$ that is detectable by some detector in the circuit.
\end{quote}
\end{assumption}
\noindent As mentioned at the remark after \Cref{ass:detectability}, in \Cref{sec:strong_convexity}, we show that \Cref{ass:detectability_cptp} is almost satisfied for any parametrization.

We can then obtain our result on the convexity of the surrogate objective function $C$ for control induced CPTP error using Eq.~\eqref{eq:pk_star_cptp} and Eq.~\eqref{eq:nabla2C_cptp}.
\begin{theorem}[Convexity of the surrogate objective for CPTP error]\label{thm:convexity_cptp}
Assume the detectability condition in \Cref{ass:detectability_cptp} holds, and we denote:
\begin{align*}
\lambda_{C}^{\min} := \min_i \lambda_{\min}(M^{(i),\text{CPTP}}) > 0.
\end{align*}
Then the surrogate objective $C(\delta\vec{\theta})$ is strictly convex on the ball $\|\delta\vec{\theta}\| < \theta_C^{(th)}$, where
\begin{align*}
\theta_C^{(th)} = \frac{2\lambda_{C}^{\min}}{L_C},
\end{align*}
with $L_C = \frac{1}{N_D}\sum_k L_k$ as defined in Eq.~\eqref{eq:pk_star_cptp}. This guarantees that within the ball $\|\delta\vec{\theta}\| < \theta_C^{(th)} $:
\begin{enumerate}
\item $C$ has a \textbf{unique global minimum} at $\delta\vec{\theta} = \vec{0}$.
\item $\nabla^2 C$ is positive semidefinite at any point, making $C$ a convex function within this regime.
\end{enumerate}
\end{theorem}

\subsection{Control induced leakage error and atom loss}\label{sec:leakage_convexity}
Leakage error is a common error in all types of quantum devices.
Unfortunately, a leakage error cannot be represented as a CPTP channel as it is not trace-preserving. 
As a result, it falls out of the convexity discussions in \appref{sec:unitary_convexity} and \appref{sec:cptp_convexity}. 
Atom loss can also be regarded as a leakage error if not detected immediately, or falls into an erasure channel, which is a CPTP map in \appref{sec:cptp_convexity}, if detected~\cite{wu2022erasure,cong2022hardware,scholl2023erasure,chow2024circuit}.
Here, assuming that the detector readout is block-diagonal in the computational-leaked decomposition, we prove a similar convexity result for $C$ when control drift can cause both computational errors and leakage out of the computational subspace.

\paragraph{The error model setup.}We assume that each physical qubit lives in a Hilbert space $\mathcal H^{\text{phys}}$ of dimension $d_i^{\text{phys}}\ge 3$. 
The computational subspace is $\mathcal H^{\text{comp}}=\operatorname{span}\{|0\rangle,|1\rangle\}$ and the leaked subspace is its orthogonal complement $\mathcal H^{\text{phys}}=\mathcal H^{\text{comp}}\oplus \mathcal H^{\text{leak}}$.
We again index each gate in the QEC circuit as $i$ and assume that it acts on $n_i$ qubits, the full Hilbert space of this gate has dimension
\begin{align*}
D_i=\prod_{a=1}^n d_a^{\text{phys}},
\end{align*}
while the computational subspace has dimension $d_i=2^{n_i}$, and the leaked subspace has dimension $D_i-d_i$.
Let $\Pi^{(i)}_C$ and $\Pi^{(i)}_L=I_{D_i}-\Pi^{(i)}_C$ be the orthogonal projectors onto these two subspaces.
Assume the error is caused by a control drift $\delta\vec{\theta}_i = \vec{\theta}_i - \vec{\theta}_i^{*}$, the error unitary reads
\begin{align*}
V_i(\delta\vec \theta_i)=U_i(\vec \theta^*_i+\delta\vec \theta_i)U(\vec \theta^*_i)^\dagger,\quad\text{with}\quad V_i(\vec 0)=I_{D_i}.
\end{align*}
Similar to the unitary case in \appref{sec:unitary_convexity}, for $\delta\vec \theta_i$ sufficiently small, choose the local logarithm
\begin{align*}
V_i(\delta\vec \theta_i)=e^{-i\mathcal H_i(\delta\vec \theta_i)},\quad\text{with}\quad
\mathcal H_i(\vec 0)=0,\quad
\mathcal H_i(\delta\vec \theta_i)=\sum_j \delta \theta_{i,j}G_{i,j}+O(\|\delta\vec \theta_i\|^2).
\end{align*}
Here, each $G_{i,j}$ is the Hermitian coefficient matrix, and $G_{i,j}$ and $V_i$ block form
\begin{align*}
G_{i,j}=\begin{pmatrix}
G_{i,j}^{CC} & G_{i,j}^{CL}\\
(G_{i,j}^{CL})^\dagger & G_{i,j}^{LL}
\end{pmatrix},\quad V_i=\begin{pmatrix}
V_i^{CC} & V_i^{CL}\\
V_i^{LC} & V_i^{LL}
\end{pmatrix},
\end{align*}
where $G_{i,j}^{CC}=\Pi^{(i)}_C G_{i,j}\Pi^{(i)}_C$, $G_{i,j}^{CL}=\Pi^{(i)}_C G_{i,j}$, and $G_{i,j}^{LL}=\Pi^{(i)}_L G_{i,j}\Pi^{(i)}_L$, and similarly, $V_i^{CC}=\Pi^{(i)}_C V_i\Pi^{(i)}_C$, $V_i^{LC}=\Pi^{(i)}_L V_i\Pi^{(i)}_C$.
Because $V_i$ is unitary, we have $(V_i^{CC})^\dagger V_i^{CC} + (V_i^{LC})^\dagger V_i^{LC} = I_{d_i}$, so $V_i^{CC}$ is a contraction:
\begin{align*}
(V_i^{CC})^\dagger V_i^{CC}\preceq I_{d_i}.
\end{align*}
We assume $\Tr(G_{i,j}^{CC}) = 0$ for all $j$. Since $G_{i,j}^{CC}$ is a $d \times d$ Hermitian operator on the computational subspace, this is the standard tracelessness condition for the computational part of the error generator. 
We do not require $\Tr(G_{i,j}) = 0$ on the full space, since $G_{i,j}^{LL}$,the leaked-subspace block, may have nonzero trace. 
However, this does not affect our analysis because the initial state is supported on the computational subspace.

For each $P\in\mathcal P_{n_i}$, define the extended Pauli $\widetilde P = P\oplus I_{D_i-d_i} = \Pi^{(i)}_C P\Pi^{(i)}_C + \Pi^{(i)}_L$.
Then $\widetilde P^\dagger=\widetilde P$, and $\widetilde P^2=I_{D_i}$.
For a computational input state $\rho^{(i)}_C$ satisfying $\Pi^{(i)}_C\rho^{(i)}_C\Pi^{(i)}_C=\rho^{(i)}_C$ and $\Pi^{(i)}_L\rho^{(i)}_C=0$, define the twirled output
\begin{align*}
\rho^{(i)}_{\text{out}}=\frac{1}{d_i^2}\sum_{P\in\mathcal P_{n_i}}\widetilde PV_i\widetilde P\rho^{(i)}_C\widetilde PV_i^\dagger\widetilde P.
\end{align*}
Since $\widetilde P\rho^{(i)}_C\widetilde P=P\rho^{(i)}_C P$, we have
\begin{align*}
\widetilde PV_i\widetilde P\rho^{(i)}_C\widetilde PV_i^\dagger\widetilde P=\widetilde PV_i(P\rho^{(i)}_C P)V_i^\dagger\widetilde P.
\end{align*}
Projecting onto the $CC$, $LL$, and $CL$ blocks gives that, for any computational input $\rho^{(i)}_C$,
\begin{align*}
\begin{split}
\Pi^{(i)}_C\rho^{(i)}_{\text{out}}\Pi^{(i)}_C&=\frac{1}{d_i^2}\sum_{P\in\mathcal P_{n_i}}
PV_i^{CC}P\rho^{(i)}_CP(V_i^{CC})^\dagger P,\\
\Pi^{(i)}_L\rho^{(i)}_{\text{out}}\Pi^{(i)}_L&=\frac{1}{d_i^2}\sum_{P\in\mathcal P_{n_i}}V_i^{LC}P\rho^{(i)}_CP(V_i^{LC})^\dagger,\\
\Pi^{(i)}_C\rho^{(i)}_{\text{out}}\Pi^{(i)}_L&=\frac{1}{d_i^2}\sum_{P\in\mathcal P_{n_i}}PV_i^{CC}P\rho^{(i)}_CP(V_i^{LC})^\dagger.
\end{split}
\end{align*}
The leaked block simplifies further because the Pauli average depolarizes the input. Formally, we have the following proposition.
\begin{proposition}\label{prop:depolarize_pauli_leakage}
For any computational density matrix $\rho^{(i)}_C$, we have
\begin{align*}
\frac{1}{d_i^2}\sum_{P\in\mathcal P_{n_i}} P\rho^{(i)}_C P = \frac{I_{d_i}}{d_i}.
\end{align*}
Consequently,
\begin{align*}
\Pi^{(i)}_L\rho^{(i)}_{\text{out}}\Pi^{(i)}_L=V_i^{LC}\frac{I_{d_i}}{d_i}(V_i^{LC})^\dagger.
\end{align*}
\end{proposition}

\begin{proof}
Expand $\rho^{(i)}_C$ in the Pauli basis:
\begin{align*}
\rho^{(i)}_C=\frac{I_{d_i}}{d_i}+\sum_{Q\neq I_{d_i}} r_Q Q.
\end{align*}
Using Pauli character orthogonality $\frac1{d_i^2}\sum_{P} P Q P = 0 \quad(Q\neq I_{d_i})$, and $\frac1{d_i^2}\sum_P P I_{d_i} P = I_{d_i}$
Hence
\begin{align*}
\frac1{d_i^2}\sum_P P\rho^{(i)}_C P = \frac{I_{d_i}}{d_i}.
\end{align*}
Substituting into the $LL$-block formula above gives the result.
\end{proof}
We also remark that  $\Pi^{(i)}_C\rho^{(i)}_{\text{out}}\Pi^{(i)}_L \neq 0$ in general, which means the twirled output is not generally block-diagonal.

We then define the detector model in the case with leakage error.
To make the detector probability rigorous, we will need the following explicit readout assumption.
\begin{assumption}
For each detector $D_k$, the effective local POVM element used to score a leakage-induced detector event is block-diagonal:
\begin{align*}
\mathcal O^{(i)}_k =\mathcal O^{(i),CC}_k\oplus \mathcal O^{(i),LL}_k,
\qquad
\mathcal O^{(i),CL}_k=0,\qquad
\Pi^{(i)}_C \mathcal O^{(i)}_k \Pi^{(i)}_L = 0.
\end{align*}
\end{assumption}
\noindent This is the natural assumption for standard energy-basis readout, where leaked levels are read as classical outcomes and no measurement is sensitive to computation-leak phase coherence.
We will provide more discussion of this assumption by the end of this subsection.
Under this assumption, the coherence block $\Pi^{(i)}_C\rho^{(i)}_{\text{out}}\Pi^{(i)}_L$ contributes exactly zero to $\Tr(\mathcal O_k^{(i)} \rho^{(i)}_{\text{out}})$.
For each gate $i$ in the detecting region $\mathcal R_k$, the local contribution to detector $D_k$ is modeled by
\begin{align*}
q_{ki}=q_{ki}^{\text{comp}}+q_{ki}^{\text{leak}},
\end{align*}
where $q_{ki}^{\text{comp}}$ is the probability that the gate induces a computational Pauli error that flips detector $D_k$, and $q_{ki}^{\text{leak}}$ is the probability that the gate induces a leaked output detected by the positive operator $\mathcal O_k^{(i),LL}$.
The exact detector rate is then
\begin{align*}
DR_k=\frac{1-\prod_{i\in\mathcal R_k}(1-2q_{ki})}{2},
\end{align*}
under the same gate-independence/parity model as in the leakage-free note. 
This is the same operational detector model as before, now with an added leakage term.
Again, for the relevant error case, the above equation can still approximately estimate the detection event rate up to the second order according to \appref{sec:unitary_convexity}.

Define the twirled computational map
\begin{align*}
\mathcal T_{i,CC}(\rho^{(i)})=\frac{1}{d_i^2}\sum_{P\in\mathcal P_{n_i}}PV_i^{CC}P\rho^{(i)} P(V_i^{CC})^\dagger P.
\end{align*}
This map is completely positive but not necessarily trace-preserving.
However, we can still show that it will diagonalize the computation subspace with the following lemma.
\begin{lemma}\label{lem:twirlinG_{i,l}eakage}
The map $\mathcal T_{i,CC}$ is diagonal in the Pauli basis. In other words, we have for any $Q\in\mathcal P_{n_i}$,
\begin{align*}
\mathcal T_{i,CC}(Q)=\tilde c^{(i)}_QQ,\qquad
\tilde c^{(i)}_Q=\frac{1}{d_i} \Tr\big(QV_i^{CC}Q(V_i^{CC})^\dagger\big).
\end{align*}
\end{lemma}
\begin{proof}
Fix $Q\in\mathcal P_{n_i}$. Then
\begin{align*}
\mathcal T_{i,CC}(Q)=\frac{1}{d_i^2}\sum_P PV_i^{CC}PQP(V_i^{CC})^\dagger P=\frac1{d_i^2}\sum_P (-1)^{s(P,Q)} PV_i^{CC}Q(V_i^{CC})^\dagger P.
\end{align*}
We expand the following into the Pauli basis
\begin{align*}
V_i^{CC}Q(V_i^{CC})^\dagger = \sum_{R\in\mathcal{P}_{n_i}} \tilde{c}^{(i)}_R R.
\end{align*}
Then
\begin{align*}
\mathcal T_{i,CC}(Q)=\sum_R \tilde{c}^{(i)}_R\left[\frac1{d_i^2}\sum_P (-1)^{s(P,Q)+s(P,R)}\right]R.
\end{align*}
By character orthogonality $\frac1{d_i^2}\sum_P (-1)^{s(P,Q)+s(P,R)} = \delta_{Q,R}$, we have $\mathcal T_{i,CC}(Q)=\tilde{c}^{(i)}_Q Q$.
Finally, we have $\tilde{c}^{(i)}_Q = \frac{1}{d_i} \Tr\big(QV_i^{CC}Q(V_i^{CC})^\dagger\big)$ as claimed.
\end{proof}
\noindent Because $\mathcal T_{i,CC}$ is Pauli diagonal and completely positive, it has a sub-normalized Pauli-channel representation
\begin{align*}
\mathcal T_{i,CC}(\rho^{(i)})=\sum_{P\in\mathcal P_{n_i}}\tilde p^{(i)}_PP\rho^{(i)} P,
\qquad
\tilde p^{(i)}_P= \frac1{d_i^2}\sum_{Q\in\mathcal P_{n_i}} (-1)^{s(P,Q)} \tilde c^{(i)}_Q\ge 0.
\end{align*}
Therefore, the total weight is
\begin{align*}
\sum_P \tilde p^{(i)}_P = \tilde c_I =  \frac{1}{d_i} \Tr\big((V_i^{CC})^\dagger V_i^{CC}\big) \le 1.
\end{align*}
Equivalently, we have $\sum_P \tilde p^{(i)}_P = 1-\mathsf{L}_{i}$,
where
\begin{align*}
\mathsf{L}_{i}:=1- \frac{1}{d_i} \Tr\big((V_i^{CC})^\dagger V_i^{CC}\big)= \frac{1}{d_i} \Tr\big((V_i^{LC})^\dagger V_i^{LC}\big)
\end{align*}
is the leakage probability from a maximally mixed computational input.

As the last part of this paragraph, we derive the second-order expansions of $V_i^{CC}$ and $V_i^{LC}$ to be used later.
We first define the first-order blocks
\begin{align*}
\begin{split}
\mathcal{H}_i^{CC}(\delta\vec{\theta})&:=\Pi^{(i)}_C\Big(\sum_j \delta \theta_{i,j} G_{i,j}\Big)\Pi^{(i)}_C = \sum_j \delta \theta_{i,j} G_{i,j}^{CC},\\
\mathcal{H}_i^{CL}(\delta\vec{\theta})&:= \Pi^{(i)}_C\Big(\sum_j \delta \theta_{i,j} G_{i,j}\Big)\Pi^{(i)}_L = \sum_j \delta \theta_{i,j} G_{i,j}^{CL}.
\end{split}
\end{align*}
Then $\mathcal{H}_i^{CC}(\delta\vec{\theta})=O(\|\delta\vec \theta_i\|)$ and $\mathcal{H}_i^{CL}(\delta\vec{\theta})=O(\|\delta\vec \theta_i\|)$. 
Using $V_i = I_{D_i} - i\mathcal H_i - \frac12 \mathcal H_i^2 + O(\|\delta\vec \theta_i\|^3)$, we obtain the block expansions
\begin{align*}
\begin{split}
V_i^{CC}&=I_{d_i} - i\mathcal{H}_i^{CC}(\delta\vec{\theta}) - \frac12\Pi^{(i)}_C \mathcal H_i^2 \Pi^{(i)}_C + O(\|\delta\vec \theta_i\|^3),\\
V_i^{LC}&=-i\Pi^{(i)}_L\mathcal H_i \Pi^{(i)}_C + O(\|\delta\vec \theta_i\|^2). 
\end{split}
\end{align*}
By utilizing the following fact
\begin{align*}
\Pi^{(i)}_C \mathcal H_i^2 \Pi^{(i)}_C=
(\Pi^{(i)}_C\mathcal H_i\Pi^{(i)}_C)^2 + (\Pi^{(i)}_C\mathcal H_i\Pi^{(i)}_L)(\Pi^{(i)}_L\mathcal H_i\Pi^{(i)}_C)=\mathcal{H}_i^{CC}(\delta\vec{\theta})^2 + \mathcal{H}_i^{CL}(\delta\vec{\theta})\mathcal{H}_i^{CL}(\delta\vec{\theta})^\dagger,
\end{align*}
we have
\begin{align*}
\begin{split}
V_i^{CC}&=I_{d_i} - i\mathcal{H}_i^{CC}(\delta\vec{\theta}) - \frac12(\mathcal{H}_i^{CC}(\delta\vec{\theta})^2+\mathcal{H}_i^{CL}(\delta\vec{\theta})\mathcal{H}_i^{CL}(\delta\vec{\theta})^\dagger) + O(\|\delta\vec \theta_i\|^3),\\
V_i^{LC}&=-i\mathcal{H}_i^{CL}(\delta\vec{\theta})^\dagger + O(\|\delta\vec \theta_i\|^2).   
\end{split}
\end{align*}

\paragraph{Convexity of the surrogate objective function $C$.}To start with, we compute the Hessian of $q_{ki}$. 
For the computational part, we use the same detector-flip sets $\mathcal F_k^{(i)}$ as in the previous leakage-free part and define $q_{ki}^{\text{comp}}=\sum_{P\in\mathcal F_k^{(i)}} \tilde p^{(i)}_P$.
We now compute the second-order expansion of $\tilde p^{(i)}_P$ for $P\neq I_{n_i}$.
Let $\widetilde{\mathcal{H}}_{i,Q}^{CC} := Q \mathcal{H}_i^{CC} Q$. Since
\begin{align*}
V_i^{CC}=I - i\mathcal{H}_i^{CC} - \frac12(\mathcal{H}_i^{CC\ 2}+\mathcal{H}_i^{CL}\mathcal{H}_i^{CL\ \dagger}) + O(\|\delta\vec \theta_i\|^3),
\end{align*}
and thus $(V_i^{CC})^\dagger=I + i\mathcal{H}_i^{CC} - \frac12(\mathcal{H}_i^{CC\ 2}+\mathcal{H}_i^{CL}\mathcal{H}_i^{CL\ \dagger}) + O(\|\delta\vec \theta_i\|^3)$,
We have
\begin{align*}
Q V_i^{CC} Q = I - i\widetilde{\mathcal{H}}_{i,Q}^{CC} - \frac12\big((\widetilde{\mathcal{H}}_{i,Q}^{CC})^2 + Q\mathcal{H}_i^{CL}\mathcal{H}_i^{CL\ \dagger} Q\big) + O(\|\delta\vec \theta_i\|^3).
\end{align*}
Multiplying the right side by $(V_i^{CC})^\dagger$, we have
\begin{align*}
\begin{split}
&Q V_i^{CC} Q (V_i^{CC})^\dagger\\
=&I + i(\mathcal{H}_i^{CC}-\widetilde{\mathcal{H}}_{i,Q}^{CC}) - \frac12\big((\widetilde{\mathcal{H}}_{i,Q}^{CC})^2 + Q\mathcal{H}_i^{CL}\mathcal{H}_i^{CL\ \dagger} Q + \mathcal{H}_i^{CC\ 2} + \mathcal{H}_i^{CL}\mathcal{H}_i^{CL\ \dagger}\big) + \widetilde{\mathcal{H}}_{i,Q}^{CC} \mathcal{H}_i^{CC} + O(\|\delta\vec \theta_i\|^3).
\end{split}
\end{align*}
Taking the normalized trace, the linear term vanishes because $\Tr(\mathcal{H}_i^{CC})=0$, and thus $\Tr(\widetilde{\mathcal{H}}_{i,Q}^{CC})=\Tr(\mathcal{H}_i^{CC})=0$.
Also, we have $\Tr((\widetilde{\mathcal{H}}_{i,Q}^{CC})^2)=\Tr(\mathcal{H}_i^{CC\ 2})$ and $\Tr(Q\mathcal{H}_i^{CL}\mathcal{H}_i^{CL\ \dagger} Q)=\Tr(\mathcal{H}_i^{CL}\mathcal{H}_i^{CL\ \dagger})$.
Hence, we can derive using \Cref{lem:twirlinG_{i,l}eakage} that
\begin{align*}
\tilde c^{(i)}_Q = 1 - \frac{1}{d_i}\Tr(\mathcal{H}_i^{CC\ 2}) - \frac{1}{d_i}\Tr(\mathcal{H}_i^{CL}\mathcal{H}_i^{CL\ \dagger}) + \frac{1}{d_i}\Tr(Q \mathcal{H}_i^{CC} Q \mathcal{H}_i^{CC}) + O(\|\delta\vec \theta_i\|^3).
\end{align*}
Expand $\mathcal{H}_i^{CC}$ in the Pauli basis as $\mathcal{H}_i^{CC} = \sum_{R\in\mathcal P_{n_i}} \alpha_{i,R} R$ where $\alpha_{i,R} = \frac{1}{d_i} \Tr(R \mathcal{H}_i^{CC})$.
Since $\Tr(\mathcal{H}_i^{CC})=0$, we have $\alpha_I=0$. Furthermore, we have $\frac{1}{d_i}\Tr(\mathcal{H}_i^{CC\ 2})=\sum_R \alpha_{i,R}^2$, and $\frac{1}{d_i}\Tr(Q \mathcal{H}_i^{CC} Q \mathcal{H}_i^{CC}) = \sum_R (-1)^{s(Q,R)} \alpha_{i,R}^2$.
Therefore, we can write the coefficients as
\begin{align*}
\tilde c^{(i)}_Q = 1 - \ell_i - 2\sum_{R:s(Q,R)=1}\alpha_{i,R}^2 + O(\|\delta\vec \theta_i\|^3),
\end{align*}
where $\ell_i:=\frac{1}{d_i}\Tr(\mathcal{H}_i^{CL}\mathcal{H}_i^{CL\ \dagger})$.
Now we perform a Fourier transformation on this expression.
For every nontrivial Pauli $P\neq I_{n_i}$, using $\tilde p^{(i)}_P=\frac1{d_i^2}\sum_Q (-1)^{s(P,Q)} \tilde c^{(i)}_Q$, we get
\begin{align*}
\tilde p^{(i)}_P=\frac1{d_i^2}\sum_Q (-1)^{s(P,Q)}\left[1-\ell_i-2\sum_{R:s(Q,R)=1}\alpha_{i,R}^2\right]+O(\|\delta\vec \theta_i\|^3).
\end{align*}
For $P\neq I_{n_i}$, as $\sum_Q (-1)^{s(P,Q)}=0$, both the constant $1$ term and the leakage constant $\ell_i$ term vanish. 
Thus, we have
\begin{align*}
\tilde p^{(i)}_P=-\frac2{d_i^2}\sum_R \alpha_{i,R}^2\sum_{Q:s(Q,R)=1} (-1)^{s(P,Q)}+O(\|\delta\vec p_i\|^3).
\end{align*}
Using the restricted character sum $\sum_{Q:s(Q,R)=1} (-1)^{s(P,Q)}=-\frac{d_i^2}{2}\delta_{P,R}$, we obtain
\begin{align*}
\tilde p^{(i)}_P = \alpha_{i,P}^2 + O(\|\delta\vec \theta_i\|^3).
\end{align*}
Since $\alpha_{i,P}=\frac{1}{d_i} \Tr(P \mathcal{H}_i^{CC})=\sum_j \delta \theta_{i,j} \frac{\Tr(P G_{i,j}^{CC})}{d_i}$, we define the following value in a similar way as the leakage-free case as
\begin{align*}
h_{P,j}^{(i)}:=\frac{\Tr(P G_{i,j}^{CC})}{d_i}.
\end{align*}
We note here that leakage affects $\tilde p_I$ at second order through the scalar term $\ell_i=\frac{1}{d_i}\Tr(\mathcal{H}_i^{CL}\mathcal{H}_i^{CL\ \dagger})$.
However, this scalar term is invisible to every nontrivial Pauli probability $\tilde p^{(i)}_P$ with $P\neq I_{n_i}$. 
Therefore, the Hessian of the computational detector-flip probability is exactly the same as in the leakage-free case.
Then we have the following corrected computational Hessian lemma.
\begin{lemma}\label{lem:qki_comp_hessian}
At $\delta\vec \theta_i=\vec 0$, we have $q_{ki}^{\text{comp}}(\vec 0)=0$, $\nabla q_{ki}^{\text{comp}}(\vec 0)=\vec 0$, and
\begin{align*}
\nabla^2 q_{ki}^{\text{comp}}\big|_{\vec 0}=2\sum_{P\in\mathcal F_k^{(i)}}\vec h_P^{(i)} (\vec h_P^{(i)})^\top\succeq 0.
\end{align*}
\end{lemma}
\begin{proof}
For $P\neq I_{n_i}$, we note that $\tilde p^{(i)}_P=\left(\sum_j \delta \theta_{i,j} h_{P,j}^{(i)}\right)^2 + O(\|\delta\vec \theta_i\|^3)$.
Summing over $P\in\mathcal F_k^{(i)}$ gives
\begin{align*}
q_{ki}^{\text{comp}}=\sum_{P\in\mathcal F_k^{(i)}} \left(\sum_j \delta \theta_{i,j} h_{P,j}^{(i)}\right)^2 +O(\|\delta\vec \theta_i\|^3),
\end{align*}
which is the stated Hessian.
\end{proof}

We then calculate the Hessian of the leakage contribution.
From \Cref{prop:depolarize_pauli_leakage} and the expansion of $V_i^{LC}$,
\begin{align*}
\begin{split}
\Pi^{(i)}_L\rho^{(i)}_{\text{out}}\Pi^{(i)}_L&=V_i^{LC}\frac{I_{d_i}}{d_i}(V_i^{LC})^\dagger=\frac{1}{d_i} \mathcal{H}_i^{CL\ \dagger} \mathcal{H}_i^{CL} + O(\|\delta\vec \theta_i\|^3)\\
&=\frac{1}{d_i} \sum_{j,l} \delta \theta_{i,j} \delta \theta_{i,l}(G_{i,j}^{CL})^\dagger G_{i,l}^{CL}+O(\|\delta\vec \theta_i\|^3).
\end{split}
\end{align*}
We note that the leakage contribution to detector $D_k$ is given by $q_{ki}^{\text{leak}}:=\Tr\big(\mathcal O_k^{LL}\Pi^{(i)}_L\rho^{(i)}_{\text{out}}\Pi^{(i)}_L\big)$. 
We thus have
\begin{align}\label{eq:qki_leak}
q_{ki}^{\text{leak}}=\frac{1}{d_i} \sum_{j,l}\delta \theta_{i,j}\delta \theta_{i,l}\Tr\big(\mathcal O_k^{LL}(G_{i,j}^{CL})^\dagger G_{i,l}^{CL}\big)+O(\|\delta\vec \theta_i\|^3).
\end{align}

\begin{lemma}\label{lem:qki_leak_hessian}
At $\delta\vec \theta=\vec 0$, we have $q_{ki}^{\text{leak}}(\vec 0)=0$, $\nabla q_{ki}^{\text{leak}}(\vec 0)=\vec 0$, and
\begin{align*}
\frac{\partial^2 q_{ki}^{\text{leak}}}{\partial(\delta \theta_{i,j})\partial(\delta \theta_{i,l})}\Big|_{\vec 0}=\frac{2}{d_i} \Re\Tr\big(\mathcal O_k^{LL}(G_{i,j}^{CL})^\dagger G_{i,l}^{CL}\big).
\end{align*}
Moreover, we have the Hessian matrix $\nabla^2 q_{ki}^{\text{leak}}=\left[\frac{2}{d_i} \Re\Tr\big(\mathcal O_k^{LL}(G_{i,j}^{CL})^\dagger G_{i,l}^{CL}\big)\right]_{j,l}\succeq 0$.
\end{lemma}

\begin{proof}
The first two statements follow because $q_{ki}^{\text{leak}}$ starts at quadratic order. For the Hessian, according to Eq.~\eqref{eq:qki_leak}, we immediately have
\begin{align*}
\frac{\partial^2 q_{ki}^{\text{leak}}}{\partial(\delta \theta_{i,j})\partial(\delta \theta_{i,l})}\Big|_{\vec 0}=\frac{2}{d_i} \Re\Tr\big(\mathcal O_k^{LL}(G_{i,j}^{CL})^\dagger G_{i,l}^{CL}\big)
\end{align*}
as claimed. 
To prove PSD, let $\vec v\in\mathbb R^{P_i}$ and set $G_{i,\vec v}:=\sum_j v_j G_{i,j}^{CL}$.
Then, we have
\begin{align*}
\vec v^\top \left[\frac{2}{d_i} \Re\Tr\big(\mathcal O_k^{LL}(G_{i,j}^{CL})^\dagger G_{i,l}^{CL}\big)\right]_{j,l} \vec v=\frac{1}{d_i} \Re\Tr\big(\mathcal O_k^{LL} G_{i,\vec v}^\dagger G_{i,\vec v}\big).
\end{align*}
Since $\mathcal O_k^{LL}\succeq 0$ and $G_{i,\vec v}^\dagger G_{i,\vec v}\succeq 0$, the operator $(\mathcal O_k^{LL})^{1/2} G_{i,\vec v}^\dagger G_{i,\vec v} (\mathcal O_k^{LL})^{1/2}$ is positive semidefinite, so its trace is nonnegative, i.e., $\Tr\big(\mathcal O_k^{LL} G_{i,\vec v}^\dagger G_{i,\vec v}\big)\ge 0$. 
Hence, $\vec v^\top \left[\frac{2}{d_i} \Re\Tr\big(\mathcal O_k^{LL}(G_{i,j}^{CL})^\dagger G_{i,l}^{CL}\big)\right]_{j,l} \vec v \ge 0$.
\end{proof}

Given \Cref{lem:qki_comp_hessian} and \Cref{lem:qki_leak_hessian}, we now compute the Hessian of the detector rate $DR_k$.
Similar to \appref{sec:unitary_convexity} and \appref{sec:cptp_convexity}, we denote $F(\{q_{ki}\}_{i\in\mathcal R_k})=\frac{1}{2}\left(1-\prod_{i\in\mathcal R_k}(1-2q_{ki})\right)$, and thus $DR_k = F(\{q_{ki}\})$.
Using Eq.~\eqref{eq:derivativeF}, we have
\begin{align*}
\nabla^2 DR_k=\sum_{i\in\mathcal R_k}\frac{\partial F}{\partial q_{ki}}\nabla^2 q_{ki}+\sum_{\substack{i,i'\in\mathcal R_k\\ i\neq i'}}\frac{\partial^2 F}{\partial q_{ki}\partial q_{ki'}}(\nabla q_{ki})(\nabla q_{ki'})^\top.
\end{align*}
At $\delta\vec \theta=\vec 0$, all $q_{ki}=0$, we thus have $\frac{\partial F}{\partial q_{ki}}\Big|_{\vec 0}=1$, and since $\nabla q_{ki}(\vec 0)=0$, all cross-gate terms vanish. 
Therefore, we have
\begin{align*}
\nabla^2 DR_k\big|_{\vec 0}=\sum_{i\in\mathcal R_k} \nabla^2 q_{ki}\big|_{\vec 0},\quad
\nabla^2 DR_k\big|_{\vec 0}=2\sum_{i\in\mathcal R_k}\left[\sum_{P\in\mathcal F_k^{(i)}} \vec h_P^{(i)}(\vec h_P^{(i)})^\top+\left[\frac{1}{d_i} \Re\Tr\big(\mathcal O_k^{LL}(G_{i,j}^{CL})^\dagger G_{i,l}^{CL}\big)\right]_{j,l}\right]\succeq 0.
\end{align*}
Recall that the surrogate objective is defined as $C(\delta\vec \theta)=\frac1{N_D}\sum_{k=1}^{N_D}DR_k(\delta\vec \theta)$.
Then, we again have $\nabla^2 C=\frac1{N_D}\sum_k \nabla^2 DR_k$.
At the optimum, $\nabla^2 C\big|_{\vec 0}=2\operatorname{blockdiag}\big(\widetilde M_{\text{ext}}^{(1)},\widetilde M_{\text{ext}}^{(2)},\ldots\big)$, where
\begin{align*}
\widetilde M_{\text{ext}}^{(i)}=\frac1{N_D}\sum_{\ell=1}^{N_D}\left[\sum_{P\in\mathcal F_\ell^{(i)}}\vec h_P^{(i)}(\vec h_P^{(i)})^\top+\left[\frac{1}{d_i} \Re\Tr\big(\mathcal O_k^{LL}(G_{i,j}^{CL})^\dagger G_{i,l}^{CL}\big)\right]_{j,l}\right].
\end{align*}
Both matrices are PSD, so $\widetilde M_{\text{ext}}^{(i)}\succeq 0$.
This yields the correct generalized detectability condition.

\begin{assumption}\label{ass:detectability_leak}
The circuit satisfies generalized detectability if $\widetilde M_{\text{ext}}^{(i)}\succ 0$ for all gates $i$.
Equivalently, for every gate $i$ and every nonzero $\vec v\in\mathbb R^{P_i}$, at least one of the following is true.
Either there exists some detector $k$ and some $P\in\mathcal F_k^{(i)}$ such that $\Tr\Big(P\sum_j v_j G_{i,j}^{CC}\Big)\neq 0$, or there exists some detector $k$ such that $\Tr\big(\mathcal O_k^{LL} G_{i,\vec v}^\dagger G_{i,\vec v}\big)>0$.
We note that Leakage does not hurt local convexity, because $\widetilde M_{\text{ext}}^{(i)} \succeq \frac1{N_D}\sum_{\ell=1}^{N_D}\sum_{P\in\mathcal F_\ell^{(i)}}\vec h_P^{(i)}(\vec h_P^{(i)})^\top$.
Instead, it may strictly improve detectability, but not in every direction.
\end{assumption}

In Appendix~\ref{sec:strong_convexity}, we again show that \Cref{ass:detectability_leak} is almost satisfied for any parametrization.
Similar to the argument for \Cref{ass:detectability_cptp}, we can consider the quotient subspace for a structural zero Hessian and smoothing from noise for a non-structural zero Hessian.

Finally, we compute the convexity radius of $C$.
Again, we define $\lambda_{\min}^{(C)}:=\min_i \lambda_{\min}\big(\widetilde M_{\text{ext}}^{(i)}\big)$.
Then, we have $\lambda_{\min}\big(\nabla^2 C(\vec 0)\big)=2\lambda_{\min}^{(C)}$.
For simplicity, we make the following assumption on the Lipschitzness of $q_{ki}$'s.
\begin{assumption}\label{ass:qki_lipschitzness}
There exists $r_0>0$ and Lipschitz constants $L^{(0)}_{ki},L^{(1)},L^{(2)}\ge 0$ such that for all $\|\delta\vec \theta\|\le r_0$,
\begin{align*}
\|\nabla^2 q_{ki}(\delta\vec \theta_i)-\nabla^2 q_{ki}(\vec 0)\|\le L^{(2)}_{ki}\|\delta\vec \theta_i\|,\quad
\|\nabla q_{ki}(\delta\vec \theta_i)\|\le L^{(1)}_{ki}\|\delta\vec \theta_i\|,\quad
|q_{ki}(\delta\vec \theta_i)|\le L^{(0)}_{ki}\|\delta\vec \theta_i\|^2.
\end{align*}
\end{assumption}
\noindent These bounds follow from the smoothness of $U(\vec \theta)$ and the finiteness of third derivatives on a compact ball. 
Then the exact chain rule in Eq.~\eqref{eq:DR_hessian_chain} and Eq.~\eqref{eq:DR_hessian_chain_cptp} give, for $\|\delta\vec \theta\|\le \theta_0$,
\begin{align*}
\|\nabla^2 DR_k(\delta\vec \theta)-\nabla^2 DR_k(\vec 0)\|\le L_k^{(1)}\|\delta\vec \theta\|+L_k^{(2)}\|\delta\vec \theta\|^2,
\end{align*}
for finite constants $L_k^{(1)} := \sum_{i\in\mathcal R_k} L^{(2)}_{ki}$ and $L_k^{(2)} := 2\sum_{i\in\mathcal R_k} L^{(0)}_{ki}\|\nabla^2 q_{ki}(\vec 0)\|+2\Big(\sum_{i\in\mathcal R_k} L^{(1)}_{ki}\Big)^2$.
Averaging over $k$, define $L_C^{(1)}:=\frac1{N_D}\sum_k L_k^{(1)}$ and $L_C^{(2)}:=\frac1{N_D}\sum_k L_k^{(2)}$.
Then, we have
\begin{align*}
\|\nabla^2 C(\delta\vec \theta)-\nabla^2 C(\vec 0)\|\le L_C^{(1)}\|\delta\vec \theta\|+L_C^{(2)}\|\delta\vec \theta\|^2.
\end{align*}
By Weyl's inequality,
\begin{align*}
\lambda_{\min}\big(\nabla^2 C(\delta\vec \theta)\big)\ge2\lambda_{\min}^{(C)}-L_C^{(1)}\|\delta\vec \theta\|-L_C^{(2)}\|\delta\vec \theta\|^2.
\end{align*}
We thus reach the following convexity radius result by the end of this subsection.
\begin{theorem}\label{thm:convexity_leakage}
Assume $\lambda_{\min}^{(C)}>0$. Then $C$ is strictly convex within the ball of any radius $\|\delta\vec \theta\|<\theta_C^{(th)}$ as long as
\begin{align*}
L_C^{(1)} \theta_C^{(th)} + L_C^{(2)} \theta_C^{(th)\ 2} < 2\lambda_{\min}^{(C)}.
\end{align*}
\end{theorem}

\subsection{Effective strong convexity of Hessians under gauge and smoothed analysis}\label{sec:strong_convexity}

In the previous subsections, the convexity radius $\theta_C^{(th)}$ requires a strictly positive minimum eigenvalue. 
For example, in the unitary case, we require $\lambda_{\min}^{(C)}:=\min_i \lambda_{\min}\left(M^{(i)}\right)>0$ in \Cref{thm:convexity_unitary}, and similarly in the CPTP and leakage cases in \Cref{thm:convexity_cptp} and \Cref{thm:convexity_leakage} we require the corresponding minimum eigenvalue of $M^{(i),\text{CPTP}}$ or $M^{(i)}_{\mathrm{ext}}$ to be strictly positive. 
However, the Hessian matrices are only naturally positive semidefinite for unitary error (\Cref{lem:pP_hessian_cptp}), CPTP error (\Cref{lem:pP_hessian_cptp}), and leakage error or atom loss (\Cref{lem:qki_leak_hessian}). 
Therefore, in general, we are only guaranteed that\ $\lambda_{\min}^{(C)} \ge 0$, and it remains to understand what happens when $\lambda_{\min}^{(C)} = 0$.

There are two different reasons why this can happen. 
The first reason is structural. 
In this case, the zero Hessian direction is a genuine gauge direction of the surrogate loss: changing the parameter along this direction does not change any detector-visible error component and hence does not change the value of the surrogate objective $C$. 
The second reason is non-structural. 
In this case, the Hessian has a zero direction only because the calibrated circuit parameters are fine-tuned to a degenerate point. 
We show below that the structural case is harmless after quotienting out the gauge directions, while the non-structural case is removed with probability one under a generic Gaussian perturbation of the circuit parameters.

\paragraph{Structural zero Hessian direction (gauges).}We first discuss the structural case. 
Let the full control parameter space have dimension $d$, and recall that $\nabla^2 C(\vec 0)$ is the Hessian of the surrogate objective at the calibrated point. 
Since all Hessian blocks derived above are positive semidefinite, we have $\nabla^2 C(\vec 0) \succeq 0$.
Suppose that $\lambda_{\min}(\nabla^2 C(\vec 0))=0$ and thus define the zero-Hessian subspace with projectors to or out of the subspace
\begin{align*}
\mathcal Z := \ker(\nabla^2 C(\vec 0)),\quad
\Pi_{\mathcal Z},\quad
\delta\vec\theta_{\perp}:=I-\Pi_{\mathcal Z}.
\end{align*}
We call $\mathcal Z$ a structural gauge subspace if the surrogate objective is locally invariant along $\mathcal Z$. 
More explicitly, for all $\delta\vec\theta$ and all sufficiently small $\vec z\in\mathcal Z$, we require $C(\delta\vec\theta+\vec z)=C(\delta\vec\theta)$.
This condition means that the parameter components along $\mathcal Z$ do not affect the detector-visible loss. 
They may change a redundant parametrization, a global phase, or another component that is not visible to the detector statistics, but they do not change the objective optimized by the calibration procedure.
Every perturbation can be decomposed as
\begin{align*}
\delta\vec\theta=\delta\vec\theta_{\mathcal Z}+\delta\vec\theta_{\perp},\quad
\delta\vec\theta_{\mathcal Z}:=\Pi_{\mathcal Z}\delta\vec\theta,\qquad
\delta\vec\theta_{\perp}:=\Pi_{\mathcal Z^\perp}\delta\vec\theta.
\end{align*}
By the structural gauge condition, we have $C(\delta\vec\theta)=C(\delta\vec\theta_{\perp})$.
Therefore, the objective only depends on the quotient parameter $\delta\vec\theta_{\perp}\in\mathcal Z^\perp$. 
The Hessian restricted to this quotient space is $\Pi_{\mathcal Z^\perp} \nabla^2 C(\vec 0) \Pi_{\mathcal Z^\perp}\big|_{\mathcal Z^\perp}$.
Since $\mathcal Z=\ker(\nabla^2 C(\vec 0))$, all zero eigenvalues of $\nabla^2 C(\vec 0)$ are removed after restricting to $\mathcal Z^\perp$. 
If the remaining directions are detectable, then
\begin{align*}
\lambda_{\min}^{(C,\perp)}:=\lambda_{\min}\left(\Pi_{\mathcal Z^\perp} \nabla^2 C(\vec 0) \Pi_{\mathcal Z^\perp}\big|_{\mathcal Z^\perp}\right)>0.
\end{align*}
In this case, the correct statement is not strict convexity in the full parameter space, but strict convexity on the quotient space $\mathbb R^{d}/\mathcal Z$.
\begin{lemma}[Structural zero Hessian (gauge) directions and quotient strong convexity]\label{lem:quotien_strong_convexity}
Suppose that $\mathcal Z=\ker(\nabla^2 C(\vec 0))$ is a structural gauge subspace and that
\begin{align*}
\lambda_{\min}^{(C,\perp)}:=\lambda_{\min}\left(\Pi_{\mathcal Z^\perp} \nabla^2 C(\vec 0) \Pi_{\mathcal Z^\perp}\big|_{\mathcal Z^\perp}\right)>0.
\end{align*}
Assume that, for all small $\delta\vec\theta$ in the previous subsections, we have $\|\Pi_{\mathcal Z^\perp}\left(\nabla^2 C(\delta\vec\theta)-\nabla^2C(\vec 0)\right)\Pi_{\mathcal Z^\perp}\|\le L_{\perp}\|\delta\vec\theta_{\perp}\|$.
Then $C$ is strictly convex on the quotient ball
\begin{align*}
\|\delta\vec\theta_{\perp}\|<\frac{\lambda_{\min}^{(C,\perp)}}{L_{\perp}}.
\end{align*}
Moreover, the set of local minimizers in the full parameter space is a flat valley of the form $\{\delta\vec\theta:\delta\vec\theta_{\perp}=\vec 0,\delta\vec\theta_{\mathcal Z}\in\mathcal Z\}$.
\end{lemma}
\begin{proof}
First, by the structural gauge condition, for every sufficiently small $\delta\vec\theta$, we have $C(\delta\vec\theta)=C(\delta\vec\theta_{\perp})$.
Therefore, it is enough to study the function restricted to $\mathcal Z^\perp$. 
Define the restricted function with its Hessian
\begin{align*}
\Pi_{\mathcal Z^\perp} C(\delta\vec\theta_\perp) \Pi_{\mathcal Z^\perp}\big|_{\mathcal Z^\perp},\qquad
\Pi_{\mathcal Z^\perp} \nabla^2 C(\delta\vec\theta_\perp) \Pi_{\mathcal Z^\perp}\big|_{\mathcal Z^\perp}.
\end{align*}
The Hessian of $C_{\perp}$ at the origin is exactly the restricted Hessian $\Pi_{\mathcal Z^\perp} \nabla^2 C(\vec 0) \Pi_{\mathcal Z^\perp}\big|_{\mathcal Z^\perp}$
By assumption, we have
\begin{align*}
\lambda_{\min}\left(\Pi_{\mathcal Z^\perp} \nabla^2 C(\vec 0) \Pi_{\mathcal Z^\perp}\big|_{\mathcal Z^\perp}\right)=\lambda_{\min}^{(C,\perp)}>0.
\end{align*}
For any sufficiently small $\delta\vec\theta_{\perp}\in\mathcal Z^\perp$, the minimum eigenvalue of the restricted Hessian satisfies
\begin{align*}
\lambda_{\min}\left(\Pi_{\mathcal Z^\perp} \nabla^2 C(\delta\vec\theta_\perp) \Pi_{\mathcal Z^\perp}\big|_{\mathcal Z^\perp}\right)\ge\lambda_{\min}\left(\Pi_{\mathcal Z^\perp} \nabla^2 C(\vec 0) \Pi_{\mathcal Z^\perp}\big|_{\mathcal Z^\perp}\right)-\left\|\Pi_{\mathcal Z^\perp} \nabla^2 C(\delta\vec\theta_\perp) \Pi_{\mathcal Z^\perp}\big|_{\mathcal Z^\perp}-\Pi_{\mathcal Z^\perp} \nabla^2 C(\vec 0) \Pi_{\mathcal Z^\perp}\big|_{\mathcal Z^\perp}\right\|.
\end{align*}
The first term is $\lambda_{+}$. The second term is bounded by the restricted Hessian Lipschitz assumption:
\begin{align*}
\left\|
\Pi_{\mathcal Z^\perp} \nabla^2 C(\delta\vec\theta_\perp) \Pi_{\mathcal Z^\perp}\big|_{\mathcal Z^\perp}-\Pi_{\mathcal Z^\perp} \nabla^2 C(\vec 0) \Pi_{\mathcal Z^\perp}\big|_{\mathcal Z^\perp}\right\|\leq L_{\perp}\|\delta\vec\theta_{\perp}\|.
\end{align*}
Thus,
\begin{align*}
\lambda_{\min}\left(\Pi_{\mathcal Z^\perp} \nabla^2 C(\delta\vec\theta_\perp) \Pi_{\mathcal Z^\perp}\big|_{\mathcal Z^\perp}\right)\ge\lambda^{(C,\perp)}_{\min}-L_{\perp}\|\delta\vec\theta_{\perp}\|.
\end{align*}
Therefore, whenever $\|\delta\vec\theta_{\perp}\|<\tfrac{\lambda_{+}}{L_{\perp}}$, we have
\begin{align*}
\lambda_{\min}\left(\Pi_{\mathcal Z^\perp} \nabla^2 C(\delta\vec\theta_\perp) \Pi_{\mathcal Z^\perp}\big|_{\mathcal Z^\perp}\right)>0.
\end{align*}
Hence $C_{\perp}$ is strictly convex inside this quotient ball.
Because $C(\delta\vec\theta)=C(\delta\vec\theta_{\perp})$, minimizing $C(\delta\vec\theta)$ in the full space is equivalent to minimizing $\Pi_{\mathcal Z^\perp} C(\delta\vec\theta_\perp) \Pi_{\mathcal Z^\perp}\big|_{\mathcal Z^\perp}$ in $\mathcal Z^\perp$. Since $\Pi_{\mathcal Z^\perp} C(\delta\vec\theta_\perp) \Pi_{\mathcal Z^\perp}\big|_{\mathcal Z^\perp}$ is strictly convex in the quotient ball and has its minimum at $\delta\vec\theta_{\perp}=\vec 0$, the quotient minimizer is unique at $\delta\vec\theta_{\perp}=\vec 0$.
However, the gauge component $\delta\vec\theta_{\mathcal Z}$ does not affect the loss. 
Therefore, the full set of minimizers is
\begin{align*}
\left\{\delta\vec\theta:\delta\vec\theta_{\perp}=\vec 0,\quad\delta\vec\theta_{\mathcal Z}\in\mathcal Z\right\}.
\end{align*}
This proves the lemma.
\end{proof}

Therefore, this structural case is harmless for calibration according to \Cref{lem:quotien_strong_convexity}. 
Moreover, we further quantify a statement that a random parametrization is unlikely to select only gauge directions. 
Let $\vec u$ be a random unit vector in $\mathbb R^d$, sampled uniformly from the unit sphere. 
If $\dim(\mathcal Z)<d$, then
\begin{align*}
\Pr[\vec u\in\mathcal Z]=0,\quad
\Pr\left[\|\Pi_{\mathcal Z^\perp}\vec u\|\le \tau\right]\le C_{d,\dim(\mathcal Z)}\tau^{d-\dim(\mathcal Z)}
\end{align*}
for any $0<\tau<1$, where $C_{d,\dim(\mathcal Z)}$ is a constant depending only on $d$ and $\dim(\mathcal Z)$.
This shows that when the gauge subspace has dimension $\dim(\mathcal Z)$ inside a $d$-dimensional parameter space, a random direction is exactly gauge with probability zero, and is nearly a gauge except for a $\tau$ orthogonal part with probability $\tau^{d-\dim(\mathcal Z)}$. 
Thus, unless $\dim(\mathcal Z)$ is very close to $d$, a randomly chosen parametrization direction will almost surely have a detector-visible component.

\paragraph{Non-structural zero Hessian direction and the smoothed-analysis framework.}We now turn to the non-structural case.
In the previous subsections, the strict convexity of the surrogate objective function $C$ near the calibrated point was reduced to the positive definiteness of the corresponding detectable projection matrices. 
For unitary errors, this matrix is $M^{(i)}$, for general CPTP errors, this matrix is $M^{(i),\text{CPTP}}$, and for or leakage errors, this matrix is $M^{(i)}_{\mathrm{ext}}$. 
A possible issue is that, for a fixed choice of circuit parameters, one of these matrices may have a zero eigenvalue. 
This corresponds to a direction in parameter space that does not change the second-order detector-visible error probability at this specific point. 
We assume that this zero Hessian is non-structural and is only a property at a specific point.
Here, we show that such zero Hessian directions are not stable under a generic Gaussian perturbation of the circuit parameters.

We start with the general CPTP case, because the unitary case is obtained by setting all dissipative first-order Kraus directions $L_{i,\alpha,j}$ to zero. 
We delay the discussion on leakage error or atom loss to later.
The argument follows the same smoothed-analysis principle as Refs.~\cite{carbery2001distributional,spielman2004smoothed,arthur2009k,chen2025quantum}: instead of asking whether every worst-case circuit parameter choice gives a strictly positive Hessian, we perturb an arbitrary circuit parameter, potentially the worst-case input resulting in the non-structural zero Hessian, by a small Gaussian vector, and show that the determinant of the Hessian block is nonzero with high probability, provided that the determinant is not identically zero as a function of the circuit parameters.
Let $\vec\theta^*$ be an arbitrary calibrated parameter vector.
We consider a smoothed calibrated parameter vector
\begin{align*}
\widetilde{\vec \theta}^* = \vec\theta^* + \sigma \vec g,
\end{align*}
where $\sigma > 0$ and $\vec g$ is a standard real Gaussian vector of the same dimension as $\vec \theta^*$. 
The control drift is still measured relative to the calibrated point, so $\delta\vec\theta = \vec\theta - \widetilde{\vec \theta}^*$. 
We can intuitively regard this Gaussian smoothing as an effect of the noisy calibration on gates. 
Thus $\delta\vec\theta = 0$ is the new calibrated point. 
The only effect of the Gaussian perturbation is that the first-order response operators appearing in the Taylor expansion around the calibrated point are now evaluated at $\widetilde{\vec \theta}^*$. 
In particular, in the CPTP case, the coefficients
\begin{align*}
G_{i,j} = G_{i,j}(\widetilde{\vec \theta}^*), \qquad
L_{i,\alpha,j} = L_{i,\alpha,j}(\widetilde{\vec \theta}^*)
\end{align*}
are now random functions of $\widetilde{\vec \theta}^*$. Consequently, the detector-visible matrix $M^{(i),\text{CPTP}}(\widetilde{\vec \theta}^*)$ is also a random matrix.

We first write down the following assumption as a formalization of the non-structural zero Hessian model.
\begin{assumption}[Non-structural zero Hessian assumption for CPTP error]\label{assumption:non_structural_cptp}
For each gate $i$, the function $\det\left(M^{(i),\text{CPTP}}(\vec\theta^*)\right)$ is not identically zero as a function of the calibrated circuit parameters $\vec\theta^*$ in a neighborhood of the original point.
Equivalently, for each gate $i$, there exists at least one nearby calibrated circuit parameter choice $\vec\theta^*$ such that $M^{(i),\text{CPTP}}(\vec\theta^*)$ is positive definite.
\end{assumption}

\noindent This assumption is necessary. 
If $\det\left(M^{(i),\text{CPTP}}(\vec\theta^*)\right)$ is identically zero, then every nearby calibrated circuit has a zero Hessian direction for gate $i$. 
This means that the direction is invisible for structural reasons and falls in the structural zero Hessian (gauge) case.
We now prove the almost-sure statement.

\begin{lemma}[Smoothed positivity of the CPTP Hessian per gate]
Suppose \Cref{assumption:non_structural_cptp} holds for gate $i$. 
Assume that the map $\vec\theta^* \mapsto M^{(i),\text{CPTP}}(\vec\theta^*)$ is well-defined in a neighborhood of the original calibrated point. 
Then
\begin{align*}
\Pr_{\vec g}\left[\lambda_{\min}\left(M^{(i),\text{CPTP}}(\vec\theta^* + \sigma\vec g)\right) > 0\right] = 1.
\end{align*}
\end{lemma}

\begin{proof}
We first recall from Eq.~\eqref{eq:M_cptp}, which we rewrite below for completeness, that
\begin{align*}
M^{(i),\text{CPTP}}_{jk}=\frac{1}{N_D}\sum_{\ell:i\in R_\ell}\sum_{P\in \mathcal F^{(i)}_\ell}\left(h^{(i)}_{P,j}h^{(i)}_{P,k}+\sum_{\alpha\ge 1}\mathrm{Re}\left(\ell^{(i)}_{\alpha,P,j}\ell^{(i)}_{\alpha,P,k}\right)\right).
\end{align*}
Therefore, for every real vector $\vec v$ in the parameter space of gate $i$,
\begin{align*}
\vec v^\top M^{(i),\text{CPTP}}\vec v=\frac{1}{N_D}\sum_{\ell:i\in R_\ell}\sum_{P\in \mathcal F^{(i)}_\ell}\left(\left(\sum_j v_j h^{(i)}_{P,j}\right)^2+\sum_{\alpha\ge 1}\left|\sum_j v_j \ell^{(i)}_{\alpha,P,j}\right|^2\right).
\end{align*}
Every summand on the right-hand side is nonnegative. 
Hence $M^{(i),\text{CPTP}}$ is positive semidefinite for every $\vec\theta^*$.
It is thus positive definite if and only if its determinant is strictly positive. Indeed, let its eigenvalues be $\lambda_1,\ldots,\lambda_{d_i}$, where $d_i$ is the number of control parameters for gate $i$. 
Since the matrix is positive semidefinite, all $\lambda_a \ge 0$ for $a=1,..,d_i$. 
Also, as $\det(M^{(i),\text{CPTP}})=\prod_{a=1}^{d_i}\lambda_a$, if $M^{(i),\text{CPTP}}$ is positive definite, then all $\lambda_a > 0$, so the determinant is positive. Conversely, if the determinant is positive and every $\lambda_a \ge 0$, then no $\lambda_a$ can be zero, so all $\lambda_a > 0$ and the matrix is positive definite.
By \Cref{assumption:non_structural_cptp}, the scalar function $\det\left(M^{(i),\text{CPTP}}(\vec\theta^*)\right)$ is not identically zero. 
A nonzero real analytic function cannot vanish on a set of positive Lebesgue measure. 
Therefore,
\begin{align*}
\{\vec\theta^* : \det\left(M^{(i),\text{CPTP}}(\vec\theta^*)\right)= 0\}
\end{align*}
has Lebesgue measure zero in the local parameter space.
Note that the random vector $\vec\theta^* + \sigma\vec g$ has a probability density with respect to Lebesgue measure. 
Therefore, the probability that it falls inside any Lebesgue-measure-zero set is zero. 
Applying this to the zero set of $\det\left(M^{(i),\text{CPTP}}\right)$ gives
\begin{align*}
\Pr_{\vec g}\left[M^{(i),\text{CPTP}}(\vec\theta^*+\sigma\vec g)=0\right]=0,\quad\Rightarrow\quad
\Pr_{\vec g}\left[M^{(i),\text{CPTP}}(\vec\theta^*+\sigma\vec g)\succ 0\right]=1.
\end{align*}
This proves the lemma.
\end{proof}

\noindent Taking a union bound over all gates gives the following consequence.

\begin{corollary}[Smoothed positivity of the CPTP Hessian]\label{coro:smoothed_cptp}
Suppose Assumption 6 holds for every gate $i$. Then, after an arbitrarily small Gaussian perturbation $\widetilde{\vec \theta}^* = \vec\theta^* + \sigma\vec g$ of the calibrated circuit parameters,
\begin{align*}
\Pr_{\vec g}\left[\nabla^2 C(\vec 0)\succ 0\right]=1.
\end{align*}
\end{corollary}

\noindent One can even give a quantitative version of the above \Cref{coro:smoothed_cptp} following the line of argument in Refs.~\cite{carbery2001distributional,spielman2004smoothed,arthur2009k,chen2025quantum}.

We now explain how the same argument applies to leakage errors (atom loss). 
In Appendix~\ref{sec:leakage_convexity}, the Hessian block for gate $i$ is
\begin{align*}
M^{(i)}_{\mathrm{ext}}=\frac{1}{N_D}\sum_{\ell=1}^{N_D}\left[\sum_{P\in \mathcal F^{(i)}_\ell}h^{(i)}_P\left(h^{(i)}_P\right)^\top+\left[\frac{1}{d_i}\mathrm{Re}\,\Tr\left(O^{LL}_\ell(G^{CL}_{i,j})^\dagger G^{CL}_{i,l}\right)\right]_{j,l}\right].
\end{align*}
For every real vector $\vec v$, define $G^{CL}_{i,\vec v}:=\sum_j v_jG^{CL}_{i,j}$, then
\begin{align*}
\vec v^\top M^{(i)}_{\mathrm{ext}}\vec v=\frac{1}{N_D}\sum_{\ell=1}^{N_D}\left[\sum_{P\in F^{(i)}_\ell}\left(\sum_j v_jh^{(i)}_{P,j}\right)^2+\frac{1}{d_i}\mathrm{Re}\Tr\left(O^{LL}_\ell(G^{CL}_{i,v})^\dagger G^{CL}_{i,v}\right)\right].
\end{align*}
The first term is a sum of squares. For the second term, $O^{LL}_\ell$ is positive semidefinite, so
\begin{align*}
\Tr\left(O^{LL}_\ell(G^{CL}_{i,v})^\dagger G^{CL}_{i,v}\right)=\Tr\left((O^{LL}_\ell)^{1/2}(G^{CL}_{i,v})^\dagger G^{CL}_{i,v}(O^{LL}_\ell)^{1/2}\right)=\Tr\left(\left(G^{CL}_{i,v}(O^{LL}_\ell)^{1/2}\right)^\dagger\left(G^{CL}_{i,v}(O^{LL}_\ell)^{1/2}\right)\right)\ge 0.
\end{align*}
Therefore, $M^{(i)}_{\mathrm{ext}}$ is positive semidefinite for every circuit parameter value.

Again, we impose the leakage analog of the non-structural assumption as a formalization of the non-structural zero Hessian model.

\begin{assumption}[Non-structural zero Hessian assumption for leakage error]\label{assumption:non_structural_leakage}
For each gate $i$, the function $\det\left(M^{(i)}_{\mathrm{ext}}(\vec\theta^*)\right)$ is not identically zero as a function of the calibrated circuit parameters $\vec\theta^*$ in a neighborhood of the original point.
Equivalently, for each gate $i$, there exists at least one nearby calibrated circuit parameter choice $\vec\theta^*$ such that $M^{(i)}_{\mathrm{ext}}(\vec\theta^*)$ is positive definite.
\end{assumption}
And we can similarly obtain the following result as the CPTP error case. 
Also, a quantitative statement can be obtained for the leakage error case following the smoothed analysis literature.

\begin{lemma}[Smoothed positivity for leakage errors]
Suppose \Cref{assumption:non_structural_leakage} holds for every gate $i$. 
For $\widetilde{\vec \theta}^* = \vec\theta^* + \sigma\vec g$, we have
\begin{align*}
\Pr_{\vec g}\left[\nabla^2C(0)\succ 0\right]=1.
\end{align*}
\end{lemma}

\begin{proof}
The proof is the same as the proof of the CPTP map error case, but we spell out the replacement.
As $M^{(i)}_{\mathrm{ext}}$ is positive semidefinite for every $\vec\theta^*$, $M^{(i)}_{\mathrm{ext}}$ is positive definite if and only if $\det\left(M^{(i)}_{\mathrm{ext}}\right)>0$.
By \Cref{assumption:non_structural_leakage}, $\det\left(M^{(i)}_{\mathrm{ext}}(\vec\theta^*)\right)$ is not identically zero. 
Since the entries of $M^{(i)}_{\mathrm{ext}}$ are real analytic functions of $\vec\theta^*$, the determinant $\det(M^{(i)}_{\mathrm{ext}})$ is also real analytic. 
The zero set of a nonzero real analytic function has Lebesgue measure zero.
As $\vec\theta^* + \sigma\vec g$ has a density with respect to Lebesgue measure, we have
\begin{align*}
\Pr_{\vec g}\left[\det\left(M^{(i)}_{\mathrm{ext}}(\vec\theta^*+\sigma\vec g)\right)=0\right]=0,\quad\Rightarrow\quad
\Pr_{\vec g}\left[M^{(i)}_{\mathrm{ext}}(\vec\theta^*+\sigma\vec g)\succ 0\right]=1
\end{align*}
for each gate $i$. 
Since the number of gates is finite, all leakage Hessian blocks are positive definite simultaneously with probability one. Because $\nabla^2 C(0)=2\mathrm{blockdiag}\left(M^{(1)}_{\mathrm{ext}},M^{(2)}_{\mathrm{ext}},\ldots\right)$, we conclude that
\begin{align*}
\Pr_{\vec g}\left[\nabla^2C(0)\succ 0\right]=1.
\end{align*}
This proves the lemma.
\end{proof}

\section{One-time control drift and optimization}
\label{app:appendix_2}
Here, based on the definition of the surrogate objective function $C$ and how we access to estimate it as in Eq.~\eqref{eq:model_Y}, Eq.~\eqref{eq:model_window}, and Eq.~\eqref{eq:model_Cemp}, we define the whole control optimization model as follows.
\begin{definition}[Optimization procedure of the surrogate model]\label{def:optimization_process}
Given a possibly time-dependent optimal parameter vector $\vec \theta^*_r$ as a function of cycle $r$, and a convex feasible set of all possible choices of parameters drift $\mathsf{P}$ such that $\delta\vec{\theta}_{t,r}=\vec{\theta}_t-\vec{\theta}^*_r\in\mathsf{P}$ at any $t$ for any control $\vec \theta_t$ we choose, we consider the following procedure.
In each epoch $t=1,...,T$, we optimize in the following model.
\begin{enumerate}
\item We choose a set of parameters $\vec \theta_t$ with the unknown drift defined as $\delta\vec \theta_{t,r}=\vec \theta_t-\vec \theta^*_r$ at each cycle $r\in\mathcal{W}_t$.
\item We implement $\vec \theta_t$, and then empirically estimate the surrogate loss function $C_t(\{\delta\vec \theta_{t,r}\}_{r\in\mathcal W_t})$ from the $m$ consecutive QEC cycles and $N_D$ detectors in the window $r\in\mathcal{W}_t$ by Eq.~\eqref{eq:model_Y}, Eq.~\eqref{eq:model_window}, and Eq.~\eqref{eq:model_Cemp}.
\item Based on the estimated $\hat{C}(\{\delta\vec \theta_{t,r}\}_{r\in\mathcal W_t})$, we update our choice of parameters $\vec \theta(t+1)$ for the next epoch.
\end{enumerate}
\end{definition}

In this section, we focus on the case of \Cref{def:optimization_process} when we are optimizing $C(\delta\vec \theta)$ in the convex regime and the optimal parameter $\vec \theta^*$ is time independent.
In other words, we assume that at the initial time $t=0$, a drift $\vec \theta^*$ occurs, and it remains unchanged throughout the whole optimization procedure.
Therefore, the drift $\delta\vec \theta_{t,r}$ is fixed across all cycles $r\in\mathcal W_t$, and we denote it as $\delta\vec \theta_t$ for convenience throughout this section.
As shown in \appref{app:appendix_1}, the surrogate objective function $C(\delta\vec \theta)$ is convex within some radius $\norm{\delta\vec \theta}<\theta_C^{(th)}$ for unitary errors, general CPTP error channels, and leakage errors. 
Within this radius, we can truncate $C$ up to the second-order, which results in a quadratic function in the form
\begin{align}\label{eq:C_quadratic_offlie}
C(\delta\vec \theta):=\frac{1}{N_D}\sum_{k=1}^{N_D}DR_k(\delta\vec \theta)=C_0+\delta\vec \theta^\top \Omega_C \delta\vec \theta+O\left(\norm{\delta\vec \theta}^3\right),
\end{align}
where $C_0$ is independent of $\delta\vec \theta$ and symmetric $\Omega_C\succeq 0$.

\subsection{The algorithm}

As we are considering a convex function $C(\delta\vec \theta)$, the feasible set $\mathsf{P}$ is the ball of radius $\theta_C^{(th)}$.
We denote the maximum eigenvalue of $\Omega_C$ as $\lambda_{\max}(\Omega_C)=\omega_C$. 
We thus have 
\begin{align}\label{eq:C_nabla_offline}
\|\nabla C(\delta\vec \theta)\| \le 2\omega_C\|\delta\vec \theta\| \le 2\omega_C\theta_C^{(th)},
\end{align}
which indicates that $C$ satisfies Lipschitz condition
\begin{align*}
\abs{C(\delta\vec \theta_1)-C(\delta\vec \theta_2)}\leq 2\omega_C\theta_C^{(th)}\norm{\delta\vec \theta_1-\delta\vec \theta_2}.
\end{align*}
We denote the dimension of $\vec \theta$, which is the number of parameters, as $\dim(\vec{\theta})=d$.
We note that we can only get access to the empirical zeroth-order information (i.e. function value) $\hat{C}(\delta\vec \theta)$ of $C$, we need to estimate the gradient based on the empirical zeroth-order information.
We employ a standard simultaneous perturbation stochastic approximation (SPSA) gradient estimator~\cite{spall1992spsa} with zeroth-order queries, which estimates $\nabla C(\delta\vec \theta)$ by finite differentiation from two points~\cite{duchi2015optimal}.
We need to choose a perturbation radius $\lambda>0$, and a corresponding shrunk feasible set
\begin{align*}
\mathsf P_\lambda := \{\delta\vec \theta\in\mathsf P:\ \delta\vec \theta+\lambda \vec u\in\mathsf P\ \text{and}\ \delta\vec \theta-\lambda \vec u\in\mathsf P\ \text{for all}\ \vec u\in\{\pm 1\}^d\}.
\end{align*}
Now, we are ready to present our algorithm as follows:
\begin{algorithm}[htbp]
\caption{Projected SPSA-SGD for one-time control drift}
\label{alg:SPSA_offline}
\DontPrintSemicolon                          
\SetKwInput{KwInput}{Input}                  
\SetKwInput{KwOutput}{Output}                
\SetKw{KwAnd}{and}                           
\SetKw{KwOr}{or}

\KwInput{Iteration budget $T$, cycles-per-window $m$, perturbation radius $\lambda>0$, step size $\eta>0$.}
\KwOutput{An estimation $\overline{\delta\vec \theta_T}$ that minimize $C(\delta\vec \theta)$ with $\delta\vec \theta=\vec \theta-\vec \theta^*$ for fixed unknown $\vec \theta^*$.}
\SetKwFor{RepTimes}{repeat}{times}{end}

\BlankLine
Initialize $\delta\vec \theta_1\in\mathsf P_\lambda$. \\
\For{$t=1,2,\dots,T$:}{
Sample $u_t\in\{\pm 1\}^d$ uniformly.\\
Query two noisy function values using independent QEC batches according to Eq.~\eqref{eq:model_Y}, Eq.~\eqref{eq:model_window}, and Eq.~\eqref{eq:model_Cemp}, and compute $\widehat C_t^+ := \widehat C(\delta\vec \theta_t+\lambda u_t)$, and $\widehat C_t^- := \widehat C(\delta\vec \theta_t-\lambda u_t)$.\\
Form the gradient estimator $\widehat g_t := \frac{\widehat C_t^+ - \widehat C_t^-}{2\lambda}u_t$.\\
Projected update $\delta\vec \theta_{t+1}:=\Pi_{\mathsf P_\lambda}\bigl(\delta\vec \theta_t-\eta \widehat g_t\bigr)$.
}
\Return{$\overline{\delta\vec \theta_T}:=\frac{1}{T}\sum_{t=1}^T \delta\vec \theta_t.$}  
\end{algorithm}

\subsection{The convergence guarantee for time-independent drift}

In this section, we show the following rigorous performance guarantee for Algorithm~\ref{alg:SPSA_offline}.
\begin{theorem}[Convergence guarantee for projected SPSA-SGD]\label{thm:offline}
Assume that $\lambda\leq \theta_C^{(th)}$, the output $\overline{\delta\vec \theta_T}$ of the projected SPSA-SGD algorithm in Algorithm~\ref{alg:SPSA_offline} satisfies
\begin{align*}
\mathbb E[C(\overline{\delta\vec \theta_T})]-C(\vec 0)\le
\frac{4(\theta_C^{(th)}-\lambda)^2}{2\eta T}+\frac{\eta}{2}\left(8d\omega_C^2\theta_C^{(th)\ 2}+\frac{d}{4\lambda^2mN_D}\right)
\end{align*}
for any step size $\eta>0$. Choosing the optimal step size $\eta=\frac{2(\theta_C^{(th)}-\lambda)}{\sqrt{T\left(8d\omega_C^2\theta_C^{(th)\ 2}+d/(4\lambda^2mN_D)\right)}}$, we have
\begin{align*}
\mathbb E[C(\overline{\delta\vec \theta_T})]-C(\vec 0)\le\frac{2(\theta_C^{(th)}-\lambda)}{\sqrt{T}}\sqrt{8d\omega_C^2\theta_C^{(th)\ 2}+\frac{d}{4\lambda^2mN_D}}.
\end{align*}
So to reach expected error $\le \varepsilon$, it suffices that
\begin{align*}
T \geq \frac{2(\theta_C^{(th)}-\lambda)}{\varepsilon^2}\sqrt{8d\omega_C^2\theta_C^{(th)\ 2}+\frac{d}{4\lambda^2mN_D}}.
\end{align*}
\end{theorem}
We note that we have assumed that the errors across different gates and QEC cycles are independent in \Cref{thm:offline}.
If the errors are relevant, we have $\text{Var}(\widehat{C}(\delta\vec \theta))\leq 1$, and the bound in \Cref{thm:offline} becomes
\begin{align*}
\mathbb E[C(\overline{\delta\vec \theta_T})]-C(\vec 0)\le\frac{2(\theta_C^{(th)}-\lambda)}{\sqrt{T}}\sqrt{8d\omega_C^2\theta_C^{(th)\ 2}+\frac{d}{4\lambda^2}}.
\end{align*}
The performance result will become worse, but we can still find the optimal drift in $O(\varepsilon^{-2})$ epochs.

Before proving \Cref{thm:offline}, we first prove some useful facts.
To start with, we show the following properties of $\hat{C}$.
\begin{proposition}\label{prop:hatc_prop_offline}
For any fixed $\delta\vec \theta$, we have $\mathbb E[\widehat C(\delta\vec \theta)]=C(\delta\vec \theta)$ and $\text{Var}(\widehat C(\delta\vec \theta))\le \frac{1}{4m}$.
If we also assume the $N_D$ detector bits within a cycle are independent, then $\text{Var}(\widehat C(\delta\vec \theta))\le \frac{1}{4mN_D}$.
\end{proposition}

\begin{proof}
We first note that $\mathbb E[\widehat C(\delta\vec \theta)]=\frac1m\sum_{r=1}^m\mathbb E[Y_r(\delta\vec \theta)]=C(\delta\vec \theta)$.
Since $0\le Y_r(\delta\vec \theta)\le 1$, $\text{Var}(Y_r(\delta\vec \theta))\le 1/4$. For independent cycles,
\begin{align*}
\text{Var}(\widehat C(\delta\vec \theta))=\frac{1}{m^2}\sum_{r=1}^m \text{Var}(Y_r(\delta\vec \theta))\le \frac{1}{m^2}\cdot m\cdot \frac14=\frac{1}{4m}.
\end{align*}
If additionally the $N_D$ detector bits are independent within each cycle, then $Y_r$ is an average of $N_D$ Bernoulli variables and $\text{Var}(Y_r)\le 1/(4N_D)$, giving $\text{Var}(\widehat C)\le 1/(4mN_D)$.
\end{proof}

Next, we show some properties of the SPSA-SGD algorithm.
The key simplification in our setting is that $C(\delta\vec \theta)$ is exactly quadratic. 
That removes the usual smoothing bias and makes the two-point SPSA estimator exactly unbiased for the true gradient.
Quantitatively, we have the following lemma.
\begin{lemma}[Unbiasedness of $\widehat g_t$ for the quadratic surrogate]\label{lem:gradient_offline}
Conditioned on $\delta\vec \theta_t$, we have $\mathbb E[\widehat g_t \mid \delta\vec \theta_t]=\nabla C(\delta\vec \theta_t)$.
\end{lemma}

\begin{proof}
Recall from Eq.~\eqref{eq:C_quadratic_offlie} and Eq.~\eqref{eq:C_nabla_offline} that $C(\delta\vec \theta)=C_0+\delta\vec \theta^\top \Omega_C \delta\vec \theta$ and $\nabla C(\delta\vec \theta)=2\Omega_C\delta\vec \theta$ as $\Omega_C$ is symmetric.
Expand $C(\delta\vec \theta_t\pm \lambda u)$ as the following
\begin{align*}
\begin{split}
C(\delta\vec \theta_t+\lambda u)&=C_0+(\delta\vec \theta_t+\lambda u)^\top \Omega_C(\delta\vec \theta_t+\lambda u)= C(\delta\vec \theta_t)+2\lambda u^\top \Omega_C\delta\vec \theta_t+\lambda^2 u^\top \Omega_C u,\\
C(\delta\vec \theta_t-\lambda u)&=C(\delta\vec \theta_t)-2\lambda u^\top \Omega_C\delta\vec \theta_t+\lambda^2 u^\top \Omega_C u.
\end{split}
\end{align*}
We subtract the above two terms and obtain
\begin{align*}
C(\delta\vec \theta_t+\lambda u)-C(\delta\vec \theta_t-\lambda u)=4\lambda u^\top \Omega_C\delta\vec \theta_t=2\lambda\nabla C(\delta\vec \theta_t)^\top u.
\end{align*}
So the noiseless SPSA estimate equals
\begin{align*}
\frac{C(\delta\vec \theta_t+\lambda u)-C(\delta\vec \theta_t-\lambda u)}{2\lambda}u=(\nabla C(\delta\vec \theta_t)^\top u)u= (u u^\top)\nabla C(\delta\vec \theta_t).
\end{align*}
Now take expectation over the random sign vector $u$. 
Since the coordinates are i.i.d. Rademacher, $\mathbb E[u u^\top]=I$. Hence
\begin{align*}
\mathbb E_u[(u u^\top)\nabla C(\delta\vec \theta_t)] = \nabla C(\delta\vec \theta_t).
\end{align*}
Finally, replacing $C$ by its noisy unbiased estimates $\widehat C_t^\pm$ preserves unbiasedness because $\mathbb E[\widehat C_t^\pm\mid \delta\vec \theta_t, u_t]=C(\delta\vec \theta_t\pm \lambda u_t)$ according to \Cref{prop:hatc_prop_offline}.
We thus reach the claimed unbiasedness of $\widehat{g_t}$.
\end{proof}

We will also need the following second moment bound for $\widehat g_t$.
\begin{lemma}[Second moment bound of $\hat{g}_t$]
\label{lem:second_moment_g_offline}
Recall from \Cref{prop:hatc_prop_offline} that $\sup_{\delta\vec \theta\in\mathsf P}\text{Var}(\widehat C(\delta\vec \theta))\le 1/(4mN_D)$.
Then for any $\delta\vec \theta_t$,
\begin{align*}
\mathbb E[\|\widehat g_t\|^2 \mid \delta\vec \theta_t]\le 2d\|\nabla C(\delta\vec \theta_t)\|^2 + \frac{d}{4mN_D\lambda^2}.
\end{align*}
In particular, using $\|\nabla C(\delta\vec \theta_t)\|\le 2\omega_C\theta_C^{(th)}$ on $\mathsf P$,
\begin{align*}
\mathbb E[\|\widehat g_t\|^2] \le 8d\omega_C^2\theta_C^{(th)\ 2} + \frac{d}{4mN_D\lambda^2}.
\end{align*}
\end{lemma}

\begin{proof}
Write $\widehat C_t^\pm = C(\delta\vec \theta_t\pm \lambda u_t) + \xi_t^\pm$, where $\xi_t$ is the empirical noise and $\mathbb E[\xi_t^\pm\mid \delta\vec \theta_t,u_t]=0$ and $\text{Var}(\xi_t^\pm\mid \delta\vec \theta_t,u_t)\le 1/(4mN_D)$. 
Define noise difference $\zeta_t:=\xi_t^+-\xi_t^-$. 
With independent QEC epochs for the two evaluations, we have $\mathbb E[\zeta_t\mid \delta\vec \theta_t,u_t]=0$ and $\text{Var}(\zeta_t\mid \delta\vec \theta_t,u_t)\le \frac{1}{2mN_D}$.
Now we decompose $\widehat g_t$ into two terms:
\begin{align*}
\widehat g_t=\frac{C(\delta\vec \theta_t+\lambda u_t)-C(\delta\vec \theta_t-\lambda u_t)}{2\lambda}u_t+\frac{\zeta_t}{2\lambda}u_t.
\end{align*}
Using $\|a+b\|^2\le 2\|a\|^2+2\|b\|^2$ and denote the first term as the signal term $S_t$ and the second term as the noise term $N_t$, we have
\begin{align*}
\mathbb E[\|\widehat g_t\|^2\mid \delta\vec \theta_t]
\le 2\mathbb E[\|S_t\|^2\mid \delta\vec \theta_t] + 2\mathbb E[\|N_t\|^2\mid \delta\vec \theta_t]
\end{align*}
as $\mathbb{E}[S_tN_t]=0$ as they are independent and $\mathbb{E}[N_t]=0$.

We then compute the signal term. 
From the exact quadratic calculation in \Cref{lem:gradient_offline}, $S_t=(u_tu_t^\top)\nabla C(\delta\vec \theta_t)$. 
Hence, we have $\|S_t\|^2 = \|(u_tu_t^\top)\nabla C\|^2 = (u_t^\top \nabla C)^2 \cdot \|u_t\|^2$.
For a sign vector, $\|u_t\|^2=d$. 
Also, we have $\mathbb E[(u_t^\top v)^2]=\|v\|^2$ for any fixed $v$ because the coordinates have mean $0$, variance $1$, and are independent. 
Therefore, we have
\begin{align*}
\mathbb E[\|S_t\|^2\mid \delta\vec \theta_t] = d\mathbb E[(u_t^\top \nabla C)^2\mid \delta\vec \theta_t] = d\|\nabla C(\delta\vec \theta_t)\|^2.
\end{align*}

We then compute the noise term. We have $\|N_t\|^2 = \frac{\zeta_t^2}{4\lambda^2}\|u_t\|^2=\frac{d}{4\lambda^2}\zeta_t^2$. Taking expectation and using $\mathbb E[\zeta_t^2]\le 1/(2mN_D)$,
\begin{align*}
\mathbb E[\|N_t\|^2\mid \delta\vec \theta_t]\le \frac{d}{4\lambda^2}\cdot \frac{1}{2mN_D} = \frac{d}{8mN_D\lambda^2}.
\end{align*}
Combine and absorb the factor $2$ from $\|a+b\|^2\le 2\|a\|^2+2\|b\|^2$ to get the displayed bound. 
\end{proof}

Finally, we are ready to prove \Cref{thm:offline}.
\begin{proof}[Proof of \Cref{thm:offline}]
We first note that a projection onto a closed convex set is non-expansive as $\|\Pi_{\mathsf P_\lambda}(a)-\Pi_{\mathsf P_\lambda}(b)\|\le \|a-b\|$.
Apply with $a=\delta\vec \theta_t-\eta \widehat g_t$ and $b=\vec 0$ (noting $\Pi(\vec{0})=\vec{0}$ since $\vec{0}\in\mathsf P_\lambda$):
\begin{align*}
\|\delta\vec \theta_{t+1}-\vec{0}\|^2\le \|\delta\vec \theta_t-\eta \widehat g_t-\vec{0}\|^2 =\|\delta\vec \theta_t-\vec{0}\|^2 -2\eta \langle \widehat g_t,\delta\vec \theta_t-\vec{0}\rangle + \eta^2\|\widehat g_t\|^2.
\end{align*}
Next, we take the conditional expectation and use unbiasedness. 
Condition on $\delta\vec \theta_t$. By \Cref{lem:gradient_offline}, we have$\mathbb E[\widehat g_t\mid \delta\vec \theta_t]=\nabla C(\delta\vec \theta_t)$. 
Hence, we have
\begin{align*}
\mathbb E[\langle \widehat g_t,\delta\vec \theta_t-\vec{0}\rangle\mid \delta\vec \theta_t]=\langle \nabla C(\delta\vec \theta_t),\delta\vec \theta_t-\vec{0}\rangle.
\end{align*}
Also by \Cref{lem:second_moment_g_offline}, we have $\mathbb E[\|\widehat g_t\|^2\mid \delta\vec \theta_t]\le 2d\|\nabla C(\delta\vec \theta_t)\|^2+\frac{d}{4mN_D\lambda^2}\le 8d\omega_C^2\theta_C^{(th)\ 2}+\frac{d}{4mN_D\lambda^2}$.
Therefore, we derive that
\begin{align*}
\mathbb E[\|\delta\vec \theta_{t+1}-\vec{0}\|^2\mid \delta\vec \theta_t]\le\|\delta\vec \theta_t-\vec{0}\|^2-2\eta \langle \nabla C(\delta\vec \theta_t),\delta\vec \theta_t-\vec{0}\rangle+\eta^2\left(8d\omega_C^2\theta_C^{(th)\ 2}+\frac{d}{4mN_D\lambda^2}\right).
\end{align*}
We then use convexity to relate the inner product to suboptimality. 
For a convex differentiable $C$, we have $C(\delta\vec \theta_t)-C(\vec{0})\le \langle \nabla C(\delta\vec \theta_t), \delta\vec \theta_t-\vec{0}\rangle$.
Plugging into the equality above, we have
\begin{align*}
2\eta\mathbb E[C(\delta\vec \theta_t)-C(\vec{0})\mid \delta\vec \theta_t]\le\|\delta\vec \theta_t-\vec{0}\|^2 - \mathbb E[\|\delta\vec \theta_{t+1}-\vec{0}\|^2\mid \delta\vec \theta_t]+\eta^2\left(8d\omega_C^2\theta_C^{(th)\ 2}+\frac{d}{4mN_D\lambda^2}\right).
\end{align*}
Now we take full expectation to remove conditioning as
\begin{align*}
2\eta\mathbb E[C(\delta\vec \theta_t)-C(\vec{0})]\le\mathbb E\|\delta\vec \theta_t-\vec{0}\|^2 - \mathbb E\|\delta\vec \theta_{t+1}-\vec{0}\|^2+\eta^2\left(8d\omega_C^2\theta_C^{(th)\ 2}+\frac{d}{4mN_D\lambda^2}\right).
\end{align*}
We then sum over $t=1,\dots,T$ and yields:
\begin{align*}
2\eta \sum_{t=1}^T \mathbb E[C(\delta\vec \theta_t)-C(\vec{0})]\le\mathbb E\|\delta\vec \theta_1-\vec{0}\|^2 - \mathbb E\|\delta\vec \theta_{T+1}-\vec{0}\|^2+ T\eta^2\left(8d\omega_C^2\theta_C^{(th)\ 2}+\frac{d}{4mN_D\lambda^2}\right).
\end{align*}
We drop the nonnegative term $\mathbb E\|\delta\vec \theta_{T+1}-\vec{0}\|^2\ge 0$, and use $\|\delta\vec \theta_1-\vec{0}\|\le 2(\theta^{(th)}_C-\lambda)$ to derive
\begin{align*}
2\eta \sum_{t=1}^T \mathbb E[C(\delta\vec \theta_t)-C(\vec{0})]\le 4(\theta^{(th)}_C-\lambda)^2 + T\eta^2\left(8d\omega_C^2\theta_C^{(th)\ 2}+\frac{d}{4mN_D\lambda^2}\right).
\end{align*}
Finally, we convert to the averaged iterate. 
By the convexity of $C$, we have $C(\overline{\delta\vec \theta_T})\le \frac{1}{T}\sum_{t=1}^T C(\delta\vec \theta_t)$.
Taking expectation and subtracting $C(\vec{0})$, we have $\mathbb E[C(\overline{\delta\vec \theta_T})]-C(\vec{0})\le\frac{1}{T}\sum_{t=1}^T \mathbb E[C(\delta\vec \theta_t)-C(\vec{0})]$.
Inserting this into the last equation, we have
\begin{align*}
\mathbb E[C(\overline{\delta\vec \theta_T})]-C(\vec{0})\le\frac{4(\theta^{(th)}_C-\lambda)^2}{2\eta T}+\frac{\eta}{2}\left(8d\omega_C^2\theta_C^{(th)\ 2}+\frac{d}{4mN_D\lambda^2}\right).
\end{align*}
Optimizing the RHS over $\eta$ gives the stated choice and bound.
\end{proof}

\subsection{The improved convergence guarantee for local codes}\label{sec:imp_offline_local}
The convergence guarantee in \Cref{thm:offline} treats $C(\delta\vec\theta)$ as a generic $d$-dimensional zeroth-order objective. 
This is sufficient for proving convergence, but it does not use the potential local structure of detector data in a QEC circuit. 
In a local code, each detector depends only on \emph{a bounded number} of control parameters, and each control parameter can influence only \emph{a bounded number} of detectors.
We now show that this locality improves the second-moment bound of the SPSA estimator.

Given the surrogate objective function as the average detector-event rate $C(\delta\vec\theta)=\tfrac{1}{N_D}\sum_{k=1}^{N_D} DR_k(\delta\vec\theta)$, we again assume that the feasible set is still the convex region $\mathsf{P}$, and that the iterates and the queried points stay inside $\mathsf{P}$. 
We also keep the same Lipschitz (smoothness) parameter $\lambda_{\max}(\Omega_C)=\omega_C$, so that, as in Eq.~\eqref{eq:C_nabla_offline}, $\|\nabla C(\delta\vec\theta)\|\leq2\omega_C\theta_C^{(\mathrm{th})}$ for every $\delta\vec\theta\in \mathsf{P}$. 
The only additional assumption here is the locality.
Formally, we make the following assumption.
\begin{assumption}[Locality condition on QEC codes]\label{ass:local_code}
For every detector $D_k$, let $\mathcal{S}_k\subseteq \{1,\ldots,d\}$ be the set of scalar control coordinates on which $DR_k$ depends, i.e., $DR_k(\delta\vec\theta)=DR_k(\delta\vec\theta_{S_k})$.
Assume that there exist constant integers $s,c>0$, independent of the code distance, such that $|S_k|\leq s$ for every detector $D_k$, and $\mathcal{K}_j:=\{k:j\in S_k\}$ satisfies $|\mathcal{K}_j|\leq c$ for every coordinate $j$. We automatically have $\frac{d}{N_D}\leq s$.
\end{assumption}
\noindent \Cref{ass:local_code} indicates that each detector depends on at most $s$ scalar control coordinates and each scalar control coordinate appears in at most $c$ detector losses.
For a fixed local QEC architecture, $s,c$ are fixed geometric constants.
The improvement comes from using detector-resolved empirical rates before averaging. For detector $D_k$, define
\begin{align*}
\widehat{DR}_k(\delta\vec\theta)=\frac{1}{m}\sum_{r\in W_t}\mathcal{D}_{k,r}(\delta\vec\theta)\quad\Rightarrow\quad
\mathbb E[\widehat{DR}_k(\delta\vec\theta)]=DR_k(\delta\vec\theta),\qquad
\operatorname{Var}(\widehat{DR}_k(\delta\vec\theta))\leq\frac{1}{4m}.
\end{align*}
If detector bits are independent within a window, then the variance of a sum of detector estimates is the sum of the variances, as in \Cref{prop:hatc_prop_offline}.

We now alter Algorithm~\ref{alg:SPSA_offline} only in the classical post-processing. 
Define a graph on the coordinates $\{1,\ldots,d\}$ by connecting $j$ and $\ell$ whenever there exists a detector $D_k$ such that $j,\ell\in \mathcal{S}_k$.
By the locality condition, the degree of this graph is at most $c(s-1)$. 
Therefore, it admits a coloring with at most $\chi\leq c(s-1)+1$ colors. 
Let $G_1,\ldots,G_\chi$ be the color classes. 
The defining property is that, for each detector $D_k$ and each color $G_a$, we have $|S_k\cap G_a|\leq 1$.

At iteration $t$, for each color $G_a$, sample signs $u_{t,j}\in\{\pm1\}$ for $j\in G_a$, and define the sparse perturbation vector $u_t^{(a)}$ by
\begin{align*}
(u_t^{(a)})_j=\begin{cases}
u_{t,j}, & j\in G_a,\\
0, & j\notin G_a.
\end{cases}
\end{align*}
For each color $G_a$, query the two controls $\delta\vec\theta_t+\lambda u_t^{(a)}$ and $ \delta\vec\theta_t-\lambda u_t^{(a)}$.
Since $\mathsf{P}_\lambda$ was defined so that $\delta\vec\theta\pm \lambda u\in \mathsf{P}$ for every full sign vector $u\in\{\pm1\}^d$, and since $u_t^{(a)}$ has entries in $[-1,1]$, convexity of $\mathsf{P}$ implies that these sparse perturbed points also lie in $\mathsf{P}$.
For $j\in G_a$, define
\begin{align*}
\widehat g_{t,j}^{\mathrm{loc}}=\frac{u_{t,j}}{2\lambda N_D}\sum_{k\in K_j}\left(\widehat{DR}_k(\delta\vec\theta_t+\lambda u_t^{(a)})-\widehat{DR}_k(\delta\vec\theta_t-\lambda u_t^{(a)})\right).
\end{align*}
The update remains $\delta\vec\theta_{t+1}=\Pi_{\mathsf{P}_\lambda}(\delta\vec\theta_t-\eta \widehat g_t^{\mathrm{loc}})$, and the output remains the averaged iterate $\delta\vec\theta_T=\frac{1}{T}\sum_{t=1}^T\delta\vec\theta_t$.

This algorithm uses $2\chi$ detector-rate queries per update instead of two. 
Since $\chi$ depends only on the local constants $s$ and $c$, this is only a constant-factor overhead for a local code family.

We have the following improved convergence guarantee for the modified version of Algorithm~\ref{alg:SPSA_offline}.

\begin{corollary}[Improved convergence guarantee for local QEC codes]\label{coro:imp_local_offline}
Assume the one-time control drift setting of Appendix~\ref{app:appendix_2}. 
Assume that $C$ is convex on $\mathsf{P}$, $\lambda\leq \theta_C^{(\mathrm{th})}$, and the locality condition in \Cref{ass:local_code} holds. 
Assume also that detector bits are independent across detectors and cycles within each empirical window. 
Then the locality-aware projected update satisfies, for any step size $\eta>0$, using $d/N_D\leq s$, we have
\begin{align*}
\mathbb E[C(\delta\vec\theta_T)]-C(\vec 0)\leq\frac{4(\theta_C^{(\mathrm{th})}-\lambda)^2}{2\eta T}+\frac{\eta}{2}\left(4\omega_C^2(\theta_C^{(\mathrm{th})})^2+\frac{s c}{8m\lambda^2N_D}\right).
\end{align*}
Choosing $\eta=\tfrac{2(\theta_C^{(\mathrm{th})}-\lambda)}{\sqrt{T\left(4\omega_C^2(\theta_C^{(\mathrm{th})})^2+\frac{s c}{8m\lambda^2N_D}\right)}}$, we obtain
\begin{align*}
\mathbb E[C(\delta\vec\theta_T)]-C(\vec 0)\leq\frac{2(\theta_C^{(\mathrm{th})}-\lambda)}{\sqrt T}\sqrt{4\omega_C^2(\theta_C^{(\mathrm{th})})^2+\frac{s c}{8m\lambda^2N_D}}.
\end{align*}
Therefore, to reach expected excess average detector-event rate at most $\varepsilon$, it is sufficient to take
\begin{align*}
T\geq\frac{4(\theta_C^{(\mathrm{th})}-\lambda)^2}{\varepsilon^2}\left(4\omega_C^2(\theta_C^{(\mathrm{th})})^2+\frac{s c}{8m\lambda^2N_D}\right).
\end{align*}
\end{corollary}

\noindent Given \Cref{ass:local_code}, the explicit $d$-dependence in \Cref{thm:offline} is removed. 
The dependence on the code size appears only through the local-code ratio $d/N_D\leq s$, and in the detector sampling term it appears as $1/N_D$.
In particular, for a fixed local architecture, the rate remains $T=O(\varepsilon^{-2})$ but the prefactor is controlled by $\omega_C$, $\theta_C^{(\mathrm{th})}$, $m$, $\lambda$, and the local constants $c,s$, rather than by the full dimension $d$ of the parameter space.

We prove the result by replacing \Cref{lem:gradient_offline} and \Cref{lem:second_moment_g_offline} with locality-aware versions, and then applying the same projected-SGD argument as in the proof of \Cref{thm:offline}.

\begin{lemma}[Unbiasedness of $\widehat g_t^{\mathrm{loc}}$ for the quadratic surrogate and local QEC codes]\label{lem:gradient_offline_local}
Conditioned on $\delta\vec\theta_t$, the estimator satisfies
\begin{align*}
\mathbb E[
\widehat g_t^{\mathrm{loc}}
\mid
\delta\vec\theta_t
]
=
\nabla C(\delta\vec\theta_t).
\end{align*}
\end{lemma}
\begin{proof}
Fix a coordinate $j$, and let $G_a$ be the color class containing $j$. 
For every detector $D_k\in K_j$, we know that $j\in S_k$ by definition, and thus by the coloring property, $|S_k\cap G_a|\leq 1$.
Since $j\in S_k\cap G_a$, it follows that $j$ is the only active perturbed coordinate in $S_k$ when we apply the sparse perturbation $u_t^{(a)}$. 
Therefore, in the quadratic surrogate regime,
\begin{align*}
DR_k(\delta\vec\theta_t+\lambda u_t^{(a)})-DR_k(\delta\vec\theta_t-\lambda u_t^{(a)})=
2\lambda\partial_jDR_k(\delta\vec\theta_t)u_{t,j}.
\end{align*}
Substituting this into the noiseless version of the estimator gives
\begin{align*}
g_{t,j}^{\mathrm{loc}}=\frac{u_{t,j}}{2\lambda N_D}\sum_{k\in K_j}2\lambda\partial_jDR_k(\delta\vec\theta_t)u_{t,j}=\frac{1}{N_D}\sum_{k\in K_j}\partial_jDR_k(\delta\vec\theta_t)=\partial_j C(\delta\vec\theta_t)
\end{align*}
since $u_{t,j}^2=1$, and that for $k\notin K_j$, $DR_k$ does not depend on coordinate $j$, so $\partial_jDR_k(\delta\vec\theta_t)=0$.
Finally, replacing $DR_k$ by $\widehat{DR}_k$ preserves the expectation because the empirical detector-rate estimates are unbiased. 
Therefore, we prove the claimed unbiasedness on each coordinate as
\begin{align*}
\mathbb E[\widehat g_{t,j}^{\mathrm{loc}}\mid\delta\vec\theta_t]=\partial_j C(\delta\vec\theta_t).
\end{align*}
Since the equality holds for every coordinate $j$, the lemma follows.
\end{proof}

\begin{lemma}[Second-moment bound of $\widehat g_t^{\mathrm{loc}}$ for local QEC codes]\label{lem:second_moment_g_offline_local}
Conditioned on $\delta\vec\theta_t$,
\begin{align*}
\mathbb E[\|\widehat g_t^{\mathrm{loc}}\|^2\mid\delta\vec\theta_t]\leq\|\nabla C(\delta\vec\theta_t)\|^2+\frac{dc}{8m\lambda^2N_D^2}\leq4\omega_C^2(\theta_C^{(\mathrm{th})})^2+\frac{sc}{8m\lambda^2N_D}.
\end{align*}
\end{lemma}

\begin{proof}
Write the empirical detector-rate estimates as
\begin{align*}
\widehat{DR}_k^{+}=DR_k(\delta\vec\theta_t+\lambda u_t^{(a)})+\xi_k^{+},\qquad
\widehat{DR}_k^{-}=DR_k(\delta\vec\theta_t-\lambda u_t^{(a)})+\xi_k^{-}.
\end{align*}
Here, the $+$ and $-$ batches are independent, and
\begin{align*}
\mathbb E[\xi_k^{+}]=\mathbb E[\xi_k^{-}]=0,\qquad
\operatorname{Var}(\xi_k^{+})\leq\frac{1}{4m},\qquad
\operatorname{Var}(\xi_k^{-})\leq\frac{1}{4m}.
\end{align*}
By \Cref{lem:gradient_offline_local}, the noiseless part of the estimator equals $\nabla C(\delta\vec\theta_t)$. 
Thus we can write
\begin{align*}
\widehat g_{t,j}^{\mathrm{loc}}=\partial_j C(\delta\vec\theta_t)+\zeta_{t,j},\qquad
\zeta_{t,j}=\frac{u_{t,j}}{2\lambda N_D}\sum_{k\in K_j}(\xi_k^{+}-\xi_k^{-}).
\end{align*}
The noise has mean zero. 
For each individual detector, we have
\begin{align*}
\operatorname{Var}(\xi_k^{+}-\xi_k^{-})=\operatorname{Var}(\xi_k^{+})+\operatorname{Var}(\xi_k^{-})\leq\frac{1}{2m}.
\end{align*}
Using independence across detectors and $|K_j|\leq c$,
\begin{align*}
\operatorname{Var}\left(\sum_{k\in K_j}(\xi_k^{+}-\xi_k^{-})\right)\leq\frac{c}{2m}.
\end{align*}
Therefore,
\begin{align*}
\operatorname{Var}(\zeta_{t,j})\leq\frac{1}{4\lambda^2N_D^2}\cdot\frac{c}{2m}=\frac{c}{8m\lambda^2N_D^2}.
\end{align*}

Now we compute the conditional second moment:
\begin{align*}
\mathbb E[(\widehat g_{t,j}^{\mathrm{loc}})^2\mid\delta\vec\theta_t]=(\partial_j C(\delta\vec\theta_t))^2+\mathbb E[\zeta_{t,j}^2\mid\delta\vec\theta_t]\leq(\partial_j C(\delta\vec\theta_t))^2+\frac{c}{8m\lambda^2N_D^2},
\end{align*}
because $\mathbb E[\zeta_{t,j}\mid \delta\vec\theta_t]=0$.
Summing over $j=1,\ldots,d$ gives
\begin{align*}
\mathbb E[\|\widehat g_t^{\mathrm{loc}}\|^2\mid\delta\vec\theta_t]\leq\|\nabla C(\delta\vec\theta_t)\|^2+\frac{dc}{8m\lambda^2N_D^2}.
\end{align*}
The claimed bound holds by using the facts in Eq.~\eqref{eq:C_nabla_offline}.
\end{proof}

Now, we are ready to prove \Cref{coro:imp_local_offline}

\begin{proof}[Proof of \Cref{coro:imp_local_offline}]
Since projection onto a closed convex set is non-expansive and $\vec 0\in \mathsf{P}_\lambda$,
\begin{align*}
\|\delta\vec\theta_{t+1}-\vec 0\|^2\leq\|\delta\vec\theta_t-\eta \widehat g_t^{\mathrm{loc}}-\vec 0\|^2\quad\Rightarrow\quad\|\delta\vec\theta_{t+1}\|^2
\leq
\|\delta\vec\theta_t\|^2
-
2\eta
\langle
\widehat g_t^{\mathrm{loc}},
\delta\vec\theta_t
\rangle
+
\eta^2
\|\widehat g_t^{\mathrm{loc}}\|^2.
\end{align*}
Condition on $\delta\vec\theta_t$. 
By \Cref{lem:gradient_offline_local} and \Cref{lem:second_moment_g_offline_local}, we have
\begin{align*}
\mathbb E[\|\delta\vec\theta_{t+1}\|^2\mid\delta\vec\theta_t]\leq\|\delta\vec\theta_t\|^2-2\eta\langle\nabla C(\delta\vec\theta_t),\delta\vec\theta_t\rangle+\eta^2\left(4\omega_C^2(\theta_C^{(\mathrm{th})})^2+\frac{dc}{8m\lambda^2N_D^2}\right).
\end{align*}
Because $C$ is convex on $\mathsf{P}$ and $\vec 0$ is the minimizer in the local convex regime, we have $C(\delta\vec\theta_t)-C(\vec 0)\leq\langle\nabla C(\delta\vec\theta_t),\delta\vec\theta_t-\vec 0\rangle$.
Therefore, we can deduce that
\begin{align*}
2\eta\mathbb E[C(\delta\vec\theta_t)-C(\vec 0)]\leq\mathbb E\|\delta\vec\theta_t\|^2-\mathbb E\|\delta\vec\theta_{t+1}\|^2+\eta^2\left(4\omega_C^2(\theta_C^{(\mathrm{th})})^2+\frac{dc}{8m\lambda^2N_D^2}\right).
\end{align*}
Summing over $t=1,\ldots,T$, we obtain
\begin{align*}
2\eta\sum_{t=1}^T\mathbb E[C(\delta\vec\theta_t)-C(\vec 0)]\leq\mathbb E\|\delta\vec\theta_1\|^2-\mathbb E\|\delta\vec\theta_{T+1}\|^2+T\eta^2\left(4\omega_C^2(\theta_C^{(\mathrm{th})})^2+\frac{dc}{8m\lambda^2N_D^2}\right).
\end{align*}
Dropping the nonnegative term $\mathbb E\|\delta\vec\theta_{T+1}\|^2$, and using $\|\delta\vec\theta_1\|\leq 2(\theta_C^{(\mathrm{th})}-\lambda)$, we get
\begin{align*}
2\eta\sum_{t=1}^T\mathbb E[C(\delta\vec\theta_t)-C(\vec 0)]\leq4(\theta_C^{(\mathrm{th})}-\lambda)^2+T\eta^2\left(4\omega_C^2(\theta_C^{(\mathrm{th})})^2+\frac{dc}{8m\lambda^2N_D^2}\right).
\end{align*}

Since $C$ is convex and $\delta\vec\theta_T=\tfrac{1}{T}\sum_{t=1}^T\delta\vec\theta_t$, we have $C(\delta\vec\theta_T)\leq\tfrac{1}{T}\sum_{t=1}^TC(\delta\vec\theta_t)$.
Taking expectation and subtracting $C(\vec 0)$ gives
\begin{align*}
\mathbb E[C(\delta\vec\theta_T)]-C(\vec 0)\leq\frac{1}{T}\sum_{t=1}^T\mathbb E[C(\delta\vec\theta_t)-C(\vec 0)]\leq\frac{4(\theta_C^{(\mathrm{th})}-\lambda)^2}{2\eta T}+\frac{\eta}{2}\left(4\omega_C^2(\theta_C^{(\mathrm{th})})^2+\frac{dc}{8m\lambda^2N_D^2}\right).
\end{align*}
Using $d/N_D\leq s$, we have $\tfrac{dc}{8m\lambda^2N_D^2}\leq\tfrac{s c}{8m\lambda^2N_D}$.
This proves the first two stated bounds. Optimizing over $\eta$ gives the displayed choice of step size and the final $T=O(\varepsilon^{-2})$ guarantee.
\end{proof}

Finally, we remark that the assumption that detector estimates inside $K_j$ are independent can be dropped
\begin{remark}\label{rem:non_independent_K}
If the detector estimates inside $K_j$ are not independent, then $\operatorname{Var}(\sum_{k\in K_j}(\xi_k^{+}-\xi_k^{-}))$ can be bounded by $c^2/(2m)$ instead of $c/(2m)$. 
In that case, the sampling term becomes
\begin{align*}
\frac{dc^2}{8m\lambda^2N_D^2}\leq\frac{s c^2}{8m\lambda^2N_D},
\end{align*}
which worsens only the local constant $c\to c^2$, and still does not introduce an explicit $d$-dependent prefactor.
\end{remark}

\section{Time-dependent control drift and online optimization}
\label{app:appendix_3}
In this section, we consider the case of \Cref{def:optimization_process} when the optimal parameter vector $\vec \theta_r^*$ is time dependent. 
In contrast to \appref{app:appendix_2}, the physical control $\vec \theta_t$ is held fixed only throughout one window $\mathcal W_t$, while the optimal parameter $\vec \theta_r^*$ may vary from cycle to cycle inside that same window. Therefore, the cycle-level drift is $\delta \vec \theta_{t,r} := \vec \theta_t - \vec \theta_r^*$ at $r\in\mathcal W_t$.
Since the optimum drifts with time, the problem is no longer a standard offline convex optimization problem with one fixed objective function. 
Instead, the natural framework is online convex optimization~\cite{zinkevich2003online,hazan2016introduction}, where the loss function changes with the epoch, and our goal is to track the moving optimum.

\subsection{Online convex optimization and regret}

We assume that for every admissible control $\vec \theta_t\in\mathsf Q$ and every $r\in\mathcal W_t$, the induced drift $\delta \vec \theta_{t,r}$ stays inside the convex regime $\mathsf P$ established in \appref{app:appendix_1}, so that the surrogate loss can be truncated to the second order cycle by cycle.
Recall from Eq.~\eqref{eq:model_window} that the epoch-$t$ window is $\mathcal W_t := \{r : r=(t-1)m+1,\ldots,tm\}$.
According to \Cref{def:optimization_process}, within epoch $t$, we choose one control vector $\vec \theta_t\in\mathsf Q$ and hold it fixed over all cycles in $\mathcal W_t$. 
For each cycle $r\in\mathcal W_t$, the expected surrogate loss is modeled by abusing the notation a bit and writing the single-cycle surrogate objective function $C_r(\vec \theta_t;\vec \theta_r^*)$ as
\begin{align*}
C_r(\vec \theta_t;\vec \theta_r^*)=C_{r,0}+(\vec \theta_t-\vec \theta_r^*)^\top \Omega_{C,r} (\vec \theta_t-\vec \theta_r^*),
\end{align*}
where $C_{r,0}$ is independent of $\vec \theta_t$, $\Omega_{C,r}\succeq 0$, and $\lambda_{\max}(\Omega_{C,r}) \le \omega_C$ uniformly for any $r$.
We then define the window-averaged loss by abusing notation and from Eq.~\eqref{eq:model_Y}, Eq.~\eqref{eq:model_window}, and Eq.~\eqref{eq:model_Cemp} as $C_t\bigl(\vec \theta_t;\{\vec \theta_r^*\}_{r\in\mathcal W_t}\bigr):=\frac{1}{m}\sum_{r\in\mathcal W_t} C_r(\vec \theta_t;\vec \theta_r^*)$.
Expanding the quadratic terms, we have
\begin{align*}
C_t\bigl(\vec \theta_t;\{\vec \theta_r^*\}_{r\in\mathcal W_t}\bigr)=\vec \theta_t^\top\left(\frac{1}{m}\sum_{r\in\mathcal W_t}\Omega_{C,r}\right)\vec \theta_t-2\left(\frac{1}{m}\sum_{r\in\mathcal W_t}\Omega_{C,r}\vec \theta_r^*\right)^\top \vec \theta_t+\bar C_{t,0},
\end{align*}
where $\bar C_{t,0}$ is independent of $\vec \theta_t$. 
Since each $\Omega_{C,r}$ is positive semidefinite, $C_t(\cdot;\{\vec \theta_r^*\})$ is convex on $\mathsf Q$.
Its gradient with respect to the control variable is
\begin{align*}
\nabla_{\vec \theta} C_t\bigl(\vec \theta_t;\{\vec \theta_r^*\}_{r\in\mathcal W_t}\bigr)=\frac{2}{m}\sum_{r\in\mathcal W_t}\Omega_{C,r}(\vec \theta_t-\vec \theta_r^*).
\end{align*}
Using the uniform eigenvalue bound and the assumption that every cycle-level drift remains in the convex region, we obtain
\begin{align*}
\bigl\|\nabla_{\vec \theta} C_t(\vec \theta_t;\{\vec \theta_r^*\}_{r\in\mathcal W_t})\bigr\|\le\frac{2}{m}\sum_{r\in\mathcal W_t}\|\Omega_{C,r}(\vec \theta_t-\vec \theta_r^*)\|\le\frac{2}{m}\sum_{r\in\mathcal W_t}\omega_C\|\vec \theta_t-\vec \theta_r^*\| \le 2\omega_C \theta_C^{(th)}.
\end{align*}
Therefore, every epoch loss is convex and uniformly Lipschitz in the control variable.
This exactly fits the online convex optimization framework: at epoch $t$, the learner chooses $\vec \theta_t\in\mathsf Q$, then incurs the convex loss $C_t\bigl(\vec \theta;\{\vec \theta_r^*\}_{r\in\mathcal W_t}\bigr)$.
Accordingly, for any comparator sequence $\{\vec v_t\}_{t=1}^T\subseteq\mathsf Q$, we define the dynamic regret by
\begin{align}\label{eq:dyn_reg_online}
\mathrm{DynReg}(T;\{\vec v_t\}):=\sum_{t=1}^T\Bigl(C_t\bigl(\vec \theta_t;\{\vec \theta_r^*\}_{r\in\mathcal W_t}\bigr)-C_t\bigl(\vec v_t;\{\vec \theta_r^*\}_{r\in\mathcal W_t}\bigr)\Bigr).
\end{align}
The natural comparator is the sequence of window-wise minimizers $\vec v_t^\star \in \arg\min_{v\in\mathsf Q}C_t\bigl(v;\{\vec \theta_r^*\}_{r\in\mathcal W_t}\bigr)$, but the regret theorem below holds for any comparator sequence.

\subsection{The algorithm}

\begin{algorithm}[htbp]
\caption{Online projected SPSA-SGD for time-dependent control drift}
\label{alg:SPSA_online}
\DontPrintSemicolon                          
\SetKwInput{KwInput}{Input}                  
\SetKwInput{KwOutput}{Output}                
\SetKw{KwAnd}{and}                           
\SetKw{KwOr}{or}

\KwInput{Iteration budget $T$, cycles-per-window $m$, perturbation radius $\lambda>0$, step size $\eta>0$.}
\KwOutput{A sequence of control parameters $\vec \theta_t$ for $t=1,...,T$ that minimize the dynamic regret in Eq.\eqref{eq:dyn_reg_online} for the time-dependent unknown drift sequence $\{\vec \theta^*_r\}_r$ and optimal choice of comparator $\{v^*_t\}_t$.}
\SetKwFor{RepTimes}{repeat}{times}{end}

\BlankLine
Initialize $\vec \theta_1\in\mathsf Q_\lambda$ \\
\For{$t=1,2,\dots,T$:}{
Sample $u_t\in\{\pm 1\}^d$ uniformly.\\
Query two noisy function values using independent QEC batches according to Eq.~\eqref{eq:model_Y}, Eq.~\eqref{eq:model_window}, and Eq.~\eqref{eq:model_Cemp}, and compute $\widehat C_t^+ := \widehat C(\vec \theta_t+\lambda u_t;\{\vec \theta_r^*\}_{r\in\mathcal W_t})$, and $\widehat C_t^- := \widehat C(\vec \theta_t-\lambda u_t;\{\vec \theta_r^*\}_{r\in\mathcal W_t})$.\\
Form the gradient estimator $\widehat g_t := \frac{\widehat C_t^+ - \widehat C_t^-}{2\lambda}u_t$.\\
Projected update $\vec \theta_{t+1}:=\Pi_{\mathsf P_\lambda}\bigl(\vec \theta_t-\eta \widehat g_t\bigr)$.
}
\Return{The sequence $\{\vec \theta_{t}\}_{t=1}^T$.}  
\end{algorithm}

Here, we introduce our algorithm as in Algorithm~\ref{alg:SPSA_online} as a slightly modified online version of Algorithm~\ref{alg:SPSA_offline}.
Again, we need to choose a perturbation radius $\lambda>0$ and define the shrunk feasible set
\begin{align*}
\mathsf Q_\lambda:=\{\vec \theta\in\mathsf Q : \vec \theta+\lambda u\in\mathsf Q\text{ and }\vec \theta-\lambda u\in\mathsf Q\text{ for all } u\in\{\pm1\}^d\}.
\end{align*}
We will always have that for any $x,y\in\mathsf Q_\lambda$, we have $\norm{x-y}\leq 2(\theta_C^{(th)}-\lambda)$.
In addition, for any $x,y\in\mathsf Q$, we have $\norm{x-y}\leq 2\theta_C^{(th)}$.

\subsection{The regret bound for time-dependent control drift and explanation}

In this section, we show the following rigorous performance guarantee for Algorithm~\ref{alg:SPSA_online}.
\begin{theorem}[Convergence guarantee for online projected SPSA-SGD]\label{thm:online}
Assume that $\lambda\leq \theta_C^{(th)}$, the output sequence $\{\vec \theta_t\}_t$ of the online projected SPSA-SGD algorithm in Algorithm~\ref{alg:SPSA_online} satisfies that, for any comparator sequence $\{\vec v_t\}_{t=1}^T\subseteq\mathsf P$ and any step size $\eta>0$,
\begin{align*}
\mathbb E\left[\sum_{t=1}^T\Bigl(C_t(\vec \theta_t;\{\vec \theta_r^*\}_{r\in\mathcal W_t})-C_t(\vec v_t;\{\vec \theta_r^*\}_{r\in\mathcal W_t})\Bigr)\right]\le\frac{2\theta_C^{(th)\,2}}{\eta}+\frac{\eta T}{2}\left(8d\omega_C^2 \theta_C^{(th)\,2}+\frac{d}{4mN_D\lambda^2}\right)+\frac{2\theta_C^{(th)}}{\eta}\sum_{t=1}^{T-1}\|\vec v_{t+1}-\vec v_t\|.
\end{align*}
Choosing the optimal step size $\eta=\sqrt{\frac{4\theta_C^{(th)\,2}+4\theta_C^{(th)}\sum_{t=1}^{T-1}\|\vec v_{t+1}-\vec v_t\|}{\left(8d\omega_C^2 \theta_C^{(th)\,2}+\frac{d}{4mN_D\lambda^2}\right)T}}$, we have
\begin{align*}
\mathbb E[\mathrm{DynReg}(T;\{\vec v_t\})]\le\sqrt{
\bigl(4\theta_C^{(th)\,2}+4\theta_C^{(th)}\sum_{t=1}^{T-1}\|\vec v_{t+1}-\vec v_t\|\bigr)\left(8d\omega_C^2 \theta_C^{(th)\,2}+\frac{d}{4mN_D\lambda^2}\right)T}.
\end{align*}
\end{theorem}
\noindent In \Cref{thm:online}, for the natural comparator $\vec v_t=\vec v_t^\star$, we obtain a sublinear regret guarantee whenever the path length of the window-wise minimizers grows sublinearly in $T$. 
Again, we have assumed that the errors across different gates and QEC cycles are independent in \Cref{thm:online}.
If the errors are relevant, we have $\text{Var}(\widehat{C}(\delta\vec \theta))\leq 1$.
Similar to \Cref{thm:offline}, the performance result will become worse, but we can still find the optimal drift in $O(\varepsilon^{-2})$ epochs.

Before proving the above theorem, we introduce some properties as tools of the proof.
Similar to the offline case in \appref{app:appendix_2}, we only observe zeroth-order empirical losses. 
In each cycle $r$, detector $D_k$ produces a bit $\mathcal D_{k,r}\in\{0,1\}$, and the per-cycle empirical detection fraction is $Y_r := \frac{1}{N_D}\sum_{k=1}^{N_D} \mathcal D_{k,r} \in [0,1]$.
The empirical window loss is $\hat C_t(\vec \theta):=\frac{1}{m}\sum_{r\in\mathcal W_t} Y_r$ according to Eq.~\eqref{eq:model_Y}, Eq.~\eqref{eq:model_window}, and Eq.~\eqref{eq:model_Cemp}.
For fixed $\vec \theta$, we can also observe that
\begin{align*}
\mathbb E[\hat C_t(\vec \theta;\{\vec \theta_r^*\}_{r\in\mathcal W_t})]=C_t\bigl(\vec \theta;\{\vec \theta_r^*\}_{r\in\mathcal W_t}\bigr).
\end{align*}
If we additionally assume that the $mN_D$ detector bits in the window are conditionally independent given $\vec \theta$, then similar to the offline case, we have $\operatorname{Var}(\hat C_t(\vec \theta)) \le \frac{1}{4mN_D}$.

We next prove the two key facts needed for the regret bound.
\begin{lemma}[Unbiasedness of the $\widehat{g}_t$]\label{lem:gt_unbiased_online}
For every epoch $t$, $\mathbb E[\hat g_t \mid \vec \theta_t]=\nabla_{\vec \theta} C_t\bigl(\vec \theta_t;\{\vec \theta_r^*\}_{r\in\mathcal W_t}\bigr)$.
\end{lemma}
\begin{proof}
We first fix a $t$. 
For each cycle $r\in\mathcal W_t$, since $C_r(\vec \theta;\vec \theta_r^*)$ is a quadratic function of $\vec \theta$ with symmetric Hessian $2\Omega_{C,r}$, we have
\begin{align*}
C_r(\vec \theta_t+\lambda u_t;\vec \theta_r^*)-C_r(\vec \theta_t-\lambda u_t;\vec \theta_r^*)=2\lambda\nabla_{\vec \theta} C_r(\vec \theta_t;\vec \theta_r^*)^\top u_t.
\end{align*}
Averaging over $r\in\mathcal W_t$ yields
\begin{align*}
C_t(\vec \theta_t+\lambda u_t;\{\vec \theta_r^*\})-C_t(\vec \theta_t-\lambda u_t;\{\vec \theta_r^*\})=2\lambda\nabla_{\vec \theta} C_t(\vec \theta_t;\{\vec \theta_r^*\})^\top u_t.
\end{align*}
Therefore, the noiseless SPSA estimator equals
\begin{align*}
\frac{C_t(\vec \theta_t+\lambda u_t;\{\vec \theta_r^*\})-C_t(\vec \theta_t-\lambda u_t;\{\vec \theta_r^*\})}{2\lambda}u_t=(u_tu_t^\top)\nabla_{\vec \theta} C_t(\vec \theta_t;\{\vec \theta_r^*\}).
\end{align*}
Since the coordinates of $u_t$ are i.i.d., we have $\mathbb E[u_tu_t^\top] = I$.
Hence, the noiseless estimator is unbiased. 
Replacing the exact losses by the empirical estimates $\hat C_t^\pm$ preserves unbiasedness because $\mathbb E[\hat C_t^\pm\mid \vec \theta_t,u_t]=C_t(\vec \theta_t\pm\lambda u_t;\{\vec \theta_r^*\}_{r\in\mathcal W_t})$.
\end{proof}

\begin{lemma}[Second-moment bound of $\widehat{g}_t$]\label{lem:second_moment_online}
For every epoch $t$, we have
\begin{align*}
\mathbb E[\|\hat g_t\|^2 \mid \vec \theta_t]\le2d\bigl\|\nabla_{\vec \theta} C_t(\vec \theta_t;\{\vec \theta_r^*\}_{r\in\mathcal W_t})\bigr\|^2+\frac{d}{4mN_D\lambda^2}.
\end{align*}
In particular, using the uniform gradient bound above, we have
\begin{align*}
\mathbb E[\|\hat g_t\|^2 \mid \vec \theta_t]\le8d\omega_C^2 \theta_C^{(th)\,2}+\frac{d}{4mN_D\lambda^2}.
\end{align*}
\end{lemma}

\begin{proof}
We write $\hat C_t^\pm=C_t(\vec \theta_t\pm\lambda u_t;\{\vec \theta_r^*\})+\xi_t^\pm$, where $\xi_t^\pm$ are the empirical noises. 
By the variance bound for $\hat C_t$, we have $\mathbb E[\xi_t^\pm\mid \vec \theta_t,u_t]=0$ and $\operatorname{Var}(\xi_t^\pm\mid \vec \theta_t,u_t)\le \frac{1}{4mN_D}$.
Define $\zeta_t:=\xi_t^+-\xi_t^-$. 
Using independent batches for the two function queries gives $\mathbb E[\zeta_t\mid \vec \theta_t,u_t]=0$ and $\operatorname{Var}(\zeta_t\mid \vec \theta_t,u_t)\le \frac{1}{2mN_D}$.
Now we decompose
\begin{align*}
\hat g_t&=\frac{C_t(\vec \theta_t+\lambda u_t;\{\vec \theta_r^*\})-C_t(\vec \theta_t-\lambda u_t;\{\vec \theta_r^*\})}{2\lambda}u_t+\frac{\zeta_t}{2\lambda}u_t.
\end{align*}
Using $\|a+b\|^2\le 2\|a\|^2+2\|b\|^2$ and denote the first term as the signal term $S_t$ and the second term as the noise term $N_t$, we have
\begin{align*}
\mathbb E[\|\hat g_t\|^2\mid \vec \theta_t]\le2\mathbb E[\|S_t\|^2\mid \vec \theta_t]+2\mathbb E[\|N_t\|^2\mid \vec \theta_t].
\end{align*}
From \Cref{lem:gt_unbiased_online}, we have $S_t=(u_tu_t^\top)\nabla_{\vec \theta} C_t(\vec \theta_t;\{\vec \theta_r^*\})$.
Hence, we have $\|S_t\|^2=(u_t^\top \nabla_{\vec \theta} C_t)^2 \|u_t\|^2$.
Since $\|u_t\|^2=d$ and $\mathbb E[(u_t^\top v)^2]=\|v\|^2$ for any fixed $v$,
\begin{align*}
\mathbb E[\|S_t\|^2\mid \vec \theta_t]=d\|\nabla_{\vec \theta} C_t(\vec \theta_t;\{\vec \theta_r^*\})\|^2.
\end{align*}
For the noise term, we have$\|N_t\|^2=\frac{\zeta_t^2}{4\lambda^2}\|u_t\|^2=
\frac{d\zeta_t^2}{4\lambda^2}$.
Taking expectation and using $\mathbb E[\zeta_t^2]\le 1/(2mN_D)$ yields
\begin{align*}
\mathbb E[\|N_t\|^2\mid \vec \theta_t]\le\frac{d}{8mN_D\lambda^2}.
\end{align*}
Combining the two bounds gives the stated result.
\end{proof}

Finally, we are ready to prove \Cref{thm:online}.

\begin{proof}[Proof of \Cref{thm:online}]
Define the auxiliary point $y_{t+1}:=\vec \theta_t-\eta \hat g_t$.
Note that $\|\vec \theta_{t+1}-\vec v_t\|^2=\|\Pi_{\mathsf Q_\lambda}(y_{t+1})-\Pi_{\mathsf Q_\lambda}(\vec v_t)\|^2\le\|y_{t+1}-\vec v_t\|^2$.
Expanding the above equation gives
\begin{align*}
\|\vec \theta_{t+1}-\vec v_t\|^2\le\|\vec \theta_t-\vec v_t\|^2-2\eta\langle \hat g_t,\vec \theta_t-\vec v_t\rangle+\eta^2\|\hat g_t\|^2.
\end{align*}
Conditioning on $\vec \theta_t$ and using \Cref{lem:gt_unbiased_online} and \Cref{lem:second_moment_online}, we obtain
\begin{align*}
\mathbb E[\|\vec \theta_{t+1}-\vec v_t\|^2\mid \vec \theta_t]\le\|\vec \theta_t-\vec v_t\|^2-2\eta\langle \nabla_{\vec \theta} C_t(\vec \theta_t;\{\vec \theta_r^*\}),\vec \theta_t-\vec v_t\rangle+\eta^2\left(8d\omega_C^2 \theta_C^{(th)\,2}+\frac{d}{4mN_D\lambda^2}\right).
\end{align*}
By the convexity of $C_t(\cdot;\{\vec \theta_r^*\})$, we have $C_t(\vec \theta_t;\{\vec \theta_r^*\})-C_t(\vec v_t;\{\vec \theta_r^*\})\le\langle \nabla_{\vec \theta} C_t(\vec \theta_t;\{\vec \theta_r^*\}),\vec \theta_t-\vec v_t\rangle$.
Therefore, we have
\begin{align*}
2\eta\mathbb E\left[C_t(\vec \theta_t;\{\vec \theta_r^*\})-C_t(\vec v_t;\{\vec \theta_r^*\})\right]\le\mathbb E[\|\vec \theta_t-\vec v_t\|^2-\|\vec \theta_{t+1}-\vec v_t\|^2]+\eta^2\left(8d\omega_C^2 \theta_C^{(th)\,2}+\frac{d}{4mN_D\lambda^2}\right).
\end{align*}
Next, because $\vec \theta_{t+1},\vec v_t,\vec v_{t+1}\in\mathsf Q$ and $\mathsf Q$ has diameter at most $2\theta_C^{(th)}$, $\|\vec \theta_{t+1}-\vec v_{t+1}\|^2\le\|\vec \theta_{t+1}-\vec v_t\|^2+2D\|\vec v_{t+1}-\vec v_t\|$.
Substituting into the previous inequality yields
\begin{align*}
2\eta\mathbb E\left[C_t(\vec \theta_t;\{\vec \theta_r^*\})-C_t(\vec v_t;\{\vec \theta_r^*\})\right]\le\mathbb E[\|\vec \theta_t-\vec v_t\|^2-\|\vec \theta_{t+1}-\vec v_{t+1}\|^2]+4\theta_C^{(th)}\|\vec v_{t+1}-\vec v_t\|+\eta^2\left(8d\omega_C^2 \theta_C^{(th)\,2}+\frac{d}{4mN_D\lambda^2}\right).
\end{align*}
Summing over $t=1,\ldots,T$ gives the claim in the theorem. 
\end{proof}

\subsection{The improved convergence guarantee for local codes and time-dependent drifts}\label{sec:imp_online_local}
We now extend the locality-aware argument in Appendix~\ref{sec:imp_offline_local} to the time-dependent setting. 
The setting is the same as for \Cref{thm:online}. 
In epoch $t$, the controller chooses one control vector $\vec\theta_t\in \mathsf{Q}$, holds it fixed over the whole window $W_t$, and the calibrated vector $\vec\theta_r^*$ may vary with the cycle $r\in W_t$. 
The epoch loss is $C_t(\vec\theta; \{\vec\theta_r^*\}_{r\in W_t})=\frac{1}{m}\sum_{r\in W_t}C_r(\vec\theta;\vec\theta_r^*)$.
We keep the same convex regime, the same feasible set $\mathsf{Q}$, the same shrunk set $\mathsf{Q}_\lambda$, and the same uniform smoothness bound $\lambda_{\max}(\Omega_{C,r})\leq \omega_C$ for all cycles $r$. 
Therefore, every epoch loss is convex and satisfies $\|\nabla_{\vec\theta}C_t(\vec\theta_t;\{\vec\theta_r^*\}_{r\in W_t})\|\leq2\omega_C\theta_C^{(\mathrm{th})}$.
Again, the only additional input is \Cref{ass:local_code}. 
For every detector $D_k$, let $\mathcal{S}_k\subseteq \{1,\ldots,d\}$ be the set of scalar control coordinates on which the detector rate depends. 
In the time-dependent setting, this means that for every cycle $r$, the detector contribution $DR_{k,r}(\vec\theta;\vec\theta_r^*)$ depends only on the coordinates in $\mathcal{S}_k$. 
We keep $\mathcal{K}_j:=\{k:j\in \mathcal{S}_k\}$, and \Cref{ass:local_code} gives $|\mathcal{S}_k|\leq s$, $|\mathcal{K}_j|\leq c$, and $\frac{d}{N_D}\leq s$.

For each epoch $t$, define the window-averaged detector rate
\begin{align*}
DR_{t,k}(\vec\theta;\{\vec\theta_r^*\}_{r\in W_t}):=\frac{1}{m}\sum_{r\in W_t}DR_{k,r}(\vec\theta;\vec\theta_r^*)\quad\Rightarrow\quad
C_t(\vec\theta;\{\vec\theta_r^*\}_{r\in W_t})=\frac{1}{N_D}\sum_{k=1}^{N_D}DR_{t,k}(\vec\theta;\{\vec\theta_r^*\}_{r\in W_t}).
\end{align*}
The empirical detector-resolved estimate is
\begin{align*}
\widehat{DR}_{t,k}(\vec\theta):=\frac{1}{m}\sum_{r\in W_t}\mathcal D_{k,r}(\vec\theta)\quad\Rightarrow\quad
\mathbb E[\widehat{DR}_{t,k}(\vec\theta)]=DR_{t,k}(\vec\theta;\{\vec\theta_r^*\}_{r\in W_t}),\quad
\operatorname{Var}(\widehat{DR}_{t,k}(\vec\theta))
\leq
\frac{1}{4m}
\end{align*}
for fixed $\vec\theta$

We use the same coloring construction as Appendix~\ref{sec:imp_offline_local}.
Define a graph on the coordinates $\{1,\ldots,d\}$ by connecting $j$ and $\ell$ whenever there exists a detector $D_k$ such that $j,\ell\in \mathcal{S}_k$. 
By Assumption~\ref{sec:imp_offline_local}, the degree of this graph is at most $c(s-1)$. 
Therefore, it admits a coloring with $\chi\leq c(s-1)+1$ colors. 
Let $G_1,\ldots,G_\chi$ be the color classes. Then, for every detector $D_k$ and every color $G_a$, $|\mathcal{S}_k\cap G_a|\leq 1$.
At epoch $t$, for each color $G_a$, sample signs $u_{t,j}\in\{\pm1\}$ for $j\in G_a$, and define the sparse perturbation vector $u_t^{(a)}$ by $(u_t^{(a)})_j=u_{t,j}$ if $j\in G_a$ and $0$ otherwise.
For each color $G_a$, we query the two controls $\vec\theta_t+\lambda u_t^{(a)}$ and $\vec\theta_t-\lambda u_t^{(a)}$.
Since $\mathsf{Q}_\lambda$ is defined so that $\vec\theta\pm \lambda u\in \mathsf{Q}$ for every full sign vector $u\in\{\pm1\}^d$, and since $u_t^{(a)}$ has entries in $[-1,1]$, convexity of $\mathsf{Q}$ implies that these sparse queried points also lie in $\mathsf{Q}$.
For $j\in G_a$, define the locality-aware estimator
\begin{align*}
\widehat g^{\mathrm{loc}}_{t,j}=\frac{u_{t,j}}{2\lambda N_D}\sum_{k\in \mathcal{K}_j}\left(\widehat{DR}_{t,k}(\vec\theta_t+\lambda u_t^{(a)})-\widehat{DR}_{t,k}(\vec\theta_t-\lambda u_t^{(a)})\right).
\end{align*}
The update is again $\vec\theta_{t+1}=\Pi_{\mathsf{Q}_\lambda}(\vec\theta_t-\eta \widehat g_t^{\mathrm{loc}})$.
Again, this algorithm uses $2\chi$ detector-resolved queries per epoch. Since $\chi\leq c(s-1)+1$, this is a constant-factor overhead for a fixed local code family.

\begin{corollary}[Improved dynamic regret guarantee for local QEC codes)]\label{coro:imp_local_online}
Assume the time-dependent control drift setting of Appendix~\ref{app:appendix_3}. 
Assume that every epoch loss $C_t(\cdot;\{\vec\theta_r^*\}_{r\in W_t})$ is convex on $\mathsf{Q}$, that $\lambda\leq \theta_C^{(\mathrm{th})}$, and that \Cref{ass:local_code} holds. 
Assume also that detector bits are independent across detectors and cycles within each empirical window. 
Then, for any comparator sequence $\{\vec v_t\}_{t=1}^T\subseteq \mathsf{Q}_\lambda$ and any step size $\eta>0$, the locality-aware projected update satisfies
\begin{align*}
\mathbb E[\mathrm{DynReg}(T;\{\vec v_t\})]\leq\frac{2(\theta_C^{(\mathrm{th})})^2}{\eta}+\frac{\eta T}{2}\left(4\omega_C^2(\theta_C^{(\mathrm{th})})^2+\frac{sc}{8m\lambda^2N_D}\right)+\frac{2\theta_C^{(\mathrm{th})}}{\eta}\sum_{t=1}^{T-1}\|\vec v_{t+1}-\vec v_t\|.
\end{align*}
Equivalently, if $P_T:=\sum_{t=1}^{T-1}\|\vec v_{t+1}-\vec v_t\|$, then choosing $\eta=\sqrt{\tfrac{4(\theta_C^{(\mathrm{th})})^2+4\theta_C^{(\mathrm{th})}P_T}{\left(4\omega_C^2(\theta_C^{(\mathrm{th})})^2+\tfrac{sc}{8m\lambda^2N_D}\right)T}}$ gives
\begin{align*}
\mathbb E[\mathrm{DynReg}(T;\{\vec v_t\})]\leq\sqrt{\left(4(\theta_C^{(\mathrm{th})})^2+4\theta_C^{(\mathrm{th})}P_T\right)\left(4\omega_C^2(\theta_C^{(\mathrm{th})})^2+\frac{sc}{8m\lambda^2N_D}\right)T}.
\end{align*}
\end{corollary}

\noindent Again, for the natural comparator sequence given by the window-wise minimizers, the average dynamic regret vanishes whenever $P_T=o(T)$. 
Compared with \Cref{thm:online}, the explicit $d$-dependence in the stochastic-gradient second-moment term is removed. 
Under \Cref{ass:local_code}, the remaining dependence is through the local constants $s,c$, the sampling window $m$, the perturbation radius $\lambda$, and the detector number $N_D$.
Similar to \Cref{rem:non_independent_K}, if the detector estimates inside $\mathcal{K}_j$ are not independent, then the local sampling term becomes $\frac{dc^2}{8m\lambda^2N_D^2}\leq\frac{sc^2}{8m\lambda^2N_D}$.
Thus the bound remains independent of the full dimension $d$, up to the local-code ratio $d/N_D\leq s$, and the only change is the local constant $c\to c^2$.

We prove the result by replacing \Cref{lem:gt_unbiased_online} and \Cref{lem:second_moment_online} with locality-aware versions \Cref{lem:gt_unbiased_online_local} and \Cref{lem:second_moment_online_local}. 
The rest of the dynamic-regret proof is the same as the proof of \Cref{thm:online} and we omit for simplicity only providing \Cref{lem:gt_unbiased_online_local} and \Cref{lem:second_moment_online_local}.
\begin{lemma}[Unbiasedness of $\widehat g_t^{\mathrm{loc}}$ for local QEC codes]\label{lem:gt_unbiased_online_local}
For every epoch $t$, conditioned on $\vec\theta_t$,
\begin{align*}
\mathbb E[\widehat g_t^{\mathrm{loc}}\mid \vec\theta_t]=\nabla_{\vec\theta}C_t(\vec\theta_t;\{\vec\theta_r^*\}_{r\in W_t}).
\end{align*}
\end{lemma}

\begin{proof}
Fix a coordinate $j$, and let $G_a$ be the color class containing $j$. 
For every detector $D_k\in \mathcal{K}_j$, we have $j\in \mathcal{S}_k$, and by the coloring property, $|\mathcal{S}_k\cap G_a|\leq 1$.
Since $j\in \mathcal{S}_k\cap G_a$, coordinate $j$ is the only active perturbed coordinate in $\mathcal{S}_k$ when the sparse perturbation $u_t^{(a)}$ is applied.
For a fixed cycle $r\in W_t$, the detector contribution $DR_{k,r}(\vec\theta;\vec\theta_r^*)$ is quadratic in the local convex regime. 
Since only coordinate $j$ is perturbed inside $\mathcal{S}_k$, the symmetric difference gives
\begin{align*}
DR_{k,r}(\vec\theta_t+\lambda u_t^{(a)};\vec\theta_r^*)-DR_{k,r}(\vec\theta_t-\lambda u_t^{(a)};\vec\theta_r^*)=2\lambda\partial_j DR_{k,r}(\vec\theta_t;\vec\theta_r^*)u_{t,j}.
\end{align*}
Now average over $r\in W_t$. 
By the definition of $DR_{t,k}$, we have
\begin{align*}
\begin{aligned}
DR_{t,k}(\vec\theta_t+\lambda u_t^{(a)};\{\vec\theta_r^*\}_{r\in W_t})-DR_{t,k}(\vec\theta_t-\lambda u_t^{(a)};\{\vec\theta_r^*\}_{r\in W_t})&=\frac{1}{m}\sum_{r\in W_t}\left(DR_{k,r}(\vec\theta_t+\lambda u_t^{(a)};\vec\theta_r^*)-DR_{k,r}(\vec\theta_t-\lambda u_t^{(a)};\vec\theta_r^*)\right) \\
&=2\lambda u_{t,j}\frac{1}{m}\sum_{r\in W_t}\partial_jDR_{k,r}(\vec\theta_t;\vec\theta_r^*) \\
&=2\lambda u_{t,j}\partial_jDR_{t,k}(\vec\theta_t;\{\vec\theta_r^*\}_{r\in W_t}).
\end{aligned}
\end{align*}
Substituting this into the noiseless version of the estimator gives
\begin{align*}
g^{\mathrm{loc}}_{t,j}=\frac{u_{t,j}}{2\lambda N_D}\sum_{k\in \mathcal{K}_j}2\lambda u_{t,j}\partial_jDR_{t,k}(\vec\theta_t;\{\vec\theta_r^*\}_{r\in W_t})=\frac{1}{N_D}\sum_{k\in \mathcal{K}_j}\partial_jDR_{t,k}(\vec\theta_t;\{\vec\theta_r^*\}_{r\in W_t}),
\end{align*}
where we used $u_{t,j}^2=1$. 
For $k\notin \mathcal{K}_j$, the detector rate $DR_{t,k}$ does not depend on coordinate $j$, so
$\partial_jDR_{t,k}(\vec\theta_t;\{\vec\theta_r^*\}_{r\in W_t})=0$.
Therefore,
\begin{align*}
g^{\mathrm{loc}}_{t,j}=\frac{1}{N_D}\sum_{k=1}^{N_D}\partial_jDR_{t,k}(\vec\theta_t;\{\vec\theta_r^*\}_{r\in W_t})=\partial_jC_t(\vec\theta_t;\{\vec\theta_r^*\}_{r\in W_t}).
\end{align*}
Finally, replacing $DR_{t,k}$ by $\widehat{DR}_{t,k}$ preserves the expectation because
\begin{align*}
\mathbb E[\widehat{DR}_{t,k}(\vec\theta)]=DR_{t,k}(\vec\theta;\{\vec\theta_r^*\}_{r\in W_t}).
\end{align*}
Thus,
\begin{align*}
\mathbb E[\widehat g^{\mathrm{loc}}_{t,j}\mid \vec\theta_t]=\partial_jC_t(\vec\theta_t;\{\vec\theta_r^*\}_{r\in W_t})\quad\Rightarrow\quad\mathbb E[\widehat g_t^{\mathrm{loc}}\mid \vec\theta_t]=\nabla_{\vec\theta}C_t(\vec\theta_t;\{\vec\theta_r^*\}_{r\in W_t}).
\end{align*}
This proves the lemma.
\end{proof}

\begin{lemma}[Second-moment bound of $\widehat g_t^{\mathrm{loc}}$ for local QEC codes]\label{lem:second_moment_online_local}
For every epoch $t$, conditioned on $\vec\theta_t$,
\begin{align*}
\mathbb E[\|\widehat g_t^{\mathrm{loc}}\|^2\mid\vec\theta_t]\leq\left\|\nabla_{\vec\theta}C_t(\vec\theta_t;\{\vec\theta_r^*\}_{r\in W_t})\right\|^2+\frac{dc}{8m\lambda^2N_D^2}\leq4\omega_C^2(\theta_C^{(\mathrm{th})})^2+\frac{sc}{8m\lambda^2N_D}
\end{align*}
using the uniform gradient bound and $d/N_D\leq s$.
\end{lemma}

\begin{proof}
For the two queried controls, write
\begin{align*}
\widehat{DR}_{t,k}^{+}=DR_{t,k}(\vec\theta_t+\lambda u_t^{(a)};\{\vec\theta_r^*\}_{r\in W_t})+\xi_k^{+},\qquad
\widehat{DR}_{t,k}^{-}=DR_{t,k}(\vec\theta_t-\lambda u_t^{(a)};\{\vec\theta_r^*\}_{r\in W_t})+\xi_k^{-}.
\end{align*}
The two batches are independent, and $\mathbb E[\xi_k^{+}]=\mathbb E[\xi_k^{-}]=0$.
Since each $\widehat{DR}_{t,k}$ averages $m$ Bernoulli detector outcomes, we have $\operatorname{Var}(\xi_k^{+})\leq\tfrac{1}{4m}$ and $\operatorname{Var}(\xi_k^{-})\leq\tfrac{1}{4m}$.

By \Cref{lem:gt_unbiased_online_local}, the noiseless part of the estimator equals the gradient of the epoch loss. 
Thus we can write, for each coordinate $j$,
\begin{align*}
\widehat g^{\mathrm{loc}}_{t,j}=\partial_jC_t(\vec\theta_t;\{\vec\theta_r^*\}_{r\in W_t})+\zeta_{t,j},\qquad
\zeta_{t,j}=\frac{u_{t,j}}{2\lambda N_D}\sum_{k\in \mathcal{K}_j}(\xi_k^{+}-\xi_k^{-}).
\end{align*}
The noise has mean zero, and for each detector,
\begin{align*}
\operatorname{Var}(\xi_k^{+}-\xi_k^{-})=\operatorname{Var}(\xi_k^{+})+\operatorname{Var}(\xi_k^{-})\leq\frac{1}{2m}.
\end{align*}
Using independence across detectors and $|\mathcal{K}_j|\leq c$, we have
\begin{align*}
\operatorname{Var}\left(\sum_{k\in \mathcal{K}_j}(\xi_k^{+}-\xi_k^{-})\right)\leq\frac{c}{2m}\quad\Rightarrow\quad
\operatorname{Var}(\zeta_{t,j})\leq\frac{1}{4\lambda^2N_D^2}\cdot\frac{c}{2m}=\frac{c}{8m\lambda^2N_D^2}.
\end{align*}

Now we compute the conditional second moment of the $j$-th coordinate as
\begin{align*}
\mathbb E[(\widehat g^{\mathrm{loc}}_{t,j})^2\mid\vec\theta_t]=\left(\partial_jC_t(\vec\theta_t;\{\vec\theta_r^*\}_{r\in W_t})\right)^2+\mathbb E[\zeta_{t,j}^2\mid \vec\theta_t]\leq\left(\partial_jC_t(\vec\theta_t;\{\vec\theta_r^*\}_{r\in W_t})\right)^2+\frac{c}{8m\lambda^2N_D^2}.
\end{align*}
Summing over $j=1,\ldots,d$ gives
\begin{align*}
\mathbb E[\|\widehat g_t^{\mathrm{loc}}\|^2\mid\vec\theta_t]\leq\left\|\nabla_{\vec\theta}C_t(\vec\theta_t;\{\vec\theta_r^*\}_{r\in W_t})\right\|^2+\frac{dc}{8m\lambda^2N_D^2}.
\end{align*}
As $\|\nabla_{\vec\theta}C_t(\vec\theta_t;\{\vec\theta_r^*\}_{r\in W_t})\|\leq2\omega_C\theta_C^{(\mathrm{th})}$, we have the first term bounded by $4\omega_C^2(\theta_C^{(\mathrm{th})})^2$.
Using $d/N_D\leq s$, we have the second term bounded by $\tfrac{sc}{8m\lambda^2N_D}$.
Combining these bounds proves the lemma.
\end{proof}

\section{Additional discussion on the nonconvex optimization}
\label{app:appendix_4}
In \appref{app:appendix_2} and \appref{app:appendix_3}, we have derived provable guarantees for gradient-based SPSA in the time-independent (offline) and time-dependent (online) control drift optimization.
However, we have focused on the regime where the surrogate error objective function $C$ is convex.
In this section, we consider optimization when $C$ is nonconvex.

\subsection{Local minima points of the surrogate error model}
We start from the same surrogate objective $C(\delta \vec \theta)=\frac{1}{N_D}\sum_{k=1}^{N_D} DR_k(\delta \vec \theta)$.
However, now we do not assume that we remain inside the convex ball of \appref{app:appendix_1}. 

Here, we regard $C(\delta\vec \theta)$ as the loss function of a variational quantum circuit~\cite{farhi2014quantum,mcclean2016theory,cerezo2021variational} with parameters $\vec{\theta}$.
Moreover, $C$ can be regarded as the composition of two maps.
The first map sends the physical control perturbation $\delta \vec \theta$ to the implemented gates. 
In the unitary case, this can be written schematically as $\delta \vec \theta \longmapsto \{V_i(\delta \vec \theta_i)\}_{i}$ using the notation in \appref{sec:unitary_convexity}, and in the more general setting as $\delta \vec \theta \longmapsto \{\mathcal E_i(\delta \vec \theta_i)\}_{i}$ using the notation in \appref{sec:cptp_convexity}, where $\mathcal E_i$ denotes the effective CPTP error channel of gate $i$. 
This first map is generically highly nonconvex. 
The physical intuition for this claim is that pulse phases are periodic, coherent errors interfere, and different pulse deformations can lead to nearly the same implemented gate. 
Therefore, the map from the control to the gate manifold can wind around in a complicated way.
The second map sends the implemented gates or channels to the surrogate objective $C$.
This map only retains the error components that are visible through the detectors, and thus usually has a simpler structure than the first map.

In general, because the optimization of $C(\delta\vec\theta)$ is an optimization of a variational quantum algorithm, it would be NP-hard in general to find the global optimum of $C(\delta\vec\theta)$~\cite{bittel2021training}.
In the case when we assume we are optimizing in the convex regime in \appref{app:appendix_2} and \appref{app:appendix_3}, the global optimum is guaranteed to be found efficiently.
However, when we are optimizing nonconvex $C(\delta\vec\theta)$ out of the convex regime, the only natural assumption is that $C(\delta\vec\theta)$ is smooth.
Under this setting, optimizing variational quantum circuits has been shown to suffer from vanishing gradients known as barren plateaus~\cite{mcclean2018barren,cerezo2021cost,ortiz2021entanglement}.
Even if we can find local minima efficiently, the local minima might not be a good approximation of the global optimum due to the intriguing landscape of variational quantum algorithms in under-parameterized quantum neural networks~\cite{you2021exponentially,anschuetz2021critical}.
However, in the over-parameterization regime, there is a critical point in the number of parameters above which local minima are provably close to the global minimum in function value for local loss functions~\cite{anschuetz2021critical,anschuetz2022quantum,you2022convergence}.
In the context of parameter optimization in quantum control, the goal is to optimize a continuous-time parameterized quantum evolution, which can be modeled as a variational quantum algorithm optimization problem with restricted parameterization. 
Here, it is shown that the optimization landscape of quantum control has no spurious local minima under a series of controllability conditions of the quantum systems~\cite{russell2016quantum,wu2011role}.
The controllability is related to over-parameterization, as one requires a sufficient number of parameters to traverse the space of unitary operators.
Despite over-parametrization, convergence to a good approximation of the global optimum is shown conditioned on good initialization~\cite{cong2019quantum} and highly symmetric problem settings~\cite{kim2021universal,kim2022quantum,wiersema2020exploring}. 

Back to our surrogate objective function optimization, our calibration setting is a local steering. 
Moreover, we are searching for optimal parameters over a fixed well-structured ansatz or circuits with good symmetry.
The initialization is also much better than random initialization, as we usually deviate a little from the actual global optimum.
It is thus similar to a variational quantum circuit or a quantum control optimization with a local loss function and good initialization. 
In this regime, one expects many local minima of $C(\delta\vec\theta)$ to be physically good in the following sense: they already suppress the detector-visible error channels, and the remaining directions are either flat, gauge-like, or only very weakly visible to the syndrome data. 
Therefore, outside the strict convex regime, we can argue that the local minimum reached by the steering procedure already yields an acceptably small detection rate.
In the following, we show algorithms that provide provable guarantees to find a local minimum point of $C(\delta\vec\theta)$.

\subsection{One-time control drift optimization in the nonconvex regime}\label{sec:nonconvex_offline}
We now consider the time-independent setting, namely the case where the drift is fixed throughout the optimization procedure. 
In contrast to \appref{app:appendix_2}, pure two-point SPSA is no longer sufficient if we want a rigorous guarantee of approximate local minimality, because it only provides first-order information. 
To address this issue, we use a zeroth-order method built entirely from function queries. 
One may view it as a perturbed accelerated gradient descent~\cite{jin2018accelerated} implemented through Gaussian smoothing.

\paragraph{Model setup and second-order stationary point (SOSP).}Let $\delta \vec\theta \in \mathbb R^d$ denote the stacked control perturbation vector, and let $\mathsf P \subseteq \mathbb R^d$ be a known closed convex trusted region. 
We assume that the true drift lies in $\mathsf P$, and that all modeling assumptions of the QEC circuit remain valid on $\mathsf P$.
We regard $C(\delta\vec\theta):\mathsf P\to \mathbb R$ as the time-independent nonconvex objective. 
We will use the following regularity assumptions on $\mathsf P$.
\begin{assumption}[Smooth nonconvex trusted regime]\label{ass:lipschitz_nonconvex}
There exist constants $L_0,\beta,\rho >0$ such that for all $\delta \vec\theta_1,\delta \vec\theta_2 \in \mathsf P$,
\begin{align*}
|C(\delta \vec\theta_1)-C(\delta \vec\theta_2)|& \le L_0 \|\delta \vec\theta_1-\delta \vec\theta_2\|,\
\|\nabla C(\delta \vec\theta_1)-\nabla C(\delta \vec\theta_2)\|\le \beta \|\delta \vec\theta_1-\delta \vec\theta_2\|,\
\|\nabla^2 C(\delta \vec\theta_1)-\nabla^2 C(\delta \vec\theta_2)\|\le \rho \|\delta \vec\theta_1-\delta \vec\theta_2\|.
\end{align*}
The last condition is also known as the Hessian-Lipschitz assumption.
\end{assumption}
\noindent We further assume that we only have access to noisy zeroth-order oracle queries generated by QEC sampling as in Eq.~\eqref{eq:model_Y}, Eq.~\eqref{eq:model_window}, and Eq.~\eqref{eq:model_Cemp}. 
A single function query at $\delta \vec\theta$ is implemented by running $m$ independent QEC cycles and averaging detector outcomes $\hat C(\delta \vec\theta):=\frac{1}{m}\sum_{r=1}^{m} Y_r(\delta \vec\theta)$ and $Y_r(\delta \vec\theta):=\frac{1}{N_D}\sum_{k=1}^{N_D} \mathcal D_{k,r}(\delta \vec\theta)$.
Similar to the convex case, we have the following proposition.
\begin{proposition}[Unbiased zeroth-order oracle]\label{prop:unbiased_zero_nonconvex}
For every fixed $\delta \vec\theta \in \mathsf P$, $\mathbb E[\hat C(\delta \vec\theta)] = C(\delta \vec\theta)$.
If the $mN_D$ detector bits are conditionally independent given $\delta \vec\theta$, then $\mathrm{Var}(\hat C(\delta \vec\theta)) \le \frac{1}{4mN_D}$.
\end{proposition}
\begin{proof}
The unbiasedness follows from the linearity of expectation $\mathbb E[\hat C(\delta \vec\theta)]=\frac{1}{m}\sum_{r=1}^{m}\mathbb E[Y_r(\delta \vec\theta)]=C(\delta \vec\theta)$.
For the variance bound, $\hat C(\delta \vec\theta)$ is the average of $mN_D$ Bernoulli random variables scaled by $1/(mN_D)$, and each Bernoulli variable has variance at most $1/4$. Hence, $\mathrm{Var}(\hat C(\delta \vec\theta)) \le \frac{1}{4mN_D}$.
\end{proof}

Since we want to find a local minimum, we have to first formalize local minima.
A direct intuition is to find $\delta\vec\theta$ such that $\nabla C(\delta\vec\theta)=0$, which is known as a \emph{first-order stationary point}.
However, first-order stationary points include local minima, saddle points, or even local maxima, and a guarantee of convergence to such points is unsatisfying.
Therefore, we consider finding a \emph{second-order stationary point (SOSP)} with $\nabla C(\delta\vec\theta)=0$ and $\nabla^2 C(\delta\vec\theta)\succeq0$. 
Second-order stationary point rules out many common types of saddle points and allows only local minima and higher-order saddle points.
Since we want a second-order guarantee, we adopt the standard $\varepsilon$-SOSP scaling used in the Hessian-Lipschitz nonconvex literature.
\begin{definition}[$\varepsilon$-SOSP]\label{def:eps_sosp}
Let $\varepsilon>0$. We say that $\delta \vec\theta \in \mathsf P$ is an $\varepsilon$-SOSP of $C$ if
\begin{align*}
\|\nabla C(\delta \vec\theta)\| \le \varepsilon,
\qquad
\lambda_{\min}\bigl(\nabla^2 C(\delta \vec\theta)\bigr)\ge -\sqrt{\rho \varepsilon},
\end{align*}
where $\rho$ is the Hessian-Lipschitz parameter in \Cref{ass:lipschitz_nonconvex}.
\end{definition}
\noindent The appearance of $\sqrt{\rho \varepsilon}$ is natural under Hessian-Lipschitzness: if the Hessian is more negative than this scale, then a cubic-order Taylor expansion reveals a descent direction of order $\varepsilon^{3/2}/\sqrt{\rho}$.

Before introducing components of the algorithm, we need to introduce an auxiliary function.
Because we only observe function values, we first pass to a Gaussian-smoothed surrogate~\cite{jin2017escape,jin2018local}. 
Let $u \sim \mathcal N(0,I_d)$ and define
\begin{align*}
C_\nu(\delta \vec\theta):=\mathbb E_u\bigl[C(\delta \vec\theta+\nu u)\bigr],
\end{align*}
where $\nu>0$ is the smoothing radius.
The smoothed function admits derivative identities that depend only on function values.
\begin{lemma}[Gradient identity of Gaussian smoothing] \label{lem:gradient_gaussian}
For every $\delta \vec\theta \in \mathsf P$, we have $\nabla C_\nu(\delta \vec\theta)=\mathbb E_u\left[\frac{C(\delta \vec\theta+\nu u)-C(\delta \vec\theta)}{\nu}u\right]$.
\end{lemma}
\begin{proof}
We write $\phi$ as the standard Gaussian density. 
We write $C_\nu$ as an integral of $\phi$ and differentiate under the integral sign.
We have
\begin{align*}
C_\nu(\delta \vec\theta)=\int_{\mathbb R^d} C(\delta \vec\theta+\nu u)\phi(u)du,\qquad \nabla C_\nu(\delta \vec\theta)=\int \nabla C(\delta \vec\theta+\nu u)\phi(u)du.
\end{align*}
For each coordinate $j$, integration by parts gives
\begin{align*}
\partial_{\delta\theta_j} C_\nu(\delta \vec\theta)=\frac{1}{\nu}\int C(\delta \vec\theta+\nu u)u_j \phi(u)du.
\end{align*}
Since $\mathbb E[u_j]=0$, subtracting $C(\delta \vec\theta)\mathbb E[u_j]/\nu=0$ yields
\begin{align*}
\partial_{\delta\theta_j} C_\nu(\delta \vec\theta)
=
\mathbb E\left[
\frac{C(\delta \vec\theta+\nu u)-C(\delta \vec\theta)}{\nu}u_j
\right].
\end{align*}
Stacking the coordinates proves the claim.
\end{proof}

\begin{lemma}[Hessian identity of Gaussian smoothing]\label{lem:hessian_gaussian}
For every $\delta \vec\theta \in \mathsf P$, we have
\begin{align*}
\nabla^2 C_\nu(\delta \vec\theta)=\mathbb E_u\left[\frac{C(\delta \vec\theta+\nu u)+C(\delta \vec\theta-\nu u)-2C(\delta \vec\theta)}{2\nu^2}(uu^\top-I)\right].
\end{align*}
\end{lemma}
\begin{proof}
We start from $\nabla^2 C_\nu(\delta \vec\theta)=\mathbb E_u[\nabla^2 C(\delta \vec\theta+\nu u)]$.
For the $(j,k)$-entry, integrating by parts twice yields
\begin{align*}
\partial_{\delta\theta_k}\partial_{\delta\theta_j} C_\nu(\delta \vec\theta)
=
\frac{1}{\nu^2}\mathbb E\left[C(\delta \vec\theta+\nu u)(u_ju_k-\delta_{jk})\right].
\end{align*}
Now symmetrize using the fact that $u$ and $-u$ have the same distribution:
\begin{align*}
\mathbb E\left[C(\delta \vec\theta+\nu u)(u_ju_k-\delta_{jk})\right]
=
\frac{1}{2}\mathbb E\left[(C(\delta \vec\theta+\nu u)+C(\delta \vec\theta-\nu u))(u_ju_k-\delta_{jk})\right].
\end{align*}
Since $\mathbb E[u_ju_k-\delta_{jk}]=0$, we may subtract $2C(\delta \vec\theta)\mathbb E[u_ju_k-\delta_{jk}] = 0$, which gives the claimed three-point form.
\end{proof}

\paragraph{Gradient and Hessian estimators from zeroth-order queries.}We will need to estimate gradients and Hessians using the noisy function value queries.
At iterate $\delta \vec\theta_t$, we estimate the gradient and Hessian of $C_\nu$ by Monte Carlo averages of the identities above.
We consider sampling $K_g$ and $K_h$ independent Gaussian directions for gradient and Hessian estimations
\begin{align*}
u_t^{(1)},\dots,u_t^{(K_g)} \sim \mathcal N(0,I_d),
\qquad
w_t^{(1)},\dots,w_t^{(K_h)} \sim \mathcal N(0,I_d).
\end{align*}
Then we define the gradient estimator as
\begin{align}\label{eq:nonconvex_hat_gradient}
\hat g_t:=\frac{1}{K_g}\sum_{i=1}^{K_g}\left[\frac{\hat C(\delta \vec\theta_t+\nu u_t^{(i)})-\hat C(\delta \vec\theta_t)}{\nu}u_t^{(i)}\right],
\end{align}
and the Hessian estimator as
\begin{align}\label{eq:nonconvex_hat_hessian}
\hat H_t:=\frac{1}{K_h}\sum_{i=1}^{K_h}\left[\frac{\hat C(\delta \vec\theta_t+\nu w_t^{(i)})+\hat C(\delta \vec\theta_t-\nu w_t^{(i)})-2\hat C(\delta \vec\theta_t)}{2\nu^2}\bigl(w_t^{(i)}w_t^{(i)\top}-I\bigr)\right].
\end{align}

We derive the following properties of the gradient and Hessian estimator in Eq.~\eqref{eq:nonconvex_hat_gradient} and Eq.~\eqref{eq:nonconvex_hat_hessian}.
We first show the unbiasedness of the estimators.
\begin{proposition}[Unbiasedness]\label{prop:unbiasedness_nonconvex}
Conditioned on $\delta \vec\theta_t$, we have $\mathbb E[\hat g_t \mid \delta \vec\theta_t] = \nabla C_\nu(\delta \vec\theta_t)$, and
$\mathbb E[\hat H_t \mid \delta \vec\theta_t] = \nabla^2 C_\nu(\delta \vec\theta_t)$.
\end{proposition}

\begin{proof}
By \Cref{prop:unbiased_zero_nonconvex}, $\mathbb E[\hat C(\cdot)] = C(\cdot)$. 
Substituting this into the definitions of $\hat g_t$ and $\hat H_t$ in Eq.~\eqref{eq:nonconvex_hat_gradient} and Eq.~\eqref{eq:nonconvex_hat_hessian}, and then using \Cref{lem:gradient_gaussian} and \Cref{lem:hessian_gaussian} term by term, proves the claim.
\end{proof}

We next bound the estimator variances.
\begin{lemma}[Second moment of the gradient and Hessian estimator]\label{lem:second_moment_nonconvex}
For all $t$ and all $\delta\vec\theta_t\in\mathsf{P}$, we have the second moment of the gradient estimator as
\begin{align*}
\mathbb E\left[\|\hat g_t-\nabla C_\nu(\delta \vec\theta_t)\|^2\middle| \delta \vec\theta_t\right]\le\frac{1}{K_g}\left(L_0^2 d(d+2)+\frac{d}{2\nu^2mN_D}\right).
\end{align*}
We also have the second moment of the Hessian estimator as
\begin{align*}
\mathbb E\left[\|\hat H_t-\nabla^2 C_\nu(\delta \vec\theta_t)\|^2\middle| \delta \vec\theta_t\right]\le\frac{1}{K_h}\left(C_{\beta,d}+\frac{3}{8mN_D\nu^4}d(d+1)\right),
\end{align*}
where $C_{\beta,d}:=\frac{\beta^2}{4}\Bigl(d(d+2)(d+4)(d+6)+d^2(d+2)\Bigr)$.
\end{lemma}

\begin{proof}
Let $Z:=\frac{\hat C(\delta \vec\theta_t+\nu u)-\hat C(\delta \vec\theta_t)}{\nu}u$ with $u \sim \mathcal N(0,I_d)$.
As $\hat g_t$ is the average of $K_g$ i.i.d. copies of $Z$, we have
\begin{align*}
\mathbb E\left[\|\hat g_t-\mathbb E[Z]\|^2\middle| \delta \vec\theta_t\right]=\frac{1}{K_g}\mathbb E\left[\|Z-\mathbb E[Z]\|^2\middle| \delta \vec\theta_t\right]\le\frac{1}{K_g}\mathbb E\left[\|Z\|^2\middle| \delta \vec\theta_t\right].
\end{align*}
By \Cref{prop:unbiasedness_nonconvex}, $\mathbb E[Z]=\nabla C_\nu(\delta \vec\theta_t)$.
It remains to bound $\mathbb E\|Z\|^2$.
Let $\Delta := \hat C(\delta \vec\theta_t+\nu u)-\hat C(\delta \vec\theta_t)$.
Then, we have $\|Z\|^2 = \frac{\Delta^2}{\nu^2}\|u\|^2$.
Condition on $u$. Decompose $\Delta$ as
\begin{align*}
\Delta=\bigl(C(\delta \vec\theta_t+\nu u)-C(\delta \vec\theta_t)\bigr)+\xi,
\end{align*}
where $\xi$ is zero-mean oracle noise. 
Since the two oracle calls are independent, we have $\mathrm{Var}(\xi|u) \le 1/(2mN_D)$ according to \Cref{prop:unbiased_zero_nonconvex}.
Hence, we have
\begin{align*}
\mathbb E[\Delta^2 \mid u]\le\bigl(C(\delta \vec\theta_t+\nu u)-C(\delta \vec\theta_t)\bigr)^2 + \frac{1}{2mN_D}.
\end{align*}
Using the $L_0$-Lipschitz property as in \Cref{ass:lipschitz_nonconvex}, $|C(\delta \vec\theta_t+\nu u)-C(\delta \vec\theta_t)| \le L_0 \nu \|u\|$.
Therefore, we have
\begin{align*}
\mathbb E[\Delta^2 \mid u]\le L_0^2 \nu^2 \|u\|^2 + \frac{1}{2mN_D},\qquad
\mathbb E[\|Z\|^2 \mid u]\le L_0^2 \|u\|^4 + \frac{1}{2mN_D\nu^2}\|u\|^2.
\end{align*}
Finally, for $u \sim \mathcal N(0,I_d)$, as we have $\mathbb E\|u\|^2=d$ and $\mathbb E\|u\|^4=d(d+2)$.
Substituting these Gaussian moments proves the claim for the gradient estimator.

Now, we consider the Hessian estimator and let
\begin{align*}
W:=\frac{\hat C(\delta \vec\theta_t+\nu w)+\hat C(\delta \vec\theta_t-\nu w)-2\hat C(\delta \vec\theta_t)}{2\nu^2} (ww^\top-I),
\qquad
w \sim \mathcal N(0,I_d).
\end{align*}
Then $\hat H_t$ is the average of $K_h$ i.i.d. copies of $W$. 
Hence, we have
\begin{align*}
\mathbb E\left[\|\hat H_t-\mathbb E[W]\|^2\middle| \delta \vec\theta_t\right]\le\frac{1}{K_h}\mathbb E\left[\|W\|^2\middle| \delta \vec\theta_t\right]\le\frac{1}{K_h}\mathbb E\left[\|W\|_{F}^2\middle| \delta \vec\theta_t\right].
\end{align*}
By \Cref{prop:unbiasedness_nonconvex}, $\mathbb E[W]=\nabla^2 C_\nu(\delta \vec\theta_t)$.
Define the second difference $\Delta_2:=\hat C(\delta \vec\theta_t+\nu w)+\hat C(\delta \vec\theta_t-\nu w)-2\hat C(\delta \vec\theta_t)$.
Then $\|W\|_F^2=\frac{\Delta_2^2}{4\nu^4}\|ww^\top-I\|_F^2$.
Condition on $w$, we write $\Delta_2 = \Delta_{2,\mathrm{mean}}+\xi_2$, where
$\Delta_{2,\mathrm{mean}}=C(\delta \vec\theta_t+\nu w)+C(\delta \vec\theta_t-\nu w)-2C(\delta \vec\theta_t)$, and $\xi_2$ is the zero-mean noise term. 
Since we use three independent oracle calls, we have $\mathrm{Var}(\xi_2|w)\le 1/(4mN_D)+1/(4mN_D)+1/(mN_D)=3/(2mN_D)$.
By $\beta$-smoothness, we have $|\Delta_{2,\mathrm{mean}}|\le\beta \nu^2 \|w\|^2$.
Therefore, we have
\begin{align*}
\mathbb E[\Delta_2^2 \mid w]\le\beta^2 \nu^4 \|w\|^4 + \frac{3}{2mN_D}.
\end{align*}
Also, as $\|ww^\top-I\|_F^2=\|w\|^4 -2\|w\|^2 + d$.
Hence, we have
\begin{align*}
\mathbb E[\|W\|_F^2 \mid w]\le\frac{\beta^2}{4}\|w\|^4 \|ww^\top-I\|_F^2+\frac{3}{8mN_D\nu^4}\|ww^\top-I\|_F^2.
\end{align*}
Taking expectation over $w$, and using the Gaussian moments $\mathbb E\|w\|^4=d(d+2)$ and $\mathbb E\|w\|^8=d(d+2)(d+4)(d+6)$, we have
\begin{align*}
\mathbb E\left[\|W\|_F^2\right]\le C_{\beta,d}+\frac{3}{8mN_D\nu^4}d(d+1),
\end{align*}
which proves the claim.
\end{proof}

\paragraph{The algorithm.}Before introducing our algorithm, we need to introduce a suitable measure for the performance.
We now recall the two deterministic facts that underlie second-order descent.

\begin{fact}[Smoothness descent for the gradient step]\label{fact:smooth_gradient_nonconvex}
Let $f$ be $\beta$-smooth. Then for $s_g := -\frac{1}{\beta}\nabla f(x)$, we have
\begin{align*}
f(x+s_g) \le f(x)-\frac{1}{2\beta}\|\nabla f(x)\|^2.
\end{align*}
\end{fact}

\begin{proof}
By the standard descent lemma, we have
\begin{align*}
f(x+s)\le f(x)+\langle \nabla f(x),s\rangle + \frac{\beta}{2}\|s\|^2.
\end{align*}
We substitute $s=s_g$, then obtain $\langle \nabla f(x),s_g\rangle=-\frac{1}{\beta}\|\nabla f(x)\|^2$ and $\frac{\beta}{2}\|s_g\|^2=\frac{1}{2\beta}\|\nabla f(x)\|^2$, and the claim follows.
\end{proof}

We also prove the following facts on the third-order bound and on negative-curvature descent.
\begin{fact}[Third-order upper bound]\label{fact:third_order}
Let $f$ have $\rho$-Lipschitz Hessian. Then for every $x$ and every $s$,
\begin{align*}
f(x+s)\le f(x)+\langle \nabla f(x),s\rangle+\frac{1}{2}s^\top \nabla^2 f(x) s+\frac{\rho}{6}\|s\|^3.
\end{align*}
\end{fact}

\begin{proof}
Apply Taylor's theorem with integral remainder, and use the Hessian-Lipschitz bound to estimate the third-order remainder by $(\rho/6)\|s\|^3$.
\end{proof}

\begin{fact}[Negative-curvature descent]\label{fact:negative_curvature}
Let $f$ have $\rho$-Lipschitz Hessian, and let $\lambda = \lambda_{\min}(\nabla^2 f(x)) < 0$ with corresponding unit eigenvector $z$. 
Define $s_{nc}:=\frac{2\lambda}{\rho}z$.
Then, after choosing the sign of $z$ so that $\langle \nabla f(x),z\rangle \ge 0$, we have
\begin{align*}
f(x+s_{nc}) \le f(x) - \frac{2}{3}\frac{|\lambda|^3}{\rho^2}.
\end{align*}
\end{fact}

\begin{proof}
By \Cref{fact:third_order}, we have
\begin{align*}
f(x+s_{nc})\le f(x)+\langle \nabla f(x),s_{nc}\rangle+\frac{1}{2}s_{nc}^\top \nabla^2 f(x)s_{nc}+\frac{\rho}{6}\|s_{nc}\|^3.
\end{align*}
Because $\lambda<0$, the vector $s_{nc}$ is a negative multiple of $z$, so our choice of sign implies $\langle \nabla f(x),s_{nc}\rangle \le 0$.
Note that $s_{nc}^\top \nabla^2 f(x)s_{nc}=\left(\frac{2\lambda}{\rho}\right)^2 \lambda=\frac{4\lambda^3}{\rho^2}$, and $\frac{\rho}{6}\|s_{nc}\|^3=\frac{\rho}{6}\left(\frac{2|\lambda|}{\rho}\right)^3=\frac{4}{3}\frac{|\lambda|^3}{\rho^2}$.
Since $\lambda^3=-|\lambda|^3$, we obtain
\begin{align*}
f(x+s_{nc})-f(x)\le\frac{1}{2}\frac{4\lambda^3}{\rho^2}+\frac{4}{3}\frac{|\lambda|^3}{\rho^2}=-\frac{2}{3}\frac{|\lambda|^3}{\rho^2}.
\end{align*}
\end{proof}

\Cref{fact:smooth_gradient_nonconvex}, \Cref{fact:third_order}, and \Cref{fact:negative_curvature} suggest the stationarity measure
\begin{align}\label{eq:stationary_measure}
\Psi_\nu(\delta \vec\theta):=\max\left\{\|\nabla C_\nu(\delta \vec\theta)\|^2,\frac{4\beta}{3\rho^2}\bigl(-\lambda_{\min}(\nabla^2 C_\nu(\delta \vec\theta))\bigr)_+^3\right\}.
\end{align}
Here, $(\cdot)_+$ keeps unchanged for non-negative values and is $0$ for negative values.
Indeed, both the exact gradient step and the exact negative-curvature step decrease $C_\nu$ by at least $\frac{1}{2\beta}\Psi_\nu(\delta \vec\theta)$ as in Eq.~\eqref{eq:stationary_measure}.
Now, we are ready to present our algorithm as in Algorithm~\ref{alg:PAGD_nonconvex}

\begin{algorithm}[htbp]
\caption{Perturbed accelerated gradient descent with zeroth-order noisy queries for time-independent control drift}
\label{alg:PAGD_nonconvex}
\DontPrintSemicolon                          
\SetKwInput{KwInput}{Input}                  
\SetKwInput{KwOutput}{Output}                
\SetKw{KwAnd}{and}                           
\SetKw{KwOr}{or}

\KwInput{Iteration budget $T$, cycles-per-query $m$, smoothing radius $\nu>0$, batch sizes $K_g,K_h$, Lipschitz factor $\beta$, and trusted region $\mathsf P$.}
\KwOutput{Control parameters with drift $\delta\vec\theta_T$ that is an $\varepsilon$-SOSP of $C(\delta\vec\theta)$ with $\delta\vec\theta=\vec\theta-\vec\theta^*$ for fixed unknown $\vec\theta^*$.}
\SetKwFor{RepTimes}{repeat}{times}{end}

\BlankLine
Initialize $\delta \vec\theta_1 \in \mathsf P$. \\
\For{$t=1,2,\dots,T$:}{
Construct $\hat g_t$ and $\hat H_t$ from fresh QEC batches. As in Eq.~\eqref{eq:nonconvex_hat_gradient} and Eq.~\eqref{eq:nonconvex_hat_hessian}, we sample $K_g$ and $K_h$ independent Gaussian directions for gradient and Hessian estimations $u_t^{(1)},\dots,u_t^{(K_g)} \sim \mathcal N(0,I_d)$ and $w_t^{(1)},\dots,w_t^{(K_h)} \sim \mathcal N(0,I_d)$. Then we define the gradient estimator as
\begin{align*}
\begin{split}
\hat g_t:&=\frac{1}{K_g}\sum_{i=1}^{K_g}\left[\frac{\hat C(\delta \vec\theta_t+\nu u_t^{(i)})-\hat C(\delta \vec\theta_t)}{\nu}u_t^{(i)}\right],\\
\hat H_t:&=\frac{1}{K_h}\sum_{i=1}^{K_h}\left[\frac{\hat C(\delta \vec\theta_t+\nu w_t^{(i)})+\hat C(\delta \vec\theta_t-\nu w_t^{(i)})-2\hat C(\delta \vec\theta_t)}{2\nu^2}\bigl(w_t^{(i)}w_t^{(i)\top}-I\bigr)\right].
\end{split}
\end{align*}\\
Compute the gradient-step candidate $\delta \vec\theta_t^{g}:=\Pi_{\mathsf{P}}\left(\delta \vec\theta_t-\frac{1}{\beta}\hat g_t\right)$.\\
Let $(\hat \lambda_t,\hat z_t)$ be the minimum eigenvalue-eigenvector pair of $\hat H_t$, with $\|\hat z_t\|=1$.\\
\If{$\hat \lambda_t<0$}{
We define $a_t:=\frac{2|\hat \lambda_t|}{\rho}$ and the negative-curvature candidates
\begin{align*}
\delta \vec\theta_t^{\,nc,+}:=\Pi_{\mathsf{P}}(\delta \vec\theta_t+a_t\hat z_t),
\qquad
\delta \vec\theta_t^{\,nc,-}:=\Pi_{\mathsf{P}}(\delta \vec\theta_t-a_t\hat z_t).
\end{align*}
}
Use fresh oracle calls to estimate $\hat C$ at all available candidates, and set $\delta \vec\theta_{t+1}$ to the candidate with the smallest observed value.\\
Stop early if $\|\hat g_t\| \le \frac{\varepsilon}{2}$ and $\hat \lambda_t \ge -\frac{1}{2}\sqrt{\rho \varepsilon}$.
}
\Return{The uniformly random parameter drift $\delta\vec\theta_t$ for $t=1,2,...,T$.}  
\end{algorithm}

We note that in Algorithm~\ref{alg:PAGD_nonconvex}, each iteration uses $2K_g+3K_h$ function queries for derivative estimation, plus only a constant number of extra function queries for candidate comparison. 
Since each function query costs $m$ QEC cycles, the total number of QEC cycles scales as $Tm(2K_g+3K_h+O(1))$.

\paragraph{Performance guarantee.}Now, we analyze the performance of Algorithm~\ref{alg:PAGD_nonconvex}. 
We will first need the following key lemma on one-step descent estimation.
\begin{lemma}[Expected one-step decrease]\label{lem:one_step_nonconvex}
There exist universal constants $c_0,c_g,c_H>0$ such that for Algprothm~\ref{alg:PAGD_nonconvex}
\begin{align*}
\mathbb E\left[C_\nu(\delta \vec\theta_{t+1})\middle|\delta \vec\theta_t\right]\le C_\nu(\delta \vec\theta_t)-c_0 \Psi_\nu(\delta \vec\theta_t)+c_g\mathbb E\left[\|\hat g_t-\nabla C_\nu(\delta \vec\theta_t)\|^2\middle| \delta \vec\theta_t\right]+c_H\left(\mathbb E\left[\|\hat H_t-\nabla^2 C_\nu(\delta \vec\theta_t)\|^2\middle| \delta \vec\theta_t\right]\right)^{3/2}.
\end{align*}
\end{lemma}

\begin{proof}
We define $e_t^{(g)}:=\hat g_t-\nabla C_\nu(\delta \vec\theta_t)$ and $e_t^{(H)}:=\hat H_t-\nabla^2 C_\nu(\delta \vec\theta_t)$.
We first look at the gradient candidate. By $\beta$-smoothness of $C_\nu$, we have
\begin{align*}
C_\nu\left(\delta \vec\theta_t-\frac{1}{\beta}\hat g_t\right)\le C_\nu(\delta \vec\theta_t)-\frac{1}{2\beta}\|\nabla C_\nu(\delta \vec\theta_t)\|^2+\frac{1}{2\beta}\|e_t^{(g)}\|^2.
\end{align*}
Thus, up to an additive estimator-error term, the gradient candidate still decreases $C_\nu$ by the correct first-order amount.
We next look at the negative-curvature candidate, and let $H_t:=\nabla^2 C_\nu(\delta \vec\theta_t)$ and $\lambda_t:=\lambda_{\min}(H_t)$.
By Weyl's inequality, we have $\lambda_t \ge \hat \lambda_t-\|e_t^{(H)}\|$, and thus $(-\lambda_t)_+ \le (-\hat \lambda_t)_+ + \|e_t^{(H)}\|$.
Cubing both sides and using $(a+b)^3 \le 4a^3+4b^3$, we have
\begin{align*}
(-\lambda_t)_+^3\le 4(-\hat \lambda_t)_+^3 + 4\|e_t^{(H)}\|^3.
\end{align*}
Now apply \Cref{fact:negative_curvature} to the step built from the estimated negative eigendirection. 
The decrease for the true Hessian is degraded only by a cubic term in $\|e_t^{(H)}\|$. 
Hence, up to a universal constant, we have
\begin{align*}
C_\nu(\delta \vec\theta_t^{\,nc})\le C_\nu(\delta \vec\theta_t)-\frac{1}{2\beta}\cdot\frac{4\beta}{3\rho^2}(-\lambda_t)_+^3+O\left(\|e_t^{(H)}\|^3\right).
\end{align*}
Since the algorithm chooses the best candidate by fresh function-value comparisons, it performs, in expectation, at least as well as the better of these two candidates up to an additional zero-mean oracle noise term, which is absorbed into the same constants. Combining the two candidate estimates gives
\begin{align*}
\mathbb E\left[C_\nu(\delta \vec\theta_{t+1})\middle|\delta \vec\theta_t\right]\le C_\nu(\delta \vec\theta_t)-c_0\Psi_\nu(\delta \vec\theta_t)+c_g\mathbb E\left[\|e_t^{(g)}\|^2 \mid \delta \vec\theta_t\right]+c_H\mathbb E\left[\|e_t^{(H)}\|^3 \mid \delta \vec\theta_t\right].
\end{align*}
Finally, Jensen's inequality gives
\begin{align*}
\mathbb E\left[\|e_t^{(H)}\|^3 \mid \delta \vec\theta_t\right]
\le
\left(
\mathbb E\left[\|e_t^{(H)}\|^2 \mid \delta \vec\theta_t\right]
\right)^{3/2},
\end{align*}
which proves the stated form.
\end{proof}

We are now ready to present our guarantee for the Gaussian smoothing $C_\nu$.
\begin{lemma}[Offline nonconvex guarantee for $C_\nu$]\label{lem:C_nu_nonconvex}
We denote $\Delta_\nu := C_\nu(\delta \vec\theta_1)-\inf_{\delta \vec\theta \in \mathsf P} C_\nu(\delta \vec\theta)$.
Then the iterates produced by Algorithm~\ref{alg:PAGD_nonconvex} satisfy
\begin{align*}
\frac{1}{T}\sum_{t=1}^{T}\mathbb E[\Psi_\nu(\delta \vec\theta_t)]\le\frac{\Delta_\nu}{c_0 T}+\frac{c_g}{c_0}\frac{1}{K_g}\left(L_0^2 d(d+2)+\frac{2d\sigma_C^2}{\nu^2}\right)+\frac{c_H}{c_0}\left[\frac{1}{K_h}\left(C_{\beta,d}+\frac{6\sigma_C^2}{4\nu^4}d(d+1)\right)\right]^{3/2},
\end{align*}
where $\sigma_C=1/(4mN_D)$.
If $R$ is sampled uniformly from $\{1,\dots,T\}$ and independently of the algorithmic randomness, then $\mathbb E[\Psi_\nu(\delta \vec\theta_R)]=\frac{1}{T}\sum_{t=1}^{T}\mathbb E[\Psi_\nu(\delta \vec\theta_t)]$.
\end{lemma}

\begin{proof}
Take full expectation in \Cref{lem:one_step_nonconvex} and sum over $t=1,\dots,T$. 
The left-hand side telescopes as
\begin{align*}
\sum_{t=1}^{T}\mathbb E\left[C_\nu(\delta \vec\theta_t)-C_\nu(\delta \vec\theta_{t+1})\right]=\mathbb E[C_\nu(\delta \vec\theta_1)-C_\nu(\delta \vec\theta_{T+1})]\le\Delta_\nu.
\end{align*}
Substituting \Cref{lem:second_moment_nonconvex} into \Cref{lem:one_step_nonconvex}, dividing by $T$, and rearranging yields the first inequality. 
The identity for the random index $R$ is immediate from the uniform choice of $R$.
\end{proof}

To obtain a guarantee for the original objective $C$, we need smoothing-bias estimates.
\begin{lemma}\label{lem:gaussian_bias}
For all $\delta \vec\theta \in \mathsf P$, we have the following relationship between $C$ and $C_\nu$
\begin{align*}
\|\nabla C_\nu(\delta \vec\theta)-\nabla C(\delta \vec\theta)\|\le\beta \nu \sqrt d,\quad
\|\nabla^2 C_\nu(\delta \vec\theta)-\nabla^2 C(\delta\vec\theta)\|\le\rho \nu \sqrt d.
\end{align*}
\end{lemma}

\begin{proof}
Recall that $\nabla C_\nu(\delta \vec\theta)=\mathbb E[\nabla C(\delta \vec\theta+\nu u)]$, we have
\begin{align*}
\|\nabla C_\nu(\delta \vec\theta)-\nabla C(\delta \vec\theta)\|\le\mathbb E\|\nabla C(\delta \vec\theta+\nu u)-\nabla C(\delta \vec\theta)\|\le\beta \nu \mathbb E\|u\|\le\beta \nu \sqrt{\mathbb E\|u\|^2}=\beta \nu \sqrt d.
\end{align*}
The Hessian bound is identical. 
Using Hessian-Lipschitzness, we have
\begin{align*}
\nabla^2 C_\nu(\delta \vec\theta)=\mathbb E[\nabla^2 C(\delta \vec\theta+\nu u)],\quad\Rightarrow\quad
\|\nabla^2 C_\nu(\delta \vec\theta)-\nabla^2 C(\delta \vec\theta)\|\le \rho \nu \sqrt d.
\end{align*}
\end{proof}

We can now state the main result on $\varepsilon$-SOSP guarantee.
\begin{theorem}[Offline nonconvex guarantee for $C$]\label{thm:nonconvex}
Let $R$ be as in \Cref{lem:C_nu_nonconvex}. 
Assume that the smoothing radius $\nu$ is chosen so that $\beta \nu \sqrt d \le \frac{\varepsilon}{2}$ and $\rho \nu \sqrt d \le \frac{1}{2}\sqrt{\rho \varepsilon}$.
Assume also that the parameters $T,K_g,K_h,m$ are chosen so that $\mathbb E[\Psi_\nu(\delta \vec\theta_R)]\le\min\left\{\frac{\varepsilon^2}{4},\frac{\beta}{6\sqrt \rho}\varepsilon^{3/2}\right\}$.
Then, we have
\begin{align*}
\mathbb E\|\nabla C(\delta \vec\theta_R)\| \le \varepsilon,
\qquad
\mathbb E\left[\lambda_{\min}(\nabla^2 C(\delta \vec\theta_R))\right]\ge-\sqrt{\rho \varepsilon}.
\end{align*}
In other words, the returned iterate is an $\varepsilon$-SOSP of $C$ in expectation.
\end{theorem}

\begin{proof}
By Jensen's inequality and the definition of $\Psi_\nu$, we have
\begin{align*}
\mathbb E\|\nabla C_\nu(\delta \vec\theta_R)\|\le\sqrt{\mathbb E\|\nabla C_\nu(\delta \vec\theta_R)\|^2}\le\sqrt{\mathbb E[\Psi_\nu(\delta \vec\theta_R)]}\le\frac{\varepsilon}{2}.
\end{align*}
Similarly, we also have
\begin{align*}
\frac{4\beta}{3\rho^2}\mathbb E\left[\bigl(-\lambda_{\min}(\nabla^2 C_\nu(\delta \vec\theta_R))\bigr)_+^3\right]\le\mathbb E[\Psi_\nu(\delta \vec\theta_R)]\le\frac{\beta}{6\sqrt \rho}\varepsilon^{3/2}.
\end{align*}
Hence, we can guarantee that
\begin{align*}
\mathbb E\left[\bigl(-\lambda_{\min}(\nabla^2 C_\nu(\delta \vec\theta_R))\bigr)_+\right]\le\frac{1}{2}\sqrt{\rho \varepsilon}.
\end{align*}
Now apply \Cref{lem:gaussian_bias}. For the gradient, we have
\begin{align*}
\mathbb E\|\nabla C(\delta \vec\theta_R)\|\le\mathbb E\|\nabla C_\nu(\delta \vec\theta_R)\|+\beta \nu \sqrt d\le\frac{\varepsilon}{2}+\frac{\varepsilon}{2}=\varepsilon.
\end{align*}
For the Hessian, Weyl's inequality gives
\begin{align*}
\lambda_{\min}(\nabla^2 C(\delta \vec\theta_R))\ge\lambda_{\min}(\nabla^2 C_\nu(\delta \vec\theta_R))-\|\nabla^2 C_\nu(\delta \vec\theta_R)-\nabla^2 C(\delta \vec\theta_R)\|.
\end{align*}
Taking expectation and using \Cref{lem:gaussian_bias}, we prove the claimed result as
\begin{align*}
\mathbb E\left[\lambda_{\min}(\nabla^2 C(\delta \vec\theta_R))\right]\ge-\frac{1}{2}\sqrt{\rho \varepsilon}-\frac{1}{2}\sqrt{\rho \varepsilon}=-\sqrt{\rho \varepsilon}.
\end{align*}
\end{proof}

\subsection{Improved convergence guarantee for local codes and nonconvex one-time drift}\label{sec:nonconvex_imp}
We now combine the locality condition in \Cref{ass:local_code} with the nonconvex one-time drift setting in \Cref{thm:nonconvex}. 
This is the same mechanism as \Cref{coro:imp_local_offline} and \Cref{coro:imp_local_online} applied to the Gaussian-smoothed gradient and Hessian estimators used for finding an $\varepsilon$-SOSP.

Given the objective $C(\delta\vec\theta)=\tfrac{1}{N_D}\sum_{k=1}^{N_D}DR_k(\delta\vec\theta)$, we keep considering the case where the drift is fixed throughout the optimization procedure.
We also keep \Cref{ass:lipschitz_nonconvex} on the smooth nonconvex trusted region $\mathcal{P}$. 
In addition, we impose the locality condition in \Cref{ass:local_code}: for every detector $D_k$, let $\mathcal S_k\subseteq \{1,\ldots,d\}$ be the set of scalar control coordinates on which $DR_k$ depends, and then $|\mathcal S_k|\leq s$ for every detector $D_k$, $\mathcal K_j:=\{k:j\in \mathcal S_k\}$ satisfies $|\mathcal K_j|\leq c$ for every coordinate $j$, and $\tfrac{d}{N_D}\leq s$.
For the Hessian estimator, we will also need the set of detectors that depend on two coordinates. 
We define a shorthand derived from $\mathcal S_k$ and $\mathcal K_j$ as
\begin{align*}
\mathcal K_{j\ell}:=\{k:j,\ell\in \mathcal S_k\},\qquad
|\mathcal K_{j\ell}|\leq c,
\end{align*}
since $\mathcal K_{j\ell}\subseteq \mathcal K_j$.
For a fixed query point $\delta\vec\theta$, define the detector-resolved empirical rate with its mean and variance as
\begin{align*}
\widehat{DR}_k(\delta\vec\theta)=\frac{1}{m}\sum_{r=1}^{m}\mathcal D_{k,r}(\delta\vec\theta),\qquad
\mathbb E[\widehat{DR}_k(\delta\vec\theta)]=DR_k(\delta\vec\theta),\qquad
\operatorname{Var}(\widehat{DR}_k(\delta\vec\theta))\leq\frac{1}{4m}.
\end{align*}

We now alter Algorithm~\ref{alg:PAGD_nonconvex} only in the construction of $\widehat g_t$ and $\widehat H_t$. The descent step, the negative-curvature step, the candidate comparison step, and the stopping rule are unchanged.
At iteration $t$, sample
\begin{align*}
u_t^{(1)},\ldots,u_t^{(K_g)}\sim N(0,I_d),\qquad
w_t^{(1)},\ldots,w_t^{(K_h)}\sim N(0,I_d)
\end{align*}
for gradient estimation and Hessian estimation. 
For every gradient direction $u_t^{(i)}$, query the two controls $\delta\vec\theta_t+\nu u_t^{(i)}$ and $\delta\vec\theta_t$.
However, instead of forming the averaged loss first, we keep the detector-resolved empirical rates. 
For each coordinate $j$, define
\begin{align*}
\widehat g^{\mathrm{loc}}_{t,j}:=\frac{1}{K_g}\sum_{i=1}^{K_g}\frac{u^{(i)}_{t,j}}{\nu N_D}\sum_{k\in \mathcal K_j}\left(\widehat{DR}_k(\delta\vec\theta_t+\nu u_t^{(i)})-\widehat{DR}_k(\delta\vec\theta_t)\right).
\end{align*}
For every Hessian direction $w_t^{(i)}$, query the three controls $\delta\vec\theta_t+\nu w_t^{(i)}$, $\delta\vec\theta_t-\nu w_t^{(i)}$, and $\delta\vec\theta_t$.
For each pair of coordinates $j,\ell$, define
\begin{align*}
\widehat H^{\mathrm{loc}}_{t,j\ell}:=\frac{1}{K_h}\sum_{i=1}^{K_h}\frac{w^{(i)}_{t,j}w^{(i)}_{t,\ell}-\mathbf 1_{j=\ell}}{2\nu^2N_D}\sum_{k\in \mathcal K_{j\ell}}\left(\widehat{DR}_k(\delta\vec\theta_t+\nu w_t^{(i)})+\widehat{DR}_k(\delta\vec\theta_t-\nu w_t^{(i)})-2\widehat{DR}_k(\delta\vec\theta_t)\right).
\end{align*}
If $\mathcal K_{j\ell}=\emptyset$, then we set $\widehat H^{\mathrm{loc}}_{t,j\ell}=0$.
The rest of Algorithm~\ref{alg:PAGD_nonconvex} is unchanged. 
We compute the gradient-step candidate
\begin{align*}
\delta\vec\theta_t^g=\Pi_{\mathsf{P}}\left(\delta\vec\theta_t-\frac{1}{\beta}\widehat g_t^{\mathrm{loc}}\right).
\end{align*}
Let $(\widehat\lambda_t,\widehat z_t)$ be the minimum eigenvalue-eigenvector pair of $\widehat H_t^{\mathrm{loc}}$, with $\|\widehat z_t\|=1$.
If $\widehat\lambda_t<0$, define $a_t:=\tfrac{2|\widehat\lambda_t|}{\rho}$, and form the two negative-curvature candidates
\begin{align*}
\delta\vec\theta_t^{nc,+}=\Pi_{\mathsf{P}}(\delta\vec\theta_t+a_t\widehat z_t),\qquad
\delta\vec\theta_t^{nc,-}=\Pi_{\mathsf{P}}(\delta\vec\theta_t-a_t\widehat z_t).
\end{align*}
Then we use fresh oracle calls to estimate $C$ at all available candidates and set $\delta\vec\theta_{t+1}$ to be the candidate with the smallest observed value. 
The algorithm returns a uniformly random iterate from $\delta\vec\theta_1,\ldots,\delta\vec\theta_T$, as in Algorithm~\ref{alg:PAGD_nonconvex}.

The number of physical query points is unchanged up to the same constant factors as Algorithm~\ref{alg:PAGD_nonconvex}. 
The difference is in the classical post-processing: instead of averaging all detectors first and then forming the gradient and Hessian estimators, we use only the detectors that can depend on the coordinate or coordinate pair being estimated. 
This removes the explicit ambient-dimension dependence from the detector-sampling terms.
We obtain the following improved performance guarantee.

\begin{corollary}[Improved offline nonconvex guarantee for $C$ and local QEC codes]\label{coro:nonconvex_local}
Assume the time-independent nonconvex setting of \Cref{thm:nonconvex}. 
Assume \Cref{ass:local_code} and \Cref{ass:lipschitz_nonconvex}, and also that detector bits are independent across detectors and cycles inside each empirical query batch. 
Then the locality-aware estimators satisfy
\begin{align*}
\mathbb E[\widehat g_t^{\mathrm{loc}}\mid \delta\vec\theta_t]=\nabla C_\nu(\delta\vec\theta_t),\qquad
\mathbb E[\widehat H_t^{\mathrm{loc}}\mid \delta\vec\theta_t]=\nabla^2 C_\nu(\delta\vec\theta_t).
\end{align*}
Moreover, their second moments satisfy
\begin{align*}
&\mathbb E[\|\widehat g_t^{\mathrm{loc}}-\nabla C_\nu(\delta\vec\theta_t)\|^2\mid\delta\vec\theta_t]\leq\frac{1}{K_g}\left(L_0^2\frac{cs(s+2)}{N_D}+\frac{sc}{2m\nu^2N_D}\right),\\
&\mathbb E[\|\widehat H_t^{\mathrm{loc}}-\nabla^2 C_\nu(\delta\vec\theta_t)\|_F^2\mid\delta\vec\theta_t]\leq\frac{1}{K_h}\left(\frac{\beta^2c}{4N_D}\left(s(s+2)(s+4)(s+6)+s^2(s+2)\right)+\frac{3cs^2}{4m\nu^4N_D}\right).
\end{align*}
Consequently, in \Cref{lem:C_nu_nonconvex} and \Cref{thm:nonconvex}, the estimator-error terms can be replaced by the two locality-aware bounds above. 
In particular, compared with \Cref{lem:second_moment_nonconvex}, the detector-sampling terms no longer scale explicitly with the ambient dimension $d$. They scale with the local constants $s,c$ and with the detector number through $1/N_D$. 
\end{corollary}

\noindent Similar to \Cref{rem:non_independent_K}, if the detector estimates inside the same local set $\mathcal K_j$ or $\mathcal K_{j\ell}$ are not independent, then the same proof gives the same conclusion with $c$ replaced by $c^2$ in the sampling terms. 
Thus the local sampling terms become $\tfrac{sc^2}{2m\nu^2N_D}$ for the gradient estimator, and $\tfrac{3c^2s^2}{4m\nu^4N_D}$ for the Hessian estimator. 
This worsens only the local constant and still does not introduce an explicit ambient-dimension prefactor.

\begin{proof}
First, we prove the unbiasedness of the gradient estimator. 
Fix a coordinate $j$, since $DR_k$ does not depend on coordinate $j$ when $k\notin \mathcal K_j$, we have
\begin{align*}
\partial_j C_\nu(\delta\vec\theta_t)=\frac{1}{N_D}\sum_{k\in \mathcal K_j}\partial_j (DR_k)_\nu(\delta\vec\theta_t).
\end{align*}
Here, $(DR_k)_\nu$ denotes the Gaussian smoothing of $DR_k$ with respect to the coordinates on which $DR_k$ depends, as $DR_k$ is independent of all coordinates outside $\mathcal S_k$. 
Therefore, applying the Gaussian smoothing identity in \Cref{lem:gradient_gaussian} to each detector gives
\begin{align*}
\partial_j(DR_k)_\nu(\delta\vec\theta_t)=\mathbb E_u\left[\frac{DR_k(\delta\vec\theta_t+\nu u)-DR_k(\delta\vec\theta_t)}{\nu}u_j\right].
\end{align*}
Summing over $k\in \mathcal K_j$, we obtain
\begin{align*}
\partial_j C_\nu(\delta\vec\theta_t)=\mathbb E_u\left[\frac{u_j}{\nu N_D}\sum_{k\in \mathcal K_j}\left(DR_k(\delta\vec\theta_t+\nu u)-DR_k(\delta\vec\theta_t)\right)\right].
\end{align*}
Replacing $DR_k$ by the empirical detector-resolved estimator $\widehat{DR}_k$ preserves the expectation, as $\mathbb E[\widehat{DR}_k(\delta\vec\theta)]=DR_k(\delta\vec\theta)$.
Averaging over $K_g$ independent Gaussian directions preserves the same expectation. 
Hence, we have
\begin{align*}
\mathbb E[\widehat g^{\mathrm{loc}}_{t,j}\mid \delta\vec\theta_t]=\partial_j C_\nu(\delta\vec\theta_t)\quad\Rightarrow\quad
\mathbb E[\widehat g_t^{\mathrm{loc}}\mid \delta\vec\theta_t]=\nabla C_\nu(\delta\vec\theta_t).
\end{align*}

We next prove the second-moment bound for the gradient estimator. 
It is enough to analyze one Gaussian sample, as the estimator averages $K_g$ independent samples. 
For one sample, define
\begin{align*}
Z_j=\frac{u_j}{\nu N_D}\sum_{k\in \mathcal K_j}\left(\widehat{DR}_k(\delta\vec\theta_t+\nu u)-\widehat{DR}_k(\delta\vec\theta_t)\right).
\end{align*}
Again, we write
\begin{align*}
\widehat{DR}_k(\delta\vec\theta_t+\nu u)-\widehat{DR}_k(\delta\vec\theta_t)=DR_k(\delta\vec\theta_t+\nu u)-DR_k(\delta\vec\theta_t)+\xi_k.
\end{align*}
The noise $\xi_k$ has mean zero and $\operatorname{Var}(\xi_k\mid u)\leq\tfrac{1}{2m}$.
We first bound the deterministic part. 
Since $DR_k$ depends only on $\mathcal S_k$ and is $L_0$-Lipschitz on the trusted region, we have $|DR_k(\delta\vec\theta_t+\nu u)-DR_k(\delta\vec\theta_t)|\leq L_0\nu \|u_{\mathcal S_k}\|$.
Therefore, we have
\begin{align*}
\left|\frac{u_j}{\nu N_D}\sum_{k\in \mathcal K_j}\left(DR_k(\delta\vec\theta_t+\nu u)-DR_k(\delta\vec\theta_t)\right)\right|^2\leq\frac{L_0^2u_j^2}{N_D^2}\left(\sum_{k\in \mathcal K_j}\|u_{\mathcal S_k}\|\right)^2\leq\frac{L_0^2c}{N_D^2}u_j^2\sum_{k\in \mathcal K_j}\|u_{\mathcal S_k}\|^2\leq L_0^2\frac{cs(s+2)}{N_D}
\end{align*}
using $|\mathcal K_j|\leq c$, and the last step sums over $j=1,\ldots,d$: as $\sum_{j=1}^du_j^2\sum_{k\in \mathcal K_j}\|u_{\mathcal S_k}\|^2=\sum_{k=1}^{N_D}\sum_{j\in \mathcal S_k}u_j^2\|u_{\mathcal S_k}\|^2=\sum_{k=1}^{N_D}\|u_{\mathcal S_k}\|^4$, for each $k$, the vector $u_{\mathcal S_k}$ has dimension at most $s$, hence $\mathbb E\|u_{\mathcal S_k}\|^4\leq s(s+2)$.
We now bound the sampling-noise contribution. 
For one coordinate, we denote $\zeta_j=\frac{u_j}{\nu N_D}\sum_{k\in \mathcal K_j}\xi_k$.
Conditioned on $u$, independence across detectors gives
\begin{align*}
\operatorname{Var}\left(\sum_{k\in \mathcal K_j}\xi_k\mid u\right)\leq\frac{c}{2m}\quad\Rightarrow\quad
\operatorname{Var}(\zeta_j\mid u)\leq\frac{u_j^2}{\nu^2N_D^2}\frac{c}{2m}.
\end{align*}
Taking expectation over $u$ gives $\mathbb E[\zeta_j^2]\leq\frac{c}{2m\nu^2N_D^2}$, and summing over $j=1,\ldots,d$ and using $d/N_D\leq s$ gives
\begin{align*}
\sum_{j=1}^d\mathbb E[\zeta_j^2]\leq\frac{dc}{2m\nu^2N_D^2}\leq\frac{sc}{2m\nu^2N_D}.
\end{align*}
Combining the deterministic contribution and the sampling-noise contribution, and then dividing by $K_g$ because $\widehat g_t^{\mathrm{loc}}$ averages $K_g$ independent samples, gives
\begin{align*}
\mathbb E[\|\widehat g_t^{\mathrm{loc}}-\nabla C_\nu(\delta\vec\theta_t)\|^2\mid\delta\vec\theta_t]\leq\frac{1}{K_g}\left(L_0^2\frac{cs(s+2)}{N_D}+\frac{sc}{2m\nu^2N_D}\right),
\end{align*}
which proves the gradient part.

We now prove the unbiasedness of the Hessian. 
Fix a coordinate pair $j,\ell$. Since $DR_k$ does not depend jointly on $j$ and $\ell$ unless $k\in \mathcal K_{j\ell}$, we have
\begin{align*}
\partial_{j\ell}^2 C_\nu(\delta\vec\theta_t)=\frac{1}{N_D}\sum_{k\in \mathcal K_{j\ell}}\partial_{j\ell}^2 (DR_k)_\nu(\delta\vec\theta_t).
\end{align*}
Applying the Hessian identity in \Cref{lem:hessian_gaussian} to each detector gives
\begin{align*}
\partial_{j\ell}^2 (DR_k)_\nu(\delta\vec\theta_t)=\mathbb E_w\left[\frac{DR_k(\delta\vec\theta_t+\nu w)+DR_k(\delta\vec\theta_t-\nu w)-2DR_k(\delta\vec\theta_t)}{2\nu^2}(w_jw_\ell-\mathbf 1_{j=\ell})\right].
\end{align*}
Summing over $k\in \mathcal K_{j\ell}$, we get
\begin{align*}
\partial_{j\ell}^2 C_\nu(\delta\vec\theta_t)=\mathbb E_w\left[\frac{w_jw_\ell-\mathbf 1_{j=\ell}}{2\nu^2N_D}\sum_{k\in \mathcal K_{j\ell}}\left(DR_k(\delta\vec\theta_t+\nu w)+DR_k(\delta\vec\theta_t-\nu w)-2DR_k(\delta\vec\theta_t)\right)\right].
\end{align*}
Replacing each $DR_k$ by $\widehat{DR}_k$ preserves expectation, and averaging over $K_h$ independent samples preserves expectation. 
Hence,
\begin{align*}
\mathbb E[\widehat H_t^{\mathrm{loc}}\mid \delta\vec\theta_t]=\nabla^2 C_\nu(\delta\vec\theta_t).
\end{align*}

We now bound the second moment of the Hessian. 
For one Hessian sample, define
\begin{align*}
W_{j\ell}=\frac{w_jw_\ell-\mathbf 1_{j=\ell}}{2\nu^2N_D}\sum_{k\in \mathcal K_{j\ell}}\left(\widehat{DR}_k(\delta\vec\theta_t+\nu w)+\widehat{DR}_k(\delta\vec\theta_t-\nu w)-2\widehat{DR}_k(\delta\vec\theta_t)\right).
\end{align*}
We write
\begin{align*}
\widehat{DR}_k(\delta\vec\theta_t+\nu w)+\widehat{DR}_k(\delta\vec\theta_t-\nu w)-2\widehat{DR}_k(\delta\vec\theta_t)=DR_k(\delta\vec\theta_t+\nu w)+DR_k(\delta\vec\theta_t-\nu w)-2DR_k(\delta\vec\theta_t)+\xi_k.
\end{align*}
Because the three empirical detector-rate estimates use independent batches, the variance satisfies $\operatorname{Var}(\xi_k\mid w)\leq\tfrac{1}{4m}+\tfrac{1}{4m}+\tfrac{1}{m}=\tfrac{3}{2m}$.
We first bound the deterministic part. 
Since $DR_k$ depends only on $\mathcal S_k$, and since it is $\beta$-smooth on the trusted region,
\begin{align*}
|DR_k(\delta\vec\theta_t+\nu w)+DR_k(\delta\vec\theta_t-\nu w)-2DR_k(\delta\vec\theta_t)|\leq\beta\nu^2\|w_{\mathcal S_k}\|^2.
\end{align*}
Therefore,
\begin{align*}
\left|\frac{w_jw_\ell-\mathbf 1_{j=\ell}}{2\nu^2N_D}\sum_{k\in \mathcal K_{j\ell}}\left(DR_k(\delta\vec\theta_t+\nu w)+DR_k(\delta\vec\theta_t-\nu w)-2DR_k(\delta\vec\theta_t)\right)\right|^2\leq\frac{\beta^2}{4N_D^2}(w_jw_\ell-\mathbf 1_{j=\ell})^2\left(\sum_{k\in \mathcal K_{j\ell}}\|w_{\mathcal S_k}\|^2\right)^2.
\end{align*}
Using $|\mathcal K_{j\ell}|\leq c$, we have $\sum_{k\in \mathcal K_{j\ell}}\|w_{\mathcal S_k}\|^2)^2\leq c\sum_{k\in \mathcal K_{j\ell}}\|w_{\mathcal S_k}\|^4$, and thus,
\begin{align*}
\left|\frac{w_jw_\ell-\mathbf 1_{j=\ell}}{2\nu^2N_D}\sum_{k\in \mathcal K_{j\ell}}\left(DR_k(\delta\vec\theta_t+\nu w)+DR_k(\delta\vec\theta_t-\nu w)-2DR_k(\delta\vec\theta_t)\right)\right|^2\leq\frac{\beta^2c}{4N_D^2}(w_jw_\ell-\mathbf 1_{j=\ell})^2\sum_{k\in \mathcal K_{j\ell}}\|w_{\mathcal S_k}\|^4.
\end{align*}
Summing over all coordinate pairs $(j,\ell)$, we obtain
\begin{align*}
\sum_{j,\ell}(w_jw_\ell-\mathbf 1_{j=\ell})^2\sum_{k\in \mathcal K_{j\ell}}\|w_{\mathcal S_k}\|^4=\sum_{k=1}^{N_D}\|w_{\mathcal S_k}\|^4\sum_{j,\ell\in \mathcal S_k}(w_jw_\ell-\mathbf 1_{j=\ell})^2.
\end{align*}
For each $k$, we have
\begin{align*}
\sum_{j,\ell\in \mathcal S_k}(w_jw_\ell-\mathbf 1_{j=\ell})^2=\|w_{\mathcal S_k}w_{\mathcal S_k}^{\top}-I_{\mathcal S_k}\|_F^2=\|w_{\mathcal S_k}\|^4-2\|w_{\mathcal S_k}\|^2+|\mathcal S_k|.
\end{align*}
Since $|\mathcal S_k|\leq s$,
\begin{align*}
\mathbb E\left[\|w_{\mathcal S_k}\|^4\|w_{\mathcal S_k}w_{\mathcal S_k}^{\top}-I_{\mathcal S_k}\|_F^2\right]\leq s(s+2)(s+4)(s+6)+s^2(s+2).
\end{align*}
Therefore, the deterministic contribution to the Hessian second moment is bounded by
\begin{align*}
\frac{\beta^2c}{4N_D}\left(s(s+2)(s+4)(s+6)+s^2(s+2)\right).
\end{align*}
We now bound the sampling-noise contribution. 
For one coordinate pair,
\begin{align*}
\zeta_{j\ell}=\frac{w_jw_\ell-\mathbf 1_{j=\ell}}{2\nu^2N_D}\sum_{k\in \mathcal K_{j\ell}}\xi_k.
\end{align*}
Conditioned on $w$, independence across detectors gives $\operatorname{Var}(\sum_{k\in \mathcal K_{j\ell}}\xi_k\mid w)\leq\tfrac{3c}{2m}$, and thus
\begin{align*}
\operatorname{Var}(\zeta_{j\ell}\mid w)\leq\frac{(w_jw_\ell-\mathbf 1_{j=\ell})^2}{4\nu^4N_D^2}\frac{3c}{2m}\quad\Rightarrow\quad\mathbb E[\zeta_{j\ell}^2]\leq\frac{3c}{4m\nu^4N_D^2}..
\end{align*}
after taking expectation over $w$ and using $\mathbb E[(w_jw_\ell-\mathbf 1_{j=\ell})^2]\leq 2$ and obtain
Note that the number of coordinate pairs $(j,\ell)$ for which $\mathcal K_{j\ell}\neq \emptyset$ is at most $\sum_{k=1}^{N_D}|\mathcal S_k|^2\leq N_Ds^2$.
Thus the total sampling-noise contribution to the Frobenius second moment is bounded by
\begin{align*}
N_Ds^2\frac{3c}{4m\nu^4N_D^2}=\frac{3cs^2}{4m\nu^4N_D}.
\end{align*}
Combining the deterministic contribution and the sampling-noise contribution, and then dividing by $K_h$ because $\widehat H_t^{\mathrm{loc}}$ averages $K_h$ independent samples, gives
\begin{align*}
\mathbb E[\|\widehat H_t^{\mathrm{loc}}-\nabla^2 C_\nu(\delta\vec\theta_t)\|_F^2\mid\delta\vec\theta_t]\leq\frac{1}{K_h}\left(\frac{\beta^2c}{4N_D}\left(s(s+2)(s+4)(s+6)+s^2(s+2)\right)+\frac{3cs^2}{4m\nu^4N_D}\right),
\end{align*}
which proves the Hessian part.

It remains to explain how the lemma modifies the final nonconvex guarantee. 
\Cref{lem:one_step_nonconvex} is deterministic once the errors
\begin{align*}
\widehat g_t-\nabla C_\nu(\delta\vec\theta_t),\quad
\widehat H_t-\nabla^2 C_\nu(\delta\vec\theta_t)
\end{align*}
are bounded in second moment. 
Therefore, replacing $\widehat g_t$ and $\widehat H_t$ by $\widehat g_t^{\mathrm{loc}}$ and $\widehat H_t^{\mathrm{loc}}$ changes only the two estimator-error terms in \Cref{lem:C_nu_nonconvex}. 
The argument in \Cref{lem:C_nu_nonconvex} is thus unchanged because it uses only the one-step decrease inequality. 
The smoothing-bias argument in \Cref{lem:gaussian_bias} is also unchanged because it depends only on the smoothness of $C$, not on how the derivatives of $C_\nu$ are estimated. 
Finally, \Cref{coro:nonconvex_local} follows in the same way as \Cref{thm:nonconvex} once $T,K_g,K_h,m$ are chosen so that the modified bound on $\mathbb E[\Psi_\nu(\delta\vec\theta_R)]$ is at most $\min\{\tfrac{\varepsilon^2}{4},\tfrac{\beta}{6\sqrt{\rho}}\varepsilon^{3/2}\}$.
Hence the locality-aware version of Algorithm~\ref{alg:PAGD_nonconvex} returns an $\varepsilon$-SOSP of $C$ in expectation, with the detector-sampling prefactors controlled by $s,c$ rather than by the full ambient dimension $d$.
\end{proof}

\subsection{Additional discussion on the time-dependent case control drifts}\label{sec:nonconvex_online}
The time-dependent drift case is fundamentally harder.
In particular, we need to optimize to a local minima of $C_t(\delta\vec\theta_t)$ in \Cref{def:optimization_process} for nonconvex $C_t$ at any $t=1,...,T$.
To the best of our knowledge, how to find local minima efficiently of a sequence of target nonconvex functions $f_1,...,f_T$ in the online setting remains unknown, especially given only function value queries.
Based on a template known as ``follow the
regularized leader'' (FTRL)~\cite{shalev2006convex,hazan2016introduction}, there are algorithms achieving a $O(\sqrt{T\sum_{t=1}^{T-1}\norm{\vec \vec v_{t+1}-\vec \vec v_t}})$ global dynamic regret bound, which compute the accumulated difference between the objective function of $\{\vec\theta_t\}$ and the objective function of the optimal comparator $\{\vec \vec v_t^*\}$, similar to \Cref{thm:online} with noisy function value queries~\cite{heliou2020online,heliou2021zeroth}.
However, computing the next $\delta\vec\theta_{t+1}$ based on the noisy function queries using the FTRL template is usually computationally inefficient.
Some other algorithms also achieve an $o(T\sum_{t=1}^{T-1}\norm{\vec \vec v_{t+1}-\vec v_t})$ dynamic regret efficiently based on additional assumptions on the nonconvex objective functions, such as the weak pseudo-convex condition~\cite{gao2018online}.
However, it is unclear whether the weak pseudo-convex condition can fit into our surrogate objective function $C_t$.
Despite dynamic regret, another widely known metric is the local regret, which is the summation of $\norm{\nabla C_t(\delta\vec\theta_t)}$.
An efficient algorithm is known to achieve a $O(\sqrt{T\sum_{t=1}^{T-1}\norm{\vec \vec v_{t+1}-\vec \vec v_t}})$ local regret~\cite{guan2023hardness}.
Unfortunately, such an algorithm can only provide guarantees on first-order stationary points, which contain unwanted local maxima and saddle points other than local minima.
In conclusion, it remains open to provide any nontrivial guarantee, such as sublinear $o(T)$ bounds for dynamics regret or local regret, for finding even global (local) minima in each iteration for our surrogate model.

\paragraph{A special case: piecewise time-independent drifts.}There is, however, an important special case in which a much better statement is possible: piecewise time-independent drifts with only finitely many jumps.
Assume that the time-dependent objective is piecewise constant. 
Namely, there exist jump times $1=\tau_0<\tau_1<\cdots<\tau_J<\tau_{J+1}=T+1$ such that for each segment $t \in \{\tau_j,\tau_j+1,\dots,\tau_{j+1}-1\}$ and $C_t(\delta \vec\theta)=C^{(j)}(\delta \vec\theta)$,
where each $C^{(j)}$ satisfies \Cref{ass:lipschitz_nonconvex} on the same trusted region $P$. Suppose also that the jump times are known and that on each segment we restart Algorithm~\ref{alg:PAGD_nonconvex} from the previous endpoint.
Let $T_{\mathrm{off}}(\varepsilon)$ denote any offline iteration budget sufficient in \Cref{thm:nonconvex} to obtain an $\varepsilon$-SOSP for a fixed objective. 
If every segment length satisfies $\tau_{j+1}-\tau_j \ge T_{\mathrm{off}}(\varepsilon)$, then on every segment there exists an iterate $\delta \vec\theta^{(j)}_{\mathrm{out}}$ that is an $\varepsilon$-SOSP of the active objective $C^{(j)}$ in expectation.
Moreover, if we define the segment-wise second-order stationarity measure
\begin{align*}
\Psi_t(\delta \vec\theta):=\max\left\{\|\nabla C_t(\delta \vec\theta)\|_2^2,\frac{4\beta}{3\rho^2}\bigl(-\lambda_{\min}(\nabla^2C_t(\delta \vec\theta))\bigr)_+^3\right\},
\end{align*}
We have the following corollary based on \Cref{thm:nonconvex}.
\begin{corollary}[Piecewise constant drift with a constant number of jumps]
Based on the setting mentioned above, the average local regret obeys
\begin{align*}
\frac{1}{T}\sum_{t=1}^{T}\mathbb E[\Psi_t(\delta \vec\theta_t)]\le\varepsilon^2+\frac{(J+1)T_{\mathrm{off}}(\varepsilon)}{T}.
\end{align*}
\end{corollary}

\begin{proof}
On each segment, the objective is time independent, so \Cref{thm:nonconvex} applies directly after restart. 
Thus, after at most $T_{\mathrm{off}}(\varepsilon)$ iterations on segment $j$, we reach an $\varepsilon$-SOSP of $C^{(j)}$ in expectation. 
The only epochs that may fail to satisfy the target stationarity level are the transient epochs immediately after each restart. 
Their total number is at most $(J+1)T_{\mathrm{off}}(\varepsilon)$. 
On all remaining epochs, the expected stationarity measure is at most $O(\varepsilon^2)$. 
Dividing by $T$ gives the stated bound.
\end{proof}

\noindent This corollary shows that the fully adversarial linear-in-$T$ barrier is not the right picture for slowly changing or jump-sparse drifts. 
If the number of jumps $J$ is $O(1)$, then the transient cost of each restart is also $O(1)$, and the average local regret summing up $\Psi_t(\delta\vec\theta)$ vanishes as $T \to \infty$. 
In that sense, piecewise constant drift is much closer to a sequence of offline calibration problems than to a genuinely adversarial online one.
For our application, this is the regime most relevant to rare recalibration events, abrupt hardware interventions, or occasional changes of the background operating condition. 
In such cases, one can hope to combine the offline nonconvex guarantee of the present subsection with a restart logic, and thereby recover a meaningful control guarantee even though the general online nonconvex problem admits only linear worst-case cumulative regret.

\section{Exact circuit-level noise simulation of Rydberg atom arrays}
\label{app:appendix_5}

\subsection{Review of quantum control of Rydberg Hamiltonian}

In neutral atom arrays, we choose to encode the qubit information with two hyperfine states ($\ket{0}$ and $\ket{1}$) of the atom. These two states can be coupled by the laser, and single qubit gates can be implemented with high fidelity by tuning the Rabi frequency and the detuning of the laser. This implementation is relatively straightforward.

To mediate interactions between atoms and build up entanglement, we utilized the third level, the highly excited Rydberg state $\ket{r}$. Specifically, the atoms are controlled with a global Rabi frequency $\Omega(t)=|\Omega(t)|e^{i\varphi(t)}$ between $\ket{1}$ and the Rydberg state $\ket{r}$, and the Hamiltonian is 
\eqs{
H(t) = \dfrac{\Omega(t)}{2}\sum_j |1\rangle_{j} {}_j\langle r|+h.c.+\sum_{jk}B_{jk}|rr\rangle_{jk}{}_{jk}\langle rr|,
}
where $B_{jk}$ is the Rydberg interaction. The native entangling gate of the neutral atom arrays is the CPHASE gate or CZ gate, which is mediated by Rydberg states with global controls \cite{Levine_Pichler}.

Usually, for the implementation of the CZ gate, we assume the two atoms are in the strongly blocked region; therefore, $\ket{rr}$ is prohibited. To achieve the time-optimal controlled-phase gate, one chooses $|\Omega(t)|=\Omega_{\max}$ for the maximum value and optimizes the phase $\varphi(t)$ profile \cite{Jandura2022timeoptimaltwothree}. For the Rydberg CZ gate $U(T)\ket{q}=e^{i\xi_q}\ket{q}$ where $q\in \{00,01,10,11\}$, since $\xi_{00}=0$, we require $\xi_{11}-\xi_{01}-\xi_{10}=\pi$. This will realize a CZ gate up to some single qubit gates. The average gate fidelity can be calculated as 
\eqs{
F = \dfrac{1}{2^N(2^N+1)}\left(|\sum_q e^{-i\xi_q }\langle q|U(T)|q\rangle|^2+\sum_q |\langle q |U(T)|q\rangle|^2\right).
}

For global controls, the control Hamiltonian and relevant states can be further simplified using the symmetric subspace. Following Ref.~\cite{Jandura2022timeoptimaltwothree}, the CZ gate is implemented by a global laser
pulse $\Omega(t) = \Omega_{\max} e^{i\varphi(t)}$ coupling $|1\rangle
\leftrightarrow |r\rangle$ on both atoms. In the Rydberg blockade regime ($B =
\infty$), the two-atom Hamiltonian is block-diagonal. For a global pulse, the
relevant blocks in the $\{|01\rangle,|0r\rangle, |11\rangle,
|W\rangle\}$ basis (with $|W\rangle = (|1r\rangle+|r1\rangle)/\sqrt{2}$) are
\begin{equation}
  H = \begin{pmatrix}
    0 & \frac{\Omega}{2} & 0 & 0 \\
    \frac{\Omega^*}{2} & 0 & 0 & 0 \\
    0 & 0 & 0 & \frac{\sqrt{2}\,\Omega}{2} \\
    0 & 0 & \frac{\sqrt{2}\,\Omega^*}{2} & 0
  \end{pmatrix}.
\end{equation}
The time-optimal phase profile $\varphi(t)$ is found by GRAPE optimization,
minimizing the averaged gate infidelity
\begin{equation}
  1 - F = 1 - \frac{1}{20}\!\left(|1 + 2a_{01} + a_{11}|^2
  + 1 + 2|a_{01}|^2 + |a_{11}|^2\right),
\end{equation}
with $a_q = e^{-i\xi_q}\langle q|\psi(T)\rangle$, over the piecewise-constant
phase $\varphi_k$ and single-qubit phase $\theta$. The CZ condition requires
$\xi_{11} - \xi_{01} - \xi_{10} = \pi$. After the pulse, single-qubit $R_z(-\theta)$
rotations on each atom yield the standard CZ gate.

For the qutrit simulation, we construct the full $9\times 9$ unitary $U_{\mathrm{CZ}}$
in the two-qutrit basis $\{|00\rangle, |01\rangle, |0r\rangle, |10\rangle, |11\rangle,
|1r\rangle, |r0\rangle, |r1\rangle, |rr\rangle\}$ by propagating each
independent block analytically:
\begin{align}
\begin{split}
  &\{|01\rangle,|0r\rangle\}: \quad \alpha = \Omega_{\max}\,\delta t/2, \\
  &\{|10\rangle,|r0\rangle\}: \quad \text{identical to above}, \\
  &\{|11\rangle,|W\rangle\}: \quad \beta = \sqrt{2}\,\Omega_{\max}\,\delta t/2,
\end{split}
\end{align}
with the dark state $|D\rangle = (|1r\rangle - |r1\rangle)/\sqrt{2}$ decoupled.
The states $|00\rangle$ and $|rr\rangle$ are trivially unchanged. This $9\times 9$
matrix is exactly unitary for any real-valued $\varphi(t)$, ensuring $\|\psi\|^2 = 1$
throughout the simulation, regardless of pulse imperfections. Using gradient descent with smoothness constraints, one obtains the ideal laser control pulse $\varphi(t)$ and $|\Omega(t)|$, as shown in \Cref{fig:phi_t} as the solid lines, with total gate fidelity $99.99\%$.

\begin{figure}
    \centering
    \includegraphics[width=0.8\linewidth]{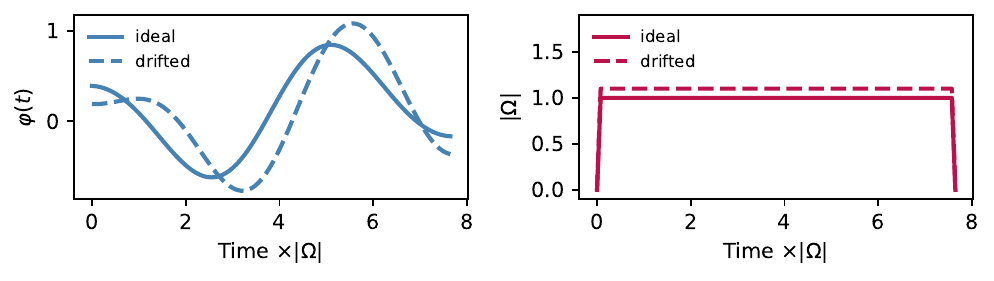}
    \caption{Native controlled-$Z$ gates are implemented by optimized laser amplitude $|\Omega|$ and phase $\varphi(t)$ shown as solid lines and dashed curves illustrate representative typical control drifts.}
    \label{fig:phi_t}
\end{figure}

\subsection{Pulse-level error model}

In the simulation, we consider that the two qubit entangling gate error is the main source of error and ignore the single qubit gate errors~\cite{995}. As shown before, there are two pulse control parameters: the laser phase $\varphi(t)$ and the laser amplitude $|\Omega|$. As shown in \Cref{fig:phi_t}, $\varphi(t)$ is a smooth function of time and only low frequency Fourier modes contribute. In our simulation, $\varphi(t)$ is discretized into 100 steps, and only the first two Fourier modes have significant non-zero contributions. 

Two physically motivated error models are considered:

\paragraph{Pulse distortion.} The phase profile is decomposed into Fourier modes,
$\varphi(t) = \sum_k c_k\, e^{i 2\pi f_k t}$, and individual coefficients $c_k$ are
perturbed (maintaining conjugate symmetry to keep $\varphi$ real). Each of the 8 CZ
gates in the syndrome circuit can receive independent perturbations, parameterized by
an $8 \times K$ matrix of Fourier coefficient shifts at $K=2$.

\paragraph{Coherent overrotation.} A unitary ZZ error $\exp(-i\theta\, Z\otimes Z/2)$
is applied after each CZ gate, where $Z_{\mathrm{qutrit}} = \mathrm{diag}(1,-1,1)$.
Each gate receives an independent error angle $\theta_i$.

\subsection{Exact simulation of \texorpdfstring{$[\![4,1,2]\!]$}{[4,1,2]} quantum code with Rydberg atom arrays and error model}
\begin{figure}[htbp]
    \centering
    \includegraphics[width=0.95\linewidth]{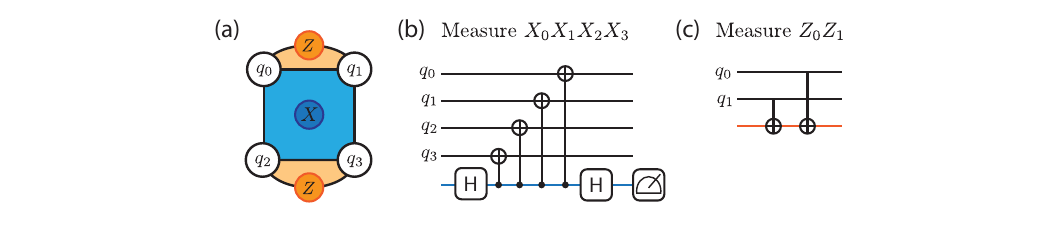}
    \caption{Small error detection code and syndrome extraction circuit.}
    \label{fig:small_code}
\end{figure}

\begin{figure}
    \centering
    \includegraphics[width=0.8\linewidth]{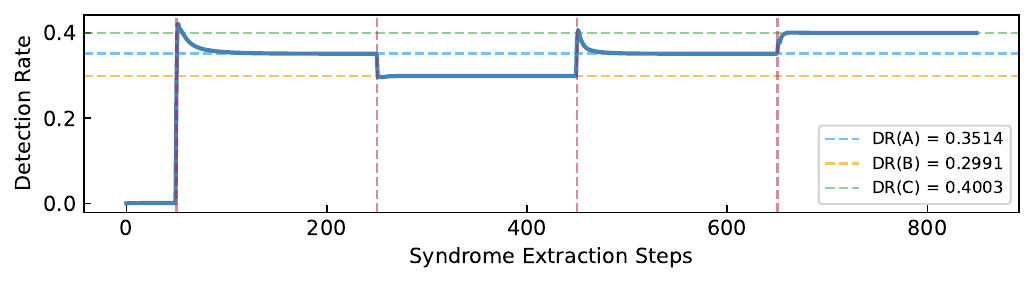}
    \caption{The steady-state detection rate is a state function of the control parameters.
Per-round detection rate of the $[\![4,1,2]\!]$ code, with the $d=32$ pulse parameters held fixed within each period and switched abruptly between three random miscalibrations along the schedule $A \to B \to A \to C$ ($200$ rounds each, preceded by $50$ drift-free rounds; vertical dashed lines mark the switches).
After a transient of a few tens of rounds the rate settles onto a plateau (horizontal dashed lines) set only by the current control point: the two $A$ periods agree to $4.8\times10^{-5}$ despite the excursion through $B$, three orders of magnitude below the $5.2\times10^{-2}$ and $1.0\times10^{-1}$ separations between distinct configurations.
}
    \label{fig:placeholder2}
\end{figure}

In this section, we provide the details in simulating a simple quantum error detection code with a logical quantum memory experiment involving pulse-level noise. The code is the $[\![4,1,2]\!]$ quantum detection code with four physical qubits and three ancillary qubits. The stabilizers of the code are
\eqs{
& S_1 = X_0 X_1 X_2 X_3,\\
& S_2 = Z_0 Z_1,\\
& S_3 = Z_2 Z_3,
}
and the logical operators are
\eqs{
& L_Z = Z_0 Z_2 \\
& L_X = X_0 X_1.
}
This code is the smallest rotated surface code, which is visualized in \Cref{fig:small_code} (a). The syndrome measurement circuits are plotted in \Cref{fig:small_code} (b) and (c). The physical device we focused on is the Rydberg atom array, where we encode qubits $\ket{0}$ and $\ket{1}$ as two hyperfine states. The single qubit gates are implemented through the Rabi frequency coupling those two hyperfine states. For simplicity, we ignore the control error of single qubit gates.

 The quantum state of the logical Z memory is $\ket{0_L}=(\ket{0000}+\ket{1111})/2$. In each round of syndrome detection, we introduce three ancillary qubits to check the parities. We can compile one round of syndrome detection circuit to native CZ gates and single qubit gates, as shown in \Cref{fig:small_circuit}.

\begin{figure}[htbp]
    \centering
    \includegraphics[width=1\linewidth]{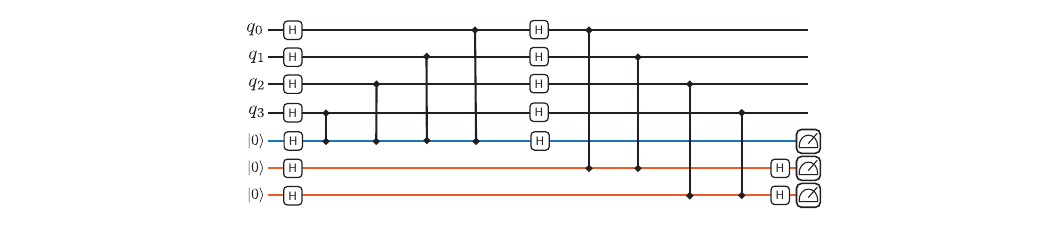}
    \caption{The quantum circuit of quantum memory experiment.}
    \label{fig:small_circuit}
\end{figure}




\subsection{Pauli twirling (randomized compiling)}

For each CZ gate, a random two-qubit Pauli $P = \sigma_a \otimes \sigma_b$ is drawn
uniformly from the 16-element set $\{I,X,Y,Z\}^{\otimes 2}$ and applied before the
gate. The correction $Q = \mathrm{CZ}\cdot P \cdot \mathrm{CZ}^\dagger$ (which is
also a Pauli up to phase since CZ is Clifford) is applied after. For an ideal CZ,
$Q^\dagger \cdot \mathrm{CZ} \cdot P = \mathrm{CZ}$, so twirling is transparent.
For a noisy gate $\tilde{U} = E \cdot \mathrm{CZ}$, the twirled channel averaged
over all Pauli choices becomes a Pauli channel, converting coherent errors into
stochastic noise. In the qutrit simulation, Pauli gates act as $\sigma_{\mathrm{qutrit}}
= \sigma_{2\times 2} \oplus (1)$ on $|r\rangle$.

\subsection{Syndrome extraction circuit}

The $[\![4,1,2]\!]$ code has stabilizers $S_1 = X_0X_1X_2X_3$, $S_2 = Z_0Z_1$,
$S_3 = Z_2Z_3$, measured using ancilla qutrits $a_0, a_1, a_2$ respectively. The
circuit consists of three phases:
\begin{enumerate}
  \item \textbf{Measure $S_1$:} Hadamard on data qubits $q_0$--$q_3$ (basis change
    $X\to Z$), Hadamard on $a_0$, four CZ gates $(q_i, a_0)$, Hadamard on $a_0$,
    Hadamard on data qubits.
  \item \textbf{Measure $S_2$:} Hadamard on $a_1$, CZ$(q_0, a_1)$, CZ$(q_1, a_1)$,
    Hadamard on $a_1$.
  \item \textbf{Measure $S_3$:} Hadamard on $a_2$, CZ$(q_2, a_2)$, CZ$(q_3, a_2)$,
    Hadamard on $a_2$.
\end{enumerate}
This yields 8 CZ gates total. Ancilla measurement outcomes
$\vec{m} = (m_X, m_{Z_1}, m_{Z_2}) \in \{0,1,2\}^3$ include leakage detection ($m=2$).

\subsection{Multi-round density matrix simulation}

For multi-round syndrome extraction, we track 27 branches of the 4-data-qutrit
density matrix ($81\times 81$), one per ancilla measurement outcome. The key
computational tool is the \emph{quantum instrument} formalism using Kraus operators.

\paragraph{Kraus operator construction.}
For a given set of (possibly noisy and twirled) CZ gates, we compute 27 Kraus
operators $K_{\vec{m}} \in \mathbb{C}^{81\times 81}$ defined by
\begin{equation}
  (K_{\vec{m}})_{d',d} = \langle d', \vec{m}\,|\, U_{\mathrm{circuit}}\,
  |d, 000\rangle,
\end{equation}
where $d, d' \in \{0,1,r\}^4$ index data qutrit basis states and $|000\rangle$ is
the ancilla reset state. These are obtained by evolving 81 basis states through the
full 2187-dimensional circuit using efficient tensor contractions (cost: $O(81 \times
2187)$ per gate, $\sim$160\,ms total).

\paragraph{Round update.}
Given the data density matrix $\rho$ entering a round, the measurement probability
and post-measurement state for outcome $\vec{m}$ are
\begin{equation}
  p(\vec{m}\,|\,\rho) = \mathrm{Tr}(K_{\vec{m}}\,\rho\,K_{\vec{m}}^\dagger),
  \qquad
  \sigma_{\vec{m}} = \frac{K_{\vec{m}}\,\rho\,K_{\vec{m}}^\dagger}{p(\vec{m}\,|\,\rho)}.
\end{equation}
For twirling, the Kraus operators are averaged over $N_{\mathrm{twirl}}$ random
instances: $\mathcal{E}_{\vec{m}}(\rho) = \frac{1}{N_{\mathrm{twirl}}}
\sum_{j=1}^{N_{\mathrm{twirl}}} K_{\vec{m}}^{(j)}\,\rho\, K_{\vec{m}}^{(j)\dagger}$.

\paragraph{Branch merging.}
After each round, branches with the same current outcome $\vec{m}^{(t)}$ but
different histories are merged:
\begin{equation}
  \tilde{p}_{\vec{m}^{(t)}} = \sum_{\vec{m}^{(t-1)}} p_{\vec{m}^{(t-1)}}
  \cdot p(\vec{m}^{(t)} | \sigma_{\vec{m}^{(t-1)}}), \qquad
  \tilde{\sigma}_{\vec{m}^{(t)}} = \frac{1}{\tilde{p}_{\vec{m}^{(t)}}}
  \sum_{\vec{m}^{(t-1)}} p_{\vec{m}^{(t-1)}} \cdot
  K_{\vec{m}^{(t)}} \sigma_{\vec{m}^{(t-1)}} K_{\vec{m}^{(t)}}^\dagger.
\end{equation}
This maintains exactly 27 branches per round regardless of the number of rounds.

\paragraph{Detection event rates.}
A detection event on detector $k \in \{X, Z_1, Z_2\}$ at round $t$ fires when the
binarized outcome changes: $\bar{m}_k^{(t)} \neq \bar{m}_k^{(t-1)}$, where
$\bar{m} = 0$ if $m = 0$ and $\bar{m} = 1$ if $m \in \{1, 2\}$. The expected
detection rate is
\begin{equation}
  DR_k^{(t)} = \sum_{\vec{m}^{(t-1)}}
  \sum_{\substack{\vec{m}^{(t)}:\\[1pt] \bar{m}_k^{(t)} \neq \bar{m}_k^{(t-1)}}}
  \tilde{p}_{\vec{m}^{(t-1)}} \cdot p(\vec{m}^{(t)} | \tilde{\sigma}_{\vec{m}^{(t-1)}}).
\end{equation}
This is computed exactly (no sampling) from the $27 \times 27$ matrix of transition
probabilities at each round, at a cost of 27 density matrix evolutions per round.

\begin{figure}
    \centering
    \includegraphics[width=1.0\linewidth]{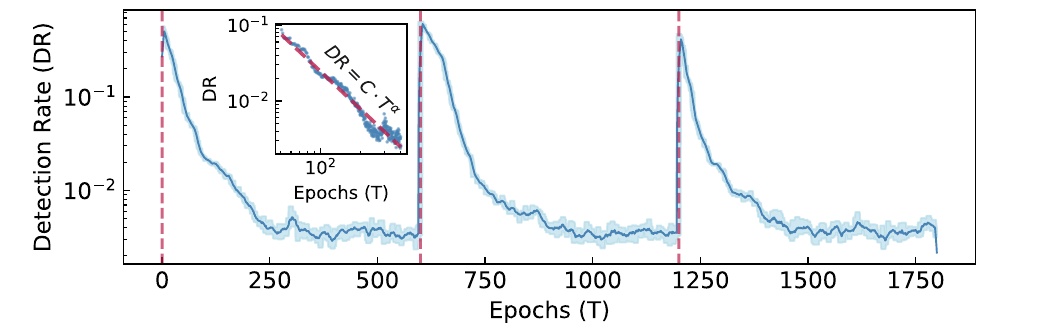}
    \caption{Adaptive $\eta=0.005/(1+T/T_{\mathrm{tot}})$}
    \label{fig:placeholder3}
\end{figure}

\begin{figure}
    \centering
    \includegraphics[width=1.0\linewidth]{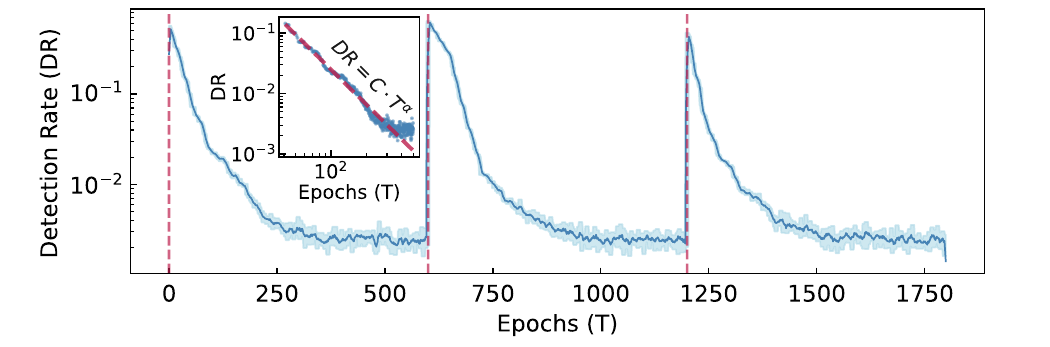}
    \caption{constant learning rate spsa with $\eta=0.005$.}
    \label{fig:placeholder4}
\end{figure}

\begin{figure}
    \centering
    \includegraphics[width=1.0\linewidth]{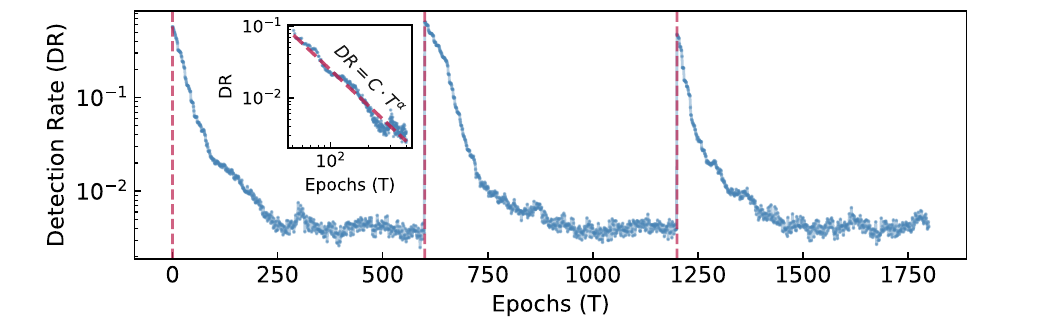}
    \caption{AGD with $\eta=0.005$ and $\theta=0.95$ for momentum parameter. The $\alpha$ are $-1.57$, $-1.9$ and $-1.23$ for the three random constant drifts.}
    \label{fig:placeholder5}
\end{figure}

\begin{figure}
    \centering
    \includegraphics[width=0.8\linewidth]{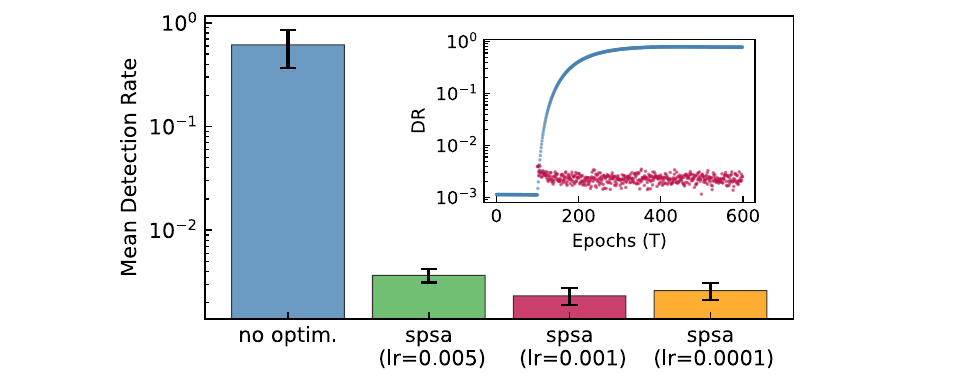}
    \caption{constant learning rate SPSA with slow linear drifts. Each control drift has the form $p^*=kT$, with $k$ uniformly sampled from [0,0.001].}
    \label{fig:placeholder6}
\end{figure}

\section{Details of the Clifford-Level Numerical Experiments}
\label{app:clifford_numerics}

This appendix collects the implementation details of the circuit-level (Clifford) simulations reported in the main text: the rotated surface code and the bivariate bicycle (BB) LDPC codes.
Both families are simulated with the \emph{same} testbed --- the same per-gate control-parameter model, the same drift protocol, the same surrogate objective, and the same optimizer --- so that the only difference between the two sets of runs is the code (and hence the Tanner-graph geometry and the detector--gate incidence structure).
All circuits are constructed and sampled with Stim~\cite{gidney2021stim}.

\subsection{Circuits}
\label{app:num_circuits}

\paragraph{Rotated surface code.}
We simulate a $Z$-memory experiment on the rotated surface code of odd distance $d \in \{3,5,7,9,11,13\}$, generated with Stim's \texttt{surface\_code:rotated\_memory\_z} template with $N_c = 2d$ syndrome rounds followed by a transversal data-qubit measurement.
The generated circuit already carries the standard circuit-level noise model~\cite{fowler2009high}: a \texttt{DEPOLARIZE1} channel after every single-qubit Clifford and on every idle data qubit, a \texttt{DEPOLARIZE2} channel after every two-qubit Clifford, and \texttt{X\_ERROR} flip channels on reset and measurement.
We then parse the repeat body and replace the \emph{uniform} depolarization probabilities by individually addressable per-slot rates, keeping the reset/measurement flip probabilities fixed at $p_{\rm base} = 10^{-3}$; these channels model readout imperfections that the control knobs considered here cannot influence.
Writing $n_{1q}$ and $n_{2q}$ for the number of one- and two-qubit depolarization slots in one syndrome round, the number of independently controllable noise locations is $n_{\rm gate} = n_{1q} + n_{2q}$, and the number of detectors is $N_D = (d^2-1) N_c$.

\paragraph{Bivariate bicycle codes.}
For the BB family, we build the syndrome-extraction circuit directly from the code definition in Ref.~\cite{bravyi2024high}, following the depth-7 CNOT schedule of Ref.~\cite{bravyi2024high}: $X$-check ancilla are prepared in $\ket{+}$ and $Z$-check ancilla in $\ket{0}$, and within a cycle, the $X$- and $Z$-check schedules are interleaved as
\begin{equation}
  \mathcal{S}_X = (\text{idle},\,B_2,\,A_2,\,A_1,\,A_3,\,B_1,\,B_3), \qquad
  \mathcal{S}_Z = (A_1^{\mathsf T},\,A_3^{\mathsf T},\,B_1^{\mathsf T},\,B_2^{\mathsf T},\,B_3^{\mathsf T},\,A_2^{\mathsf T},\,\text{idle}),
\end{equation}
so that every layer is a perfect matching on the Tanner graph.
Noise is inserted with exactly the same philosophy as above: a \texttt{DEPOLARIZE2} channel after each of the $12 \cdot (n/2)$ CNOTs per cycle, a \texttt{DEPOLARIZE1} channel on each of the $4 \cdot (n/2)$ idle slots per cycle, and fixed reset/measurement flip channels at $p_{\rm base} = 10^{-3}$.
Detectors compare consecutive $X$- and $Z$-check outcomes ($Z$-checks only in the first cycle, since the data register starts in $\ket{0}^{\otimes n}$), giving $N_D = n N_c$ detectors for $N_c$ cycles; we take $N_c = d$.
We simulate the $[\![72,12,6]\!]$, $[\![90,8,10]\!]$ and $[\![108,8,10]\!]$ instances.

Because each evaluation of the objective requires rebuilding the circuit with a fresh set of $n_{\rm gate}$ depolarization probabilities, we precompile the circuit once into a text template with one hole per noise slot; an evaluation is then a string substitution plus Stim's C\raisebox{0.2ex}{\tiny ++} parser ($\sim 7$\,ms for $[\![72,12,6]\!]$) rather than an instruction-by-instruction Python rebuild ($\sim 165$\,ms). The template is verified against the instruction-level rebuild by comparing detector error models at construction time.

\subsection{Control parameters, miscalibration, and drift}
\label{app:num_error_model}

Each noise slot $i \in \{1,\dots,n_{\rm gate}\}$ carries $P$ local control parameters (we use $P = 5$ throughout), collected in $p_i \in \mathbb{R}^P$, and the corresponding hardware setpoint drifts to $p_i^{\rm drift}$.
The physical error rate of that gate is the surrogate quadratic model of \Cref{sec:surrogate_model_convexity},
\begin{equation}
  \varepsilon_i \;=\; \varepsilon_i^{\rm base} \;+\; \delta p_i^{\mathsf T}\, \Omega_i \, \delta p_i,
  \qquad \delta p_i \;=\; p_i^{\rm drift} - p_i,
  \label{eq:num_error_model}
\end{equation}
with $\varepsilon_i^{\rm base} = 10^{-3}$ the irreducible error floor and $\Omega_i \succeq 0$ a $P \times P$ sensitivity matrix.
The $\Omega_i$ are generated once per run (seed $42$) as block-diagonal matrices with four blocks plus cross-block coupling of relative strength $0.5$ and overall scale $0.015$, then symmetrized and projected onto the PSD cone; the cross-block terms are what makes the landscape genuinely multi-parameter rather than a collection of decoupled scalars.
Control parameters live in the box $[-1,1]^{P}$, and the total dimension of the optimization problem is $n_{\rm gate} \cdot P$ (e.g.\ $4805$ for $d=13$, $2880$ for $[\![72,12,6]\!]$).

The drift protocol is a step (``calibration event'') profile: the run begins with $30$ evaluations at $\delta p = 0$ to fix the drift-free reference, after which every parameter jumps simultaneously to
\begin{equation}
  p_{i,j}^{\rm drift} \;=\; A \, s_{ij} \, u_{ij}, \qquad s_{ij} \sim \mathrm{Unif}\{-1,+1\}, \quad u_{ij} \sim \mathrm{Unif}[0.5,1],
\end{equation}
with amplitude $A = 0.35$, and is then held fixed while the optimizer runs.
This is the harshest setting for a tracking algorithm --- the entire parameter vector is displaced at once --- and it isolates the transient recovery behavior.
The code also supports smooth sinusoidal drift with per-parameter random frequency, amplitude and phase, and mixed profiles with $L>1$ successive drift events; these were used for the consistency checks quoted in the main text.

\subsection{Surrogate objective and locality-aware estimator}
\label{app:num_objective}

The optimizer never sees the logical state. Its objective is the mean detection rate
\begin{equation}
  DR(p) \;=\; \frac{1}{N_D}\sum_{k=1}^{N_D} DR_k(p), \qquad
  DR_k(p) \;=\; \Pr[\text{detector } k \text{ fires}],
\end{equation}
estimated from $m$ Monte Carlo shots of the compiled detector sampler. Its statistical error scales as $1/\sqrt{m N_D}$ independently of $d$, which is precisely why it is usable as an online signal at large code distance while the logical error rate is not.

Rather than perturbing all $n_{\rm gate} P$ coordinates simultaneously (vanilla SPSA~\cite{spall1992spsa}), we exploit the locality established in \Cref{sec:imp_local}:
\begin{enumerate}
  \item \textbf{Incidence probing.} For each noise slot $i$ we build a probe circuit with \emph{only} that depolarization channel active (flip channels off) and read its detector error model, obtaining the set $K_i$ of detectors that slot $i$ can influence. This uses the same slot indexing as the evaluator, so $K_i$ is exact rather than inferred from the code geometry.
  \item \textbf{Conflict-graph coloring.} Two slots conflict if they share a detector. Greedy coloring of the conflict graph (best of largest-first, smallest-last, and DSATUR) yields $\chi_{\rm gate}$ classes with the property that no detector sees two members of the same class; this property is asserted at runtime. A coordinate class is a pair (gate class, parameter index), so there are $\chi = \chi_{\rm gate} P$ classes.
  \item \textbf{Masked gradient.} For each class, fresh Rademacher signs $u$ are drawn and paired probes at $p \pm \lambda u$ are evaluated with per-detector resolution. The estimator for coordinate $j = (i,q)$ is
  \begin{equation}
    g_j \;=\; \frac{u_j}{2 \lambda N_D} \sum_{k \in K_i}\big[\widehat{DR}_k(p + \lambda u) - \widehat{DR}_k(p - \lambda u)\big],
  \end{equation}
  i.e.\ only the detectors that the coordinate actually moves contribute, which removes the $O(N_D)$ variance inflation of the global estimator.
  \item \textbf{Update.} The estimate is fed to Adam ($\beta_1 = 0.9$, $\beta_2 = 0.999$, learning rate $0.01$) and the result is projected back onto the shrunk box $[-1+\lambda,\, 1-\lambda]^P$ so that both probes stay feasible.
\end{enumerate}
One iteration costs $2\chi + 1$ sampler evaluations ($2\chi$ probes plus one evaluation of the current policy for monitoring).
We use perturbation radius $\lambda = 0.05$, $2000$ iterations, $2000$--$20{,}000$ shots per probe, and $20{,}000$ shots for the monitoring evaluations.
Every run is compared against two reference policies evaluated on the same drift realization: the \emph{fixed} policy $p \equiv 0$ (no recalibration) and the \emph{oracle} policy $p \equiv p^{\rm drift}$ (perfect knowledge of the drift, $\delta p = 0$), which lower-bounds the achievable detection rate at $\varepsilon^{\rm base}$.

\subsection{Reinforcement-learning baseline}
\label{app:num_rl}

The reinforcement-learning (RL) controller against which we benchmark in the scaling comparison of the main text is \emph{our implementation} of the masked policy-gradient algorithm introduced in Ref.~\cite{sivak2025reinforcement}, and the data only reflects the performance of our implementation.
We wrote it according to Algorithm~1 in its single-step, no-replay configuration.
It runs on exactly the same testbed as the locality-aware SPSA (LSPSA) of \Cref{app:num_objective}

\paragraph{Policy and gradient estimator.}
The controller does not learn a state-to-action map; the policy is a factorized Gaussian over the control vector itself,
\begin{equation}
  \pi(p \mid \mu, \sigma) \;=\; \prod_{i,q} \mathcal{N}\!\left(p_{iq} \,\middle|\, \mu_{iq},\, \sigma_{iq}^2\right),
\end{equation}
with one $(\mu_{iq}, \ln \sigma_{iq})$ pair per coordinate, so the learnable dimension is $2 n_{\rm gate} P$ and the deployed policy is $\mu$.
In each epoch a batch of $B$ candidates is sampled and projected onto the control box, each candidate is evaluated with a single detector-sampler call returning the \emph{per-detector} rates $o^{(k)} \in \mathbb{R}^{N_D}$, and the reward vector is $r^{(k)} = -o^{(k)}$.
Advantages are taken against a learned baseline vector $b \in \mathbb{R}^{N_D}$, and the key scalability ingredient of Ref.~\cite{sivak2025reinforcement}, i.e. gradient masking, restricts the advantage seen by a coordinate to the detectors that its gate can actually influence,
\begin{equation}
  A_i^{(k)} \;=\; \frac{1}{N_D}\sum_{j \in K_i} \alpha_j^{(k)},
  \qquad \alpha^{(k)} = r^{(k)} - b .
  \label{eq:rl_masked_advantage}
\end{equation}
This uses the same incidence sets $K_i$ as the LSPSA estimator of \Cref{app:num_objective}, so the two algorithms are given identical locality information and the comparison isolates the update rule rather than the amount of structure supplied to each.
Because the algorithm collects on-policy samples and takes a single optimizer step per epoch, all importance ratios equal one at the point of differentiation. Therefore, PPO clipping in Eq.~(18) of Ref.\cite{sivak2025reinforcement} is inactive, and the policy gradient reduces exactly to the masked score-function form used below.
The normalization $1/N_D$ in \Cref{eq:rl_masked_advantage} matches the corresponding factor in the LSPSA gradient so that a single set of hyperparameters is meaningful across code distances.

\begin{algorithm}[htbp]
\caption{Masked policy-gradient (RL) steering, as implemented for the baseline reported in the main text}
\label{alg:rl_baseline}
\DontPrintSemicolon
\SetKwInput{KwInput}{Input}
\SetKwInput{KwOutput}{Output}

\KwInput{Epoch budget $T$; batch size $B$; shots per candidate $m$; learning rates $\eta_\mu, \eta_\sigma, \eta_b$; entropy weight $\beta_{\rm ent}$; exploration bounds $[\sigma_{\min}, \sigma_{\max}]$; incidence sets $\{K_i\}_{i=1}^{n_{\rm gate}}$; control box $\mathsf{P} = [p_{\rm lo}, p_{\rm hi}]$.}
\KwOutput{Deployed policy $\mu_T$.}

\BlankLine
Initialize $\mu_1 \leftarrow 0$, $\sigma_1 \leftarrow \sigma_{\rm init}$, baseline $b$ uninitialized.\;
\For{$t=1,2,\dots,T$}{
  Sample $\epsilon^{(k)} \sim \mathcal N(0, I)$ and set $p^{(k)} = \Pi_{\mathsf P}\!\left(\mu_t + \sigma_t \odot \epsilon^{(k)}\right)$ for $k=1,\dots,B$.\;
  Evaluate each candidate with one sampler call of $m$ shots, recording per-detector rates $o^{(k)} \in \mathbb{R}^{N_D}$; set $r^{(k)} = -o^{(k)}$.\;
  \lIf{$b$ uninitialized}{$b \leftarrow \frac{1}{B}\sum_k r^{(k)}$}
  Advantages $\alpha^{(k)} = r^{(k)} - b$; masked per-gate advantages $A_i^{(k)} = \frac{1}{N_D}\sum_{j \in K_i} \alpha_j^{(k)}$.\;
  Standardized deviations $z^{(k)} = (p^{(k)} - \mu_t)/\sigma_t$.\;
  Masked score-function gradients\;
  \qquad $g_\mu = -\frac{1}{B}\sum_k A^{(k)} \odot z^{(k)} / \sigma_t$, \quad
         $g_{\ln\sigma} = -\frac{1}{B}\sum_k A^{(k)} \odot \left(z^{(k)2} - 1\right) - \beta_{\rm ent}$.\;
  Single Adam step: $\mu_{t+1} = \Pi_{\mathsf P}\!\left(\mu_t + \mathrm{Adam}_\mu(g_\mu)\right)$, \;
  \qquad $\ln\sigma_{t+1} = \mathrm{clip}\!\left(\ln\sigma_t + \mathrm{Adam}_\sigma(g_{\ln\sigma}),\ \ln\sigma_{\min},\ \ln\sigma_{\max}\right)$.\;
  Baseline update $b \leftarrow b + 2\eta_b\left(\frac{1}{B}\sum_k r^{(k)} - b\right)$.\;
}
\Return{$\mu_T$.}
\end{algorithm}

Here $\mathrm{Adam}$ returns the descent update $-\eta\,\hat m/(\sqrt{\hat v}+\epsilon)$ with $\beta_1 = 0.9$, $\beta_2 = 0.999$, so passing $g_\mu = -\langle A z/\sigma\rangle$ performs ascent on $\mathbb{E}[A \log \pi]$, i.e. reward maximization.  The $-\beta_{\rm ent}$ term in $g_{\ln\sigma}$ is the derivative of the Gaussian entropy and acts as a constant outward pressure on the exploration width.

\paragraph{Hyperparameters and shot accounting.}
The hyperparameters of Ref.~\cite{sivak2025reinforcement} are not published (that work reports a shallow manual optimization), so ours were fixed by our own scan and then held constant across all code distances: $B = 40$, $\sigma_{\rm init} = 0.05$ (equal to the LSPSA perturbation radius $\lambda$), $\sigma \in [10^{-3}, 0.3]$, $\eta_\mu = 0.01$ (equal to the LSPSA Adam learning rate), $\eta_\sigma = 5\times10^{-3}$, $\beta_{\rm ent} = 3\times10^{-4}$, $\eta_b = 0.1$, $m = 2\times10^{4}$ shots per candidate, and $T = 2000$ epochs after the same $30$-evaluation drift-free reference.
One RL epoch therefore consumes $B m = 8\times10^{5}$ sampler shots, against $(2\chi+1)\,m_{\rm probe}$ for one LSPSA iteration, i.e.\ $\approx 4.8\times10^{5}$ shots at $m_{\rm probe} = 2000$ and $\approx 4.8\times10^{6}$ at $m_{\rm probe} = 2\times10^{4}$.

\Cref{tab:rl_baseline} reports the outcome. The RL controller tracks the drift well at every distance, reducing the detection rate by $77$--$79\%$ relative to the uncorrected policy.

\begin{table}[t]
\centering
\begin{tabular}{lrrrrr}
\hline
$d$ & $N_D$ & $n_{\rm gate}P$ & $DR_{\rm fixed}$ & $DR_{\rm RL}$ & $DR_{\rm oracle}$ \\
\hline
$3$  &   48 &  205 & 0.0624 & 0.0145 & 0.0130 \\
$5$  &  240 &  645 & 0.0863 & 0.0182 & 0.0152 \\
$7$  &  672 & 1325 & 0.0922 & 0.0197 & 0.0163 \\
$9$  & 1440 & 2245 & 0.0930 & 0.0204 & 0.0169 \\
$11$ & 2640 & 3405 & 0.0988 & 0.0210 & 0.0173 \\
$13$ & 4368 & 4805 & 0.0988 & 0.0213 & 0.0176 \\
\hline
\end{tabular}
\caption{
Reinforcement-learning baseline on the rotated surface code, under the step-drift protocol of \Cref{app:num_error_model} and the hyperparameters listed in the text. $DR_{\rm fixed}$ and $DR_{\rm oracle}$ are the uncorrected and drift-omniscient references evaluated on the same drift realization. Each run uses $2000$ epochs at $8\times10^{5}$ sampler shots per epoch.
\label{tab:rl_baseline}
}
\end{table}

\paragraph{Post-convergence diffusion of the deployed policy.}
After the deployed policy $\mu_t$ first reaches its optimum, its detection rate rises slowly and monotonically before saturating; in a $2\times10^{4}$-epoch run at $d=5$ it plateaus around epoch $7000$ at an excess of $4.7\times10^{-3}$ over the oracle ($+30\%$), with $\|\mu_t - p^{\rm drift}\|$ saturating at the same time.
This is exploration-driven diffusion, and it is structural rather than an artifact of a particular parameter choice.
The noise source is perturbation crosstalk: each detector rate responds to the $\sigma$-sized perturbations of \emph{all} gates in its neighborhood, so for any single coordinate the others contribute noise of variance $\propto \sigma^2$ (shot noise is subdominant at $m = 2\times10^{4}$).
Near the optimum the candidates no longer differ systematically in performance, the score-function signal-to-noise ratio falls to zero, and Adam's per-coordinate normalization --- which cannot distinguish signal from noise --- keeps taking steps of full size $\eta_\mu$. The deployed policy therefore executes a random walk with drift rate $\tfrac{1}{2}\eta_\mu^2 \sum_j h_j$, $h_j$ the curvature eigenvalues, until the walk equilibrates against the weak restoring force and the control box.

We verified each link of this chain by ablation at $d=5$.
Sweeping $\beta_{\rm ent}$ over $\{10^{-3}, 10^{-2}, 10^{-1}\}$ leaves the trajectory unchanged, so the \emph{value} of the entropy weight is irrelevant.
Setting $\beta_{\rm ent} = 0$ slows the drift by $17\times$ and freezing $\sigma$ at $\sigma_{\rm init}$ slows it by $25\times$, both consistent with the $\sigma^2$ scaling expected from $(\sigma_{\max}/\sigma_{\rm init})^2 = 36$ --- but in neither case does the drift disappear: it is present for any $\sigma > 0$.
Varying $\eta_\mu \in \{0.005, 0.01, 0.02\}$ gives rate ratios $0.30 : 1 : 3.39$ against the predicted $0.25 : 1 : 4$.
Finally, the measured rate matches $\tfrac{1}{2}\eta_\mu^2\sum_j h_j$ with the curvature sum obtained independently from $2\Delta DR/\|\delta\|^2$ ($1.9\times10^{-6}$ versus $1.7\times10^{-6}$ per epoch at $d=7$), and a random-walk simulation containing no RL at all --- $\eta_\mu$-sized isotropic kicks inside the box, scored by the measured curvatures --- reproduces both the trajectory and the plateau to within a factor of two.
The entropy term is thus an amplifier rather than the cause: it ratchets $\sigma$ up to $\sigma_{\max}$ and so multiplies the diffusion constant, but the mechanism operates without it.

\end{appendix}
 
\end{document}